\documentclass[final]{amsart}

\usepackage[T1]{fontenc}
\usepackage[utf8]{inputenc}
\usepackage{lmodern}
\usepackage{amsmath,amssymb}
\usepackage{mathrsfs}
\usepackage{amsfonts,amsthm,latexsym}
\usepackage[arrow, matrix, curve]{xy}
\usepackage{mathtools}
\usepackage{nicefrac}
\usepackage{tikz}
\usepackage{bracket-hack}
\usepackage{ming}
\usepackage[english]{babel}
\usepackage{graphicx}
\usepackage{color}
\usepackage{paralist}
\usepackage[final]{hyperref}
\usepackage[colorinlistoftodos,obeyDraft]{todonotes}
\usepackage[notcite,notref]{showkeys}
\usepackage{marginnote}
\usepackage{ifdraft}
\usepackage[a4paper,twoside]{geometry}

\mathtoolsset{showonlyrefs=true}
\mathtoolsset{centercolon}
\setdefaultenum{i)}{}{}{}
\usepackage{auto-eqlabel}

\theoremstyle{definition}
\newtheorem{defi}{Definition}[section]
\newtheorem{nota}[defi]{Notation}

\newtheorem{rema}[defi]{Remark}

\theoremstyle{plain}
\newtheorem{prop}[defi]{Proposition}
\newtheorem{theo}[defi]{Theorem}
\newtheorem{conj}[defi]{Conjecture}
\newtheorem{lemm}[defi]{Lemma}
\newtheorem{coro}[defi]{Corollary}

\newenvironment{subaligned}{\left\{\aligned}{\endaligned\right.}

\newcommand{\R}{{\mathcal R}}

\renewcommand{\L}{{\mathcal L}}
\newcommand{\D}{{\mathcal D}}
\newcommand{\E}{{\mathcal E}}
\newcommand{\K}{{\mathcal K}}

\renewcommand{\O}{{\mathcal O}}
\newcommand{\C}{{\mathcal C}}

\newcommand{\maths}[1]{{\ensuremath{\mathbb #1}}}
\newcommand{\RR}{\maths{R}}
\newcommand{\NN}{\maths{N}}

\renewcommand{\ggg}{{\ensuremath{\mathfrak g}}}

\newcommand{\ie}{i.\,e.\ }
\newcommand{\st}{such that }
\newcommand{\wrt}{with respect to}
\newcommand{\resp}{resp.\ }
\newcommand{\woutlog}{without loss of generality}
\newcommand{\iif}{if and only if}
\newcommand{\rhs}{right-hand side}
\newcommand{\lhs}{left-hand side}

\newcommand{\equ}{equation}
\newcommand{\equs}{equations}
\newcommand{\diffeo}{diffeomorphism}

\newcommand{\mf}{mani\-fold}

\newcommand{\riem}{Rie\-mann\-ian}

\newcommand{\liegr}{Lie group}

\newcommand{\liealg}{Lie al\-ge\-bra}
\newcommand{\liealgs}{Lie al\-ge\-bras}

\newcommand{\fundform}{second fun\-da\-men\-tal form}
\newcommand{\curv}{curvature}

\newcommand{\onorm}{or\-tho\-nor\-mal}
\newcommand{\ogon}{or\-tho\-go\-nal}

\newcommand{\levi}{Levi-Civita}

\newcommand{\mghd}{ma\-xi\-mal glo\-bally hy\-per\-bo\-lic de\-ve\-lop\-ment}
\newcommand{\desitter}{de~Sitter}

\newcommand{\eps}{\varepsilon}
\newcommand{\tr}{\operatorname{tr}}

\newcommand{\diag}{\operatorname{diag}}

\newcommand{\absval}[1]{\lvert #1 \rvert}

\newcommand{\scalprod}[2]{\langle #1,#2 \rangle}

\newcommand{\ad}{\operatorname{ad}}
\newcommand{\liebr}[2]{\left[ #1,#2 \right]}

\newcommand{\kasnerparabola}{\K}
\newcommand{\planewave}[1]{\L_{#1}}

\begin{document}

\title[SCC in Bianchi~B perfect fluids and vacuum]{Strong Cosmic Censorship in orthogonal Bianchi class~B perfect fluids and vacuum models}
\author[K. Radermacher]{Katharina Radermacher}
\address{Department of Mathematics, KTH Royal Institute of Technology, SE-10044 Stockholm, Sweden}
\curraddr{}
\email{kmra@kth.se}
\urladdr{}
\dedicatory{}
\date{\today}
\translator{}

\newcommand{\slimit}{\operatorname s}
\newcommand{\rfactornew}{\hat{\operatorname r}}
\newcommand{\hateps}{\hat{\eps}}
\newcommand{\maxdecayleft}{\Pi}
\newcommand{\maxdecayright}{\overline{\Pi}}
\newcommand{\integraltwominusq}{\R}
\newcommand{\functionforscc}{\mathsf f}
\newcommand{\betaN}{\beta_{N_+}}
\newcommand{\betaD}{\beta_\Delta}
\newcommand{\kretschmann}{R_{\alpha\beta\gamma\delta}R^{\alpha\beta\gamma\delta}}
\newcommand{\submfsplanewave}{L'}
\newcommand{\submfsplanewavezero}{\submfsplanewave'}
\newcommand{\submfsplanewavetheorem}{L}
\newcommand{\LHS}{LHS}
\newcommand{\RHS}{RHS}
\newcommand{\centreunstablespecialkasner}{M_{\kasnerparabola}}
\newcommand{\centreunstableplanewaveleft}{M_{\planewave {\bparamk},\operatorname{left}}}
\newcommand{\centreunstableplanewaveright}{M_{\planewave {\bparamk},\operatorname{right}}}
\newcommand{\centreunstableplanewavespecial}{M_{\planewave {\bparamk},\operatorname{special}}}
\newcommand{\centreunstableplanewavezero}{M_{\planewave {\bparamk},0}}
\newcommand{\taubone}{\operatorname{T1}}
\newcommand{\taubtwo}{\operatorname{T2}}
\newcommand{\bparamk}{\kappa}
\newcommand{\binvparam}{\eta}
\newcommand{\idmetric}{h}
\newcommand{\idfundform}{k}
\newcommand{\functiondecaylemma}{\zeta}
\newcommand{\functionprooftaubtwo}{\zeta}
\newcommand{\setconvplanewave}{\C}
\newcommand{\basic}{{\operatorname{basic}}}
\newcommand{\initialdata}{{\operatorname{i.d.}}}
\newcommand{\difftildesigma}{F_{\tilde\sigma}}
\newcommand{\difftilden}{F_{\tilde n}}
\newcommand{\diffdelta}{F_\delta}
\newcommand{\diffconstraint}{F_{\operatorname{constraint}}}
\newcommand{\lrs}{\operatorname{LRS}}
\newcommand{\flrw}{\operatorname{FLRW}}

\begin{abstract}
	The Strong Cosmic Censorship conjecture states that for generic initial data to Einstein's field equations, the \mghd\ is inextendible. We prove this conjecture in the class of \ogon\ Bianchi class~B perfect fluids and vacuum spacetimes, by showing that unboundedness of certain curvature invariants such as the Kretschmann scalar is a generic property. The only spacetimes where this scalar remains bounded exhibit local rotational symmetry or are of plane wave equilibrium type.

	We further investigate the qualitative behaviour of solutions towards the initial singularity. To this end, we work in the expansion-normalised variables introduced by Hewitt--Wainwright and show that a set of full measure, which is also a countable intersection of open and dense sets in the state space, yields convergence to a specific subarc of the Kasner parabola. We further give an explicit construction enabling the translation between these variables and geometric initial data to Einstein's equations.
\end{abstract}

\maketitle

\section{Introduction}

In her ground-breaking work~\cite{fouresbruhat_thmexistencecertainespde}, Choquet-Bruhat showed that Einstein's field equations can be formulated as an initial value problem, where the initial data is given on a spacelike Cauchy hypersurface. To given initial data, there is a \mghd\ which is unique up to isometry, as was shown by Choquet-Bruhat and Geroch in~\cite{choquetbruhatgeroch_globalaspectsgr}.
As understood in physics, this implies that this spacetime is uniquely determined by the initial data. The question arises whether one can find a larger development than this when dropping the requirement of global hyperbolicity, and whether this larger development is still unique in a meaningful sense. If two inequivalent developments existed, this would imply that determinism breaks, as the data on the initial Cauchy hypersurface does no longer suffice for determining which of the two developments the universe chooses.
The Strong Cosmic Censorship conjecture states that this does not happen, at least not generically, as it conjectures that there is no development larger than the \mghd.
\begin{conj}[Strong Cosmic Censorship]
\label{conj_scc}
	For generic initial data to Einstein's \equs, the \mghd\ (MGHD) is inextendible.
\end{conj}
As the currently preferred spacetime models exhibit singularities at early times, this conjecture is of additional interest: It claims that the occurrence of such a singularity is not a consequence of a particularly bad choice of initial data hypersurface, \ie of our point of view, but inherent to the four-dimensional spacetime.

In the form given in Conjecture~\ref{conj_scc}, the statement will not hold, one has to add some form of boundary conditions such as spatial compactness or asymptotic flatness, or homogeneity.
To this date, a full answer in the general case has not been given, and it is therefore of interest to consider subsets of initial data with additional properties.
The natural starting point for this discussion are the maximally symmetric \desitter, Anti-\desitter\ and Minkowski spacetimes, which in a first step generalise to the spatially homogeneous and isotropic Friedman-Lemaître-Robertson-Walker spacetimes.
These can be further generalised to the spatially homogeneous Bianchi spacetimes which admit a three-dimensional symmetry group. Some of the Bianchi spacetimes in turn appear as special cases of~$G_2$ cosmologies, where the symmetry group is only of dimension~two. In this paper, we discuss the case of Bianchi spacetimes. The expectation is that the results obtained here will in part translate to the more general~$G_2$ cosmologies and pave the way for the fully general setting.

The terms generic and inextendible in the conjecture have to be made precise in order to obtain a meaningful statement.
What we mean by  genericity will become clear below.
There are first results on inextendibility in the $C^0$-sense (for Schwarzschild,~\cite{sbierski_C0inextschwarzschild}, and for certain spherically symmetric spacetimes,~\cite{christodoulou_instabnakedsinggravcollapsescalfield} and~\cite{gallowayling_remarksC0inextend}), but here we consider inextendibility in the~$C^2$-sense, as we can then replace Conjecture~\ref{conj_scc} with the following, stronger conjecture (see~\cite[Conj.~17.2]{ringstrom_cauchyproblem}).
\begin{conj}[Curvature blow up]
\label{conj_blowup}
	For generic initial data to Einstein's \equs, the Kretsch\-mann scalar $R_{\alpha\beta\gamma\delta}R{}^{\alpha\beta\gamma\delta}$ is unbounded in the incomplete directions of causal geodesics in the \mghd\ (MGHD).
\end{conj}
\begin{rema}
\label{rema_blowupwithmatter}
	Whenever the spacetime is not vacuum, \ie the Ricci tensor does not vanish, it may be enough to determine whether the contraction of the Ricci \curv~$R_{\alpha\beta}R^{\alpha\beta}$ is unbounded in the incomplete directions of causal geodesics in the MGHD. Just like the Kretschmann scalar, this is a four-dimensional quantity invariant under isometries, and its unboundedness contradicts the existence of a~$C^2$ extension.
\end{rema}

\subsection{Bianchi perfect fluid spacetimes}
We restrict our discussion to spacetimes~$(M,g)$ satisfying Einstein's field equations
\begin{equation}
\label{eqn_einsteineqn}
	R_{\alpha\beta}-\frac12Sg_{\alpha\beta}=T_{\alpha\beta},
\end{equation}
with~$R_{\alpha\beta}$ and~$S$ the Ricci and scalar curvature of the spacetime
and~$T_{\alpha\beta}$ the stress-energy tensor of a perfect fluid
\begin{equation}
\label{eqn_stressenergyperfectfluid}
	T_{\alpha\beta}=\mu u_\alpha u_\beta+p(g_{\alpha\beta}+u_\alpha u_\beta),
\end{equation}
where~$u$ is a unit timelike vector field. We assume that the pressure~$p$ and the energy density~$\mu$ satisfy a linear equation of state
\begin{equation}
\label{eqn_lineareqnofstate}
	p=(\gamma-1)\mu
\end{equation}
for some constant~$\gamma$.
The case of vacuum is included as~$\mu=0$.
We note at this point that Greek indices $\alpha,\beta,\ldots$ range from~$0$ to~$3$, while lower case Latin ones $i,j,\ldots$ range from~$1$ to~$3$.

We consider
Bianchi perfect fluids, which are perfect fluid spacetimes such that additionally
there is a three-dimensional \liegr~$G$ acting on the spacetime~$(M,g)$
\begin{equation}
	G\times M\rightarrow M.
\end{equation}
The metric~$g$ is invariant under this action, and the unit timelike vector field~$u$ appearing in the stress-energy tensor~\eqref{eqn_stressenergyperfectfluid} of the perfect fluid is \ogon\ to the orbits of this action. This defines non-tilted perfect fluids.
The name Bianchi spacetimes stems from the classification of three-dimensional \liegr s by Bianchi from~1903, see~\cite{krasinskibehrschueckingestabrookwahlquistellisjantzenkundt_bianchiclass} for a historical overview:
Unimodular \liegr s are called class~A, non-unimodular ones class~B, and depending on the form of the structure constants of the corresponding \liealgs, they can be further separated into the following types:
\begin{itemize}
	\item Bianchi class~A: types~I, II, VI$_0$, VII$_0$, VIII, IX;
	\item Bianchi class~B: types~IV, V, VI$_\binvparam$, VII$_\binvparam$, where~$\binvparam\in\RR$ is a parameter.
\end{itemize}
In this paper, we focus on Bianchi class~B models, but types~I and~II appear as boundary cases.
The models of class~A have been the subject of more detailed study in the past, and a number of results have been achieved,
including an affirmative answer to the Strong Cosmic Censorship conjecture. We refer to~\cite{ringstrom_cauchyproblem} for a detailed exposition, in particular Prop.~22.23 therein, and are going to relate our own findings to a number of other results further down.

In terms of initial data, we arrive at the following setting: The initial data \mf\ is a three-dimensional \liegr~$G$ with a metric and \fundform\ which are left-invariant, meaning that they are invariant under the action of~$G$ on itself via multiplication from the left. Further, a given constant represents the initial matter configuration.
\begin{defi}
\label{defi_generalinitialdatabianchib}
	Bianchi perfect fluid initial data consists of a \liegr~$G$, a left-invariant \riem\ metric~${\idmetric}$ on~$G$, a left-invariant symmetric covariant two-tensor~${\idfundform}$ on~$G$, and a constant $\mu_0\ge0$, satisfying the constraint equations
	\begin{align}
		\overline R-{\idfundform}_{ij}{\idfundform}^{ij}+(\tr_{\idmetric}{\idfundform})^2={}&2\mu_0,\\
		\overline\nabla_i\tr_{\idmetric}{\idfundform}-\overline\nabla^j{\idfundform}_{ij}={}&0.
	\end{align}
	In these equations, $\overline\nabla$ is the \levi\ connection of~${\idmetric}$, and~$\overline R$ the corresponding scalar curvature. Indices are lowered and raised by~${\idmetric}$.
\end{defi}
Here and in the following, we use the Einstein summation convention and sum over indices which occur twice, both as a sub- and as a superindex.

The symmetry of the metric and \fundform\ is preserved under Einstein's equations, and the resulting \mghd\ is isometric to the spatially homogeneous four-dimensional spacetime
\begin{equation}
\label{eqn_mghdform}
	(I\times G\,,\,g={-}dt^2+{}^t\overline g),
\end{equation}
with~$I$ an open interval and $\{{}^t\overline g\}_{t\in I}$ a family of left-invariant Riemannian metrics on~$G\cong\{t\}\times G$. We give an explicit construction of the \mghd\ in Section~\ref{section_equivalenceinitialdataexpansionnormalised}.

\subsection{Results for orthogonal perfect fluid initial data}
Let us now state the precise setting we discuss in this paper and give the main results we obtain regarding the Strong Cosmic Censorship conjecture.

We focus our attention on Bianchi perfect fluid initial data with a \liegr\ of Bianchi class~B. Type~I and~II appear as boundary cases. In particular, we exclude \liegr s of type~VIII and~IX. The \liealg~$\ggg$ associated to the initial data \liegr~$G$ then admits an Abelian subalgebra, see Lemma~\ref{lemm_abeliansubalgstructconstants}. In case of a \liegr\ of Bianchi class~B, this subalgebra can be characterised geometrically as the kernel of a certain one-form, see Lemma~\ref{lemm_abeliansubalginvardefi}. Using this geometric construction, the subalgebra is unique and denoted by~$\ggg_2$.

We introduce an \onorm\ basis $e_1,e_2,e_3$ of~$\ggg$ such that~$e_2,e_3$ span the Abelian subalgebra~$\ggg_2$ and~$e_1$ is \ogon\ to it \wrt\ the initial metric. By the previous argument, this can be done uniquely up to rotation in the~$e_2e_3$-plane and a choice of sign in~$e_1$, provided the \liegr\ is of class~B.

For Bianchi perfect fluid initial data with a \liegr\ of class~B, one realises that the momentum constraint gives an algebraic relation for the components of the initial symmetric two-tensor~$\idfundform$ with respect to this basis.
For all initial data apart from certain cases where the \liegr\ is of type~VI$_{{-}1/9}$, this relation implies that the off-diagonal components~$\idfundform_{12}$ and~$\idfundform_{13}$ of this tensor vanish, see Lemma~\ref{lemm_exceptionalbianchiasinitialdata}. In the remaining special cases, the sets of initial data admit an additional degree of freedom and their \mghd\ is a so-called `exceptional` Bianchi cosmological spacetime, denoted Bbii in~\cite{ellismaccallum_classofhomogcosmmodels}.
This term should be understood as 'having exceptional behaviour'.
We refer to Remark~\ref{rema_exceptionalityexplained} for more details on this terminology.

In this paper, we exclude initial data with a \liegr\ of Bianchi type~VI$_{{-}1/9}$ and thereby ensure that the resulting spacetime is a 'non-exceptional' or '\ogon ' Bianchi spacetime.
\begin{defi}
\label{defi_initialdatabianchib}
	Orthogonal perfect fluid Bianchi class~B initial data consists of a \liegr~$G$ of class~B other than type~VI$_{{-}1/9}$, a left-invariant \riem\ metric~${\idmetric}$ on~$G$, a left-invariant symmetric covariant two-tensor~${\idfundform}$ on~$G$, and a constant $\mu_0\ge0$, satisfying the constraint equations
	\begin{align}
		\overline R-{\idfundform}_{ij}{\idfundform}^{ij}+(\tr_{\idmetric}{\idfundform})^2={}&2\mu_0,\\
		\overline\nabla_i\tr_{\idmetric}{\idfundform}-\overline\nabla^j{\idfundform}_{ij}={}&0.
	\end{align}

	Orthogonal perfect fluid Bianchi type~I and~II initial data for Einstein's orthogonal perfect fluid \equ s is defined similarly, allowing any type~I or~II \liegr ~$G$.
\end{defi}
The reason for this restriction is the technique we apply in this paper. We wish to transform initial data sets into a specific set of variables and prove statements in this setting before translating them back. For \liegr s of type~VI$_{{-}1/9}$ we encounter difficulties as this transformation, described in detail in Section~\ref{section_equivalenceinitialdataexpansionnormalised}, can only be carried out for spacetimes which are `non-exceptional`.

Certain initial data sets with higher symmetry will be of importance in the discussion of the Strong Cosmic Censorship conjecture. They are defined using a suitably adapted basis of the \liealg~$\ggg$:
\begin{defi}[Locally rotationally symmetric initial data]
\label{defi_lrsinitialdata}
	Consider initial data $(G,{\idmetric},{\idfundform},\mu_0)$ as in Def.~\ref{defi_initialdatabianchib}, and denote
	by~$\ggg$ the corresponding \liealg\ with two-dimensional Abelian subalgebra~$\ggg_2$. Let
	$e_1,e_2,e_3$ an \onorm\ basis of~$\ggg$ such that $e_2,e_3$ span~$\ggg_2$.
	The initial data is said to be locally rotationally symmetric (LRS)
	if the basis can be chosen such that
	\begin{itemize}
		\item $e_2$ commutes with~$e_1$ and~$e_3$
		\begin{equation}
			\liebr{e_2}{e_1}=0=\liebr{e_2}{e_3},
		\end{equation}
		\item the commutator~$\liebr{e_1}{e_3}$ is a multiple of~$e_2$,
		\item the two-tensor~${\idfundform}_{ij}$ is diagonal, with ${\idfundform}_{11}={\idfundform}_{33}$.
	\end{itemize}
\end{defi}
\begin{rema}
	The notion locally rotationally symmetric in the previous definition stems from the fact that a rotation in the~$e_1e_3$-plane is a \liegr\ isomorphism and an isometry of the initial data.

	Considering a three-dimensional \liegr~$G$ of class~B with corresponding \liealg~$\ggg$, we find that the following holds: In case~$\liebr\ggg\ggg$ is two-dimensional, there is no locally rotationally symmetric initial data on this \liegr. If instead~$\liebr\ggg\ggg$ is one-dimensional, then the vector spanning this set defines a rotation axis contained in~$\ggg_2$ and leaving the \liealg\ invariant. Initial data on the given \liegr\ is locally rotationally symmetric \iif\ the two-tensor~$\idfundform$ is invariant under this rotation. For more details, we refer to Subsection~\ref{constr_highersymmetrysolutions}.
\end{rema}
\begin{defi}[Plane wave equilibrium initial data]
\label{defi_planewaveinitialdata}
	Consider initial data $(G,{\idmetric},{\idfundform},\mu_0)$ as in Def.~\ref{defi_initialdatabianchib}, and denote
	by~$\ggg$ the corresponding \liealg\ with two-dimensional Abelian subalgebra~$\ggg_2$. Let
	$e_1,e_2,e_3$ be an \onorm\ basis of~$\ggg$ such that $e_2,e_3$ span~$\ggg_2$, and denote by~$\gamma_{ij}^k$ the structure constants, \ie
	\begin{equation}
		\liebr{e_i}{e_j}=\gamma_{ij}^ke_k.
	\end{equation}
	The initial data is said to be of plane wave equilibrium type
	if the basis can be chosen such that
	\begin{equation}
		\gamma_{1A}^B+\gamma_{1B}^A={-}2{\idfundform}_{AB}.
	\end{equation}
\end{defi}
We note here that upper case Latin indices $A,B,\ldots$ range from~$2$ to~$3$, where we assume the frame elements~$e_2,e_3$ to span~$\ggg_2$.
\begin{defi}
\label{defi_lrsplanewavemghd}
	A spacetime which is the \mghd\ of locally rotationally symmetric initial data is called a locally rotationally symmetric spacetime.
	A spacetime which is the \mghd\ of initial data of plane wave equilibrium type is called a plane wave equilibrium spacetime.
\end{defi}
We show in Section~\ref{section_equivalenceinitialdataexpansionnormalised} how to construct the \mghd~\eqref{eqn_mghdform} to given initial data as in Def.~\ref{defi_initialdatabianchib}.

To answer the Strong Cosmic Censorship conjecture, we determine whether geometric quantities invariant under isometries of this spacetime remain bounded in the incomplete directions of causal geodesics.
For vacuum models we find that a bounded Kretschmann scalar~$\kretschmann$ corresponds to a spacetime with local rotational symmetry or of plane wave equilibrium type.
Whenever matter is present and $\gamma\not=0$, we do not have to compute the full Kretschmann scalar, but it is enough to determine whether the contraction of the Ricci \curv~$R_{\alpha\beta}R{}^{\alpha\beta}$ is unbounded. We show that this is the case for all causal geodesics, there are no exceptions.
The arguments for the matter case work similarly as in the Bianchi~A case, which was discussed in~\cite{ringstrom_asymptbianchiaspacetimes}.
The full statement about curvature blow-up in the incomplete directions of causal geodesics is collected in the following two theorems.
\begin{theo}
\label{theo_curvblowupinitialdatamatter}
	Consider orthogonal perfect fluid Bianchi class~B initial data~$(G,{\idmetric},{\idfundform},\mu_0)$ for matter~$\mu_0>0$. Let $(M,g,\mu)$ be the \mghd\ of the data, solving Einstein's equations for a perfect fluid~\eqref{eqn_stressenergyperfectfluid} with linear equation of state~\eqref{eqn_lineareqnofstate}, where~$\gamma>0$.
	Then the contraction of the Ricci tensor with itself $R_{\alpha\beta}R^{\alpha\beta}$
	is unbounded in the incomplete directions of causal geodesics.
\end{theo}
\begin{theo}
\label{theo_curvblowupinitialdatavacuum}
	Consider orthogonal perfect fluid Bianchi class~B initial data~$(G,{\idmetric},{\idfundform},\mu_0)$ for vacuum~$\mu_0=0$, which is neither locally rotationally symmetric of Bianchi type~VI$_{{-}1}$ nor of plane wave equilibrium type.
	Let $(M,g)$ be the \mghd\ of the data, solving Einstein's equations for vacuum~$T_{\alpha\beta}=0$.
	Then the Kretschmann scalar $R_{\alpha\beta\gamma\delta}R^{\alpha\beta\gamma\delta}$
	is unbounded in the incomplete directions of causal geodesics.
\end{theo}
\begin{rema}
	We recall that orthogonal perfect fluid Bianchi class~B initial data excludes \liegr s of Bianchi type~VI$_{{-}1/9}$, see Definition~\ref{defi_initialdatabianchib}. In case of such initial data, more precisely for such initial data which is 'exceptional', see Remark~\ref{rema_exceptionalityexplained}, our statements do not apply. In fact, their behaviour is expected to differ significantly: close to the initial singularity, the corresponding spacetimes are expected to show chaotic oscillatory behaviour, much the same as Bianchi type~VIII and~IX spacetime in class~A.
	
	Due to the additional complications one meets in these models, they have not been studied extensively to this date. An approach similar to the one we choose here is made in~\cite{hewitthorwoodwainwright_asymptdynamexceptbianchicosm}. In future works, one can hope that their results and conjectures can be used to find an answer to Strong Cosmic Censorship which also applies to the remaining Bianchi type we do not treat here.
\end{rema}

\begin{rema}
	Initial data such that both the Kretschmann scalar $R_{\alpha\beta\gamma\delta}R^{\alpha\beta\gamma\delta}$ and the contraction of the Ricci tensor with itself $R_{\alpha\beta}R^{\alpha\beta}$ remain bounded in the incomplete directions of causal geodesics in the \mghd\ are
	locally rotationally symmetric or of plane wave equilibrium type.
	In particular, such initial data has additional symmetry and can therefore be considered non-generic. As a consequence, Strong Cosmic Censorship holds in the class of orthogonal perfect fluid Bianchi class~B initial data.
	
	It is interesting to note here that we only in the vacuum case have to exclude certain non-generic initial data sets. In the matter case, each initial data set has a corresponding development which is inextendible.
	The presence of matter appears to simplify things, at least in terms of the Strong Cosmic Censorship conjecture.
\end{rema}
\begin{rema}
\label{rema_curvblowupcosmconst}
	In case~$\mu>0$ and~$\gamma=0$, the stress-energy tensor~\eqref{eqn_stressenergyperfectfluid} with linear equation of state~\eqref{eqn_lineareqnofstate} takes the form
	\begin{equation}
		p={-}\mu
	\end{equation}
	and can be interpreted as the stress-energy tensor of vacuum with a positive cosmological constant. In this case, we find a statement similar to the previous one: We can show that the \mghd~$(M,g)$ is of the form~\eqref{eqn_mghdform}, with $I=(t_-,t_+)$ an interval,
	and every timeslice~$\{t\}\times G$ has positive mean \curv\ \wrt\ the normal vector~$\partial_t$. Then the Kretschmann scalar $R_{\alpha\beta\gamma\delta}R^{\alpha\beta\gamma\delta}$ or the contraction of the Ricci tensor with itself $R_{\alpha\beta}R^{\alpha\beta}$
	is unbounded in the incomplete directions of causal geodesics, with the following possible exceptions:
	\begin{itemize}
		\item local rotationally symmetric initial data of Bianchi type~VI$_{{-}1}$
		\item initial data of Bianchi class~B, and in scale free variables the metric and \fundform\ of every slice $\{t\}\times G$ converge to initial data of plane wave equilibrium type,
		as $t\rightarrow t_-$.
	\end{itemize}
	Note that we did not make precise in what sense the initial data converges, and that for initial data on a \liegr\ of type other than Bianchi~VI$_{{-}1}$, we have not clearly stated any property which can be interpreted as non-generic apart from this convergence. Further, we have used the existence of a specific foliation in order to formulate the statements.
	Precise results are given and proven in scale free variables using a notion of convergence which we introduce in the following. The results, both for the two previous theorems and the special case of positive cosmological constant in vacuum, then follow directly from translating back into the current formulation.
\end{rema}

\subsection{Expansion-normalised variables and previous results}
Most of the work in this paper is carried out in a setting different from the one above. Instead of starting with initial data~$(G,{\idmetric},{\idfundform},\mu_0)$ to Einstein's equations and proving properties of the resulting \mghd, we adopt a different view-point. The information given by initial data $(G,{\idmetric},{\idfundform},\mu_0)$ can be translated into the form of expansion-normalised and dimensionless variables
\begin{equation}
	(\Sigma_+\,,\,\tilde\Sigma\,,\,\Delta\,,\,\tilde A\,,\,N_+),
\end{equation}
which have been introduced in~\cite{hewittwainwright_dynamicalsystemsapproachbianchiorthogonalB}. Einstein's field equations translate into an ordinary differential equation in these variables, defined on a compact subset of~$\RR^5$ and with a changed time coordinate. It is in these variables that we work and obtain the results on curvature blow-up. Every initial data set corresponds to a point in the compact subset of~$\RR^5$, and once one has determined the solution curve in expansion-normalised variables, the \mghd~\eqref{eqn_mghdform} which solves Einstein's field equations with correct initial data can be determined. The detailed construction is carried out in Section~\ref{section_equivalenceinitialdataexpansionnormalised}. Via this construction, our results in expansion-normalised variables can be carried over to the setting of geometric initial data.

The evolution of these expansion-normalised variables is defined in the following way: To start with, one fixes two parameters, ${\bparamk}\in\RR$ and~$\gamma\in[0,2]$. The evolution equations for the variables~$(\Sigma_+,\tilde\Sigma,\Delta,\tilde A,N_+)$ are
\begin{equation}
\label{eqns_evolutionbianchib}
\begin{subaligned}
	\Sigma_+'={}&(q-2)\Sigma_+-2\tilde N,\\
	\tilde\Sigma'={}&2(q-2)\tilde\Sigma-4\Sigma_+\tilde A-4\Delta N_+,\\
	\Delta'={}&2(q+\Sigma_+-1)\Delta+2(\tilde\Sigma-\tilde N)N_+,\\
	\tilde A'={}&2(q+2\Sigma_+)\tilde A,\\
	N_+'={}&(q+2\Sigma_+)N_++6\Delta,
\end{subaligned}
\end{equation}
where~$'$ denotes differentiation ${d}/{d\tau}$ with respect to the dimensionless time~$\tau$.
The deceleration parameter~$q$ appearing in the equations satisfies
\begin{equation}
\label{eqn_definitionqgeneral}
	q=\frac32(2-\gamma)(\Sigma_+^2+\tilde\Sigma)+\frac12(3\gamma-2)(1-\tilde A-\tilde N),
\end{equation}
and one has set
\begin{equation}
\label{eqn_definitiontilden}
	\tilde N=\frac13(N_+^2-{\bparamk}\tilde A).
\end{equation}
The evolution is constrained to the set of points in~$\RR^5$ satisfying
\begin{equation}
\label{eqn_constraintgeneralone}
	\tilde\Sigma\tilde N-\Delta^2-\Sigma_+^2\tilde A=0,
\end{equation}
and
\begin{equation}
\label{eqn_constraintgeneraltwo}
	\tilde\Sigma\ge0, \qquad \tilde A\ge0, \qquad\tilde N\ge0, \qquad \Sigma_+^2+\tilde\Sigma+\tilde A+\tilde N\le 1.
\end{equation}
One further introduces the density parameter
\begin{equation}
\label{eqn_omegageneral}
	\Omega=1-\Sigma_+^2-\tilde\Sigma-\tilde A-\tilde N,
\end{equation}
which satisfies~$\Omega\ge0$ and, as a consequence of the differential equations~\eqref{eqns_evolutionbianchib}, evolves according to
\begin{equation}
\label{eqn_evolutionomega}
	\Omega'=(2q-(3\gamma-2))\Omega.
\end{equation}
Note that the relations~\eqref{eqn_constraintgeneralone} and~\eqref{eqn_constraintgeneraltwo} are preserved by the evolution equations~\eqref{eqns_evolutionbianchib}, see Remark~\ref{rema_expnormvarconstraintspreserved}.

We now give a rough explanation as to how the expansion-normalised variables are related to the \mghd\ to given initial data~$(G,{\idmetric},{\idfundform},\mu_0)$ as in Def.~\ref{defi_initialdatabianchib}. For this, assume that we are given a spacetime as in~\eqref{eqn_mghdform}, together with a function~$\mu$, which solve Einstein's field equations~\eqref{eqn_einsteineqn} for a perfect fluid~\eqref{eqn_stressenergyperfectfluid} with linear equation of state~\eqref{eqn_lineareqnofstate}. Assume that~$\partial_t$, the vector field in the~$I$ direction which is unit timelike and \ogon\ to every~$\{t\}\times G$, coincides with the vector field~$u$ from the stress-energy tensor. Assume further that the induced metric and second fundamental form on~$\{0\}\times G$ coincide with~${\idmetric}$ and~${\idfundform}$ from the initial data, \ie $(I\times G,g,\mu)$ is a development of the data.
One obtains consistency between the evolution in expansion-normalised variables and the initial data approach if one chooses the parameter~$\gamma$ from the expansion-normalised variables to coincide with the same-named constant from the linear equation of state~\eqref{eqn_lineareqnofstate}. Further, if the \liegr~$G$ has a group parameter~${\binvparam}$, then the parameter~${\bparamk}$ from the expansion-normalised variables is set to satisfy~${\bparamk}=1/{\binvparam}$. If the \liegr\ has no group parameter, then~${\bparamk}=0$.

For every dimensionless time~$\tau$, the point $(\Sigma_+,\tilde\Sigma,\Delta,\tilde A,N_+)(\tau)$ in expansion-normalised variables describes the geometric and dynamical properties of the timeslice $\{t(\tau)\}\times G$.
The variables~$\Delta,\tilde A,N_+$ carry information on the three-dimensional metric of this timeslice, while~$\Sigma_+,\tilde\Sigma$ describe the shear of the vector field~$\partial_t$.
A perfect fluid spacetime is described by a curve in~$\RR^5$ solving the constrained evolution equations~\eqref{eqns_evolutionbianchib}--\eqref{eqn_constraintgeneraltwo}, and we sometimes call such a solution an orbit.
Transforming the momentum constraint from the point of view of geometric initial data into the expansion-normalised variables yields the constraint equation~\eqref{eqn_constraintgeneralone}. Similarly, the Hamiltonian constraint justifies the definition of~$\Omega$ in~\eqref{eqn_omegageneral}, which is the expansion-normalised version of the energy density~$\mu$.

Statements about geometric initial data can be transformed into statements about specific points in the variables~$(\Sigma_+,\tilde\Sigma,\Delta,\tilde A,N_+)$. Conversely, we show in Section~\ref{section_equivalenceinitialdataexpansionnormalised} how one can, given the solution in expansion-normalised variables through this point, construct the corresponding \mghd.
Hence, statements about spacetimes solving Einstein's perfect fluid equation on the one hand and orbits solving the evolution equations in expansion-normalised variables~\eqref{eqns_evolutionbianchib}--\eqref{eqn_evolutionomega} on the other hand carry equivalent information.
It is for this reason that we can carry out our analysis wholly in the setting of expansion-normalised variables, and only at the end translate the results back to the point of view of geometric initial data. The setting of expansion-normalised variables is a rather approachable one, as we discuss a polynomial ordinary differential equation on a compact subset of~$\RR^5$. For instance, this allows the use of techniques from the theory of dynamical systems.

\smallskip

The system of differential equations~\eqref{eqns_evolutionbianchib}-\eqref{eqn_definitiontilden} with constraints~\eqref{eqn_constraintgeneralone}-\eqref{eqn_constraintgeneraltwo} was first introduced and studied in~\cite{hewittwainwright_dynamicalsystemsapproachbianchiorthogonalB}. In this reference, Hewitt and Wainwright identify invariant subsets representing the different Bianchi types, and we recall these subsets in Tables~\ref{table_bianchibsubsets} and~\ref{table_bianchiasubsets} below, together with several other invariant subsets representing models with additional symmetry, Table~\ref{table_bianchihighersymmetry}. Additionally, \cite{hewittwainwright_dynamicalsystemsapproachbianchiorthogonalB} discuss several sets of equilibrium points of the evolution equations. Of importance to our work are the Kasner parabola with the two Taub points and the plane wave equilibrium points, as convergence to points in these sets corresponds to the incomplete direction of causal geodesics in the \mghd, which is what we need to investigate in light of the Strong Cosmic Censorship conjecture.
\begin{defi}
\label{defi_kasnerparabola}
	The Kasner parabola is the subset~$\kasnerparabola$ defined by
	\begin{equation}
		\Sigma_+^2+\tilde\Sigma=1, \qquad \Delta=\tilde A=N_+=0.
	\end{equation}
\end{defi}
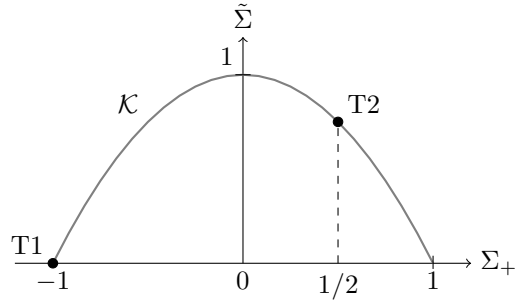
\begin{figure}[ht]
	\begin{tikzpicture}[xscale=2.5,yscale=2.5]
	\draw [->] (-1.2,0) -- (1.2,0);
	\draw [->] (0,0) -- (0,1.2);
	\draw[gray, thick, domain=-1:1] plot (\x, {(\x+1)*(1-\x)});
	\draw[fill] (-1,0) circle [radius=0.025];
	\node [below] at (0,0) {$0$};
	\draw (1,-0.03) -- (1,0.04);
	\node [below] at (1,0) {$1$};
	\node [below] at (-1,0) {${-}1$};
	\draw (-0.04,1) -- (0.04,1);
	\node [above left] at (0,1) {$1$};
	\node [right] at (1.2,0) {$\Sigma_+$};
	\node [above] at (0,1.2) {$\tilde\Sigma$};
	\node [above left] at (-1,0) {$\taubone$};
	\draw [dashed] (0.5,0) -- (0.5,0.75);
	\node [below] at (0.5,0) {$1/2$};
	\draw[fill] (0.5,0.75) circle [radius=0.025];
	\node [above right] at (0.5,0.75) {$\taubtwo$};
	\node [above left] at (-0.5,0.75) {$\kasnerparabola$};
	\end{tikzpicture}
	\caption{The Kasner parabola~$\kasnerparabola$, projected to the $\Sigma_+\tilde\Sigma$-plane. The special points Taub~1 and~2 satisfy~$\Sigma_+={-}1$ and $\Sigma_+=1/2$.}
	\label{figure_kasnerparabola}
\end{figure}
\begin{defi}
\label{defi_taubpoints}
	On the Kasner parabola~$\kasnerparabola$, the points Taub~1 and 2 are defined by
	\begin{itemize}
		\item Taub~1: $\taubone\coloneqq\left\{(\Sigma_+,\tilde\Sigma,\Delta,\tilde A,N_+)=({-}1,0,0,0,0)\right\}$,
		\item Taub~2: $\taubtwo\coloneqq\left\{(\Sigma_+,\tilde\Sigma,\Delta,\tilde A,N_+)=\left(1/2,3/4,0,0,0\right)\right\}$.
	\end{itemize}
\end{defi}
\begin{defi}
\label{defi_planewaveexpansionnorm}
	The plane wave equilibrium points are the elements of the set~$\planewave {\bparamk}$ defined by
	\begin{equation}
	\label{eqn_defiplanewave}
	\begin{split}
		&\Sigma_+>-1,\qquad
		\tilde\Sigma={-}\Sigma_+(1+\Sigma_+),\qquad
		\Delta=0,\qquad
		\tilde A=(1+\Sigma_+)^2,\\
		&N_+^2=(1+\Sigma_+)({\bparamk}(1+\Sigma_+)-3\Sigma_+).
	\end{split}
	\end{equation}
\end{defi}
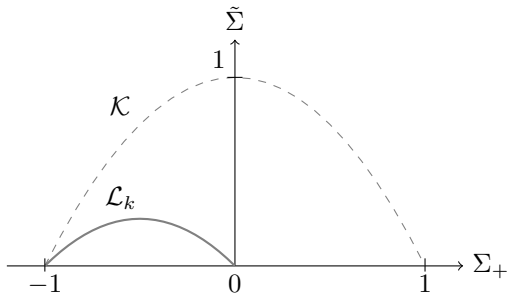
\begin{figure}[ht]
	\begin{tikzpicture}[xscale=2.5,yscale=2.5]
	\draw [->] (-1.2,0) -- (1.2,0);
	\draw [->] (0,0) -- (0,1.2);
	\draw[gray, dashed, domain=-1:1] plot (\x, {(\x+1)*(1-\x)});
	\draw[gray, thick, domain=-1:0] plot (\x, {-\x*(1+\x)});
	\node [below] at (0,0) {$0$};
	\draw (1,-0.03) -- (1,0.04);
	\node [below] at (1,0) {$1$};
	\draw (-1,-0.03) -- (-1,0.04);
	\node [below] at (-1,0) {${-}1$};
	\draw (-0.04,1) -- (0.04,1);
	\node [above left] at (0,1) {$1$};
	\node [right] at (1.2,0) {$\Sigma_+$};
	\node [above] at (0,1.2) {$\tilde\Sigma$};
	\node [above] at (-0.6,0.25) {$\planewave k$};
	\node [above left] at (-0.5,0.75) {$\kasnerparabola$};
	\end{tikzpicture}
	\caption{The plane wave equilibrium points~$\planewave {\bparamk}$, projected to the $\Sigma_+\tilde\Sigma$-plane. For reference, the Kasner parabola~$\kasnerparabola$ is plotted as a dashed line.}
	\label{figure_planewave}
\end{figure}

Both the Kasner parabola~$\kasnerparabola$ and the plane wave equilibrium points~$\planewave k$ satisfy~$\Omega=0$, are contained in the set defined by the constraint equations~\eqref{eqn_constraintgeneralone}--\eqref{eqn_constraintgeneraltwo} and consist of equilibrium points, \ie the \rhs\ of the evolution equations~\eqref{eqns_evolutionbianchib} is zero.

Information about the local stability can be drawn from the linearised evolution equations in the extended five-dimensional space, by which one means the linear approximation of the evolution \equs ~\eqref{eqns_evolutionbianchib}--\eqref{eqn_definitiontilden}, without restricting to the constraint equations~\eqref{eqn_constraintgeneralone}--\eqref{eqn_constraintgeneraltwo}. We give the explicit form of this vector field for points on the Kasner parabola~$\kasnerparabola$ in Appendix~\ref{subsect_appendixlinearisedevolutionkasner}.
The eigenvalues of this vector field are
\begin{equation}
\label{eqn_eigenvalueskasner}
	0\qquad
	2(1+\Sigma_+\pm\sqrt{3(1-\Sigma_+^2)})\qquad
	4(1+\Sigma_+)\qquad
	3(2-\gamma).
\end{equation}
We note at this point that there appear to be typos in the eigenvalues in both~\cite[Sect.~4.4]{hewittwainwright_dynamicalsystemsapproachbianchiorthogonalB} and~\cite[Sect.~7.2.3]{wainwrightellis_dynamsystemsincosm}. We give the corrected eigenvalues and state the corresponding eigenvectors in Appendix~\ref{subsect_appendixlinearisedevolutionkasner}.
The signs of the individual eigenvalues in these two references however are given correctly, and with this information Hewitt--Wainwright are able to identify the points on the Kasner parabola to the right of Taub~2 (with $1/2<\Sigma_+\le1$) as local sources, and the points to the left of Taub~2 (with ${-}1<\Sigma_+<1/2$) as saddles. For the two Taub points, the linearisation of the evolution equations alone does not determine the local stability, as two, or even three in case of the point Taub~1, of the eigenvalues vanish in these two points.

Similarly, one considers the linearised evolution equations in the extended five-dimensional state space for plane wave equilibrium points. The eigenvalues of this vector field are
\begin{equation}
\label{eqn_eigenvaluesplanewave}
	0\qquad
	{-}4(1+\Sigma_+)\qquad
	{-}4\Sigma_+-(3\gamma-2)\qquad
	{-}2(1+\Sigma_+)\pm4iN_+,
\end{equation}
with~$N_+$ as in \equ ~\eqref{eqn_defiplanewave}. See Appendix~\ref{subsect_appendixlinearisedevolutionplanewave} for more details. Again, the number of positive, negative, and zero eigenvalues can be used to determine the local stability, and~\cite{hewittwainwright_dynamicalsystemsapproachbianchiorthogonalB} identify points with~$\Sigma_+>{-}(3\gamma-2)/4$ as local sources, and points with~$\Sigma_+<{-}(3\gamma-2)/4$ as saddles. In the point~$\Sigma_+={-}(3\gamma-2)/4$ two eigenvalues vanish, and the local stability cannot be determined this way.

Using dynamical systems method, Hewitt--Wainwright show that the Kasner parabola~$\kasnerparabola$ is of central importance for the asymptotic behaviour of solutions to the evolution equation as~$\tau\rightarrow{-}\infty$, as it contains the~$\alpha$-limit set of (non-constant, generic) solutions. We state and prove a refined version of~\cite[Prop.~5.1, Prop.~5.2]{hewittwainwright_dynamicalsystemsapproachbianchiorthogonalB} in Prop.~\ref{prop_alphalimitsets_vacuum_inflat}.

\subsection{New results: convergence behaviour in expansion-normalised variables, curvature blow-up}

In the present paper, we obtain more detailed results about the behaviour of solutions to the evolution equations~\eqref{eqns_evolutionbianchib}--\eqref{eqn_evolutionomega} as~$\tau\rightarrow{-}\infty$.
We refine a statement on the $\alpha$-limit sets of solutions by~\cite{hewittwainwright_dynamicalsystemsapproachbianchiorthogonalB}: all~$\alpha$-limit points are contained in the union of the Kasner parabola~$\kasnerparabola$, the plane wave equilibrium points~$\planewave {\bparamk}$ and the point $\{\Sigma_+=\tilde\Sigma=\Delta=\tilde A=N_+=0\}$, see Prop.~\ref{prop_alphalimitsets_vacuum_inflat}.

We further prove that every solution has one unique~$\alpha$-limit point, \ie solutions to the evolution equations converge, to a limit point contained the plane wave equilibrium points~$\planewave k$ or the Kasner parabola~$\kasnerparabola$, see Prop.~\ref{prop_convergencetoplanewave} and Prop.~\ref{prop_convergencetokasner}. Only constant solutions converge to the point $\{\Sigma_+=\tilde\Sigma=\Delta=\tilde A=N_+=0\}$, see Prop.~\ref{prop_alphalimitsets_vacuum_inflat}.

Different subsets of this set of limit points have very different qualitative properties:
While solutions converging to the plane wave equilibrium points~$\planewave {\bparamk}$, to the point~$\{\Sigma_+=\tilde\Sigma=\Delta=\tilde A=N_+=0\}$ and to the subarc of the Kasner parabola~$\kasnerparabola$ with~${-}1\le\Sigma_+\le1/2$ are contained in a `small` set, namely a countable union of $C^1$ sub\mf s of positive codimension (Prop.~\ref{prop_alphalimitsets_vacuum_inflat}, Prop.~\ref{prop_taubone}, Thm~\ref{theo_taubtwocharacterisationoforbits}, Thm~\ref{theo_leftoftaubtworesult}, Thm~\ref{theo_centreunstablegloballyplanewave}),
the remaining arc of the Kasner parabola contains the limit points of all remaining solutions.
Considering the set of all possible limit points as a whole, this yields the following statement:
\begin{theo}
\label{theo_fullmeasurelimitset}
	Assume either vacuum or inflationary matter, \ie either $\Omega=0$ or $\Omega>0$, $\gamma\in\left[0,2/3\right)$. Then the following holds for solutions to the evolution \equs ~\eqref{eqns_evolutionbianchib}--\eqref{eqn_evolutionomega}:
	\begin{itemize}
		\item Consider the sets describing Bianchi type~VI$_{\binvparam}$ or~VII$_{\binvparam}$, \ie the sets
		\begin{align}
			B(VI_{\binvparam}) ={}&\{\eqref{eqn_constraintgeneralone}-\eqref{eqn_constraintgeneraltwo}\text{ hold},\,{\bparamk}=\nicefrac1{\binvparam}<0,\, \tilde A>0\}, \\
			B^\pm(VII_{\binvparam}) ={}&\{\eqref{eqn_constraintgeneralone}-\eqref{eqn_constraintgeneraltwo}\text{ hold},\, {\bparamk}=\nicefrac1{\binvparam}>0,\,\tilde A>0,\,N_+>0\textit{ or }N_+<0\}.
		\end{align}
		Then the subset of points such that the corresponding solution converges to a point in~$\kasnerparabola\cap\{\Sigma_+>1/2\}$, as~$\tau\rightarrow{-}\infty$, is
		of full measure and a countable intersection of open and dense sets
		in B(VI$_{\binvparam}$) or B$^\pm$(VII$_{\binvparam}$), respectively.
		\item Consider the set describing Bianchi type~IV, \ie the set
		\begin{equation}
			B^\pm(IV) =\{\eqref{eqn_constraintgeneralone}-\eqref{eqn_constraintgeneraltwo}\text{ hold},\,{\bparamk}=0,\,\tilde A>0,\,N_+>0\textit{ or }N_+<0\}.
		\end{equation}
		Then the subset of points such that the corresponding solution converges to a point in~$\kasnerparabola\cap\{\Sigma_+>1/2\}$ or to the point~$\kasnerparabola\cap\{\Sigma_+=0\}$, as~$\tau\rightarrow{-}\infty$, is
		of full measure and a countable intersection of open and dense sets
		in B$^\pm$(IV).
		\item Consider the set describing Bianchi type~V, \ie the set
		\begin{equation}
			B(V)=\{\eqref{eqn_constraintgeneralone}-\eqref{eqn_constraintgeneraltwo}\text{ hold},\,{\bparamk}=0,\,\tilde A>0,\,\Sigma_+=\Delta=N_+=0\}.
		\end{equation}
		Then every non-constant solutions converges to the point~$\kasnerparabola\cap\{\Sigma_+=0\}$, as~$\tau\rightarrow{-}\infty$.
	\end{itemize}
\end{theo}
For Kasner points situated to the right of Taub~2, we further find that they govern the behaviour of all solutions starting close to them.
In fact, we show in Prop.~\ref{prop_rightoftaub2source} that for every element of the subarc~$\kasnerparabola\cap\{\Sigma_+>1/2\}$, points contained in a sufficiently small neighborhood also converge to this subarc as~$\tau\rightarrow{-}\infty$, to a limit point close to the original one.

For the subarc of the Kasner parabola~$\kasnerparabola$ to the left of the point Taub~2, \ie satisfying~${-}1<\Sigma_+<1/2$, we additionally find the following restriction, compare Prop.~\ref{prop_main_generalasymptoticproperties}.
\begin{prop}
\label{prop_specialrelationbparamkslimit}
	Let $\gamma\in[0,2)$ and consider a solution to \equs ~\eqref{eqns_evolutionbianchib}--\eqref{eqn_evolutionomega} converging to $(\slimit,1-\slimit^2,0,0,0)$ as~$\tau\rightarrow{-}\infty$. If $\slimit\in\left[{-}1,1/2\right]$, then
	\begin{equation}
		\tilde A(3\slimit^2+{\bparamk}(1-\slimit^2))=0
	\end{equation}
	along the whole orbit.
\end{prop}
As a consequence of this result, for a given parameter~$\bparamk$, a solution converging to a point on the Kasner parabola~$\kasnerparabola$ situated to the left of Taub~2 has to be either a Bianchi class~A ($\tilde A=0$) solution, or it is a Bianchi class~B ($\tilde A>0$) solution but can only converge to one of two specific points on~$\kasnerparabola$, namely one with~$\Sigma_+=\pm\sqrt{\bparamk/(\bparamk-3)}$. For this to hold, the parameter~$\bparamk$ has to be non-positive. In particular, no Bianchi type~VII$_\binvparam$ solution can converge to a Kasner point situated to the left of Taub~2, as in this Bianchi type one has~$\bparamk=1/\binvparam>0$.

This statement is very different in spirit than the property of a subarc being a local source or saddle and goes far beyond what can be obtained by considering only the local stability. Instead, we discuss the asymptotic convergence behaviour of the individual expansion-normalised variables one by one, taking into account the full non-linear properties of the evolution equations~\eqref{eqns_evolutionbianchib}--\eqref{eqn_definitiontilden} in combination with the exact form of the constraint equation~\eqref{eqn_constraintgeneralone}.

Using the same methods, we also obtain a detailed statement about the Taub points, which due to the multiple zero eigenvalues of the linearised evolution equations in these points cannot be treated by local stability at all. We prove that only the constant orbit converges to the point Taub~1, see Prop.~\ref{prop_taubone} and Prop.~\ref{prop_tauboneothermatter}. Solutions converging to the point Taub~2 have to be locally rotationally symmetric, see Definition~\ref{defi_lrsexpansionnorm} and Thm~\ref{theo_taubtwocharacterisationoforbits}, and these locally rotationally symmetric solutions appear as possible exceptions to an unbounded Kretschmann scalar below in Thm~\ref{theo_curvblowupexpansionnormalisedvacuum}.

For the plane wave equilibrium points with more than one vanishing eigenvalue, we make use of techniques from the theory of dynamical systems and show that solutions are contained in a countable union of sub\mf s of dimension at most two and of positive codimension, see Prop.~\ref{theo_centreunstablegloballyplanewave} and Remark~\ref{rema_planewaveconvergencebianchiseparated}.

In addition to the statements about asymptotic behaviour in the limit~$\tau\rightarrow{-}\infty$, we also obtain the following result on the late time behaviour.
\begin{rema}
	In case of inflationary matter, \ie $\Omega>0$, $\gamma\in\left[0,2/3\right)$, all solutions converge to the point $(0,0,0,0,0)$ as~$\tau\rightarrow{+}\infty$, see Prop.~\ref{prop_omegalimitsetinflationary}. Transforming this statement back to the setting of geometric initial data and its \mghd, this implies isotropisation at late times.
\end{rema}

The main objective of the present paper is to prove the Strong Cosmic Censorship conjecture in the setting of \ogon\ Bianchi class~B perfect fluid and vacuum spacetimes. This is done by showing that boundedness of the Kretschmann scalar~$R_{\alpha\beta\gamma\delta}R^{\alpha\beta\gamma\delta}$ or the contraction of the Ricci \curv\ with itself~$R_{\alpha\beta}R^{\alpha\beta}$ in the incomplete direction of causal geodesics is a non-generic property.
Making use of the transformation between geometric initial data with corresponding \mghd~\eqref{eqn_mghdform} on the one hand, and points in the state space of expansion-normalised variables with corresponding solutions to the evolution equations~\eqref{eqns_evolutionbianchib}--\eqref{eqn_constraintgeneraltwo} on the other hand, we can express these two curvature expressions in terms of the expansion-normalised variables. Both quantities then depend only on the initial mean \curv\ and energy density, the parameters~$\gamma$ and~$\bparamk$, the time~$\tau$ and the point in expansion-normalised variables, as we explain in the beginning of Section~\ref{section_proofmainthms}. The incomplete direction of causal geodesics corresponds to the limit~$\tau\rightarrow{-}\infty$, see Prop.~\ref{prop_mghdincompletedirections}.
The statements about curvature blow-up in expansion-normalised variables are as follows.
\begin{theo}
\label{theo_curvblowupexpansionnormalisedmatter}
	Consider a solution to \equs ~\eqref{eqns_evolutionbianchib}--\eqref{eqn_evolutionomega} with~$\Omega>0$ and~$\gamma>0$. Then the contraction of the Ricci tensor with itself $R_{\alpha\beta}R^{\alpha\beta}$ is unbounded as $\tau\rightarrow{-}\infty$.
\end{theo}
\begin{rema}
	We remark at this point that the statement in the matter case is obtained without knowledge on the detailed asymptotic behaviour of the individual variables. In fact, only the evolution equation~\eqref{eqn_evolutionomega} for~$\Omega$ is considered to prove this statement, combined with details on how the expansion-normalised variables are obtained from geometric initial data sets. 
	In contrast, the vacuum case treated in the next theorem is rather intricate, and necessitates a detailed discussion of all variables.
	
	It is interesting to see that in the vacuum case, there are certain exceptions to unboundedness of the curvature, while there are none in matter. In a way, one could therefore consider the matter case as the easy case, while it is the vacuum case where interesting---and ultimately more difficult---behaviour becomes visible.
	In Bianchi class~A models, a similar observation has been made, see~\cite{ringstrom_asymptbianchiaspacetimes} and~\cite{ringstrom_cauchyproblem}.
\end{rema}
\begin{theo}
\label{theo_curvblowupexpansionnormalisedvacuum}
	Consider a solution to \equs ~\eqref{eqns_evolutionbianchib}--\eqref{eqn_evolutionomega} with~$\Omega=0$ which is neither
	the constant solution in the point Taub~1, nor
	a locally rotationally symmetric Bianchi type~I, II or~VI$_{{-}1}$ solution, nor a plane wave equilibrium solution.
	Then the Kretschmann scalar $R_{\alpha\beta\gamma\delta}R^{\alpha\beta\gamma\delta}$
	is unbounded as $\tau\rightarrow{-}\infty$.
\end{theo}
\begin{rema}
	In the previous theorem, locally rotationally symmetric solutions are defined by being contained in the set
	\begin{equation}
				3\Sigma_+^2=\tilde\Sigma, \qquad
		\Sigma_+N_+=\Delta, \qquad
		({\bparamk}+1)\tilde A=0,
	\end{equation}
	see Definition~\ref{defi_lrsexpansionnorm}, and consult Tables~\ref{table_bianchibsubsets}, \ref{table_bianchiasubsets}, and~\ref{table_bianchihighersymmetry} for more details on the separation into the different Bianchi types.

	In vacuum~$\Omega=0$,
	a bounded Kretschmann scalar is possible only for solutions converging to one of the Taub points or to the plane wave equilibrium points~$\planewave k$. Only in case of Taub~2 can such a solution be non-constant, and it necessarily has to be locally rotationally symmetric.
\end{rema}

\begin{theo}
\label{theo_curvblowupexpansionnormalisedcosmconst}
	Consider a solution to \equs ~\eqref{eqns_evolutionbianchib}--\eqref{eqn_evolutionomega} with~$\Omega>0$ and~$\gamma=0$ which is neither	the constant solution
	in the point $\Sigma_+=\tilde\Sigma=\Delta=\tilde A=N_+=0$,
	nor a locally rotationally symmetric Bianchi type~I, II or~VI$_{{-}1}$ solution,
	nor a solution contained in the set
	\begin{align}
		\setconvplanewave\coloneqq
			&\{\text{solutions to equations~\eqref{eqns_evolutionbianchib}--\eqref{eqn_evolutionomega} with }\Omega>0,\,\gamma=0,\\
			&\qquad\text{ and converging to an element of }\planewave k\text{ as } \tau\rightarrow{-}\infty\}.
	\end{align}
	Then the Kretschmann scalar $R_{\alpha\beta\gamma\delta}R^{\alpha\beta\gamma\delta}$ is unbounded as $\tau\rightarrow{-}\infty$.

	Further, there is a countable family of~$C^1$ sub\mf s~$\{\submfsplanewavetheorem_m\}_{m\in\NN}$ of dimension at most~two such that
	\begin{equation}
		\setconvplanewave\subset\bigcup_m \submfsplanewavetheorem_m.
	\end{equation}
	For certain Bianchi types, or equivalently certain values of the parameter~$\bparamk$, the following additional restrictions hold:
	\begin{itemize}
		\item Bianchi type~VI$_{\binvparam}$, which implies~${\bparamk}=1/{\binvparam}<0$: Every solution in~$\setconvplanewave$ converges to an element of the plane wave equilibrium points~$\planewave {\bparamk}$ with~$\Sigma_+=\slimit$ satisfying
		\begin{equation}
			\slimit\le\frac{{\bparamk}}{3-{\bparamk}}<0.
		\end{equation}
		\item Bianchi type~$V$, which implies~$\bparamk=0$: Every solution in~$\setconvplanewave$ is contained in~$\Sigma_+=\tilde\Sigma=\Delta=N_+=0$, and~$\tilde A$ decreases monotonically from~$1$ to~$0$.
	\end{itemize}
\end{theo}
\begin{rema}
	In the state space in expansion-normalised variables, which is the subset of~$\RR^5$ given by the constraint equations~\ref{eqn_constraintgeneralone}--\eqref{eqn_constraintgeneraltwo}, the different Bianchi types are represented by different invariant subsets, see Table~\ref{table_bianchibsubsets}. In case of Bianchi type~VI$_\binvparam$, VII$_\binvparam$, or~IV, this is a subset of dimension four. Consequently, the sets~$\submfsplanewavetheorem_m$ defined in the previous theorem are of positive codimension.
	The invariant subset describing Bianchi type~V solutions is of dimension two, and by the additional restriction stated in the theorem the set~$\setconvplanewave$ is contained in a set of dimension one, which thus is also of positive codimension.
\end{rema}
Compared to the statements in terms of initial data to Einstein's equations, Thm~\ref{theo_curvblowupinitialdatamatter} and Thm~\ref{theo_curvblowupinitialdatavacuum}, there is a larger number of exceptions to unboundedness of either geometric scalar in Thm~\ref{theo_curvblowupexpansionnormalisedmatter}, Thm~\ref{theo_curvblowupexpansionnormalisedvacuum} and Thm~\ref{theo_curvblowupexpansionnormalisedcosmconst}. The reason is that the expansion-normalised formulation allows for every initial data to the evolution equations~\eqref{eqns_evolutionbianchib} which satisfies the constraint equations~\eqref{eqn_constraintgeneralone} and~\eqref{eqn_constraintgeneraltwo}. This in particular includes certain points in~$\RR^5$ which correspond to Bianchi class~A initial data. In Thm~\ref{theo_curvblowupinitialdatamatter} and Thm~\ref{theo_curvblowupinitialdatavacuum}, these are excluded by the assumption on the initial data set.

\subsection*{Acknowledgements}

The author is indebted to Hans Ringström who suggested the topic, for stimulating discussions, ongoing supervision, and his numerous comments on the article.
Thanks is owed to Christopher Nerz for proof-reading parts of an earlier version of this paper.
Further thanks goes to Jaap Eldering for his hints on existing literature via MathOverflow that lead to App~\ref{section_appendixdynamsystheo}, as well as Ingemar Bengtsson for his insightful comments concerning previous work on the \desitter\ spacetime.

The author would like to acknowledge the support of the Göran Gustafsson Foundation for Research in Natural Sciences and Medicine.

Part of this material is based upon work supported by the National Science Foundation under Grant No.~0932078~000, while the author was in residence at the Mathematical Sciences Research Institute in Berkeley, California, during
the winter semester of~2013.

\section{Structure of the paper}

As the scope of this article is rather large, touching a number of different areas and techniques, we wish to give the reader a guide on how to read it, depending on the interests and prerequisites. We also shortly present the main ideas of the proofs.

We have given, in \equs ~\eqref{eqns_evolutionbianchib}--\eqref{eqn_evolutionomega}, the evolution \equ s which correspond to Einstein's \equ s for \ogon\ perfect fluid Bianchi~B initial data, and stated our main theorems in the introduction.
In Sections~\ref{section_basicpropertiesexpnormevolution}-\ref{section_asymptoticsplanewave}, we discuss various properties of this evolution problem. We explain their content in more detail a bit further down.

In Section~\ref{section_equivalenceinitialdataexpansionnormalised}, we move our focus away from a detailed discussion of the evolution equations and towards the relation of these equations to spacetimes solving Einstein's equations. It is here that we
prove equivalence between the evolution problem in expansion-normalised variables and the problem of finding, to given geometric initial data, the \mghd.
We recall the classification of three-dimensional \liegr s and how it relates to \liegr s as initial data sets.
This leads, in and after Lemma~\ref{lemm_exceptionalbianchiasinitialdata}, to an explanation of why we excluded Bianchi~VI$_{{-}1/9}$ \liegr s in the definition of initial data.
We recall the construction of expansion-normalised variables for \ogon\ Bianchi class~B cosmological models, proposed by~\cite{hewittwainwright_dynamicalsystemsapproachbianchiorthogonalB}, and show how this can be used to construct a spacetime which is shown to by the maximal globally hyperbolic development.
To the author's knowledge, these constructions and proofs have not been given before.

In the final section, we give the proofs of the main theorems, which are stated in the introduction. Those which are formulated in expansion-normalised variables, Thm~\ref{theo_fullmeasurelimitset} Thm~\ref{theo_curvblowupexpansionnormalisedmatter}, Thm~\ref{theo_curvblowupexpansionnormalisedvacuum}, and Thm~\ref{theo_curvblowupexpansionnormalisedcosmconst} are treated first.
Apart from computations of a number of \curv\ quantities in expansion-normalised variables which we carry out there, most of the work has been done in the sections before, and the proofs of these theorems merely collect the necessary results from Sections~\ref{section_basicpropertiesexpnormevolution}-\ref{section_asymptoticsplanewave}.
In order to then prove the remaining main theorems, Thm~\ref{theo_curvblowupinitialdatamatter} and Thm~\ref{theo_curvblowupinitialdatavacuum}, we use the equivalence between the setting of geometric initial data to Einstein's equation and that of the evolution equations in expansion-normalised variables. With this information at hand, proving the remaining main theorem is equivalent to translating the results from Thm~\ref{theo_curvblowupexpansionnormalisedmatter} and Thm~\ref{theo_curvblowupexpansionnormalisedvacuum} back to the setting of geometric initial data.

\smallskip

Let us now give a more detailed description of the content of Sections~\ref{section_basicpropertiesexpnormevolution}-\ref{section_equivalenceinitialdataexpansionnormalised}, which is where the main part of the work is carried out. These sections all discuss more and more detailed properties of the evolution in expansion-normalised variables.
In the end, we wish to make a statement about Strong Cosmic Censorship in the $C^2$-sense, and one way of doing so is to determine for which initial data both the Kretschmann scalar $R_{\alpha\beta\gamma\delta}R^{\alpha\beta\gamma\delta}$ and the contraction of the Ricci tensor with itself $R_{\alpha\beta}R^{\alpha\beta}$ remain bounded in the incomplete direction of causal geodesics in the \mghd, see Conjecture~\ref{conj_scc}, Conjecture~\ref{conj_blowup} and Remark~\ref{rema_blowupwithmatter}.
For spacetimes corresponding to solutions to \equs ~\eqref{eqns_evolutionbianchib}--\eqref{eqn_evolutionomega}, this direction corresponds to $\tau\rightarrow{-}\infty$. Consequently, we have to understand the asymptotic behaviour of solutions as $\tau\rightarrow{-}\infty$, and determine boundedness of the Kretschmann scalar and the contraction of the Ricci tensor with itself along such solutions. If we can show that the set of solutions along which both quantities are bounded at negative times is suitably small and can be considered non-generic, we have proven Strong Cosmic Censorship for this class of initial data.

Section~\ref{section_basicpropertiesexpnormevolution} starts with a general discussion of basic properties of \equ s~\eqref{eqns_evolutionbianchib}--\eqref{eqn_evolutionomega}, such as compactness of the state space, the smoothness of the constraint surface and its division into different invariant subsets corresponding to the different Bianchi types. In Section~\ref{section_kasnerparabolaplanewave}, we examine the so-called $\alpha$-limit set, which is the set of points where solutions accumulate as $\tau\rightarrow{-}\infty$.
In the case of vacuum and inflationary matter, the $\alpha$-limit set is contained in the Kasner parabola~$\kasnerparabola$ and the plane wave equilibrium points~$\planewave {\bparamk}$, together with one additional point in the case of inflationary matter, see Prop.~\ref{prop_alphalimitsets_vacuum_inflat}.
This result about accumulation points is subsequently strengthened to convergence, as $\tau\rightarrow{-}\infty$, in Prop.~\ref{prop_convergencetoplanewave} and Prop.~\ref{prop_convergencetokasner}, where the latter requires additional assumptions.
In particular, one needs to argue that the only solution with the point Taub~1 as an $\alpha$-limit point is the constant orbit, which is the main result we achieve in Section~\ref{section_taubone}.
We additionally prove a statement about isotropisation at late times: In case of inflationary matter, all solutions converge to the same point as~$\tau\rightarrow{+}\infty$, see Prop.~\ref{prop_omegalimitsetinflationary}.

\smallskip

Before we continue with a summary of the results of Sections~\ref{section_asymptoticdecaykasner}-\ref{section_equivalenceinitialdataexpansionnormalised}, let us briefly describe the main techniques.
In the end, we wish to understand exactly which solutions or sets of solutions converge to specific points, for example to those points which are of interest for the Strong Cosmic Censorship conjecture. To do so, we determine exponential decay or convergence properties of the individual variables using the evolution \equ s with constraints~\eqref{eqns_evolutionbianchib}--\eqref{eqn_evolutionomega}, for example the variable~$\tilde A$:
For a solution converging to a point $(\slimit,1-\slimit^2,0,0,0)$, $\slimit\in[{-}1,1]$, on the Kasner parabola~$\kasnerparabola$, as $\tau\rightarrow{-}\infty$,
the evolution \equ
\begin{equation}
	\tilde A'=2(q+2\Sigma_+)\tilde A
\end{equation}
immediately implies that either $\tilde A\equiv0$ or $\tilde A=\O(e^{(4+4\slimit-\eps)\tau})$.
This is a consequence of Lemma~\ref{lemm_decaylemmageneral}, and in this particular case stated as Lemma~\ref{lemm_decaytildea}.
Further in our discussion, we can improve this exponential decay estimate to not only include the slowest, but also the second-slowest exponential term, see also Prop.~\ref{prop_leftoftaub2additionaldecay}.

Similar convergence properties are obtained for the remaining variables upon convergence to the Kasner parabola.
We find that the five variables split into two groups, with identical decay rates in each group.
The variables $\Sigma_+$, $\tilde\Sigma$ and $\tilde A$, converge to their respective limit value exponentially to order~$4+4\slimit-\eps$, while~$\Delta$ and~$N_+$ decay to zero exponentially to order~$2+2\slimit\pm2\sqrt{3(1-\slimit^2)}-\eps$, where the sign depends on the sign of $\Delta N_+$ at sufficiently negative times. In some cases, this sign is determined by where the limit point on the Kasner parabola is located \wrt\ the Taub point~2.
This splitting into two different decay rates leads to a tension which we exploit on several occasion: The constraint \equ~\eqref{eqn_constraintgeneralone} written in the form
\begin{equation}
	\tilde\Sigma N_+^2-3\Delta^2=(3\Sigma_+^2+{\bparamk}\tilde\Sigma)\tilde A
\end{equation}
separates the variables according to their decay rate. The fast decay of the \lhs\ can only be achieved if the other side vanishes altogether, which is what we use in Lemma~\ref{lemm_deltanplussamesign} to show the relation stated in Prop.~\ref{prop_specialrelationbparamkslimit} between the parameter~$\bparamk$ and the limit value~$\slimit$ in Bianchi~B solutions. We further use this splitting to obtain lower bounds on the decay rates or even exclude certain subsets of the state space altogether upon convergence to specific subarcs of the Kasner parabola, see the proofs of Lemmata~\ref{lemm_deltanplussamesign},~\ref{lemm_deltanplusoppositesign}, and~\ref{lemm_deltanplusrightoftaub2oppositesign}.
Once these decay and convergence rates are found, they are used to determine the subsets containing all solutions converging to the point Taub~2, or the subarcs of the Kasner parabola to the left and right of it, respectively.

\smallskip

The most important of the convergence statements achieved in Section~\ref{section_asymptoticdecaykasner} are the following: For all variables, the lowest order exponential terms upon convergence to a point on the Kasner parabola $(\slimit,1-\slimit^2,0,0,0)$ with $\slimit\in[{-}1,1/2]$ are obtained in Prop.~\ref{prop_main_generalasymptoticproperties}.
For specific cases, this statement is refined in Prop.~\ref{prop_leftoftaub2additionaldecay}, where we determine the lowest order exponential term more precisely and find the second-lowest, again for all variables. Convergence to a point on the Kasner parabola which is situated to the left of the point Taub~2 requires that $\Delta N_+\ge0$, and we have to treat the case of convergence to the other subarc of the Kasner parabola while $\Delta N_+<0$ separately. In Lemma~\ref{lemm_deltanplusrightoftaub2oppositesign}, we determine the lowest order exponential terms for all variables for this situation. This section is the foundation of most of the results obtained later.

In Section~\ref{section_taubtwo}, we determine exactly which solutions converge to the point Taub~2. We rely on the results from the previous section and find that all such solutions have a local rotational symmetry.
This result is directly linked to Strong Cosmic Censorship: The only non-constant solutions converging to a point on the Kasner parabola for which both the Kretschmann scalar $R_{\alpha\beta\gamma\delta}R^{\alpha\beta\gamma\delta}$ and the contraction of the Ricci tensor with itself $R_{\alpha\beta}R^{\alpha\beta}$ possibly remain bounded as $\tau\rightarrow{-}\infty$ are those which converge to the point Taub~2. As all such solutions are locally rotationally symmetric, convergence to the point Taub~2 and thereby boundedness of these two curvature invariants can be considered non-generic.

Sections~\ref{section_leftoftaubtwo} and~\ref{section_asymptoticsplanewave} treat the equivalent question for the arc of the Kasner parabola~$\kasnerparabola$ to the left of Taub~2 and the plane wave equilibrium points~$\planewave {\bparamk}$. We give qualitative results on the set of solutions with this convergence behaviour, using a theorem from dynamical systems theory which we recall in Appendix~\ref{theo_centremftheory}. The statements we prove using dynamical systems theory are Thm~\ref{theo_leftoftaubtworesult} and Thm~\ref{theo_centreunstablegloballyplanewave}, which build upon where the eigenvalues to the linearised evolution equations are located in the complex plane.

Convergence towards one of the remaining points on the Kasner parabola~$\kasnerparabola$, located to the right of the point Taub~2, is treated in Section~\ref{section_rightoftaubtwo}. In~\cite{hewittwainwright_dynamicalsystemsapproachbianchiorthogonalB}, these points have been identified is local sources, and we make this statement more precise. For such perfect fluids where the convergence behaviour at early times is sufficiently well understood, we show that given a solution converging to a Kasner point to the right of Taub~2, all solutions intersecting a sufficiently small neighborhood of its limit point converge to a Kasner point close to this limit point.

\smallskip

In Appendix~\ref{section_appendixbianchivariables}, we discuss the linear approximation of the evolution equations in the extended state space, which we investigate both on the Kasner parabola~$\kasnerparabola$ and on the plane wave equilibrium points~$\planewave {\bparamk}$.

In the second appendix, Appendix~\ref{section_appendixdynamsystheo}, we present a result from the theory of dynamical systems.
This statement is an important prerequisite for Sections~\ref{section_leftoftaubtwo} and~\ref{section_asymptoticsplanewave}, where we discuss the properties of orbits converging to the Kasner parabola~$\kasnerparabola$ to the left of the point Taub~2 and to the plane wave equilibrium points~$\planewave {\bparamk}$.

\section{Basic properties of the expansion-normalised evolution}
\label{section_basicpropertiesexpnormevolution}

In this section, we start discussing properties of solutions to the evolution \equ s with constraints in expansion-normalised variables, \equ s~\eqref{eqns_evolutionbianchib}--\eqref{eqn_evolutionomega}.
We discuss the range of the individual variables and describe several subsets of the state space which remain invariant under the evolution. These correspond to specific Bianchi types, matter models, or families of models with additional symmetry. For a classification of Bianchi \liegr s in terms of their structure constants, we refer to Subsection~\ref{subsect_appendixbianchiclassification}.

One sees from the evolution \equ\ of~$\Omega$, \equ ~\eqref{eqn_evolutionomega}, that the sets $\Omega=0$ and $\Omega>0$ are invariant, and in the latter case the behaviour of~$\Omega$ depends on the value of the constant~$\gamma$. One distinguishes between the following cases:
\begin{defi}
	A solution to \equs ~\eqref{eqns_evolutionbianchib}--\eqref{eqn_evolutionomega} is called
	\begin{itemize}
		\item a vacuum solution if $\Omega\equiv0$ along the orbit,
		\item an inflationary matter solution if $\Omega>0$ along the orbit and $\gamma\in\left[0,2/3\right)$,
		\item a stiff fluid solution if $\Omega>0$ along the orbit and $\gamma=2$.
	\end{itemize}
\end{defi}
In this paper we also cover the case $\gamma\in[2/3,2)$. Of particular interest are the values $\gamma=1$ and $\gamma=4/3$, which correspond to dust and radiation respectively.
Due to the definition of the variable~$\Omega$ as the expansion-normalised version of the energy density~$\mu$, see \equ ~\eqref{eqn_definitionomega} below, the previous definition coincides with the distinction between vacuum and matter via $\mu=0$ and $\mu>0$ \resp

In vacuum, one sees that the expression for~$q$, \equ ~\eqref{eqn_definitionqgeneral}, simplifies to
\begin{equation}
\label{eqn_definitionqvacuum}
	q=2(\Sigma_+^2+\tilde\Sigma)=2(1-\tilde A-\tilde N),
\end{equation}
and in the general case the following expressions are of use
\begin{equation}
\label{eqn_definitionqmatter}
\begin{split}
	q={}&\frac32(2-\gamma)(\Sigma_+^2+\tilde\Sigma)+\frac12(3\gamma-2)(1-\tilde A-\tilde N),\\
		={}&2(1-\tilde A-\tilde N)-\frac32(2-\gamma)\Omega,\\
		={}&2(\Sigma_+^2+\tilde\Sigma)+\frac12(3\gamma-2)\Omega.
\end{split}
\end{equation}
We note that the differential \equs\ with constraints~\eqref{eqns_evolutionbianchib}--\eqref{eqn_evolutionomega} as well as the additional expressions for~$q$, \equ s~\eqref{eqn_definitionqvacuum}--\eqref{eqn_definitionqmatter} are invariant under the symmetry
\begin{equation}
	(\Delta,N_+) \mapsto {-}(\Delta,N_+).
\end{equation}
\begin{rema}
\label{rema_expnormvarconstraintspreserved}
	The constraint equations~\eqref{eqn_constraintgeneralone}--\eqref{eqn_constraintgeneraltwo} are invariant under the evolution equations~\eqref{eqns_evolutionbianchib} with~$q$ and~$\tilde N$ as in~\eqref{eqn_definitionqgeneral} and~\eqref{eqn_definitiontilden}. For the first constraint, equation~\eqref{eqn_constraintgeneralone}, this follows from
	\begin{equation}
		(\tilde\Sigma\tilde N-\Delta^2-\Sigma_+^2\tilde A)'=4(q+\Sigma_+-1)(\tilde\Sigma\tilde N-\Delta^2-\Sigma_+^2\tilde A).
	\end{equation}
	The second and fourth relation in~\eqref{eqn_constraintgeneraltwo} are an immediate consquence of the evolution equations of~$\tilde A$ in~\eqref{eqns_evolutionbianchib} and of~$\Omega$ in~\eqref{eqn_evolutionomega}, which imply that~$\tilde A=0$ and~$\Omega=0$ are invariant sets.

	For the remaining two inequalities~$\tilde\Sigma\ge0$ and~$\tilde N\ge0$, notice that equality in both cases at some time implies~$\Delta=0=\Sigma_+\tilde A$ at that time by the constraint equation~\eqref{eqn_constraintgeneralone}.
	The evolution equations for~$\tilde\Sigma$ and~$\Delta$ in~\eqref{eqns_evolutionbianchib} together with those for~$\tilde N$ and~$\Sigma_+\tilde A$,
	\begin{align}
		\tilde N'={}&2(q+2\Sigma_+)\tilde N+4\Delta N_+,\\
		(\Sigma_+\tilde A)'={}&(3q+4\Sigma_+-2)\Sigma_+\tilde A-2\tilde A\tilde N,
	\end{align}
	then reveal that~$\tilde\Sigma\equiv0\equiv\tilde N$ (as well as~$\Delta\equiv0\equiv\Sigma_+\tilde A$) at all times. Further, $\tilde\Sigma$ and~$\tilde N$ cannot change sign at different times, as~$\tilde\Sigma\tilde N<0$ is excluded by the constraint equation~\eqref{eqn_constraintgeneralone}. Hence, the inequalities~$\tilde\Sigma\ge0$ and~$\tilde N\ge 0$ are preserved.
\end{rema}

\begin{rema}
\label{rema_statespacecompact}
	The choice of expansion-normalised variables $(\Sigma_+,\tilde\Sigma,\Delta,\tilde A,N_+)$ has two advantages. Firstly, normalisation yields a compact state space~\eqref{eqn_constraintgeneralone}--\eqref{eqn_omegageneral}, which facilitates the discussion of dynamical properties. In fact, \equs ~\eqref{eqn_constraintgeneraltwo},~\eqref{eqn_omegageneral} and~\eqref{eqn_definitionqmatter} in combination with
	$\gamma\in\left[0,2\right]$ yield
	\begin{equation}
	\label{eqn_constraintgeneralthree}
		\Sigma_+\in[{-}1,1],\quad
		\tilde\Sigma\in[0,1],\quad
		\tilde A\in[0,1],\quad
		\tilde N\in[0,1],\quad
		\Omega\in[0,1],\quad
		q\in[{-}1,2].
	\end{equation}
	The second advantage is that the different Bianchi~B types are represented by invariant subsets of the same state space and can thus be discussed simultaneously.
	\begin{table}[htp]
		\centering
		\caption{Bianchi~B invariant sets.}
	\label{table_bianchibsubsets}
		\begin{tabular}{ll}
			Notation & Restrictions \\
			\hline
			\noalign{\vskip 1mm}
			B(VI$_{\binvparam}$) & ${\bparamk}=\nicefrac1{\binvparam}<0$, $\tilde A>0$ \\
			B$^\pm$(VII$_{\binvparam}$) & ${\bparamk}=\nicefrac1{\binvparam}>0$, $\tilde A>0$, $N_+>0$ or $N_+<0$ \\
			B$^\pm$(IV) & ${\bparamk}=0$, $\tilde A>0$, $N_+>0$ or $N_+<0$ \\
			B(V) & ${\bparamk}=0$, $\tilde A>0$, $\Sigma_+=\Delta=N_+=0$
		\end{tabular}
	\end{table}
Table~\ref{table_bianchibsubsets} lists the subsets and their names according to the Bianchi classification of Lie groups. We recall this classification in Subsection~\ref{subsect_appendixbianchiclassification}. This table was already given by~\cite{hewittwainwright_dynamicalsystemsapproachbianchiorthogonalB}, as were Tables~\ref{table_bianchiasubsets} and~\ref{table_bianchihighersymmetry}\footnote{Note that only one of the three LRS subsets in Table~\ref{table_bianchihighersymmetry} was mentioned in~\cite{hewittwainwright_dynamicalsystemsapproachbianchiorthogonalB}, namely LRS Bianchi~II. LRS Bianchi~VI$_{{-}1}$ first appeared in~\cite{wainwrightellis_dynamsystemsincosm}, it is also called LRS Bianchi~III in the literature.

	In the literature, the subset LRS Bianchi~II in Table~\ref{table_bianchihighersymmetry} is defined with $\Delta^2=\Sigma_+^2N_+^2$.
	Considering the evolution equations of~$3\Sigma_+^2-\tilde\Sigma$ and~$\Sigma_+N_+\pm\Delta$, we realise however that it is the case with negative sign~$\Sigma_+N_+-\Delta=0$ which is preserved, not the one with positive sign, see also equations~\eqref{eqn_derivativefunctionforscc}.
	}.

	In Table~\ref{table_bianchiasubsets}, two types of Bianchi~A models are given which appear as boundary sets of Bianchi~B invariant sets.
	\begin{table}[htp]
		\centering
		\caption{Bianchi~A invariant sets.}
	\label{table_bianchiasubsets}
		\begin{tabular}{ll}
			Notation & Restrictions \\
			\hline
			\noalign{\vskip 1mm}
			B$^\pm$(II) & $\tilde A=0$, $N_+>0$ or $N_+<0$ \\
			B(I) & $\tilde A=\Delta=N_+=0$
		\end{tabular}
	\end{table}
\end{rema}
Table~\ref{table_bianchihighersymmetry} lists several invariant subsets which describe Bianchi models with higher symmetry.
\begin{table}[htp]
	\centering
	\caption{Bianchi invariant sets with higher symmetry. (Abbreviations: FLRW= Friedman-Lemaître-Robertson-Walker, LRS= locally rotationally symmetric)}
\label{table_bianchihighersymmetry}
	\begin{tabular}{lll}
		Notation & Class of models & Restrictions \\
		\hline
		\noalign{\vskip 1mm}
		$S_{\lrs}$(I) & LRS Bianchi~I & $\Delta=\tilde A=N_+=0$, $3\Sigma_+^2=\tilde\Sigma$ \\
		$S_{\lrs}^\pm$(II) & LRS Bianchi~II & $\tilde A=0$, $3\Sigma_+^2=\tilde\Sigma$, $\Delta=\Sigma_+N_+$ \\
		$S_{\lrs}^\pm$(VI$_{{-}1}$)$=S_{\lrs}^\pm$(III) & LRS Bianchi~VI$_{{-}1}$ & ${\bparamk}={-}1$, $\tilde A>0$, $3\Sigma_+^2=\tilde\Sigma$, $\Delta=\Sigma_+N_+$ \\
		$S_0$(VI$_{\binvparam}$) & Bianchi~VI$_{\binvparam}$, $\tr n=0$ & $\Delta=N_+=0$, $3\Sigma_+^2+{\bparamk}\tilde\Sigma=0$, $\tilde A>0$ \\
		$S_{\flrw}$(V) & Bianchi~V FLRW & ${\bparamk}=0$, $\Sigma_+=\tilde\Sigma=\Delta=N_+=0$\\
		$S_{\flrw}^\pm$(VII$_{\binvparam}$) & Bianchi~VII$_{\binvparam}$ FLRW & $\Sigma_+=\tilde\Sigma=\Delta=0$, ${\bparamk}\tilde A=N_+^2>0$
	\end{tabular}
\end{table}

We want to give particular attention to the locally rotationally symmetric (LRS) models. The property of local rotational symmetry in expansion-normalised variables is defined as follows:
\begin{defi}[Local rotational symmetry (expansion-normalised)]
\label{defi_lrsexpansionnorm}
	A solution to \equs ~\eqref{eqns_evolutionbianchib}--\eqref{eqn_constraintgeneraltwo} is called locally rotationally symmetric (LRS) if
	\begin{equation}
	\label{eqn_lrsexpansionnorm}
		3\Sigma_+^2=\tilde\Sigma, \qquad
		\Sigma_+N_+=\Delta, \qquad
		({\bparamk}+1)\tilde A=0.
	\end{equation}
\end{defi}
Comparison with the three LRS subsets given in Table~\ref{table_bianchihighersymmetry} shows that the union of these subsets equals the union of all locally rotationally symmetric solutions.
We have now defined the notion of local rotational symmetry twice, in terms of expansion-normalised variables in Def.~\ref{defi_lrsexpansionnorm} and in terms of initial data to Einstein's \ogon\ perfect fluid equations in Def.~\ref{defi_lrsinitialdata}. In Subsection~\ref{constr_highersymmetrysolutions}, we clarify why these two definitions are equivalent under suitable transformation between the initial data setting and the expansion-normalised variables.

Considering the derivatives
\begin{align}
\label{eqn_derivativefunctionforscc}
	(3\Sigma_+^2-\tilde\Sigma)'={}&2(q-2)(3\Sigma_+^2-\tilde\Sigma)-4N_+(\Sigma_+N_+-\Delta)+4\tilde A\Sigma_+(1+{\bparamk}),\\
	(\Sigma_+N_+-\Delta)'={}&2(q-2-2\Sigma_++1)(\Sigma_+N_+-\Delta)+2N_+(3\Sigma_+^2-\tilde\Sigma),
\end{align}
one sees that the three LRS subsets are invariant under the evolution \equ s~\eqref{eqns_evolutionbianchib}. For the remaining sets given in Tables~\ref{table_bianchibsubsets}, \ref{table_bianchiasubsets} and~\ref{table_bianchihighersymmetry}, invariance follows by direct computation, see also Remark~\ref{rema_expnormvarconstraintspreserved} for Bianchi VII$_\binvparam$ FLRW.

\begin{rema}
\label{rema_singularconstraintequ}
	The evolution \equ s~\eqref{eqns_evolutionbianchib}--\eqref{eqn_definitiontilden}
	define a dynamical system in~$\RR^5$, and we will on several occasions apply dynamical systems methods to this. Whenever we restrict our attention to only these evolution \equ s, \ie without assuming the constraint \equ s~\eqref{eqn_constraintgeneralone}--\eqref{eqn_constraintgeneraltwo}, we call this the \textit{evolution \equ s in the extended state space}.

	The evolution in the physical, \ie constrained, state space, is obtained via restriction to the set defined by the constraint \equ s~\eqref{eqn_constraintgeneralone}--\eqref{eqn_constraintgeneraltwo}.
	The first equation is invariant under the evolution and describes a sub\mf\ as long as the gradient of its \lhs
	\begin{equation}
	\label{eqn_gradientconstraintequ}
		({-}2\Sigma_+\tilde A\,,\,\frac13(N_+^2-{\bparamk}\tilde A)\,,\,{-}2\Delta\,,\,{-}(\Sigma_+^2+\frac13{\bparamk}\tilde\Sigma)\,,\,\frac23\tilde\Sigma N_+)
	\end{equation}
	does not vanish, implying that the hypersurface is non-singular. The only exceptions to a non-vanishing gradient are:
	\begin{itemize}
		\item If ${\bparamk}>0$: $\Sigma_+=\tilde\Sigma=\Delta=0$, $N_+^2={\bparamk}\tilde A$. This defines an invariant set of dimension one.
		\item If ${\bparamk}=0$: $\Sigma_+=N_+=\Delta=0$. This defines an invariant set of dimension two.
		\item If ${\bparamk}<0$: $\Delta=\tilde A=N_+=0$, ${\bparamk}\tilde\Sigma+3\Sigma_+^2=0$. This defines an invariant set of dimension one.
	\end{itemize}
\end{rema}
\begin{rema}
	There is a similar set of expansion-normalised coordinates which is used to describe Bianchi class~A models, but does not apply to class~B models. It was introduced by Wainwright and Hsu in~\cite{wainwrighthsu_dynamicalsystemsapproachbianchiorthogonalA} and motivated the definition of the present coordinates. In the cases of Bianchi~I and II, models can be decribed in both sets of variables.

	Certain Bianchi perfect fluid spacetimes with a \liegr\ of type VI$_{{-}1/9}$ cannot be described by the evolution equations~\eqref{eqns_evolutionbianchib}--\eqref{eqn_evolutionomega}. This is the case for the so-called 'exceptional' Bianchi~B perfect fluids, a notion which we explain in Remark~\ref{rema_exceptionalityexplained}. Initial data sets in these spacetimes admit an additional degree of freedom compared to the 'non-exceptional' ones given in Definition~\ref{defi_initialdatabianchib}. In these 'exceptional' cases as well, it is possible to introduce a set of expansion-normalised variables, this has been done in~\cite{hewitthorwoodwainwright_asymptdynamexceptbianchicosm}. Due to the additional freedom, these spacetimes are described by evolution equations in six dimensions instead of five.
\end{rema}

\section{The Kasner parabola and the plane wave equilibrium points}
\label{section_kasnerparabolaplanewave}

The dynamical system for the expansion-normalised variables described by \equ s~\eqref{eqns_evolutionbianchib}--\eqref{eqn_evolutionomega} possesses a number of equilibrium points and sets, \ie points where the \rhs\ of the evolution equations~\eqref{eqns_evolutionbianchib} becomes zero. These equilibrium sets have been studied in~\cite{hewittwainwright_dynamicalsystemsapproachbianchiorthogonalB}, to whom we refer for more details. For our present discussion, the Kasner parabola~$\kasnerparabola$ and the plane wave equilibrium points~$\planewave k$ are of importance, see Definition~\ref{defi_kasnerparabola} and Definition~\ref{defi_planewaveexpansionnorm}

On the Kasner parabola~$\kasnerparabola$, one finds $\Omega=0$. Furthermore, in case $\gamma<2$, the set~$\kasnerparabola$ is characterised by $q=2$. A closer look at the evolution equations reveals that the Kasner parabola is a curve in the $\Sigma_+\tilde\Sigma$-plane consisting of individual equilibrium points.
Information about the local stability can be drawn from the linearised evolution equations in the extended five-dimensional space.
We give the explicit form of this vector field for points on the Kasner parabola~$\kasnerparabola$ in Appendix~\ref{subsect_appendixlinearisedevolutionkasner}.
The eigenvalues of this vector field are given in~\eqref{eqn_eigenvalueskasner}.
The number of positive, negative and zero eigenvalues corresponds to the quali\-tative behaviour of orbits close to the Kasner parabola. This is a result from dynamical systems theory, which we state in Appendix~\ref{section_appendixdynamsystheo}, but do not make use of in this section. The eigenvalues to the linearised evolution equations will appear further down as the exponential decay rates of certain linear combination of the variables.

We notice that for two special points on the Kasner parabola~$\kasnerparabola$ the number of zero eigenvalues is greater than one, namely the points Taub~1 and Taub~2, see Definition~\ref{defi_taubpoints}.
Especially the latter point will play a dominant role in our discussion, as it is the limit point of locally rotationally symmetric solutions, which constitute exceptions to the Strong Cosmic Censorship conjecture. To the right of the point Taub~2 ($1/2<\Sigma_+\le1$) all non-vanishing eigenvalues are positive, while to the left of Taub~2 ($-1<\Sigma_+<1/2$) exactly one of the four non-vanishing eigenvalues is negative.
This difference reflects a difference in qualitative behaviour which we explore in more detail in this paper.

The plane wave equilibrium points~$\planewave {\bparamk}$ form a curve consisting of individual equilibrium points with $\Omega=0$, as was the case for the Kasner parabola.
Using the function
\begin{equation}
\label{eqn_definitionfunctionz}
	Z\coloneqq(1+\Sigma_+)^2-\tilde A,
\end{equation}
which is discussed in~\cite{hewittwainwright_dynamicalsystemsapproachbianchiorthogonalB}, the set~$\planewave {\bparamk}$ can be characterised as follows: Direct computation shows that the constraint equation~\eqref{eqn_constraintgeneralone} is equivalent to
\begin{equation}
\label{eqn_functionzextended}
	(\tilde\Sigma-\tilde N)^2+4\Delta^2+2(\tilde\Sigma+\tilde N)\Omega+(\Omega-Z)^2=4Z(\Sigma_+^2+\frac12(\tilde\Sigma+\tilde N)),
\end{equation}
which due to the non-negativity of~$\tilde\Sigma$, $\tilde N$ and~$\Omega$, see~\eqref{eqn_constraintgeneraltwo}, means that the function~$Z$ is non-negative, and vanishes if and only if
\begin{equation}
\label{eqn_functionzzeros}
	\Omega=0,\qquad \Delta=0,\qquad -\Sigma_+(1+\Sigma_+)=\tilde\Sigma=\tilde N.
\end{equation}
Note that the identity $-\Sigma_+(1+\Sigma_+)=\tilde\Sigma$ is a reformulation of $\Omega=Z$, using the identities~\eqref{eqn_functionzzeros}.
Consequently, $Z=0$ characterises the plane wave equilibrium points together with the Kasner point Taub~1.
In case $\tilde A=0$, $Z=0$ characterises exactly the point Taub~1.
The derivative of the function~$Z$ is
\begin{equation}
	Z'={-}2(2-q)Z+3(2-\gamma)(1+\Sigma_+)\Omega
\end{equation}
due to the evolution \equs ~\eqref{eqns_evolutionbianchib}.
\begin{rema}
	Both the Kasner parabola~$\kasnerparabola$ and the plane wave equilibrium points~$\planewave {\bparamk}$ can be interpreted in terms of initial data to Einstein's \ogon\ perfect fluid equations. This is done in detail in Subsection~\ref{constr_highersymmetrysolutions}, where we explain how to establish a relation between initial data and solutions to Einstein's equations on the one hand and initial data and solutions to the evolution \equ s~\eqref{eqns_evolutionbianchib}--\eqref{eqn_evolutionomega} on the other. We find that
	\begin{itemize}
		\item The Kasner parabola~$\kasnerparabola$ corresponds to vacuum Bianchi type~I initial data.
		\item The set of plane wave equilibrium points~$\planewave {\bparamk}$ together with the point Taub~1 correspond to plane wave equilibrium initial data as in Def.~\ref{defi_planewaveinitialdata}.
		\item The point Taub~1 corresponds to initial data of Bianchi type~I which is of plane wave equilibrium type.
		\item The point Taub~2 corresponds to initial data of Bianchi type~I which is locally rotationally symmetric and the symmetric two-tensor~$\idfundform$ additionally satisfies~${\idfundform}_{11}>{\idfundform}_{22}$.
	\end{itemize}
\end{rema}

The Kasner parabola and the plane wave equilibrium points are of central importance when investigating the asymptotic behaviour of orbits as $\tau\rightarrow{-}\infty$. In~\cite{hewittwainwright_dynamicalsystemsapproachbianchiorthogonalB}, the $\alpha$-limit set of (non-constant, generic) orbits in vacuum and in inflationary matter ($\gamma\in\left[0,2/3\right)$) are determined. We state and prove a refined version of their Prop.~5.1 and Prop.~5.2.
\begin{prop}[Alpha-limit sets in vacuum and inflationary matter]
\label{prop_alphalimitsets_vacuum_inflat}
	Assume vacuum, \ie $\Omega=0$. If $\Gamma(\tau)=(\Sigma_+,\tilde\Sigma,\Delta,\tilde A,N_+)(\tau)$, $\tau\in\RR$, is a solution to \equs ~\eqref{eqns_evolutionbianchib}--\eqref{eqn_evolutionomega}, then the $\alpha$-limit set of~$\Gamma$ satisfies
	\begin{equation}
		\alpha(\Gamma)\subset\kasnerparabola\cup\planewave {\bparamk},
	\end{equation}
	which means that for every sequence $\tau_n\rightarrow{-}\infty$ all accumulation points of $\Gamma(\tau_n)$ are contained in $\kasnerparabola\cup\planewave {\bparamk}$. The only solutions with
	\begin{equation}
		\alpha(\Gamma)\cap\planewave {\bparamk}\not=\emptyset
	\end{equation}
	are constant solutions.

	Assume inflationary matter, \ie $\Omega>0$, $\gamma\in\left[0,2/3\right)$. If $\Gamma(\tau)=(\Sigma_+,\tilde\Sigma,\Delta,\tilde A,N_+)(\tau)$, $\tau\in\RR$, is a solution to \equs ~\eqref{eqns_evolutionbianchib}--\eqref{eqn_evolutionomega}, then the $\alpha$-limit set of~$\Gamma$ satisfies
	\begin{equation}
		\alpha(\Gamma)\subset\kasnerparabola\cup\planewave {\bparamk}\cup\{\Sigma_+=\tilde\Sigma=\Delta=\tilde A=N_+=0\}.
	\end{equation}
	There are no solutions whose $\alpha$-limit set intersects both~$\kasnerparabola\setminus\taubone$ and~$\planewave {\bparamk}$.
	The only solution with
	\begin{equation}
		\{\Sigma_+=\tilde\Sigma=\Delta=\tilde A=N_+=0\}\subset\alpha(\Gamma)
	\end{equation}
	is the constant solutions.
\end{prop}
Alongside with the proof of this proposition, we find a statement about asymptotic behaviour in the future time direction in case of inflationary matter: for~$\tau\rightarrow{+}\infty$, all non-constant solutions converge to one point. Note that this is the only time that we consider this time direction, all other statements treat the case~$\tau\rightarrow{-}\infty$.
\begin{prop}
\label{prop_omegalimitsetinflationary}
	Assume inflationary matter, \ie $\Omega>0$, $\gamma\in\left[0,2/3\right)$, and consider a solution to \equs ~\eqref{eqns_evolutionbianchib}--\eqref{eqn_evolutionomega}. Then
	\begin{equation}
		\lim_{\tau\rightarrow{+}\infty}(\Sigma_+,\tilde\Sigma,\Delta,\tilde A,N_+)(\tau)=(0,0,0,0,0).
	\end{equation}
\end{prop}
\begin{proof}[Proof of Prop.~\ref{prop_alphalimitsets_vacuum_inflat}]
	We start with the case of a vacuum solution $\Omega=0$. If~$q<2$ and $Z>0$, then the derivative of~$Z$ defined in~\eqref{eqn_definitionfunctionz} is negative, which means that the function~$Z$ is strictly monotone decreasing.
	The $\alpha$-limit set is therefore contained in the union of $\{q=2\}$ and $\{Z=0\}$. The first set is the Kasner parabola~$\kasnerparabola$, while the second describes the plane wave equilibrium points $\planewave {\bparamk}$ together with the point Taub~1.
	Suppose there is a solution with an $\alpha$-limit point in~$\planewave {\bparamk}\cup\taubone$. Then there is a sequence of times $\tau_k\rightarrow{-}\infty$ \st
	\begin{equation}
		Z(\tau_k)\le\frac1k.
	\end{equation}
	In combination with $Z\ge0$ and monotonicity, this implies that~$Z$ is vanishing identically along the whole orbit. The orbit is therefore contained in the zero set of~$Z$, which are the plane wave equilibrium points together with the point Taub~1. Hence the orbit is the constant orbit.

	For inflationary matter solutions $\Omega>0$, $\gamma\in[0,2/3)$, one reformulates the evolution \equ\ for~$\Omega$ into the form
	\begin{equation}
		\Omega'=(4(\Sigma_+^2+\tilde\Sigma)-(3\gamma-2)(1-\Omega))\Omega
	\end{equation}
	using~\eqref{eqn_definitionqmatter}. Due to the restrictions on the individual variables~\eqref{eqn_constraintgeneralthree}, this shows that~$\Omega$ is monotone increasing and implies that the $\alpha$-limit set is contained in the union of $\{\Omega=0\}$ and $\{\Omega=1\}$, as the latter is equivalent to the bracket vanishing. It follows immediately from monotonicity that if there is an~$\alpha$-limit point in $\{\Omega=1\}$, then the whole solution is contained in this set.
	The condition $\Omega=1$ characterises the point $\Sigma_+=\tilde\Sigma=\Delta=\tilde A=N_+=0$, which proves the last statement.

	For $0<\Omega<1$, one computes
	\begin{equation}
		(\frac{Z}{\Omega})'={-}(6-3\gamma)\frac{Z}{\Omega}+(6-3\gamma)(1+\Sigma_+),
	\end{equation}
	and concludes from
	\begin{equation}
		(e^{(6-3\gamma)\tau}\frac{Z}{\Omega})'=(6-3\gamma)(1+\Sigma_+)e^{(6-3\gamma)\tau}\ge0
	\end{equation}
	and the fact that both~$\Omega$ and~$Z$ are non-negative
	that the quantity $e^{(6-3\gamma)\tau}{Z}/{\Omega}$ decreases monotonically to some non-negative constant~$B_1$ as $\tau\rightarrow{-}\infty$. 
	This has to be understood in the sense that the expression decreases when going backwards in time, here and on similar occurences of this formulation further down.
	Monotone convergence to~$B_1$ implies
	\begin{equation}
		\frac{Z}{\Omega}=B_1e^{-(6-3\gamma)\tau}+\int_{{-}\infty}^\tau(6-3\gamma)(1+\Sigma_+)e^{(6-3\gamma)(s-\tau)}ds,
	\end{equation}
	where the integral is well-defined and bounded by~$2$ due to the bound on~$\Sigma_+$ from~\eqref{eqn_constraintgeneralthree}.

	There are two cases to consider. If~$B_1>0$, one can reformulate the previous equality into
	\begin{equation}
		e^{-(6-3\gamma)\tau}\Omega=\frac1{B_1}(Z-\Omega\int_{{-}\infty}^\tau(6-3\gamma)(1+\Sigma_+)e^{(6-3\gamma)(s-\tau)}ds),
	\end{equation}
	and sees that the \rhs\ is bounded. 
	The evolution \equ\ for the \lhs\ is
	\begin{equation}
		(e^{-(6-3\gamma)\tau}\Omega)'=2(q-2)e^{-(6-3\gamma)\tau}\Omega,
	\end{equation}
	which is of the form~$f'=g\cdot f$. For~$C^1$ functions~$f$, $g$ and~$f$ bounded as~$\tau\rightarrow{-}\infty$, integration of~$f'/f$ implies that~$g$ is integrable on~$({-}\infty,0)$. 
	In our case, one concludes that~$q-2$ is integrable on the interval $({-}\infty,0)$.
	One further knows from the evolution equations that the derivative of~$q-2$ is a polynomial in the expansion-normalised variables and consequently bounded, and thus can conclude that~$q\rightarrow2$ as $\tau\rightarrow{-}\infty$. The $\alpha$-limit set is therefore contained in $\{q=2\}$, which characterises the Kasner parabola.

	If $B_1=0$, then we have already argued above that
	\begin{equation}
		\frac Z\Omega\le2.
	\end{equation}
	As $\Omega\rightarrow0$ for $\tau\rightarrow{-}\infty$ due to monotonicity and the fact that the $\alpha$-limit set is contained in $\{\Omega=0\}$, this implies $Z\rightarrow0$. One concludes that the $\alpha$-limit set is contained in the set $\{Z=0\}$, which characterises the plane wave equilibrium points~$\planewave {\bparamk}$ together with the point Taub~1. This proves the first statement for inflationary matter models.

	The two cases for~$B_1$ are mutually exclusive, and the only point satisfying both $q=2$ and $Z=0$ is the point Taub~1, which is contained in the Kasner parabola~$\kasnerparabola$.
	This concludes the proof.
\end{proof}
\begin{proof}[Proof of Prop.~\ref{prop_omegalimitsetinflationary}]
	We have found in the previous proof that~$\Omega$ is monotone increasing with
	\begin{equation}
		\Omega'=(4(\Sigma_+^2+\tilde\Sigma)-(3\gamma-2)(1-\Omega))\Omega.
	\end{equation}
	Due to the restrictions on the individual variables~\eqref{eqn_constraintgeneralthree}, this shows that all solutions with~$\Omega>0$ have to satisfy $\Omega\rightarrow1$ as $\tau\rightarrow{+}\infty$, because the condition~$\Omega=1$ is equivalent to the bracket vanishing. Due to the definition of~$\Omega$ in equation~\eqref{eqn_omegageneral} and the restriction of the individual variables in~\eqref{eqn_constraintgeneralthree}, this concludes the proof.
\end{proof}

In the next proposition, we strengthen the result about the $\alpha$-limit set: In inflationary matter models, there cannot be more than one $\alpha$-limit point in the plane wave equilibrium points~$\planewave {\bparamk}$, \ie we find convergence.
\begin{prop}[Convergence to plane wave equilibrium points]
\label{prop_convergencetoplanewave}
	Assume inflationary matter, \ie $\Omega>0$, $\gamma\in\left[0,2/3\right)$, and consider a solution to \equs ~\eqref{eqns_evolutionbianchib}--\eqref{eqn_evolutionomega} whose $\alpha$-limit set is contained in $\planewave {\bparamk}\cup\taubone$. Then there is an $\slimit\in\left[{-}1,0\right]$ such that
	\begin{equation}
		\lim_{\tau\rightarrow{-}\infty}(\Sigma_+,\tilde\Sigma,\Delta,\tilde A,N_+)(\tau)=(\slimit,-\slimit(1+\slimit),0,(1+\slimit)^2,\nu_+(\slimit)),
	\end{equation}
	where $\nu_+^2(\slimit)=(1+\slimit)({\bparamk}(1+\slimit)-3\slimit)$.
\end{prop}
\begin{proof}
	If the point Taub~1 is the only $\alpha$-limit point, there is nothing to show. We assume therefore \woutlog\ that there is at least one $\alpha$-limit point contained in~$\planewave {\bparamk}$.
 	Due to the assumption on the $\alpha$-limit set, the solution satisfies $Z\rightarrow0$ as $\tau\rightarrow{-}\infty$. The zero set of~$Z$ is characterised by expressions~\eqref{eqn_functionzzeros}, 
 	and inserting these into the third expression for~$q$ in~\eqref{eqn_definitionqmatter}, we conclude that the solution satisfies
	\begin{equation}
	\label{eqn_convergenceqinflationary}
		q+2\Sigma_+\rightarrow0
	\end{equation}
	as $\tau\rightarrow{-}\infty$.
	All plane wave equilibrium points~$\planewave {\bparamk}$ satisfy ${-}1<\Sigma_+\le0$, which means that~$q$ converges to a non-negative value. In combination with the fact that~$3\gamma-2$ is a strictly negative constant, this implies that the factor $2q-(3\gamma-2)$ in the evolution \equ ~\eqref{eqn_evolutionomega} is strictly positive for $\tau\le\tau_0$ sufficiently negative. Consequently, $\Omega$ decays to zero exponentially as $\tau\rightarrow{-}\infty$.

	We have seen in the proof of Prop.~\ref{prop_alphalimitsets_vacuum_inflat}, that the solutions under consideration satisfy $B_1=0$, and consequently $0\le Z\le2\Omega$.
	Therefore, $Z$ decays to zero exponentially as well.
	Boundedness of the state space, see Remark~\ref{rema_statespacecompact}, therefore reveals that the \rhs\ of equation~\eqref{eqn_functionzextended} decays to zero exponentially, and therefore the same has to hold for all terms appearing on the \lhs. Using the identities~\eqref{eqn_functionzzeros} to rewrite~$\Omega-Z\rightarrow0$, we find 
	exponential decay for
	\begin{equation}
		\Sigma_+(1+\Sigma_+)+\tilde\Sigma\rightarrow0,\qquad
		\Delta\rightarrow0,
	\end{equation}
	and as a consequence also for $q+2\Sigma_+$, applying the third expression in equation~\eqref{eqn_definitionqmatter}. Inserting this into the evolution \equ ~\eqref{eqns_evolutionbianchib} shows the same decay for $\Sigma_+'$, therefore $\Sigma_+$ converges to some $\slimit\in\left[{-}1,0\right]$. The limiting values for the remaining variables follows from the definition of the plane wave equilibrium points~$\planewave {\bparamk}$, Def.~\ref{defi_planewaveexpansionnorm}.
\end{proof}
Consider now a non-constant solution to~\eqref{eqns_evolutionbianchib}--\eqref{eqn_evolutionomega} in vacuum or inflationary matter whose $\alpha$-limit set has a non-empty intersection with~$\kasnerparabola\setminus\taubone$ instead.
Due to Prop.~\ref{prop_alphalimitsets_vacuum_inflat}, this implies that the whole $\alpha$-limit set is contained in the Kasner parabola~$\kasnerparabola$.
Then
\begin{equation}
\label{eqn_convergenceratestepone}
	\Sigma_+^2+\tilde\Sigma\rightarrow1,\quad\Delta\rightarrow0,\quad \tilde A\rightarrow0,\quad N_+\rightarrow0,\quad\tilde N\rightarrow0,\quad \Omega\rightarrow0,\quad q\rightarrow 2,
\end{equation}
as $\tau\rightarrow{-}\infty$, as otherwise compactness of the state space, see Remark~\ref{rema_statespacecompact}, would yield an $\alpha$-limit point which does not lie on the Kasner parabola.
A convergence result similar to the previous statement is achieved further down in Prop.~\ref{prop_convergencetokasner}. Its proof needs some additional work.

For the discussion of more detailed properties of solutions close to the $\alpha$-limit points, we frequently make use of the following lemma which appears with slightly different notation as~\cite[Lemma~8]{ringstrom_curvblowupbianchiviiiandixvacuumspacetimes}.
\begin{lemm}
\label{lemm_decaylemmageneral}
	Consider a positive function $M:\RR\rightarrow(0,\infty)$ satisfying $M'=\functiondecaylemma M$, where $\functiondecaylemma:\RR\rightarrow\RR$ and $\functiondecaylemma(\tau)\rightarrow\iota$ as $\tau\rightarrow{-}\infty$. Then for all $\eps>0$ there is a $\tau_\eps>{-}\infty$ such that $\tau\le \tau_\eps$ implies
	\begin{equation}
		e^{(\iota+\eps)\tau}\le M(\tau)\le e^{(\iota-\eps)\tau}.
	\end{equation}
\end{lemm}
This lemma can be applied to the density parameter~$\Omega$ without any further assumptions: In vacuum $\Omega=0$ holds, while inflationary matter models satisfy $\Omega>0$ with $\gamma\in\left[0,2/3\right)$. If the $\alpha$-limit set is contained in the Kasner parabola~$\kasnerparabola$, one uses~\eqref{eqn_convergenceratestepone} to find convergence to zero at rate
\begin{equation}
\label{eqn_decayomega}
	e^{(6-3\gamma+\eps)\tau}\le \Omega(\tau) \le e^{(6-3\gamma-\eps)\tau},
\end{equation}
for $\tau\le\tau_\eps$ sufficiently negative.
\begin{nota}
	When we are only interested in the upper bound as $\tau\rightarrow{-}\infty$, we also make use of the big~$\O$ notation, \ie we write
	\begin{equation}
		\Omega=\O(e^{(6-3\gamma)\tau}),
	\end{equation}
	as $\tau\rightarrow{-}\infty$, if we want to say that there is a $\tau_0>{-}\infty$ and a~$C_\Omega>0$ such that $\tau\le \tau_0$ implies
	$
		\Omega(\tau) \le C_\Omega e^{(6-3\gamma)\tau}
	$.
	As we are in the present paper not interested in any other limit than $\tau\rightarrow{-}\infty$, we frequently 
	omit the range of~$\tau$.
\end{nota}
If the function~$\functiondecaylemma$ in the previous lemma converges exponentially, the statement can be improved as follows.
\begin{lemm}
\label{lemm_decaylemmageneralimproved}
	Consider a postitive function $M:\RR\rightarrow(0,\infty)$ satisfying $M'=\functiondecaylemma M$, where $\functiondecaylemma:\RR\rightarrow\RR$ and $\functiondecaylemma(\tau)=\iota+\O(e^{\xi\tau})$ as $\tau\rightarrow{-}\infty$, for some constants $\xi>0$, $\iota$. Then there are constants~$c_M,C_M>0$ and a~$\tau_0$ such that $\tau\le \tau_0$ implies
	\begin{equation}
		c_M e^{\iota\tau}\le M(\tau)\le C_M e^{\iota\tau}.
	\end{equation}
\end{lemm}
\begin{proof}
	Integration of
	\begin{equation}
		(\ln M)'=\frac{M'}{M}=\functiondecaylemma=\iota+\O(e^{\xi\tau})
	\end{equation}
	yields
	\begin{equation}
		M(\tau)=e^{\iota\tau}M(\tau_0)e^{-\iota\tau_0+\O(e^{\xi\tau_0})}.
	\end{equation}
	The last two factors are contained in some interval $\left[c_M,C_M\right]\subset(0,\infty)$ for~$\tau$ smaller than a fixed number~$\tau_0$.
\end{proof}

\section{Convergence to Taub~1}
\label{section_taubone}

In this section, we show that in vacuum and inflationary matter the only orbit with Taub~1 as an $\alpha$-limit point is the constant orbit. As a consequence, this special Kasner point can be neglected when we determine in more detail the asymptotic behaviour close to the Kasner parabola.
\begin{prop}
\label{prop_taubone}
	Assume either vacuum or inflationary matter, \ie either $\Omega=0$ or $\Omega>0$, $\gamma\in\left[0,2/3\right)$, and consider a solution to \equs ~\eqref{eqns_evolutionbianchib}--\eqref{eqn_evolutionomega} such that the point Taub~1
	\begin{equation}
		(\Sigma_+,\tilde\Sigma,\Delta,\tilde A,N_+)=({-}1,0,0,0,0)
	\end{equation}
	is contained in the $\alpha$-limit set. Then the solution is the constant orbit.
\end{prop}
\begin{prop}
\label{prop_tauboneothermatter}
	Assume $\Omega>0$ and $\gamma\in\left[2/3,2\right)$, and consider a solution to \equs ~\eqref{eqns_evolutionbianchib}--\eqref{eqn_evolutionomega} converging to the point Taub~1
	\begin{equation}
		(\Sigma_+,\tilde\Sigma,\Delta,\tilde A,N_+)=({-}1,0,0,0,0)
	\end{equation}
	as $\tau\rightarrow{-}\infty$. Then the solution is the constant orbit.
\end{prop}

\begin{rema}
\label{rema_tauboneexcludedinnonvacuum}
	As the point Taub~1 is contained in the vacuum set~$\Omega=0$, we conclude from these propositions that if~$\Omega>0$, then the point Taub~1 is not allowed as an $\alpha$-limit point.
\end{rema}
\begin{proof}[Proof of Prop.~\ref{prop_taubone}]
	The proof revolves around the function~$Z$ defined by equation~\eqref{eqn_definitionfunctionz}, and we start with the case of a vacuum solution $\Omega=0$. As the point Taub~1 is an $\alpha$-limit point but does not lie in~$\planewave {\bparamk}$, Prop.~\ref{prop_alphalimitsets_vacuum_inflat} yields that the $\alpha$-limit set is contained in the Kasner parabola~$\kasnerparabola$, which is characterised by $q=2$.
	In the point Taub~1, the function~$Z$ vanishes. This point being an~$\alpha$-limit point therefore implies that $Z\rightarrow0$ along a time sequence~$\tau_k\rightarrow{-}\infty$.
	In the proof of Prop.~\ref{prop_alphalimitsets_vacuum_inflat}, we have used the existence of such a time sequence together with monotonicity of the function~$Z$ to conclude that~$Z$ is vanishing constantly along the whole orbit. This argument applies to the present case, and we conclude that the orbit is contained in the zero set of~$Z$. The set~$Z=0$ equals the union of~$\planewave k$ and the point Taub~1 and consists solely of equilibrium points. As a consequence, the solution has to be the constant one in  the point Taub~1, as this is the only point which satisfies both~$Z=0$ and~$q=2$.

	In the case of matter $\Omega>0$, one first realises that due to its evolution \equ ~\eqref{eqns_evolutionbianchib}, if~$\tilde A$ vanishes at one time, then it vanishes along the whole orbit. One then reformulates the evolution of~$\Sigma_+$ using the definition of~$\tilde N$ from~\eqref{eqn_definitiontilden} together with \equ ~\eqref{eqn_definitionqmatter} for~$q$ to find
	\begin{equation}
		\Sigma_+'={-}2\Sigma_+\tilde A+\frac23(\Sigma_++1){\bparamk}\tilde A-\frac23(\Sigma_++1)N_+^2-\frac32(2-\gamma)\Omega.
	\end{equation}
	If $\tilde A=0$, then $\Sigma_+'<0$, as~$\Sigma_+$ is contained in the interval $[{-}1,1]$ by Remark~\ref{rema_statespacecompact}. Hence, $\Sigma_+$ is monotone decreasing, and an argument similar to the one for~$Z$ in the vacuum case applies:
	The assumption on the $\alpha$-limit set gives a sequence of times $\tau_k\rightarrow{-}\infty$ \st
	\begin{equation}
		\Sigma_+(\tau_k)\le{-}1+\frac1k.
	\end{equation}
	In combination with $\Sigma_+\ge{-}1$ and monotonicity, this shows that $\Sigma_+={-}1$ along the whole orbit.
	This in turn implies that~$Z$ vanishes along the orbit and concludes the proof for $\tilde A=0$, as then $\Omega=0$ due to~\eqref{eqn_functionzzeros}, a contradiction.

	Assume therefore $\tilde A>0$ and $\Omega>0$.
	We first prove convergence, \ie that Taub~1 is the unique $\alpha$-limit point, then show that the only orbit converging to this point is the constant one. From Prop.~\ref{prop_alphalimitsets_vacuum_inflat} we know that the $\alpha$-limit set is either contained in~$\kasnerparabola$ or in~$\planewave {\bparamk}\cup\taubone$, and together with Prop.~\ref{prop_convergencetoplanewave} this implies that we only have to show convergence for solutions where the $\alpha$-limit set is contained in the Kasner parabola~$\kasnerparabola$.

	As shown in~\eqref{eqn_decayomega}, $\Omega$ then decays as
	\begin{equation}
		e^{(6-3\gamma+\eps)\tau}\le \Omega(\tau) \le e^{(6-3\gamma-\eps)\tau},
	\end{equation}
	for $\tau\le\tau_\eps$ sufficiently small. Setting
	$\tilde\integraltwominusq\coloneqq\int_\tau^0(q-2)ds$, one finds that $\tilde\integraltwominusq$ is non-positive for~$\tau\le0$ due to Remark~\ref{rema_statespacecompact}. The derivative of this function is $\tilde\integraltwominusq'=2-q$. Hence
	\begin{equation}
		\frac{d}{d\tau}(e^{2\tilde\integraltwominusq}Z)=e^{2\tilde\integraltwominusq}3(2-\gamma)(1+\Sigma_+)\Omega=\O(e^{(6-3\gamma-\eps)\tau})
	\end{equation}
	for $\tau\le\tau_\eps$,
	since all factors apart from~$\Omega$ are at least bounded.
	Therefore, there exists a non-negative constant~$B_2$ such that
	\begin{equation}
		e^{2\tilde\integraltwominusq}Z=B_2+\O(e^{(6-3\gamma-\eps)\tau}).
	\end{equation}
	In case $B_2>0$, the function~$Z$ is bounded away from zero, which excludes Taub~1 as an $\alpha$-limit point.
	Consequently~$B_2=0$, which means
	\begin{equation}
		e^{2\tilde\integraltwominusq}Z=\O(e^{(6-3\gamma-\eps)\tau}).
	\end{equation}
	The convergence of~$q$ to~$2$, see~\eqref{eqn_convergenceratestepone}, implies $\absval{q-2}<\eps$ for sufficiently negative times, hence
	\begin{equation}
		Z=\O(e^{(6-3\gamma-2\eps)\tau}).
	\end{equation}
	As $Z\rightarrow0$, and additionally~$\tilde A\rightarrow0$ from the assumption that the $\alpha$-limit set is contained in the Kasner parabola~$\kasnerparabola$, one concludes that~$\Sigma_+\rightarrow{-}1$. The only point in~$\kasnerparabola$ with this property is the point Taub~1, which implies convergence to this point.

	It remains to exclude non-constant solutions with $\tilde A>0$ and $\Omega>0$ which converge to the point Taub~1 as $\tau\rightarrow{-}\infty$.
	Knowing that $\Sigma_+\rightarrow{-}1$ and $q\rightarrow2$, one can apply Lemma~\ref{lemm_decaylemmageneral} to the evolution \equ\ of~$\tilde A$, \equ~\eqref{eqns_evolutionbianchib}, to obtain
	\begin{equation}
	\label{eqn_tauboneslowdecaytildea}
		e^{\eps\tau}\le \tilde A(\tau) \le e^{-\eps\tau}
	\end{equation}
	for $\tau\le\tau_\eps$.
	Using the decay of the function~$Z$ together with boundedness of the state space, Remark~\ref{rema_statespacecompact}, in equation~\eqref{eqn_functionzextended} yields $\Delta^2=\O(e^{(6-3\gamma-\eps)\tau})$.
	With this, one computes
	\begin{equation}
		(\frac{\tilde N}{\tilde A})'=\frac{4\Delta N_+}{\tilde A}=\O(e^{(3-\frac32\gamma-2\eps)\tau}),
	\end{equation}
	as~$N_+$ is at least bounded, and finds
	\begin{equation}
		\frac{\tilde N}{\tilde A}=B_3+\O(e^{(3-\frac32\gamma-2\eps)\tau})
	\end{equation}
	for some constant~$B_3$. The constraint \equ ~\eqref{eqn_constraintgeneralone} then reads
	\begin{equation}
		\O(e^{(6-3\gamma-2\eps)\tau})=\Delta^2=\tilde\Sigma\tilde N-\Sigma_+^2\tilde A=(B_3\tilde\Sigma-\Sigma_+^2)\tilde A+\O(e^{(3-\frac32\gamma-3\eps)\tau}),
	\end{equation}
	and therefore
	\begin{equation}
		(B_3\tilde\Sigma-\Sigma_+^2)\tilde A=\O(e^{(3-\frac32\gamma-3\eps)\tau}).
	\end{equation}
	However, due to~\eqref{eqn_tauboneslowdecaytildea}~$\tilde A$ decays at most as $e^{\eps\tau}$, and the bracket on the \lhs\ converges to~$-1$, a contradiction. Thus $\tilde A>0$ is not possible, which concludes the proof.
\end{proof}
\begin{proof}[Proof of Prop.~\ref{prop_tauboneothermatter}]
	The proof is similar to the one of Prop.~\ref{prop_taubone} for inflationary matter, \ie $\Omega>0$ with~$\gamma\in[0,2/3)$. In that setting, we first had to show that the point Taub~1 is the unique~$\alpha$-limit point, but for the current statement, this holds by assumption.
	Showing that there are no non-constant solution converging to the point Taub~1 then hinged on the fact that
	\begin{equation}
		3-\frac32\gamma-3\eps>0.
	\end{equation}
	In the case~$\gamma\in[2/3,2)$, it is still possible to choose~$\eps>0$ sufficiently small that this holds. The argument at the end of the proof of~Prop.\ref{prop_taubone} excluding non-constant solutions then applies without any change.
\end{proof}

\section{Convergence properties and asymptotic decay towards the Kasner parabola}
\label{section_asymptoticdecaykasner}

This section is the longest and most technical in our discussion of the evolution equations~\eqref{eqns_evolutionbianchib}--\eqref{eqn_evolutionomega}, and it is here that we prove the main statements we build upon in the following.

We focus our attention on solutions whose $\alpha$-limit set is contained in the Kasner parabola~$\kasnerparabola$. For vacuum models, all non-constant solutions satisfy this property, while in the case of inflationary matter we have to additionally assume that the $\alpha$-limit set does not intersect the plane wave equilibrium points~$\planewave {\bparamk}$, see Prop.~\ref{prop_alphalimitsets_vacuum_inflat}.

In a first step we show that these non-constant solutions with $\alpha$-limit set in~$\kasnerparabola$ converge, \ie every such solution has a unique accumulation point. This is done in Prop.~\ref{prop_convergencetokasner}. An equivalent convergence result for inflationary matter solutions with an~$\alpha$-limit point in~$\planewave {\bparamk}$ has been obtained in Prop.~\ref{prop_convergencetoplanewave}, and we therefore find convergence for all inflationary matter and vacuum solutions.

The main aim of this section is to now obtain decay and convergence rates of the individual variables under the assumption of convergence to a limit point on the Kasner parabola~$\kasnerparabola$.
The behaviour we discover is exponential decay or convergence, and the different exponents coincide with specific eigenvalues to the linearised evolution equations in the extended state space, see~\eqref{eqn_eigenvalueskasner}.
We further find that the rates of convergence depend on where the limit point is situated relative to the point Taub~2. For solutions converging to the point Taub~2 or to a limit point to the left of this point, the lowest order exponential terms are determined in Prop.~\ref{prop_main_generalasymptoticproperties}. In certain situations, we can refine this statement to even include the second-lowest term, see Prop.~\ref{prop_leftoftaub2additionaldecay}. Solutions which converge to such a point on the Kasner parabola necessarily have to satisfy that~$\Delta N_+>0$ for sufficiently negative times or $\Delta\equiv0\equiv N_+$, as we conclude from Lemma~\ref{lemm_deltanplussigns} and Lemma~\ref{lemm_deltanplusoppositesign}. In case~$\Delta N_+<0$ for sufficiently negative times, we determine the exponential convergence rates in Lemma~\ref{lemm_deltanplusrightoftaub2oppositesign}.

We remark that even though we prove convergence only in the case of vacuum and inflationary matter, we then drop this restriction on the matter in all the following statements and only assume convergence to a limit point on the Kasner parabola.
The results we show hold
for all matter models apart from, for some statements, the stiff fluid case $\gamma=2$.
\begin{prop}[Convergence to the Kasner parabola]
\label{prop_convergencetokasner}
	Assume either vacuum or inflationary matter, \ie either $\Omega=0$ or $\Omega>0$, $\gamma\in\left[0,2/3\right)$, and consider a non-constant solution to \equs ~\eqref{eqns_evolutionbianchib}--\eqref{eqn_evolutionomega}. In the inflationary matter case, assume additionally that the $\alpha$-limit does not intersect the plane wave equilibrium points~$\planewave {\bparamk}$.
	Then there is an $\slimit\in\left({-}1,1\right]$ such that
	\begin{equation}
		\lim_{\tau\rightarrow{-}\infty}(\Sigma_+,\tilde\Sigma,\Delta,\tilde A,N_+)(\tau)=(\slimit,1-\slimit^2,0,0,0).
	\end{equation}
\end{prop}
We prove this statement below.
\begin{prop}
\label{prop_main_generalasymptoticproperties}
	Let $\gamma\in[0,2)$ and consider a solution to \equs ~\eqref{eqns_evolutionbianchib}--\eqref{eqn_evolutionomega} converging to $(\slimit,1-\slimit^2,0,0,0)$. If $\slimit\in\left[{-}1,1/2\right]$, then
	\begin{equation}
		\tilde A(3\slimit^2+{\bparamk}(1-\slimit^2))=0
	\end{equation}
	along the whole orbit, and
	\begin{align}
		\Sigma_+={}&\slimit+\O(e^{(\maxdecayleft-\eps)\tau}),\\
		\tilde\Sigma={}&1-\slimit^2+\O(e^{(\maxdecayleft-\eps)\tau}),\\
		\tilde N={}&\O(e^{(4+4\slimit-\eps)\tau}),\\
		q={}&2+\O(e^{(\maxdecayleft-\eps)\tau}),\\
	\end{align}
	as $\tau\rightarrow{-}\infty$, for every $\eps>0$. Here $\maxdecayleft\coloneqq\min(6-3\gamma,4+4\slimit)$ if $\Omega>0$, and $\maxdecayleft\coloneqq 4+4\slimit$ if $\Omega=0$. Furthermore, the following properties hold:
	\begin{itemize}
		\item If $\Delta$ and $N_+$ do not both vanish identically, then $\Delta N_+>0$ along the whole orbit and for all $\eps>0$ there is a $\tau_\eps>{-}\infty$ such that $\tau\le \tau_\eps$ implies
		\begin{align}
			e^{(2+2\slimit+2\sqrt{3(1-\slimit^2)}+\eps)\tau}\le {}&\absval{\Delta}\le e^{(2+2\slimit+2\sqrt{3(1-\slimit^2)}-\eps)\tau},\\
			e^{(2+2\slimit+2\sqrt{3(1-\slimit^2)}+\eps)\tau}\le {}&\absval{N_+}\le e^{(2+2\slimit+2\sqrt{3(1-\slimit^2)}-\eps)\tau}.
		\end{align}
		\item Either $\tilde A=0$ along the whole orbit, or for all $\eps>0$ there is a $\tau_\eps>{-}\infty$ such that $\tau\le \tau_\eps$ implies
		\begin{equation}
			e^{(4+4\slimit+\eps)\tau}\le \tilde A\le e^{(4+4\slimit-\eps)\tau}.
		\end{equation}
		\item Either $\Omega=0$ along the whole orbit (vacuum), or for all $\eps>0$ there is a $\tau_\eps>{-}\infty$ such that $\tau\le \tau_\eps$ implies
		\begin{equation}
			e^{(6-3\gamma+\eps)\tau}\le\Omega\le e^{(6-3\gamma-\eps)\tau}.
		\end{equation}
	\end{itemize}
\end{prop}
The proof is divided into several steps which will have additional individual use later. The arguments revolve around the constraint \equ ~\eqref{eqn_constraintgeneralone} written in the form
\begin{equation}
\label{eqn_constraintforproofs}
	\tilde\Sigma N_+^2-3\Delta^2=(3\Sigma_+^2+{\bparamk}\tilde\Sigma)\tilde A
\end{equation}
which is then used to determine the asymptotic decay properties of the individual variables.

\begin{lemm}
\label{lemm_monotonicitysigmaplus}
	Consider a solution to \equs ~\eqref{eqns_evolutionbianchib}--\eqref{eqn_evolutionomega}. If $\Sigma_+\in[0,1]$, then $\Sigma_+'\le0,$ \ie $\Sigma_+$ is monotonically decreasing.
\end{lemm}
\begin{proof}
	This follows from inspection of the evolution \equ ~\eqref{eqns_evolutionbianchib} for~$\Sigma_+'$, using the range of the variables given by the constraints~\eqref{eqn_constraintgeneraltwo} and Remark~\ref{rema_statespacecompact}.
\end{proof}
\begin{lemm}
\label{lemm_deltanplussigns}
	Consider a solution to \equs ~\eqref{eqns_evolutionbianchib}--\eqref{eqn_evolutionomega} whose $\alpha$-limit set is contained in~$\kasnerparabola$ and such that $\tilde\Sigma(\tau)>\delta>0$ for $\tau\le\tau_0$. Then one of the following statements holds:
	\begin{enumerate}
		\item $\Delta=N_+=0$ along the whole orbit,
		\item There is $\tau_1\in\RR$ \st $\Delta N_+(\tau)>0$ for all $\tau\le\tau_1$,
		\item There is $\tau_1\in\RR$ \st $\Delta N_+(\tau)<0$ for all $\tau\le\tau_1$.
	\end{enumerate}
\end{lemm}
Note that this statement does not require any assumption on the matter model but holds for all values of~$\gamma\in[0,2]$.
\begin{proof}
	The set $\Delta=0=N_+$ is invariant under the evolution \equs ~\eqref{eqns_evolutionbianchib}, which means that every orbit with $\Delta=0=N_+$ at one time will satisfy this property at all times.
	One can therefore assume that $\Delta^2+N_+^2>0$.
	The evolution \equs ~\eqref{eqns_evolutionbianchib} for~$\Delta$ and~$N_+$ yield
	\begin{equation}
		(\Delta N_+)'=6\Delta^2+(3q+4\Sigma_+-2)\Delta N_++2(\tilde\Sigma-\tilde N)N_+^2,
	\end{equation}
	and due to the convergence relation~\eqref{eqn_convergenceratestepone} and the assumption on~$\tilde\Sigma$, the coefficient to~$N_+^2$ is strictly positive for $\tau\le\tau_2$ sufficiently negative. Therefore, for every time $\tau_1<\tau_2$ where the product $\Delta N_+(\tau_1)$ becomes zero, the derivative $(\Delta N_+)'(\tau_1)$ is strictly positive, which means that $\Delta N_+$ changes sign from negative to positive. Consequently, the product $\Delta N_+$ can become zero at most once, which concludes the proof.
\end{proof}
\begin{proof}[Proof of Prop.~\ref{prop_convergencetokasner}]
	Using Prop.~\ref{prop_alphalimitsets_vacuum_inflat}, we can conclude that under the given assumptions, the $\alpha$-limit points are contained in the Kasner parabola~$\kasnerparabola$, both for the vacuum and the inflationary case.
	We show in the following that no solution can have $\alpha$-limit points with different $\Sigma_+$-values. As points on the Kasner parabola are uniquely identified by their $\Sigma_+$-value, this implies convergence.

	In case $\Sigma_+(\tau_0)>0$ at some time $\tau_0$, monotonicity of~$\Sigma_+$ shown in Lemma~\ref{lemm_monotonicitysigmaplus} contradicts the existence of two $\alpha$-limit points with different~$\Sigma_+$ values. Consequently, there is exactly one $\alpha$-limit point, which is equivalent to convergence.
	Recall that in the proof of Prop.~\ref{prop_taubone} we have computed the evolution of~$\Sigma_+$ to be
	\begin{equation}
	\label{eqn_proofdeltanplussignsevolutionsigmaplus}
		\Sigma_+'={-}2\Sigma_+\tilde A+\frac23(\Sigma_++1){\bparamk}\tilde A-\frac23(\Sigma_++1)N_+^2-\frac32(2-\gamma)\Omega.
	\end{equation}
	Hence, monotonicity of~$\Sigma_+$ holds in general in case $\tilde A=0$, as then~\eqref{eqn_proofdeltanplussignsevolutionsigmaplus} gives $\Sigma_+'\le0$ due to $\Sigma_+\in\left[{-}1,1\right]$ and $\Omega\ge0$.
	Assume therefore that $\Sigma_+\le0$ and $\tilde A>0$.

	According to Prop.~\ref{prop_taubone}, either the orbit is the constant one in the point Taub~1, or this Kasner point is not contained in the $\alpha$-limit set. The first case is excluded by the assumption. In the latter case, one can assume that for sufficiently negative times the orbit is bounded away from the point Taub~1, which implies that~$\Sigma_+$ is bounded from below by some constant greater than~${-}1$ for sufficiently negative times~$\tau$.
	Consequently, $\tilde\Sigma$ is bounded away from~$0$, as all possible $\alpha$-limit points are contained in $\kasnerparabola\cap\{\Sigma_+\le0\}$ and the point Taub~1 is excluded.
 	Using additionally that $q\rightarrow2$ due to~\eqref{eqn_convergenceratestepone}, the term $2(q+2\Sigma_+)$ is bounded from below by some suitable constant $D_1>0$ for sufficiently negative times, and the evolution \equ ~\eqref{eqns_evolutionbianchib} for~$\tilde A$ reads
	\begin{equation}
		\tilde A'\ge D_1\tilde A.
	\end{equation}
	Consequently, one finds
	\begin{equation}
		\tilde A=\O(e^{D_1\tau}).
	\end{equation}

	One computes from the evolution \equs ~\eqref{eqns_evolutionbianchib} that
	\begin{equation}
		(\frac{N_+^2}{\tilde A})'=\frac{12\Delta N_+}{\tilde A}
	\end{equation}
	and sees that according to Lemma~\ref{lemm_deltanplussigns} this derivative does not change sign for~$\tau$ sufficiently negative. As a consequence, the term ${N_+^2}/{\tilde A}$ either converges to a non-negative real number or diverges to~$\infty$ as $\tau\rightarrow{-}\infty$. In the latter case, the evolution of~$\Sigma_+$ is dominated by~$N_+^2$ alone, in the sense that~$\Sigma_+'$ as in~\eqref{eqn_proofdeltanplussignsevolutionsigmaplus} has negative sign for sufficiently negative times. One concludes as for $\Sigma_+>0$ that there is a unique $\alpha$-limit point. If the limit of ${N_+^2}/{\tilde A}$ is finite, this means~$N_+^2=\O(e^{D_1\tau})$.
	In combination with the decay estimate~\eqref{eqn_decayomega} on~$\Omega$, this yields that~$\Sigma_+'$ is integrable and implies convergence.
\end{proof}
As a direct consequence of convergence which we have shown in Prop.~\ref{prop_convergencetokasner}, we can apply Lemma~\ref{lemm_decaylemmageneral} to~$\tilde A$.
\begin{lemm}
\label{lemm_decaytildea}
	Consider a solution to \equs ~\eqref{eqns_evolutionbianchib}--\eqref{eqn_evolutionomega} converging to $(\slimit,1-\slimit^2,0,0,0)$ with $\slimit\in\left[{-}1,1\right]$. Then either $\tilde A=0$ along the whole orbit, or there exist for every $\eps>0$ a $\tau_\eps>{-}\infty$ such that $\tau\le\tau_\eps$ implies
	\begin{equation}
		e^{(4+4\slimit+\eps)\tau}\le \tilde A(\tau)\le e^{(4+4\slimit-\eps)\tau}.
	\end{equation}
\end{lemm}
The next lemma is of a technical nature and will be used in the following.
\begin{lemm}
\label{lemm_functionintegraltwominusq}
	Let $\gamma\in[0,2)$ and consider a solution to \equs ~\eqref{eqns_evolutionbianchib}--\eqref{eqn_evolutionomega} converging to $(\slimit,1-\slimit^2,0,0,0)$ with $\slimit\in\left[{-}1,1\right]$. Then the function
	\begin{equation}
		\integraltwominusq\coloneqq\int_{{-}\infty}^\tau(2-q)ds
	\end{equation}
	is well-defined.
\end{lemm}
\begin{proof}
	If $\slimit={-}1$, Prop.~\ref{prop_taubone} and Prop.~\ref{prop_tauboneothermatter} imply that the orbit is the constant orbit, in which case $\integraltwominusq=0$ along the whole orbit. Assume therefore that $\slimit>{-}1$.
	One computes using the evolution \equs ~\eqref{eqns_evolutionbianchib}
	\begin{equation}
	\label{eqn_evolutionsigmaplusplusone}
	\begin{split}
	(\Sigma_++1)'={}&(q-2)\Sigma_+-2\tilde N\\
		={}&(q-2)(\Sigma_++1)-(q-2)-2\tilde N\\
		={}&(q-2)(\Sigma_++1)+2\tilde A+\frac32(2-\gamma)\Omega.
	\end{split}
	\end{equation}
	Suppose that the function $\tilde\integraltwominusq\coloneqq\int_\tau^0(q-2)ds$ is unbounded as $\tau\rightarrow{-}\infty$, meaning that the integral tends to~${-}\infty$. As
	\begin{equation}
		\frac{d}{d\tau}(e^{\tilde\integraltwominusq}(\Sigma_++1))=(2\tilde A+\frac32(2-\gamma)\Omega)e^{\tilde\integraltwominusq}
	\end{equation}
	and~$\tilde A$ and~$\Omega$ decay as in Lemma~\ref{lemm_decaytildea} and \equ ~\eqref{eqn_decayomega}, integration yields
	\begin{equation}
		e^{\tilde\integraltwominusq}(\Sigma_+(\tau)+1)=\int_{{-}\infty}^\tau(2\tilde A+\frac32(2-\gamma)\Omega)e^{\tilde\integraltwominusq}ds=\O(e^{(\maxdecayleft-\eps)\tau}),
	\end{equation}
	where $\maxdecayleft\coloneqq\min(6-3\gamma,4+4\slimit)$ if $\Omega>0$, and $\maxdecayleft\coloneqq 4+4\slimit$ if $\Omega=0$. The assumption on~$\gamma$ ensures that~$\maxdecayleft>0$.
	Due to convergence to the Kasner parabola which implies $2-q\le\eps$ for sufficiently negative times, one finds that for some suitably chosen constant~$D_2>0$ and sufficiently negative times~$\tau$
	\begin{equation}
		\tilde\integraltwominusq\ge\eps\tau-D_2.
	\end{equation}
	Therefore, one concludes
	\begin{equation}
		\Sigma_++1=\O(e^{(\maxdecayleft-2\eps)\tau}),
	\end{equation}
	which is a contradiction to $\Sigma_+\rightarrow\slimit>{-}1$.
	As a consequence, the function $\tilde\integraltwominusq$ is bounded on $({-}\infty,0)$, and
	\begin{equation}
		\integraltwominusq(\tau)=\int_{{-}\infty}^\tau(2-q)ds=\int_{{-}\infty}^0(2-q)ds+\tilde\integraltwominusq
	\end{equation}
	is well-defined.
\end{proof}

Having found detailed decay properties for~$\tilde A$ and~$\Omega$, the next step is to determine the asymptotic behaviour of~$\Delta$ and~$N_+$. Their decay rates are intertwined: One searches for a linear combination of~$\Delta$ and~$N_+$ such that the evolution \equ\ has a form suitable for Lemma~\ref{lemm_decaylemmageneral}, \ie
\begin{equation}
	(\Delta+\rfactornew N_+)'=\functiondecaylemma(\Delta+\rfactornew N_+),
\end{equation}
and~$\functiondecaylemma$ converging as $\tau\rightarrow{-}\infty$. It turns out that the limit of~$\functiondecaylemma$ not only depends on the value of~$\slimit$ but also on the sign of~$\Delta N_+$, and one recovers exactly the eigenvalues $2+2\slimit\pm2\sqrt{3(1-\slimit^2)}$ of the linearised evolution equations on the Kasner parabola, see~\eqref{eqn_eigenvalueskasner} and Appendix~\ref{subsect_appendixlinearisedevolutionkasner}.
\begin{lemm}
\label{lemm_decaydeltanplus}
	Consider a solution to \equs ~\eqref{eqns_evolutionbianchib}--\eqref{eqn_evolutionomega} converging to $(\slimit,1-\slimit^2,0,0,0)$ with $\slimit\in\left({-}1,1\right)$. If $\Delta N_+(\tau)>0$ for all $\tau\le\tau_0$, then for every $\eps>0$ there exists $\tau_\eps>{-}\infty$ such that $\tau\le\tau_\eps$ implies
	\begin{equation}
		e^{(2+2\slimit+2\sqrt{3(1-\slimit^2)}+\eps)\tau}\le \absval{\Delta+\sqrt{\frac{1-\slimit^2}3}N_+}\le e^{(2+2\slimit+2\sqrt{3(1-\slimit^2)}-\eps)\tau}.
	\end{equation}
	If on the other hand $\Delta N_+(\tau)<0$ for all $\tau\le\tau_0$, then for every $\eps>0$ there exists $\tau_\eps>{-}\infty$ such that $\tau\le\tau_\eps$ implies
	\begin{equation}
		e^{(2+2\slimit-2\sqrt{3(1-\slimit^2)}+\eps)\tau}\le \absval{\Delta-\sqrt{\frac{1-\slimit^2}3}N_+}\le e^{(2+2\slimit-2\sqrt{3(1-\slimit^2)}-\eps)\tau}.
	\end{equation}

	Consider a solution to \equs ~\eqref{eqns_evolutionbianchib}--\eqref{eqn_evolutionomega} converging to $(1,0,0,0,0)$. If $\Delta N_+(\tau)>0$ for all $\tau\le\tau_0$, then for every $\eps>0$ there exists~$\hat\eps>0$ and~$\tau_\eps>{-}\infty$ such that $\tau\le\tau_\eps$ implies
	\begin{equation}
		e^{(4+\eps)\tau}\le \absval{\Delta+\hat\eps N_+}\le e^{(4-\eps)\tau}.
	\end{equation}
	If on the other hand $\Delta N_+(\tau)<0$ for all $\tau\le\tau_0$, then for every $\eps>0$ there exists~$\hat\eps>0$ and~$\tau_\eps>{-}\infty$ such that $\tau\le\tau_\eps$ implies
	\begin{equation}
		e^{(4+\eps)\tau}\le \absval{\Delta-\hat\eps N_+}\le e^{(4-\eps)\tau}.
	\end{equation}
\end{lemm}
We remark at this point that for this statement, no restriction on~$\gamma$ is imposed.
\begin{proof}
	For notational convenience, set
	\begin{equation}
	\label{eqn_definitionrfactor}
		\rfactornew\coloneqq\sqrt{\frac{1-\slimit^2}3}
	\end{equation}
	and note that
	\begin{equation}
	\label{eqn_propertyrfactor}
		2(1-\slimit^2)=6\rfactornew^2.
	\end{equation}
	In case $\Delta N_+>0$ and $\slimit\in({-}1,1)$, one computes
	\begin{equation}
		(\Delta+\rfactornew N_+)'=(2q+2\Sigma_+-2+6\rfactornew)\Delta
		+(\rfactornew q+2\rfactornew\Sigma_++2\tilde\Sigma-2\tilde N)N_+
	\end{equation}
	and notices
	\begin{align}
		&\lim_{\tau\rightarrow{-}\infty}(2q+2\Sigma_+-2+6\rfactornew) =2+2\slimit+6\rfactornew,\\
		&\lim_{\tau\rightarrow{-}\infty}(\rfactornew q+2\rfactornew\Sigma_++2\tilde\Sigma-2\tilde N) =2\rfactornew+2\rfactornew\slimit+2(1-\slimit^2).
	\end{align}
	Therefore
	\begin{equation}
	\label{eqn_evolutionrdeltaplusnplus}
	\begin{split}
		(\Delta+\rfactornew N_+)'
			={}&(2+2\slimit+6\rfactornew+f_1)\Delta+(2\rfactornew+2\rfactornew\slimit+6\rfactornew^2+f_2)N_+\\
			={}&(2+2\slimit+6\rfactornew+\frac{f_1\Delta+f_2N_+}{\Delta+\rfactornew N_+})(\Delta+\rfactornew N_+)\\
			={}&(2+2\slimit+2\sqrt{3(1-\slimit^2)}+\frac{f_1\Delta+f_2N_+}{\Delta+\rfactornew N_+})(\Delta+\rfactornew N_+),
	\end{split}
	\end{equation}
	with two functions $f_1,f_2$ converging to~$0$ as $\tau\rightarrow{-}\infty$. We do not need the explicit form of these two functions here, but use them in the proof of Prop.~\ref{prop_leftoftaub2additionaldecay}. Note that this computation makes use of the fact that~$\Delta$ and~$N_+$ have the same sign.
	As the last term in the bracket vanishes asymptotically,
	Lemma~\ref{lemm_decaylemmageneral} yields the decay of $\Delta+\rfactornew N_+$ in case both~$\Delta$ and~$N_+$ are positive. If both are negative, the statement follows due to the invariance of the evolution equations~\eqref{eqns_evolutionbianchib}--\eqref{eqn_evolutionomega} under a change of sign in these two variables.

	In case~$\slimit=1$, we find~$\rfactornew=0$ and can no longer conclude that the quotient in the last line in equation~\eqref{eqn_evolutionrdeltaplusnplus} vanishes asymptotically. We assume~$\Delta N_+>0$ and compute
	\begin{equation}
		(\Delta+\hateps N_+)'=(2q+2\Sigma_+-2+6\hateps)\Delta
		+(\hateps q+2\hateps\Sigma_++2\tilde\Sigma-2\tilde N)N_+,
	\end{equation}
	for~$\hateps>0$.
	By similar argument as above, we obtain the requested statement.

	In order to treat the cases where $\Delta N_+<0$ it is enough to replace every occurence of~$\rfactornew$ and~$\hateps$ by~${-}\rfactornew$ and~${-}\hateps$, respectively.
\end{proof}
The decay of~$\Delta$ and~$N_+$, depending on whether they have the same or opposite sign, determines the decay of the remaining variables.
\begin{lemm}
\label{lemm_deltanplussamesign}
	Let $\gamma\in[0,2)$ and consider a solution to \equs ~\eqref{eqns_evolutionbianchib}--\eqref{eqn_evolutionomega} converging to $(\slimit,1-\slimit^2,0,0,0)$ with $\slimit\in\left[{-}1,1\right]$.
	Assume that $\Delta N_+(\tau)>0$ for all $\tau\le\tau_0$. Then
	\begin{align}
		\Sigma_+={}&\slimit+\O(e^{(\maxdecayleft-\eps)\tau}),\\
		\tilde\Sigma={}&1-\slimit^2+\O(e^{(\maxdecayleft-\eps)\tau}),\\
		\Delta={}&\O(e^{(2+2\slimit+2\sqrt{3(1-\slimit^2)}-\eps)\tau}),\\
		N_+={}&\O(e^{(2+2\slimit+2\sqrt{3(1-\slimit^2)}-\eps)\tau}),\\
		\tilde N={}&\O(e^{(4+4\slimit-\eps)\tau}),\\
		q={}&2+\O(e^{(\maxdecayleft-\eps)\tau}),
	\end{align}
	as $\tau\rightarrow{-}\infty$, for every $\eps>0$. Here $\maxdecayleft\coloneqq\min(6-3\gamma,4+4\slimit)$ if $\Omega>0$, and $\maxdecayleft\coloneqq 4+4\slimit$ if $\Omega=0$.
	Furthermore,
	\begin{equation}
		\tilde A(3\slimit^2+{\bparamk}(1-\slimit^2))=0
	\end{equation}
	holds along the whole orbit, and $\tilde A>0$ implies that
	\begin{equation}
		3\Sigma_+^2+{\bparamk}\tilde\Sigma =\O(e^{(4+4\slimit+4\sqrt{3(1-\slimit^2)}-\eps)\tau})
	\end{equation}
	as $\tau\rightarrow{-}\infty$, for every $\eps>0$.

	Under the additional restriction that $\slimit\in({-}1,1)$, there is for each $\eps>0$ a $\tau_\eps>{-}\infty$ such that $\tau\le \tau_\eps$ implies
	\begin{align}
	\label{eqn_lowerbounddelta}
		e^{(2+2\slimit+2\sqrt{3(1-\slimit^2)}+\eps)\tau}\le {}&\absval{\Delta},\\
	\label{eqn_lowerboundnplus}
		e^{(2+2\slimit+2\sqrt{3(1-\slimit^2)}+\eps)\tau}\le {}&\absval{N_+}.
	\end{align}
\end{lemm}
\begin{rema}
\label{rema_auxiliarydecayconstraintexpr}
	As a consequence of this result, one finds that in the setting of this lemma
	\begin{equation}
		\frac13\tilde\Sigma N_+^2-\Delta^2=(\Sigma_+^2+\frac {\bparamk}3\tilde\Sigma)\tilde A=\O(e^{(8+8\slimit+4\sqrt{3(1-\slimit^2)}-\eps)\tau})
	\end{equation}
	both for $\tilde A=0$ and $\tilde A>0$.
\end{rema}
\begin{proof}[Proof of Lemma~\ref{lemm_deltanplussamesign}]
	If $\slimit={-}1$, then due to Prop.~\ref{prop_taubone} and Prop.~\ref{prop_tauboneothermatter} the solution is the constant orbit for which all the required properties hold. Let us therefore assume $\slimit\in({-}1,1)$.
	It follows from Lemma~\ref{lemm_decaydeltanplus} that for sufficiently negative~$\tau$
	\begin{equation}
		\absval{\Delta+\rfactornew N_+}\le e^{(2+2\slimit+2\sqrt{3(1-\slimit^2)}-\eps)\tau},
	\end{equation}
	with~$\rfactornew$ as in \equ ~\eqref{eqn_definitionrfactor}. In particular $\rfactornew>0$, which implies that~$\Delta$ and~$\rfactornew N_+$ have the same sign and yields the upper bound of~$\Delta$ and~$N_+$ in the statement. In case~$\slimit=1$, we can use the same argument with~$\rfactornew$ replaced by~$\hat\eps$.

	Comparing the decay of~$N_+$ with the decay~$\tilde A=\O(e^{(4+4\slimit-\eps)\tau})$ from Lemma~\ref{lemm_decaytildea}, we find
	\begin{equation}
		\tilde N=\frac13(N_+^2-{\bparamk}\tilde A)=\O(e^{(4+4\slimit-\eps)\tau}).
	\end{equation}
	Due to its definition~\eqref{eqn_definitionqmatter}, the quantity~$q$ inherits its convergence rate either from the decay of~$\Omega$ as in \equ ~\eqref{eqn_decayomega} or from the one of~$\tilde N+\tilde A$, whichever is slower:
	\begin{equation}
		q=2(1-\tilde A-\tilde N)-\frac32(2-\gamma)\Omega=2+\O(e^{(\maxdecayleft-\eps)\tau})
	\end{equation}
	for $\maxdecayleft\coloneqq\min(6-3\gamma,4+4\slimit)$.
	Vacuum is defined by $\Omega=0$, and one sees that the statement holds if one sets $\maxdecayleft\coloneqq 4+4\slimit$ in this case.
	Using the above information in the evolution \equs~\eqref{eqns_evolutionbianchib} for~$\Sigma_+$ and~$\tilde\Sigma$ yields
	\begin{equation}
		\Sigma_+'=\O(e^{(\maxdecayleft-\eps)\tau}),\qquad \tilde\Sigma'=\O(e^{(\maxdecayleft-\eps)\tau}),
	\end{equation}
	and thus gives the convergence rates for~$\Sigma_+$ and~$\tilde\Sigma$.

	\smallskip

	Writing the constraint \equ\ as in~\eqref{eqn_constraintforproofs} and applying the convergence rates of~$\Sigma_+$ and~$\tilde\Sigma$, one sees that
	\begin{equation}
	\label{eqn_constraintinproofsamesign}
		\tilde\Sigma N_+^2-3\Delta^2=(3\Sigma_+^2+{\bparamk}\tilde\Sigma)\tilde A=(3\slimit^2+{\bparamk}(1-\slimit^2)+\O(e^{(\maxdecayleft-\eps)\tau}))\tilde A.
	\end{equation}
	The \lhs\ is of order $\O(e^{(4+4\slimit+4\sqrt{3(1-\slimit^2)}-\eps)\tau})$ due to the above. The \rhs---if non-vanishing---consists of the bracket with its explicitly given decay and the factor~$\tilde A$ which---if non-vanishing---decays at most as $e^{(4+4\slimit+\eps)\tau}$, see Lemma~\ref{lemm_decaytildea}. In order for \equ ~\eqref{eqn_constraintinproofsamesign} to be consistent, either $\tilde A=0$ or $3\slimit^2+{\bparamk}(1-\slimit^2)=0$ has to hold.

	For the term in the bracket, one further computes from the evolution \equs ~\eqref{eqns_evolutionbianchib} and $\tilde N=(N_+^2-{\bparamk}\tilde A)/3$ that
	\begin{align}
	(3\Sigma_+^2+{\bparamk}\tilde\Sigma)'={}&6\Sigma_+\Sigma_+'+{\bparamk}\tilde\Sigma'\\
		={}&6(q-2)\Sigma_+^2-12\Sigma_+\tilde N+2{\bparamk}(q-2)\tilde\Sigma-4{\bparamk}\Sigma_+\tilde A-4{\bparamk}\Delta N_+\\
		={}&2(q-2)(3\Sigma_+^2+{\bparamk}\tilde\Sigma)-4\Sigma_+N_+^2-4{\bparamk}N_+\Delta\\
		={}&2(q-2)(3\Sigma_+^2+{\bparamk}\tilde\Sigma)-4N_+(\Sigma_+N_++{\bparamk}\Delta).
	\end{align}
	In Lemma~\ref{lemm_functionintegraltwominusq}, the function
	\begin{equation}
		\integraltwominusq(\tau)=\int_{{-}\infty}^\tau(2-q)ds
	\end{equation}
	is found to be well-defined, hence one obtains
	\begin{equation}
		\frac{d}{d\tau}(e^{2\integraltwominusq}(3\Sigma_+^2+{\bparamk}\tilde\Sigma))={-}e^{2\integraltwominusq}4N_+(\Sigma_+N_++{\bparamk}\Delta),
	\end{equation}
	and consequently
	\begin{align}
		e^{2\integraltwominusq}(3\Sigma_+^2+{\bparamk}\tilde\Sigma)
			={}&3\slimit^2+{\bparamk}(1-\slimit^2)-\int_{{-}\infty}^\tau e^{2\integraltwominusq}4N_+(\Sigma_+N_++{\bparamk}\Delta)ds\\
			={}&3\slimit^2+{\bparamk}(1-\slimit^2)+\O(e^{(4+4\slimit+4\sqrt{3(1-\slimit^2)}-\eps)\tau})
	\end{align}
	due to the above. If $\tilde A>0$, then the constant term on the \rhs\ vanishes due to our previous arguments, and one can conclude that
	\begin{equation}
		3\Sigma_+^2+{\bparamk}\tilde\Sigma=\O(e^{(4+4\slimit+4\sqrt{3(1-\slimit^2)}-\eps)\tau}).
	\end{equation}
	This yields that the \rhs\ of \equ ~\eqref{eqn_constraintforproofs}
	\begin{equation}
		\frac13\tilde\Sigma N_+^2-\Delta^2=(\Sigma_+^2+\frac {\bparamk}3\tilde\Sigma)\tilde A
	\end{equation}
	decays exponentially to order at least $8+8\slimit+4\sqrt{3(1-\slimit^2)}-\eps$, both for $\tilde A=0$ and $\tilde A>0$: In the first case, the \rhs\ vanishes identically, while in the latter case we have determined the decay properties for both factors on the \rhs\ individually.

	\smallskip

	It remains to show the lower bound for~$\Delta$ and~$N_+$ in case~$\slimit\in({-}1,1)$. From the last argument we conclude that
	\begin{equation}
		\frac13\tilde\Sigma(\frac{N_+}{\Delta+\rfactornew N_+})^2-(\frac{\Delta}{\Delta+\rfactornew N_+})^2=\frac{(\Sigma_+^2+\frac {\bparamk}3\tilde\Sigma)\tilde A}{({\Delta+\rfactornew N_+})^2}\rightarrow0
	\end{equation}
	as $\tau\rightarrow{-}\infty$, due to the lower bound on the denominator found in Lemma~\ref{lemm_decaydeltanplus}. As $1-\slimit^2>0$, the variable~$\tilde\Sigma$ is bounded away from zero for sufficiently negative times. If one of the two squared terms tends to zero along a time sequence, so does the other one (along the same time sequence), which would imply that
	\begin{equation}
		\frac{\Delta}{\Delta+\rfactornew N_+}+\rfactornew\frac{N_+}{\Delta+\rfactornew N_+}\rightarrow0,
	\end{equation}
	a contradiction. Consequently, both
	\begin{equation}
		\absval{\frac{\Delta}{\Delta+\rfactornew N_+}} \qquad \text{and}\qquad \absval{\frac{N_+}{\Delta+\rfactornew N_+}}
	\end{equation}
	are bounded from below by a positive constant for $\tau\le\tau_\eps$ sufficiently negative. This gives the lower bounds on~$\Delta$ and~$N_+$ and concludes the proof.
\end{proof}
\begin{lemm}
\label{lemm_deltanpluszero}
	Let $\gamma\in[0,2)$ and consider a solution to \equs ~\eqref{eqns_evolutionbianchib}--\eqref{eqn_evolutionomega} converging to $(\slimit,1-\slimit^2,0,0,0)$ with $\slimit\in[{-}1,1]$.
	Assume that $\Delta=N_+=0$ along the orbit.
	Then
	\begin{equation}
		\tilde A(3\slimit^2+{\bparamk}(1-\slimit^2))=0
	\end{equation}
	and
	\begin{align}
		\Sigma_+={}&\slimit+\O(e^{(\maxdecayleft-\eps)\tau}),\\
		\tilde\Sigma={}&1-\slimit^2+\O(e^{(\maxdecayleft-\eps)\tau}),\\
		\tilde N={}&\O(e^{(4+4\slimit-\eps)\tau}),\\
		q={}&2+\O(e^{(\maxdecayleft-\eps)\tau}),\\
	\end{align}
	as $\tau\rightarrow{-}\infty$, for every $\eps>0$. Here $\maxdecayleft\coloneqq\min(6-3\gamma,4+4\slimit)$ if $\Omega>0$, and $\maxdecayleft\coloneqq 4+4\slimit$ if $\Omega=0$.
\end{lemm}
\begin{proof}
	The first equation follows immediately from the constraint \equ ~\eqref{eqn_constraintforproofs} for $\tau\rightarrow{-}\infty$. For the decay properties, there is nothing to show for $\slimit={-}1$, as this is the constant orbit due to Prop.~\ref{prop_taubone} and Prop.~\ref{prop_tauboneothermatter}.
	In the other cases, one simplifies equation~\eqref{eqn_definitiontilden} for~$\tilde N$, equation~\eqref{eqn_definitionqmatter} for~$q$ and the evolution \equs ~\eqref{eqns_evolutionbianchib} for~$\Sigma_+$ and~$\tilde\Sigma$ using $\Delta=N_+=0$ to find
	\begin{align}
		\tilde N={}&-\frac13{\bparamk}\tilde A,\\
		q={}&2(1-\tilde A+\frac13{\bparamk}\tilde A)-\frac32(2-\gamma)\Omega,\\
		\Sigma_+'={}&(q-2)\Sigma_++\frac23{\bparamk}\tilde A,\\
		\tilde\Sigma'={}&2(q-2)\tilde\Sigma-4\Sigma_+\tilde A,
	\end{align}
	then inserts the decay of~$\tilde A$ and~$\Omega$ from Lemma~\ref{lemm_decaytildea} and \equ ~\eqref{eqn_decayomega} respectively to conclude the proof.
\end{proof}
For an orbit with $\Delta N_+\ge0$ for sufficiently negative times, there is no a priori restriction on where its limit point on the Kasner parabola is located with respect to Taub~2. In the case of opposite signs, the limit point has to be situated to the right of Taub~2, as the next lemma shows.
\begin{lemm}
\label{lemm_deltanplusoppositesign}
	Let $\gamma\in[0,2)$ and consider a solution to \equs ~\eqref{eqns_evolutionbianchib}--\eqref{eqn_evolutionomega} converging to $(\slimit,1-\slimit^2,0,0,0)$. If $\Delta N_+(\tau_k)<0$ for a sequence $\tau_k\rightarrow{-}\infty$, then $\slimit\in(1/2,1]$.
\end{lemm}
\begin{proof}
	As $\Sigma_+\in[{-}1,1]$ due to Remark~\ref{rema_statespacecompact}, the limit value~$\slimit$ is contained in the same interval.
	The case $\slimit={-}1$ can be excluded by Prop.~\ref{prop_taubone} and Prop.~\ref{prop_tauboneothermatter}.
	Due to Lemma~\ref{lemm_deltanplussigns} and under the assumption that $\slimit\in({-}1,1)$, we know that~$\Delta N_+$ has a fixed sign for sufficiently negative~$\tau$, hence it is enough to show that orbits with $\Delta N_+(\tau)<0$ for $\tau\le\tau_0$ cannot converge to a Kasner point with $\slimit\in({-}1,1/2]$.

	Consider first the case $-1<\slimit<1/2$ and recall from Lemma~\ref{lemm_decaydeltanplus}
	\begin{equation}
		\absval{\Delta-\rfactornew N_+}\ge e^{(2+2\slimit-2\sqrt{3(1-\slimit^2)}+\eps)\tau}
	\end{equation}
	for~$\tau$ sufficiently negative, where~$\rfactornew$ is as in \equ ~\eqref{eqn_definitionrfactor}.
	The bracket in the exponent is strictly negative for sufficiently small~$\eps$, which means that $\absval{\Delta-\rfactornew N_+}$ grows exponentially as $\tau\rightarrow{-}\infty$. This contradicts the fact that~$\Delta$ and~$N_+$ converge to zero.

	We can therefore assume that $\slimit=1/2$. One finds the special value $\rfactornew=1/2$ and notices
	\begin{equation}
		2+2\slimit-2\sqrt{3(1-\slimit^2)}=0.
	\end{equation}
	Hence
	\begin{equation}
		\absval{2\Delta-N_+}\ge e^{\eps\tau}
	\end{equation}
	for~$\tau$ sufficiently negative. This estimate holds for~$\Delta$ and~$N_+$ individually, as the following argument shows: The constraint \equ ~\eqref{eqn_constraintforproofs} can be reformulated into
	\begin{align}
		(3\Sigma_+^2+{\bparamk}\tilde\Sigma)\tilde A={}&\tilde\Sigma N_+^2-3\Delta^2\\
			={}&(\tilde\Sigma N_+-\frac34N_++\frac34N_+-\frac32\Delta)N_++\frac32(N_+-2\Delta)\Delta.
	\end{align}
	The \rhs\ divided by $3(N_+-2\Delta)/4$ equals
	\begin{equation}
		(\frac{4(\tilde\Sigma-\frac34)}{3(N_+-2\Delta)}N_++1)N_++2\Delta=(1+f_3)N_++2\Delta,
	\end{equation}
	with an asymptotically vanishing function~$f_3$, while the \lhs\ divided by the same expression decays exponentially to order $\O(e^{(4+4\slimit-2\eps)\tau})$ due to Lemma~\ref{lemm_decaytildea}. This yields
	\begin{equation}
		\absval{N_+}\ge e^{\eps\tau},\qquad \absval{\Delta}\ge e^{\eps\tau}.
	\end{equation}
	Using the definition of~$\tilde N$, the decay of~$\tilde A$ from Lemma~\eqref{lemm_decaytildea}, and \equ ~\eqref{eqn_definitionqmatter} for~$q$, one additionally finds
	\begin{equation}
		\tilde N\ge e^{\eps\tau},\qquad 2-q\ge e^{\eps\tau}.
	\end{equation}

	\smallskip

	The rest of the proof aims at constructing a contradiction to the slow decay behaviour of~$N_+$. This becomes possible by considering the evolution \equ s in more detail and relating the decay of several quantities.
	One starts by noting that Taylor expansion of the square root applied to the constraint \equ ~\eqref{eqn_constraintforproofs} yields
	\begin{equation}
	\label{eqn_proofoppositesignsdelta}
		\Delta={-}\frac{\sqrt3}{3}\tilde\Sigma^{\nicefrac12}N_++\O(e^{(6-2\eps)\tau}),
	\end{equation}
	where we used Lemma~\ref{lemm_decaytildea}, and the sign stems from our assumption $\Delta N_+<0$.
	We then multiply the constraint \equ ~\eqref{eqn_constraintgeneralone} by~$N_+^2$ and reformulate the resulting expression using \equ~\eqref{eqn_definitiontilden} defining~$\tilde N$ and \equ~\eqref{eqn_definitionqmatter} in the form
	\begin{equation}
		\tilde N=\frac12(2-q)-\tilde A-\frac34(2-\gamma)\Omega
	\end{equation}
	to find
	\begin{align}
		\Delta^2N_+^2={}&\tilde\Sigma\tilde NN_+^2-\Sigma_+^2\tilde AN_+^2\\
		={}&3\tilde\Sigma\tilde N^2+{\bparamk}\tilde\Sigma\tilde A\tilde N-\Sigma_+^2\tilde AN_+^2\\
		={}&\frac34(2-q)^2\tilde\Sigma+\O(e^{(6-3\gamma-\eps)\tau}).
	\end{align}
	The last line follows from convergence to the point Taub~2, \ie $\slimit=1/2$, in combination with Lemma~\ref{lemm_decaytildea}: $\tilde A$ decays as $\O(e^{(6-\eps)\tau})$, which is faster than the decay of~$\Omega$, given by $\O(e^{(6-3\gamma-\eps)\tau})$.
	Using once more Taylor expansion of the square root yields
	\begin{equation}
		\Delta N_+={-}\frac{\sqrt3}{2}(2-q)\tilde\Sigma^{\nicefrac12} +\O(e^{(6-3\gamma-2\eps)\tau}).
	\end{equation}

	In the next step one determines the behaviour of~$\Sigma_+$ and~$\tilde\Sigma$. The evolution of~$\Sigma_++1$ has been computed in~\eqref{eqn_evolutionsigmaplusplusone}. With the function~$\integraltwominusq$ defined as in Lemma~\ref{lemm_functionintegraltwominusq}, integration of $(e^\integraltwominusq(\Sigma_++1))'$ yields
	\begin{equation}
		e^{\integraltwominusq}(\Sigma_+(\tau)+1)=\frac32+\int_{{-}\infty}^\tau(2\tilde A+\frac32(2-\gamma)\Omega)e^{\integraltwominusq}ds,
	\end{equation}
	and thus
	\begin{equation}
	\label{eqn_proofoppositesignssigmaplus}
		\Sigma_++1=\frac32e^{-\integraltwominusq(\tau)}+\O(e^{(6-3\gamma-2\eps)\tau}).
	\end{equation}
	For~$\tilde\Sigma$, one finds
	\begin{align}
		\tilde\Sigma'={}&2(q-2)\tilde\Sigma-4\Delta N_+-4\Sigma_+\tilde A\\
			={}&2(q-2)(\tilde\Sigma-\sqrt3\tilde\Sigma^{\nicefrac12})+\O(e^{(6-3\gamma-2\eps)\tau}).
	\end{align}
	Hence
	\begin{equation}
		(\tilde\Sigma^{\nicefrac12}-\sqrt3)'
		=\frac12\frac{\tilde\Sigma'}{\tilde\Sigma^{\nicefrac12}}
		=(q-2)(\tilde\Sigma^{\nicefrac12}-\sqrt3)+\O(e^{(6-3\gamma-2\eps)\tau}),
	\end{equation}
	from which one concludes
	\begin{equation}
		\frac{d}{d\tau}(e^\integraltwominusq(\tilde\Sigma^{\nicefrac12}-\sqrt3))
			=\O(e^{(6-3\gamma-2\eps)\tau}).
	\end{equation}
	Integration yields
	\begin{equation}
	\label{eqn_proofoppositesignstildesigma}
		\tilde\Sigma^{\nicefrac12}-\sqrt3
		={-}\frac{\sqrt3}2e^{-\integraltwominusq(\tau)}+\O(e^{(6-3\gamma-2\eps)\tau}).
	\end{equation}
	One can eliminate~$\integraltwominusq$ using a linear combination of \equs ~\eqref{eqn_proofoppositesignssigmaplus} and~\eqref{eqn_proofoppositesignstildesigma} and finds
	\begin{equation}
	\label{eqn_proofoppositesignscombination}
		\Sigma_++\sqrt3\tilde\Sigma^{\nicefrac12}-2=\O(e^{(6-3\gamma-2\eps)\tau}).
	\end{equation}

	Now, we consider the evolution of~$N_+$. Due to \equs ~\eqref{eqn_proofoppositesignsdelta} and~\eqref{eqn_proofoppositesignscombination} and the slow decay of~$N_+$, one finds
	\begin{align}
		N_+'={}&(q+2\Sigma_+)N_++6\Delta\\
			={}&(q+2\Sigma_+-2\sqrt3\tilde\Sigma^{\nicefrac12})N_++\O(e^{(6-2\eps)\tau})\\
			={}&(q-2+4\Sigma_+-2+\O(e^{(6-3\gamma-2\eps)\tau}))N_++\O(e^{(6-2\eps)\tau})\\
			={}&(q-2+4\Sigma_+-2+\O(e^{(6-3\gamma-3\eps)\tau}))N_+,
	\end{align}
	which, recalling the discussion of~$\integraltwominusq$ in the proof of Lemma~\ref{lemm_functionintegraltwominusq}, implies that there is a function~$f_4$ which is integrable on $({-}\infty,0)$ and satisfies
	\begin{equation}
		N_+'=(4\Sigma_+-2+f_4)N_+.
	\end{equation}
	Set
	\begin{equation}
		F_4(\tau)\coloneqq\int_{{-}\infty}^\tau f_4(s)ds
	\end{equation}
	and compute
	\begin{equation}
		\frac{d}{d\tau}(e^{-F_4}N_+)=(4\Sigma_+-2)e^{-F_4}N_+.
	\end{equation}
	As $\Sigma_+'\le0$ asymptotically due to Lemma~\ref{lemm_monotonicitysigmaplus}, every solution converging to the point Taub~2 satisfies $4\Sigma_+\le2$, which means that the function $\absval{e^{-F_4}N_+}$ increases as $\tau\rightarrow{-}\infty$. It follows that~$N_+$ does not converge to~$0$, a contradiction.
\end{proof}
With the results found above, we are finally in a position to prove Prop.~\ref{prop_main_generalasymptoticproperties}.
\begin{proof}[Proof of Prop.~\ref{prop_main_generalasymptoticproperties}]
	In case $\slimit={-}1$, the solution is the constant solution (Prop.~\ref{prop_taubone} and Prop.~\ref{prop_tauboneothermatter}), for which the statement trivially holds. Otherwise, the decay of~$\Omega$ and~$\tilde A$ follows from \equ ~\eqref{eqn_decayomega} and Lemma~\ref{lemm_decaytildea}. According to Lemma~\ref{lemm_deltanplussigns}, there are three cases to be considered regarding the sign of~$\Delta N_+$. As $\slimit\in({-}1,1/2]$, the case $\Delta N_+<0$ is excluded by Lemma~\ref{lemm_deltanplusoppositesign}, and the remaining two cases are discussed in Lemma~\ref{lemm_deltanplussamesign} and~\ref{lemm_deltanpluszero}.
\end{proof}

In the proof of the previous lemma, we have discussed the asymptotic behaviour of orbits converging to a Kasner point to the left of Taub~2, as $\tau\rightarrow{-}\infty$. In particular, we could exclude $\Delta N_+<0$ asymptotically.
For orbits converging to a Kasner point to the right of the point Taub~2, \ie with~$\Sigma_+$ converging to $\slimit\in(1/2,1]$, we cannot exclude the negative sign. This stems from the fact that the non-positive eigenvalue $2+2\slimit-2\sqrt{3(1-\slimit^2)}$ which was used to construct a contradiction in the proof of Lemma~\ref{lemm_deltanplusoppositesign} becomes positive when changing from the Kasner arc to the left of Taub~2 to the one on the right. Orbits which satisfy~$\Delta N_+<0$ asymptotically and converge to a point on the Kasner parabola to the right of Taub~2 exist, and as in the case of~$\Delta N_+>0$ asymptotically, one finds that the rates of convergence are related to eigenvalues of the linearised evolution equation: For~$\Delta N_+>0$ asymptotically, we found exponential decay of order~$2+2\slimit+2\sqrt{3(1-\slimit^2)}$ (see Prop.~\ref{prop_main_generalasymptoticproperties}), while for~$\Delta N_+<0$ asymptotically, we obtain exponential decay of order~$2+2\slimit-2\sqrt{3(1-\slimit^2)}$.
\begin{lemm}
\label{lemm_deltanplusrightoftaub2oppositesign}
	Let $\gamma\in[0,2)$ and consider a solution to \equs ~\eqref{eqns_evolutionbianchib}--\eqref{eqn_evolutionomega} converging to $(\slimit,1-\slimit^2,0,0,0)$ with $\slimit\in\left(1/2,1\right]$.
	Assume that $\Delta N_+(\tau)<0$ for all $\tau\le\tau_0$.
	Then
	\begin{align}
		\Sigma_+={}&\slimit+\O(e^{(\maxdecayright-\eps)\tau}),\\
		\tilde\Sigma={}&1-\slimit^2+\O(e^{(\maxdecayright-\eps)\tau}),\\
		\Delta={}&\O(e^{(2+2\slimit-2\sqrt{3(1-\slimit^2)}-\eps)\tau}),\\
		N_+={}&\O(e^{(2+2\slimit-2\sqrt{3(1-\slimit^2)}-\eps)\tau}),\\
		\tilde N={}&\O(e^{(4+4\slimit-4\sqrt{3(1-\slimit^2)}-\eps)\tau}),\\
		q={}&2+\O(e^{(\maxdecayright-\eps)\tau})
	\end{align}
	as $\tau\rightarrow{-}\infty$, for every $\eps>0$. Here $\maxdecayright\coloneqq\min(6-3\gamma,4+4\slimit-4\sqrt{3(1-\slimit^2)})$ if $\Omega>0$, and $\maxdecayright\coloneqq4+4\slimit-4\sqrt{3(1-\slimit^2)}$ if $\Omega=0$.

	Under the additional restriction that $\slimit\in(1/2,1)$, there is for each $\eps>0$ a $\tau_\eps>{-}\infty$ such that $\tau\le \tau_\eps$ implies
	\begin{align}
		e^{(2+2\slimit-2\sqrt{3(1-\slimit^2)}+\eps)\tau}\le {}&\absval{\Delta},\\
		e^{(2+2\slimit-2\sqrt{3(1-\slimit^2)}+\eps)\tau}\le {}&\absval{N_+}.
	\end{align}
\end{lemm}
\begin{proof}
	If $\slimit={-}1$, then due to Prop.~\ref{prop_taubone} and Prop.~\ref{prop_tauboneothermatter} the solution is the constant orbit for which all the required properties hold. Let us therefore assume $\slimit\in({-}1,1]$.
	As in the beginning of the proof of Lemma~\ref{lemm_deltanplussamesign}, we conclude the upper bounds for~$\Delta$ and~$N_+$ from Lemma~\ref{lemm_decaydeltanplus}.

	The decay for~$\Omega$ and~$\tilde A$ has been determined in \equ ~\eqref{eqn_decayomega} and Lemma~\ref{lemm_decaytildea} independently of the sign of~$\Delta N_+$. In particular, $\tilde A$ decays faster than~$\Delta^2$ and~$N_+^2$, and consequently the evolution is no longer dominated by~$\tilde A$, but rather by those two variables. More precisely
	\begin{align}
		\tilde N={}&\frac13(N_+^2-{\bparamk}\tilde A)=\O(e^{(4+4\slimit-4\sqrt{3(1-\slimit^2)}-\eps)\tau}),\\
		q-2={}&{-}2\tilde A-2\tilde N-\frac32(2-\gamma)\Omega=\O(e^{(\maxdecayright-\eps)\tau}),
	\end{align}
	with $\maxdecayright$ defined as in the statement of the lemma.
	Using these properties in the evolution \equ s for~$\Sigma_+$ and~$\tilde\Sigma$ yields
		\begin{equation}
		\Sigma_+'=\O(e^{(\maxdecayright-\eps)\tau}),\qquad \tilde\Sigma'=\O(e^{(\maxdecayright-\eps)\tau}),
	\end{equation}
	and hence gives the convergence rates for~$\Sigma_+$ and~$\tilde\Sigma$.

	\smallskip

	In order to show the lower bounds for~$\Delta$ and~$N_+$ in case $\slimit\in({-}1,1)$, one considers the expression
	\begin{equation}
		\frac13\tilde\Sigma(\frac{N_+}{\Delta-\rfactornew N_+})^2-(\frac{\Delta}{\Delta-\rfactornew N_+})^2=\frac{(\Sigma_+^2+\frac {\bparamk}3\tilde\Sigma)\tilde A}{({\Delta-\rfactornew N_+})^2},
	\end{equation}
	which is a reformulation of the constraint \equ~\eqref{eqn_constraintforproofs}. Due to the previous argument, the numerator of the \rhs\ decays as $\O(e^{(4+4\slimit-\eps)\tau})$, while the denominator is bounded from below by $e^{(2+2\slimit-2\sqrt{3(1-\slimit^2)}+\eps)\tau}$ due to Lemma~\ref{lemm_decaydeltanplus}. The \lhs\ consequently converges to~$0$ as $\tau\rightarrow{-}\infty$. We conclude the same way we did when proving the lower bounds of Lemma~\ref{lemm_deltanplussamesign}.
\end{proof}

So far, we have obtained exponential convergence rates with exponents that always included an~$\eps>0$. Such convergence rates hold for all variables.
As an immediate consequence, we can apply Lemma~\ref{lemm_decaylemmageneralimproved} in order to eliminate the~$\eps$ in the decay of~$\tilde A$ and~$\Omega$.
\begin{lemm}
\label{lemm_improveddecaytildeaomega}
	Let $\gamma\in[0,2)$ and consider a solution to \equs ~\eqref{eqns_evolutionbianchib}--\eqref{eqn_evolutionomega} converging to $(\slimit,1-\slimit^2,0,0,0)$ with $\slimit\in\left[{-}1,1\right)$. Then the following holds:
	\begin{itemize}
		\item Either $\tilde A=0$ along the whole orbit, or there are constants $c_{\tilde A},C_{\tilde A}>0$ and~$\tau_0$ such that $\tau\le\tau_0$ implies
		\begin{equation}
			c_{\tilde A} e^{(4+4\slimit)\tau}\le \tilde A(\tau)\le C_{\tilde A} e^{(4+4\slimit)\tau}.
		\end{equation}
		\item Either $\Omega=0$ along the whole orbit, or there are constants $c_{\Omega},C_{\Omega}>0$ and~$\tau_0$ such that $\tau\le\tau_0$ implies
		\begin{equation}
			c_{\Omega} e^{(6-3\gamma)\tau}\le \Omega(\tau)\le C_{\Omega} e^{(6-3\gamma)\tau}.
		\end{equation}
	\end{itemize}
\end{lemm}
\begin{proof}
	In case $\slimit={-}1$, nothing has to be shown, as this is the constant orbit (Prop.~\ref{prop_taubone} and Prop.~\ref{prop_tauboneothermatter}).
	For the remaining values, all that has to be shown to apply Lemma~\ref{lemm_decaylemmageneralimproved} is that the expressions
	\begin{equation}
		2(q+2\Sigma_+),\qquad 2q-(3\gamma-2)
	\end{equation}
	appearing in the evolution equations of~$\tilde A$ and~$\Omega$ converge at exponential rates.
	From Lemma~\ref{lemm_deltanplussigns}, one knows that~$\Delta N_+$ has constant sign for sufficiently negative times. The case $\Delta N_+>0$ is covered in Lemma~\ref{lemm_deltanplussamesign}. If $\Delta=N_+=0$ along the whole orbit, then Lemma~\ref{lemm_deltanpluszero} yields the result. In case $\Delta N_+<0$ exponential convergence is shown in Lemma~\ref{lemm_deltanplusrightoftaub2oppositesign}.
	Note however that this case of opposite sign can be excluded for $\slimit\in({-1},1/2]$ by Lemma~\ref{lemm_deltanplusoppositesign}.
\end{proof}
These more detailed estimates can now be used to determine the convergence of the remaining variables in more detail. So far, we have found the constant term and the slowest order of exponential convergence. This is improved in the following way: We can relate the slowest order of exponential convergence of several variables, and determine which is the next non-vanishing order.

We carry this out for orbits converging to the left of the point Taub~2, as it is in this case that we make use of the more detailed convergence properties in Section~\ref{section_leftoftaubtwo}.
We point out, however, that the same approach can also be used on the remaining Kasner limit points.
\begin{prop}
\label{prop_leftoftaub2additionaldecay}
	Let $\gamma\in[0,2)$ and consider a solution to \equs ~\eqref{eqns_evolutionbianchib}--\eqref{eqn_evolutionomega} converging to $(\slimit,1-\slimit^2,0,0,0)$ with $\slimit\in({-1},1/2)$. If $\tilde A>0$, and if $\Delta$ and $N_+$ do not both vanish identically along the solution, then $\slimit=\pm\sqrt{{\bparamk}/({\bparamk}-3)}$, and there are constants $\alpha>0$, $\omega\ge0$, $\betaD\not=0$ and $\betaN\not=0$ \st
	\begin{align}
		\Delta={}&\betaD e^{(2+2\slimit+2\sqrt{3(1-\slimit^2)})\tau}+\O(e^{(2+2\slimit+2\sqrt{3(1-\slimit^2)}+\maxdecayleft-\eps)\tau}),\\
		\tilde A={}&\alpha e^{(4+4\slimit)\tau}+\O(e^{(4+4\slimit+\maxdecayleft-\eps)\tau}),\\
		N_+={}&\betaN e^{(2+2\slimit+2\sqrt{3(1-\slimit^2)})\tau}+\O(e^{(2+2\slimit+2\sqrt{3(1-\slimit^2)}+\maxdecayleft-\eps)\tau}),\\
		\tilde N={}&{-}\frac {\bparamk}3\alpha e^{(4+4\slimit)\tau}+\O(e^{(4+4\slimit+\maxdecayleft-\eps)\tau}),\\
		q={}&2+2\alpha(\frac {\bparamk}3-1)e^{(4+4\slimit)\tau}-\frac32(2-\gamma)\omega e^{(6-3\gamma)\tau}+\O(e^{(2\maxdecayleft-\eps)\tau}),
	\end{align}
	as $\tau\rightarrow{-}\infty$, for every $\eps>0$. Here $\maxdecayleft\coloneqq\min(6-3\gamma,4+4\slimit)$ if $\Omega>0$, and $\maxdecayleft\coloneqq 4+4\slimit$ if $\Omega=0$, and irrespectively of whether~$\Omega$ is positive or not
	\begin{equation}
		\betaD=\sqrt{(1-\slimit^2)/3}\betaN.
	\end{equation}
	Furthermore, if $\slimit\not=0$, then
	\begin{align}
	\label{eqn_leftoftaub2additionaldecay_otherssigmaplus}
		\Sigma_+={}&\slimit+\frac{\alpha {\bparamk}}{6\slimit}e^{(4+4\slimit)\tau}-\frac{\slimit}2\omega e^{(6-3\gamma)\tau}+\O(e^{(2\maxdecayleft-\eps)\tau}),\\
	\label{eqn_leftoftaub2additionaldecay_otherstildesigma}
		\tilde\Sigma={}&1-\slimit^2-\alpha e^{(4+4\slimit)\tau}-(1-\slimit^2)\omega e^{(6-3\gamma)\tau}+\O(e^{(2\maxdecayleft-\eps)\tau}),
	\end{align}
	otherwise
	\begin{align}
	\label{eqn_leftoftaub2additionaldecay_specialssigmaplus}
		\Sigma_+={}&\O(e^{(4+4\sqrt3)\tau}),\\
	\label{eqn_leftoftaub2additionaldecay_specialstildesigma}
		\tilde\Sigma={}&1-\alpha e^{4\tau}-\omega e^{(6-3\gamma)\tau}+\O(e^{(2\maxdecayleft-\eps)\tau}).
	\end{align}
	Additionally, either $\Omega\equiv0$ and then $\omega=0$, or
	\begin{equation}
		\Omega=\omega e^{(6-3\gamma)\tau}+\O(e^{(6-3\gamma+\maxdecayleft-\eps)\tau})
	\end{equation}
	with $\omega>0$.
\end{prop}

\begin{proof}
	It follows from Lemma~\ref{lemm_deltanplusoppositesign} that $\Delta N_+>0$ for sufficiently negative times, and Lemma~\ref{lemm_deltanplussamesign} collects the convergence properties which we obtained so far. From the relation
	\begin{equation}
		\tilde A(3\slimit^2+{\bparamk}(1-\slimit^2))=0,
	\end{equation}
	which was shown there, we find the special value for~$\slimit$. In particular, if~$\bparamk>0$ then there are no solutions converging to a point on the Kasner parabola with~$\slimit\in({-}1,1/2)$.
	We now deduce more precise decay properties for the individual variables.

	\smallskip

	For~$\tilde A$, the decay established in Lemma~\ref{lemm_improveddecaytildeaomega} implies that $e^{-(4+4\slimit)\tau}\tilde A$ is bounded for sufficiently negative~$\tau$. As
	\begin{equation}
		(e^{-(4+4\slimit)\tau}\tilde A)'=(2(q-2)+4(\Sigma_+-\slimit))e^{-(4+4\slimit)\tau}\tilde A
	\end{equation}
	and $2(q-2)+4(\Sigma_+-\slimit)=\O(e^{(\maxdecayleft-\eps)\tau})$, the \lhs\ decays as $\O(e^{(\maxdecayleft-\eps)\tau})$, and one finds
	\begin{equation}
	\label{eqn_auxiliarybetterdecaytildea}
		e^{-(4+4\slimit)\tau}\tilde A(\tau)=\lim_{\tau\rightarrow{-}\infty}e^{-(4+4\slimit)\tau}\tilde A(\tau)+\int_{{-}\infty}^\tau\O(e^{(\maxdecayleft-\eps)s})ds.
	\end{equation}
	Setting
	\begin{equation}
		\alpha\coloneqq \lim_{\tau\rightarrow{-}\infty}e^{-(4+4\slimit)\tau}\tilde A(\tau),
	\end{equation}
	equation~\eqref{eqn_auxiliarybetterdecaytildea} is equivalent to
	\begin{equation}
		\tilde A=\alpha e^{(4+4\slimit)\tau}+\O(e^{(4+4\slimit+\maxdecayleft-\eps)\tau}).
	\end{equation}
	The exact value of~$\alpha$ is not known, only that it is positive due to $\tilde A>0$. But this form shows that there are no terms of exponential order between $4+4\slimit$ and $4+4\slimit+\maxdecayleft$. By the same method one finds that
	\begin{equation}
		\Omega=\omega e^{(6-3\gamma)\tau}+\O(e^{(6-3\gamma+\maxdecayleft-\eps)\tau})
	\end{equation}
	with a constant $\omega>0$, if not $\Omega\equiv0$ along the whole orbit.

	\smallskip

	The improved decay properties of~$\tilde A$ and~$\Omega$ together with the ones from Lemma~\ref{lemm_deltanplussamesign} imply
	\begin{align}
		\tilde N={}&\frac13(N_+^2-{\bparamk}\tilde A)={-}\frac {\bparamk}3\alpha e^{(4+4\slimit)\tau}+\O(e^{(4+4\slimit+\maxdecayleft-\eps)\tau}),\\
		q={}&2(1-\tilde A-\tilde N)-\frac32(2-\gamma)\Omega\\
			={}&2+2\alpha(\frac {\bparamk}3-1)e^{(4+4\slimit)\tau}-\frac32(2-\gamma)\omega e^{(6-3\gamma)\tau}+\O(e^{(2\maxdecayleft-\eps)\tau}),
	\end{align}
	as $4\sqrt{3(1-\slimit^2)}\ge4+4\slimit$ \iif\ $\slimit\in\left[{-}1,1/2\right]$.
	The evolution \equs\ for~$\Sigma_+$ and~$\tilde\Sigma$ thus read
	\begin{align}
	\Sigma_+'={}&(q-2)\Sigma_+-2\tilde N\\
		={}&2\alpha(\frac13{\bparamk}\slimit-\slimit+\frac13{\bparamk}) e^{(4+4\slimit)\tau}-s\omega \frac32(2-\gamma)e^{(6-3\gamma)\tau}+\O(e^{(2\maxdecayleft-\eps)\tau}),\\
	\tilde\Sigma'={}&2(q-2)\tilde\Sigma-4\Delta N_+-4\Sigma_+\tilde A\\
		={}&4\alpha(\frac13{\bparamk}-\frac13{\bparamk}\slimit^2-1+\slimit^2-\slimit )e^{(4+4\slimit)\tau}-3\omega(1-\slimit^2)(2-\gamma) e^{(6-3\gamma)\tau}+\O(e^{(2\maxdecayleft-\eps)\tau}).
	\end{align}
	Integrating these expressions and making use of
	the relation $\slimit^2={\bparamk}/({\bparamk}-3)$
	yields that for $\slimit\not=0$ (or equivalently ${\bparamk}\not=0$) one finds~\eqref{eqn_leftoftaub2additionaldecay_otherssigmaplus} and~\eqref{eqn_leftoftaub2additionaldecay_otherstildesigma}. In case~$\slimit=0$ and~${\bparamk}=0$, one obtains equation~\eqref{eqn_leftoftaub2additionaldecay_specialstildesigma} for $\tilde\Sigma$, and
	\begin{equation}
		\Sigma_+=\O(e^{(2\maxdecayleft-\eps)\tau}).
	\end{equation}
	Note that~$2\maxdecayleft\le8<4+4\sqrt3$. We are going to obtain a stronger estimate further down.

	\smallskip

	To find the improved decay for the remaining variables~$\Delta$ and~$N_+$, we recall from the proof of Lemma~\ref{lemm_decaydeltanplus} that
	\begin{equation}
	\label{eqn_deltaplusrnplusforimproveddecay}
		(\Delta+\rfactornew N_+)'=(2+2\slimit+2\sqrt{3(1-\slimit^2)}+\frac{f_1\Delta+f_2N_+}{\Delta+\rfactornew N_+})(\Delta+\rfactornew N_+),
	\end{equation}
	which there appeared as \equ ~\eqref{eqn_evolutionrdeltaplusnplus}. Here, one has $\rfactornew=\sqrt{(1-\slimit^2)/3}$, and a closer look at the computation carried out in that proof reveals that the two functions~$f_1$, $f_2$ which vanish asymptotically as $\tau\rightarrow{-}\infty$ have the form
	\begin{align}
		f_1={}&2q-4+2\Sigma_+-2\slimit,\\
		f_2={}&\rfactornew q-2\rfactornew+2\rfactornew\Sigma_+-2\rfactornew\slimit+2\tilde\Sigma-2(1-\slimit^2)-2\tilde N.
	\end{align}
	Due to the previous results, both functions decay as $\O(e^{(\maxdecayleft-\eps)\tau})$.
	As~$\Delta$ and~$\rfactornew N+$ have the same sign, we see that the quotient with the functions~$f_1$, $f_2$ in \equ ~\eqref{eqn_deltaplusrnplusforimproveddecay} inherits the asymptotic behaviour of these two functions, \ie decays exponentially to order $\O(e^{(\maxdecayleft-\eps)\tau})$.
	Hence, we have
	\begin{equation}
		(e^{-(2+2\slimit+2\sqrt{3(1-\slimit^2)})\tau}(\Delta+\rfactornew N_+))'=e^{-(2+2\slimit+2\sqrt{3(1-\slimit^2)})\tau}(\Delta+\rfactornew N_+)\O(e^{(\maxdecayleft-\eps)\tau}),
	\end{equation}
	and find
	\begin{equation}
		e^{-(2+2\slimit+2\sqrt{3(1-\slimit^2)})\tau}(\Delta+\rfactornew N_+)=\beta+\O(e^{(\maxdecayleft-\eps)\tau}),
	\end{equation}
	for some constant $\beta\not=0$ having the same sign as~$\Delta$ and~$N_+$. Consequently, we obtain
	\begin{equation}
	\label{eqn_improveddecaydeltaplusrnplus}
		(\Delta+\rfactornew N_+)=\beta e^{(2+2\slimit+2\sqrt{3(1-\slimit^2)})\tau}+\O(e^{(2+2\slimit+2\sqrt{3(1-\slimit^2)}+\maxdecayleft-\eps)\tau}).
	\end{equation}
	We further see from Remark~\ref{rema_auxiliarydecayconstraintexpr} that
	\begin{equation}
		(\sqrt{\frac{\tilde\Sigma}{3}}N_++\Delta)(\sqrt{\frac{\tilde\Sigma}{3}}N_+-\Delta)=\frac13(3\Sigma_+^2+{\bparamk}\tilde\Sigma)\tilde A=\O(e^{(8+8\slimit+4\sqrt{3(1-\slimit^2)}-\eps)\tau}).
	\end{equation}
	The very first bracket is of order $e^{(2+2\slimit+2\sqrt{3(1-\slimit^2)})\tau}$ and not faster.
	Thus,
	\begin{equation}
	\label{eqn_relationdeltanplusimproveddecay}
		\Delta=\rfactornew N_++\O(e^{(2+2\slimit+2\sqrt{3(1-\slimit^2)}+\maxdecayleft-\eps)\tau}).
	\end{equation}
	Together with \equ~\eqref{eqn_improveddecaydeltaplusrnplus}, this implies the decay expressions for~$\Delta$ and~$N_+$ in the statement, with $\betaD=\rfactornew\betaN$.

	\smallskip

	It remains to show equation~\eqref{eqn_leftoftaub2additionaldecay_specialssigmaplus}. We know from the above that~$\tilde N=N_+^2/3=\O(e^{(4+4\sqrt{3})\tau})$.
	Using the function~$\integraltwominusq$ defined as in Lemma~\ref{lemm_functionintegraltwominusq},
	we integrate
	\begin{equation}
		(\Sigma_+e^{\integraltwominusq})'={-}2\tilde N e^{\integraltwominusq}=\O(e^{(4+4\sqrt{3})\tau})
	\end{equation}
	and obtain
	\begin{equation}
		\Sigma_+e^{\integraltwominusq}=\O(e^{(4+4\sqrt{3})\tau}).
	\end{equation}
	As the function~$\integraltwominusq$ is non-negative, this gives the decay for~$\Sigma_+$ and concludes the proof.
\end{proof}

\section{Asymptotics at Taub~2}
\label{section_taubtwo}

In this section, we determine precisely which orbits converge to the point Taub~2, \ie which solutions to \equs ~\eqref{eqns_evolutionbianchib}--\eqref{eqn_evolutionomega} satisfy
\begin{equation}
	\lim_{\tau\rightarrow{-}\infty}(\Sigma_+,\tilde\Sigma,\Delta,\tilde A,N_+)(\tau)=\left(\frac12,\frac34,0,0,0\right).
\end{equation}
We have to restrict ourselves to~$\gamma\in[0,2)$ in order to apply our results from the previous section. From Prop.~\ref{prop_main_generalasymptoticproperties}, we know that along every solution converging to the point Taub~2
\begin{equation}
	\tilde A({\bparamk}+1)=0
\end{equation}
has to hold, and we found precise decay conditions for all variables.

We deduce in the present section that only locally rotationally symmetric models of Bianchi type~I, II and~VI$_{{-}1}$ are possible.
In order to prove this statement, we show that the invariant set consisting of these three Bianchi types, compare Table~\ref{table_bianchihighersymmetry} and Definition~\ref{defi_lrsexpansionnorm}, contains all solution converging to the point Taub~2.
\begin{prop}
\label{prop_taub2specialandnonspecialkvalue}
	Let $\gamma\in[0,2)$ and consider a solution to~\eqref{eqns_evolutionbianchib}--\eqref{eqn_evolutionomega} which satisfies
	\begin{equation}
		\lim_{\tau\rightarrow{-}\infty}(\Sigma_+,\tilde\Sigma,\Delta,\tilde A,N_+)(\tau)=\left(\frac12,\frac34,0,0,0\right).
	\end{equation}
	Then the solution is contained in the invariant set
	\begin{equation}
		3\Sigma_+^2=\tilde\Sigma\qquad \Sigma_+N_+=\Delta.
	\end{equation}
\end{prop}
The proof revolves around the function
\begin{equation}
	\functionforscc=(3\Sigma_+^2-\tilde\Sigma)^2+2(\Sigma_+N_+-\Delta)^2.
\end{equation}
It is constructed in such a way that its zero set characterises the invariant set from Prop.~\ref{prop_taub2specialandnonspecialkvalue}.
The proof consists of showing that the function vanishes along all orbits converging to the point Taub~2. This is an adaptation of an approach which has already been used successfully in~\cite[Sect.~4]{ringstrom_curvblowupbianchiviiiandixvacuumspacetimes} in the case of Bianchi~A vacuum models. However, our method of showing that the function~$\functionforscc$ indeed vanishes is slightly different.
\begin{proof}
	The derivatives of the constituents of~$\functionforscc$ have been computed in~\eqref{eqn_derivativefunctionforscc}. Combining this with
	\begin{equation}
		\tilde A(\bparamk+1)=0
	\end{equation}
	which holds due to Prop.~\ref{prop_main_generalasymptoticproperties},
	this immediately yields that the set $\functionforscc=0$ is invariant.
	From the definition of~$\functionforscc$ as the sum of two squares
	we see that either~$\functionforscc\equiv0$ or~$\functionforscc>0$. In the first case, the statement follows. Assume therefore that~$\functionforscc>0$ holds.

	Let us have a closer look at the derivatives of the two terms which compose $\functionforscc$, see equation~\eqref{eqn_derivativefunctionforscc}. As
	\begin{equation}
		2\Sigma_+\rightarrow1,\qquad N_+\rightarrow0,\qquad q\rightarrow2,
	\end{equation}
	due to the assumption on the limit point, we conclude that
	\begin{equation}
		\functionforscc'\le \functionprooftaubtwo\functionforscc,
	\end{equation}
	for some function~$\functionprooftaubtwo$ which is integrable on~$({-}\infty,T]$ for some sufficiently negative time~$T$. This follows from the convergence rates obtained in Prop.~\ref{prop_main_generalasymptoticproperties}. Consequently, integration yields
	\begin{equation}
		\functionforscc(T)\le e^C\functionforscc(\tau)
	\end{equation}
	for all~$\tau\le T$, where~$C>0$ is a constant. Due to~$\functionforscc(\tau)\rightarrow0$ as~$\tau\rightarrow{-}\infty$, this implies~$\functionforscc(T)=0$, a contradiction.
\end{proof}
As a direct consequence of Prop.~\ref{prop_taub2specialandnonspecialkvalue}, the following theorem characterises the models whose orbits converge to the Kasner point Taub~2.
\begin{theo}
\label{theo_taubtwocharacterisationoforbits}
	Let $\gamma\in[0,2)$ and consider a solution to \equs ~\eqref{eqns_evolutionbianchib}--\eqref{eqn_evolutionomega} converging to $\left(1/2,3/4,0,0,0\right)$ as $\tau\rightarrow{-}\infty$. Then this solution is one of the following:
	\begin{itemize}
		\item Bianchi~I~LRS, \ie $\Delta=\tilde A=N_+=0$, $3\Sigma_+^2=\tilde\Sigma$,
		\item Bianchi~II~LRS, \ie $\tilde A=0$, $3\Sigma_+^2=\tilde\Sigma$, $\Sigma_+N_+=\Delta$,
		\item Bianchi~VI$_{{-}1}$~LRS, \ie ${\bparamk}={-}1$, $\tilde A>0$, $3\Sigma_+^2=\tilde\Sigma$, $\Sigma_+N_+=\Delta$.
	\end{itemize}
\end{theo}

\section{Asymptotics to the left of Taub~2}
\label{section_leftoftaubtwo}

In this section, we discuss the set of solutions converging to a Kasner point to the left of Taub~2 as $\tau\rightarrow{-}\infty$, \ie with a limit point $(\slimit,1-\slimit^2,0,0,0)$ satisfying~${-}1<\slimit<1/2$.
We have shown in Prop.~\ref{prop_main_generalasymptoticproperties} that for such orbits either $\tilde A=0$ has to hold, or the limit value~$\slimit$ and the parameter~$\bparamk$ have to satisfy the relation~$3\slimit^2+\bparamk(1-\slimit^2)=0$. This is the same restriction we appealed to in the previous section, when discussing convergence to the point Taub~2.
In other words, for a given~${\bparamk}$, any solution to \equ s~\eqref{eqns_evolutionbianchib}--\eqref{eqn_evolutionomega} converging to a point $(\slimit,1-\slimit^2,0,0,0)$ on the Kasner parabola with $-1<\slimit<1/2$ is either a Bianchi~A solution or has to converge to the limit point with $\slimit=\pm\sqrt{{\bparamk}/({\bparamk}-3)}$. This last special case, which can only occur if ${\bparamk}\le0$, is discussed in more detail in this section, using concepts, notation and results from dynamical systems theory which we recall in Appendix~\ref{section_appendixdynamsystheo}.

The object we are interested in is the sub\mf\ called centre-unstable \mf~$\C^u$, as it contains the maximal negatively invariant set~$A^-(U)$ of a suitable open neighborhood~$U$ of the point $(\slimit,1-\slimit^2,0,0,0)$. This latter set consists of all points which remain in~$U$ under the evolution in the negative time direction. In particular, all solutions converging to a point in~$U$ as $\tau\rightarrow{-}\infty$ are contained in this set~$A^-(U)$ for sufficiently negative times.

To determine the properties of the centre-unstable \mf~$\C^u$, we consider
the linearised evolution equations in the extended five-dimensional state space, by which we mean the linear approximation of the evolution \equ s~\eqref{eqns_evolutionbianchib}, with~$q$ and~$\tilde N$ defined as in equations~\eqref{eqn_definitionqgeneral} and~\eqref{eqn_definitiontilden}, but without assuming the constraint equations~\eqref{eqn_constraintgeneralone} and~\eqref{eqn_constraintgeneraltwo}.
The corresponding matrix is a linear transformation of the five-dimensional tangent space to~$\RR^5$ at equilibrium points of the evolution. We give the explicit form of the linear mapping and its eigenvectors and -values in Appendix~\ref{subsect_appendixlinearisedevolutionkasner}. For every Kasner point to the left of Taub~2, there are exactly four eigenvectors such that the corresponding eigenvalues are non-negative.

The centre-unstable \mf~$\C^u$ which we want to understand is tangent to the set~$E_m{}^c\oplus E_m{}^u$ which in its turn is spanned by the eigenvectors to eigenvalues with non-negative real part. The information on the number of such eigenvalues therefore translates into properties of the centre-unstable \mf, and thus into information on the set which solutions converging to a Kasner point to the left of Taub~2 have to eventually be contained in. We find the following statement: There is a four-dimensional \mf\ in~$\RR^5$ such that every solution converging to a point $(\slimit,1-\slimit^2,0,0,0)$ on the Kasner parabola with $-1<\slimit<1/2$ is contained in this \mf\ for sufficiently negative times.
\begin{prop}
\label{prop_centreunstablelocallyleftoftaub2}
	Consider the evolution \equ s~\eqref{eqns_evolutionbianchib}--\eqref{eqn_definitiontilden}
	in the extended state space, \ie without assuming the constraint \equ s~\eqref{eqn_constraintgeneralone}--\eqref{eqn_constraintgeneraltwo}.
	If ${\bparamk}\le0$ and $\slimit\coloneqq\pm\sqrt{{\bparamk}/({\bparamk}-3)}\in({-}1,1/2)$, then
	there is a neighborhood~$U$ of the point $(\slimit,1-\slimit^2,0,0,0)$ and a four-dimensional
	$C^1$~sub\mf
	~$\centreunstablespecialkasner$ of~$\RR^5$ in~$U$ with the following properties:
	\begin{itemize}
		\item $\centreunstablespecialkasner$ contains the point $(\slimit,1-\slimit^2,0,0,0)$ and in this point is tangent to those eigenvectors of the linearised evolution equations in the extended five-dimensional space, given in Appendix~\ref{subsect_appendixlinearisedevolutionkasner}, which correspond to eigenvalues with non-negative real part.
		\item Points in~$U$ are either contained in~$\centreunstablespecialkasner$ or
		their evolution under \equ s~\eqref{eqns_evolutionbianchib}--\eqref{eqn_definitiontilden} leaves~$U$ as $\tau\rightarrow{{-}\infty}$.
	\end{itemize}
\end{prop}
\begin{rema}
	We see in the proof that we could change the regularity of the manifold~$\centreunstablespecialkasner$ to~$C^r$, for some~$r<\infty$. Further, as all eigenvalues on~$\kasnerparabola$ are real, the \mf~$\centreunstablespecialkasner$ is tangent to those eigenvectors which correspond to non-negative eigenvalues.
\end{rema}
\begin{rema}
	This statement in particular applies to every solution to the evolution \equ s~\eqref{eqns_evolutionbianchib}--\eqref{eqn_evolutionomega} which converges to the point $(\slimit,1-\slimit^2,0,0,0)$, as $\tau\rightarrow{-}\infty$, with limit value $\slimit\coloneqq\pm\sqrt{{\bparamk}/({\bparamk}-3)}\in({-}1,1/2)$: There is a time~$\tau_0$ such that the solution is contained in the neighborhood~$U$ for all times~$\tau\le\tau_0$. Therefore, the solution has to be contained in the sub\mf~$\centreunstablespecialkasner$ for~$\tau\le\tau_0$.
\end{rema}
\begin{proof}
	Every point on the Kasner parabola is a zero-dimensional \mf\ consisting of equilibrium points of the evolution \equ s. Using the notation of Appendix~\ref{section_appendixdynamsystheo}, we can therefore apply Thm~\ref{theo_centremftheory} to this zero-dimensional sub\mf\ of equilibrium points to find a centre-unstable \mf~$\C^u$ near this point, and
	a neighborhood~$U$ of $(\slimit,1-\slimit^2,0,0,0)$ such that the maximal negatively invariant set~$A^-(U)$ is contained in~$\C^u$.
	The point~$(\slimit,1-\slimit^2,0,0,0)$ is contained in the manifold~$\C^u$ by Def.~\ref{defi_centremf}. Without loss of generality, we therefore restrict the manifold to~$U$.

	By definition, the \mf~$\C^u$ is tangential to $E_m{}^c\oplus E_m{}^u$, which are the subspaces of the tangent space at $(\slimit,1-\slimit^2,0,0,0)$ associated with eigenvalues on the imaginary axis and in the right half-plane:
	\begin{equation}
		0\qquad
		2(1+\slimit+\sqrt{3(1-\slimit^2)})\qquad
		4(1+\slimit)\qquad
		3(2-\gamma).
	\end{equation}
	The evolution equations~\eqref{eqns_evolutionbianchib} are polynomial and consequently~$C^\infty$.
	We can therefore apply Thm~\ref{theo_centremftheory} with some finite~$r$, for example~$r=1$.
	Consequently, $\C^u$ is a four-dimensional sub\mf, and it has the requested properties.
\end{proof}
\begin{rema}
	This statement could have been achieved using results from the theory of dynamical systems which are less powerful than the one we used here, Thm~\ref{theo_centremftheory}, as we only apply it to zero-dimensional \mf s and the individual points they contain. However, in Section~\ref{section_asymptoticsplanewave}, we make use of this theorem again, this time using it to a fuller extent.
\end{rema}

The previous statement gives information on the centre-unstable \mf\ corresponding to the evolution in the extended state space.
In order to understand the evolution in the non-extended state space, \ie restricted to the set of points
\st the constraint \equ s~\eqref{eqn_constraintgeneralone} and~\eqref{eqn_constraintgeneraltwo} are satisfied, we need to understand the relation between the \mf\ we found and the constraint surface defined by equation~\eqref{eqn_constraintgeneralone}. More precisely, we are interested in the codimension of their intersection, which is what we determine in Thm~\ref{theo_leftoftaubtworesult}. As the constraint \equ~\eqref{eqn_constraintgeneralone} defines a set which becomes singular in certain points, see Remark~\ref{rema_singularconstraintequ}, we cannot deduce this codimension solely from the knowledge about normal and tangent directions derived above. In addition, we use the convergence behaviour of the individual variables from Prop.~\ref{prop_leftoftaub2additionaldecay}, where we determined the slowest order exponential term and the next non-vanishing one, for convergence to Kasner points to the left of the point Taub~2.

\begin{theo}
\label{theo_leftoftaubtworesult}
	Let~$0\not={\bparamk}\in\RR$, $\gamma\in[0,2)$ and consider the set of solutions to \equs ~\eqref{eqns_evolutionbianchib}--\eqref{eqn_evolutionomega} converging to $(\slimit,1-\slimit^2,0,0,0)$ with $\slimit\in\left({-}1,1/2\right)$, as $\tau\rightarrow{-}\infty$. If~$\tilde A>0$, then
	\begin{equation}
		\slimit=\pm\sqrt{\frac{\bparamk}{{\bparamk}-3}}
	\end{equation}
	and the solution is contained in one of the following subsets:
	\begin{itemize}
		\item The invariant set satisfying $\tilde A>0$, $\Delta=0=N_+$, $3\Sigma_+^2+{\bparamk}\tilde\Sigma=0$.
		\item A countable union of $C^1$ sub\mf s satisfying $\tilde A>0$, and~$\Delta$, $N_+$ not both vanishing identically. The sub\mf s are contained either in the set of non-vacuum solutions or the set of vacuum solutions, and in the respective sets have codimension at least~one.
	\end{itemize}
\end{theo}
\begin{proof}
	The relation between the parameter~${\bparamk}$ and the value of~$\slimit$ is an immediate consequence of Prop.~\ref{prop_main_generalasymptoticproperties}, as we assume~$\tilde A>0$. If the solution satisfies~$\Delta=0=N_+$ at some time, then it does so at all times, as this property is invariant under the evolution equations~\eqref{eqns_evolutionbianchib}. The relation between~$\Sigma_+$ and~$\tilde\Sigma$ follows from equation~\eqref{eqn_constraintgeneralone},
	which concludes the proof for the first case.
	It thus remains to calculate the maximal dimension and codimension of the submanifolds in the second case.

	Consider a solution as in the statement which satisfies that~$\Delta$ and~$N_+$ do not both vanish identically. Due to convergence, the solution is contained in the neighborhood~$U$ from the previous proposition for sufficiently negative times, and therefore the solution has to lie in the sub\mf ~$\centreunstablespecialkasner$ for sufficiently negative times.
	At the same time, the solution satisfies the constraint \equ ~\eqref{eqn_constraintgeneralone}.
	We therefore have to show that the constraint surface and the sub\mf ~$\centreunstablespecialkasner$ either have an empty intersection or intersect transversally in a set of the correct dimension, when additionally restricted to $\tilde A>0$ and~$\Delta$, $N_+$ not both vanishing identically.
	To do this, we compare the normal direction of the constraint surface to the vectors spanning~$\centreunstablespecialkasner$. By counting the number of spanning vectors which are \ogon\ to the gradient direction, we find the dimension and properties of the set in question.

	As we consider solutions with~$\tilde A>0$ and $\Delta$, $N_+$ not both vanishing, and
	\begin{equation}
		\kappa=\frac{3\slimit^2}{\slimit^2-1}<0
	\end{equation}
	by Prop.~\ref{prop_main_generalasymptoticproperties}, we conclude from Remark~\ref{rema_singularconstraintequ}
	that the gradient of the constraint equation does not vanish along such solutions.
	Consequently, the set of points defined by equation~\eqref{eqn_constraintgeneralone} is smooth along the solutions under consideration.
	However, due to the special value of~${\bparamk}$ the constraint surface becomes singular in the limit point~$(\slimit,1-\slimit^2,0,0,0)$.
	We therefore cannot simply compute the scalar product between the (vanishing) gradient of the constraint equation and the vectors spanning the sub\mf ~$\centreunstablespecialkasner$ to show transversality, at least not in the limit point itself. Instead, we consider the gradient in its general form~\eqref{eqn_gradientconstraintequ}, but replace the fourth component using the constraint \equ ~\eqref{eqn_constraintforproofs} to obtain
	\begin{equation}
		({-}2\Sigma_+\tilde A\,,\,\frac13(N_+^2-{\bparamk}\tilde A)\,,\,{-}2\Delta\,,\,{-}\frac1{3\tilde A}(\tilde\Sigma N_+^2-3\Delta^2)\,,\,\frac23\tilde\Sigma N_+).
	\end{equation}
	Applying the improved convergence properties found in Prop.~\ref{prop_leftoftaub2additionaldecay} as well as Remark~\ref{rema_auxiliarydecayconstraintexpr} shows that the decay behaviour of this gradient is
	\begin{equation}
		({-}2\slimit,{-}\frac {\bparamk}3,0,0,0)\cdot\alpha e^{(4+4\slimit)\tau}+\O(e^{(\min(2+2\slimit+2\sqrt{3(1-\slimit^2)},4+4\slimit+\maxdecayleft-\eps))\tau}),
	\end{equation}
	where we used the special value of~${\bparamk}$. The last term here denotes a vector in~$\RR^5$ whose every component has the denoted decay. This decay is faster than the one of the first vector, which decays to order $4+4\slimit$. Consequently, we can normalise the gradient vector by multiplying with $e^{{-}(4+4\slimit)\tau}/\alpha$
	to eliminate the highest order of decay. Rescaling in this manner gives a well-defined non-vanishing gradient direction even up to the singular point $(\slimit,1-\slimit^2,0,0,0)$.

	Let us now turn to~$\centreunstablespecialkasner$. In the limit point~$(\slimit,1-\slimit^2,0,0,0)$, this four-dimensional sub\mf\ is spanned by the eigenvectors to non-negative eigenvalues, see Appendix~\ref{subsect_appendixlinearisedevolutionkasner} for the explicit form of these eigenvectors.
	Direct computation of the scalar product shows that the rescaled gradient of the constraint equation, \ie the vector $({-}2\slimit,{-}{\bparamk}/3,0,0,0)$, is \ogon\ to the eigenvectors to eigenvalues $2+2\slimit\pm2\sqrt{3(1-\slimit^2)}$, $4(1+\slimit)$ and $3(2-\gamma)$, but not the eigenvector to~$0$.

	Similarly, we can compare the spanning directions of~$\centreunstablespecialkasner$ to the normal direction of the set~$\Omega=0$, \ie the gradient of~$\Omega$. Due to equation~\eqref{eqn_omegageneral}, this gradient is
	\begin{equation}
		({-}2\slimit,{-}1,0,\frac{\bparamk}3-1,0)
	\end{equation}
	for the point on the Kasner parabola~$\kasnerparabola$ with~$\Sigma_+=\slimit$. Direct computation of the scalar product shows that this gradient is \ogon\ to the eigenvectors to eigenvalues $2+2\slimit\pm2\sqrt{3(1-\slimit^2)}$, $4(1+\slimit)$ and $0$, but not the eigenvector to~$3(2-\gamma)$.

	As the \mf~$\centreunstablespecialkasner$ is $C^1$, the spanning vectors depend continuously on the point. Consequently, in a sufficiently small neighborhood of the point~$(\slimit,1-\slimit^2,0,0,0)$, the \mf~$\centreunstablespecialkasner$ and the constraint surface intersect transversally in a sub\mf\ of dimension at most~three if they intersect at all. The solutions for non-vacuum, \ie $\Omega>0$, form a set of dimension four in the constraint surface, see also Table~\ref{table_bianchibsubsets}, as~$\bparamk<0$ by assumption. Hence, in a sufficiently small neighborhood the intersection of the \mf~$\centreunstablespecialkasner$ and the constraint surface is of codimension at least one in the set of all non-vacuum solutions.

	For vacuum solutions, we realise that one of the four eigenvectors in question is non-\ogon\ to the gradient of~$\Omega$, and another one is non-\ogon\ to the rescaled gradient of the constraint equation. Consequently, restricting the \mf~$\centreunstablespecialkasner$ first to the set~$\Omega=0$ and then additionally to the constraint surface by an argument similar to the one for~$\Omega>0$ yields that
	in a sufficiently small neighborhood of the point~$(\slimit,1-\slimit^2,0,0,0)$, the intersection of the \mf~$\centreunstablespecialkasner$, the constraint surface and the set~$\Omega=0$ is a sub\mf\ of dimension at most~two. As all vacuum solutions form a set of dimension three, in a sufficiently small neighborhood this intersection is of codimension at least one in the set of all vacuum solutions.

	We now apply the flow corresponding to the evolution equations and integer times to this intersection \mf. As the flow is a diffeomorphism coming from a polynomial evolution equation, the resulting set is a countable union of $C^1$ sub\mf s of codimension at least~one in the respective set of solutions, and by construction contains all solutions satisfying the properties listed in the statement. Further, the set~$\Omega=0$ is invariant under the flow.
\end{proof}
\begin{rema}
	In case~${\bparamk}=0$, we cannot apply the same reasoning, as the improved convergence properties from Prop.~\ref{prop_leftoftaub2additionaldecay} and the relation between~$\betaD$ and~$\betaN$ found in the same proposition yield that the gradient~\eqref{eqn_gradientconstraintequ} of the constraint equation decays as
	\begin{equation}
		(0,0,{-}1,0,\frac{\sqrt3}3)\cdot2\betaD e^{(2+2\sqrt{3})\tau}
		+\O(e^{(2+2\sqrt{3}+\maxdecayleft-\eps)\tau}),
	\end{equation}
	where we used Remark~\ref{rema_auxiliarydecayconstraintexpr} for the fourth component.
	One can normalise this vector by multiplication with $e^{{-}(2+2\sqrt{3})\tau}/(2\betaD)$ but the resulting direction in the limit point~$(0,1,0,0,0)$ is \ogon\ to all four eigenvectors to non-negative eigenvalues. This means that the sub\mf~$\centreunstablespecialkasner$ and the constraint surface do not intersect transversally, but are tangent in the limit point. This does not give any additional information on the set containing possible solutions.
\end{rema}

\section{Asymptotics to the right of Taub~2}

\label{section_rightoftaubtwo}

In this section, we turn our attention to solutions with a limit point on the Kasner parabola to the right of Taub~2, \ie a limit point in $\kasnerparabola\cap\{\Sigma_+>1/2\}$.
For such points, all eigenvalues but one are positive as soon as $\gamma<2$. One can therefore expect that this arc of equilibrium points acts as a source, even in the extended state space, as mentioned by Hewitt--Wainwright in~\cite{hewittwainwright_dynamicalsystemsapproachbianchiorthogonalB}.
As one considers an arc of equilibrium points, it is desirable to make a more thorough analysis. We carry out this analysis, using not the signs of the eigenvalues but the explicit evolution equations.
In our understanding, an arc of the Kasner parabola can be considered a source if for any point on this arc, every orbit which
enters a sufficiently small neighborhood of that point
also converges to the Kasner parabola as~$\tau\rightarrow{-}\infty$, and the limiting point is close to that particular point on the arc. Corollary~\ref{coro_rightoftaubtwovacuuminflationary} states that this holds for the arc of the Kasner parabola which lies to the right of the point Taub~2, for vacuum and inflationary matter models. In general, we are able to show that the~$\Sigma_+$ coordinate cannot differ too much.
\begin{prop}
\label{prop_rightoftaub2source}
	Let $1/2<\bar\slimit\le1$ and $0<\eps<\bar\slimit-1/2$. Then there is a neighborhood~$U$ of $(\bar\slimit,1-\bar\slimit^2,0,0,0)$
	with the following properties: For every solution to \equs ~\eqref{eqns_evolutionbianchib}--\eqref{eqn_evolutionomega} which intersects~$U$ at time~$\tau_0$, the~$\Sigma_+$-value
	satisfies
	\begin{equation}
		\absval{\Sigma_+(\tau)-\bar\slimit}<\eps
	\end{equation}
	for all~$\tau\le\tau_0$.
\end{prop}
This estimate holds in particular for the~$\Sigma_+$-value of every $\alpha$-limit point of the solution.
In case one knows that all $\alpha$-limit points satisfying $\Sigma_+>1/2$ are located on the Kasner parabola, this immediately gives an even stronger statement:
\begin{coro}
\label{coro_rightoftaubtwovacuuminflationary}
	Assume either vacuum or inflationary matter, \ie either $\Omega=0$ or $\Omega>0$, $\gamma\in\left[0,2/3\right)$, let $1/2<\bar\slimit\le1$ and $0<\eps<\bar\slimit-1/2$.
	Then there is a neighborhood~$U$ of $(\bar\slimit,1-\bar\slimit^2,0,0,0)$
	such that every solution to \equs ~\eqref{eqns_evolutionbianchib}--\eqref{eqn_evolutionomega} which intersects~$U$ converges to a Kasner point $(\slimit,1-\slimit^2,0,0,0)$ with $\absval{\slimit-\bar\slimit}<\eps$.
\end{coro}
\begin{proof}
	It follows from the previous proposition in combination with Prop.~\ref{prop_alphalimitsets_vacuum_inflat} and the fact that
	\begin{equation}
		\planewave {\bparamk} \cap \left\{\Sigma_+\ge\frac12\right\}=\emptyset
	\end{equation}
	that the~$\alpha$-limit set of every such solution is contained in the Kasner parabola~$\kasnerparabola$ and does not intersect the plane wave equilibrium points~$\planewave {\bparamk}$. Due to Prop.~\ref{prop_convergencetokasner}, the solution converges to a Kasner point $(\slimit,1-\slimit^2,0,0,0)$, and the estimate on~$\slimit$ follows from applying Prop.~\ref{prop_rightoftaub2source} again.
\end{proof}
To prove the previous proposition, the main idea is to use the fact that to the right of Taub~2, the Kasner parabola has a slope which is steeper than~${-}1$. Then, one shows that one can bound solutions from below by some straight line with this slope. This is the statement of Lemma~\ref{lemm_rightoftaub2deltanplusbothsigns} below and we provide a visualisation in Figure~\ref{figure_rightoftaubtwo}.
As the convergence point has to lie above this line but below the Kasner parabola, one gains control over where the convergence point has to be situated exactly.
\begin{lemm}
\label{lemm_rightoftaub2deltanplusbothsigns}
	Consider a solution to \equs ~\eqref{eqns_evolutionbianchib}--\eqref{eqn_evolutionomega} such that $\Sigma_+(\tau_0)>1/2$ for some fixed~$\tau_0$,
	and let $\eps_1>0$. Then there is a $\delta=\delta(\eps_1)>0$ such that $\tilde A(\tau_0)<\delta$ implies
	\begin{equation}
	\label{eqn_rightoftaub2estimates}
		\Sigma_+(\tau)\ge\Sigma_+(\tau_0),\qquad
		\tilde\Sigma(\tau)\ge{-}\Sigma_+(\tau)+(\Sigma_++\tilde\Sigma)(\tau_0)-\eps_1,\qquad
		\tilde A<\delta e^{\tau-\tau_0},
	\end{equation}
	for all $\tau\le\tau_0$.
\end{lemm}
We visualise the two first inequalities of this statement in Figure~\ref{figure_rightoftaubtwo}: Given that the~$\Sigma_+$- and~$\tilde\Sigma$-values are known at some time~$\tau_0$, the solution at earlier times~$\tau\le\tau_0$ has to be contained in the shaded area which is bounded by the vertical~$\Sigma_+=const$ line, the straight line with slope ${-}1$ through the point~$(\Sigma_+(\tau_0),\tilde\Sigma(\tau_0)-\eps_1)$, and the Kasner parabola~$\kasnerparabola$.
\begin{figure}[!ht]
	\begin{tikzpicture}[xscale=5.45,yscale=5.45]
	\draw [->] (0,0) -- (1.2,0);
	\draw [->] (0,0) -- (0,0.85);
	\draw [gray, thick, domain=0.35:1] plot (\x, {(\x+1)*(1-\x)});
	\node [below left] at (0,0) {$0$};
	\draw (1,-0.02) -- (1,0.02);
	\node [below] at (1,-0.01) {$1$};
	\draw (-0.02,0.75) -- (0.02,0.75);
	\node [left] at (0,0.75) {$3/4$};
	\node [right] at (1.2,0) {$\Sigma_+$};
	\node [above] at (0,0.85) {$\tilde\Sigma$};
	\draw (0.5,-0.02) -- (0.5,0.02);
	\node [below] at (0.5,-0.01) {$1/2$};
	\draw[fill] (0.5,0.75) circle [radius=0.01];
	\node [below left] at (0.5,0.75) {$\taubtwo$};
	\node [above left] at (0.7,0.6) {$\kasnerparabola$};
	\newcommand{\splusnull}{0.7}
	\newcommand{\stildenull}{0.4}
	\path [fill=gray,fill opacity=0.7] plot [smooth,samples=100,domain=\splusnull:0.958](\x,{(\x+1)*(1-\x)}) -- plot [smooth,samples=100,domain=0.958:\splusnull] (\x,{-\x +\splusnull +\stildenull-0.06});
	\draw [dotted] (0,\stildenull) -- (\splusnull,\stildenull);
	\draw [dotted] (\splusnull,0) -- (\splusnull,0.3);
	\draw (-0.02,\stildenull) -- (0.02,\stildenull);
	\node [left] at (0,\stildenull) {$\tilde\Sigma(\tau_0)$};
	\draw [dotted] (0,0.34) -- (\splusnull,0.34);
	\draw (-0.02,0.34) -- (0.02,0.34);
	\node [left] at (0,0.34) {$\tilde\Sigma(\tau_0)-\eps_1$};
	\draw (\splusnull,-0.02) -- (\splusnull,0.02);
	\node [below] at (\splusnull,-0.01) {$\Sigma_+(\tau_0)$};
	\draw[fill] (\splusnull,\stildenull) circle [radius=0.012];
	\node [above left] at (\splusnull,\stildenull) {$(\Sigma_+,\tilde\Sigma)(\tau_0)$};
	\draw [dashed] (\splusnull,0.3)--(\splusnull,0.55);
	\draw [dashed, domain=0.67:1.01] plot (\x, {-\x +\splusnull +\stildenull-0.06});
	\end{tikzpicture}
	\caption{Visualisation of Lemma~\ref{lemm_rightoftaub2deltanplusbothsigns}: For times~$\tau\le\tau_0$, the projection of the solution to the~$\Sigma_+\tilde\Sigma$-plane has to be contained in the shaded area.}
	\label{figure_rightoftaubtwo}
\end{figure}
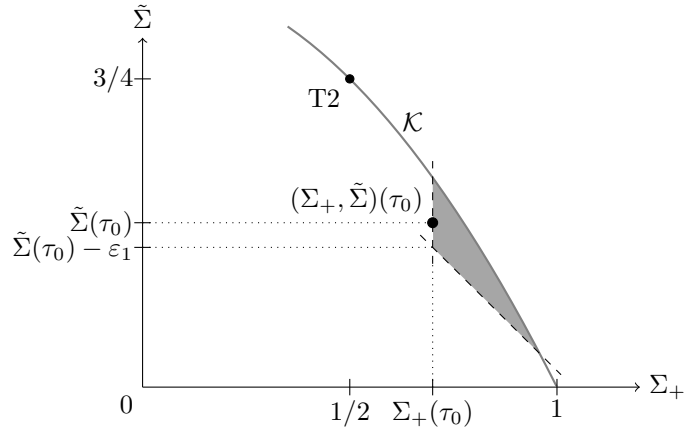
\begin{proof}
	The estimate for~$\Sigma_+$ is an immediate consequence of the monotonicity shown in Lemma~\ref{lemm_monotonicitysigmaplus}. It then follows that $\Sigma_+(\tau)>1/2$ for all~$\tau\le\tau_0$, and
	equation~\eqref{eqn_definitionqmatter} implies
	\begin{equation}
		q=2(\Sigma_+^2+\tilde\Sigma)+\frac12(3\gamma-2)\Omega\ge2\Sigma_+^2-\Omega\ge{-}\frac12
	\end{equation}
	due to Remark~\ref{rema_statespacecompact}.
	Using this estimate in the evolution \equ\ for~$\tilde A$, \equ ~\eqref{eqns_evolutionbianchib},
	one finds
	\begin{equation}
		\tilde A(\tau)\le\tilde A(\tau_0)e^{\tau-\tau_0}
	\end{equation}
	for all~$\tau\le\tau_0$.
	Consequently, the estimate on~$\tilde A$ holds for every choice of~$\delta$.

	To prove the remaining estimate, we fix a time~$\tau_1\le\tau_0$ and
	distinguish between the two cases that~$\Delta N_+(\tau_1)\ge0$ or~$\Delta N_+(\tau_1)<0$. We prove that in both cases
	\begin{equation}
	\label{eqn_evolutionsigmaplusplustildesigma}
		(\Sigma_++\tilde\Sigma)'(\tau_1)\le M\sqrt{\tilde A(\tau_1)},
	\end{equation}
	for some constant~$M>0$ independent of~$\tau_1$. This estimate is then used to conclude the proof.

	We start with the case $\Delta N_+(\tau_1)\ge 0$.
	As $\Sigma_+(\tau)>1/2$ holds for all~$\tau\le\tau_0$ due to the estimate on~$\Sigma_+$, we find
	\begin{equation}
		(\Sigma_++\tilde\Sigma)'(\tau_1)\le0,
	\end{equation}
	which gives the desired statement.

	Next, we consider the case $\Delta N_+(\tau_1)<0$.
	The constraint \equ ~\eqref{eqn_constraintgeneralone} in the form
	\begin{equation}
		\Delta^2=\frac13\tilde\Sigma N_+^2-(\Sigma_+^2+\frac{\bparamk}3\tilde\Sigma)\tilde A
	\end{equation}
	implies that
	\begin{equation}
		\absval{\Delta}\le \sqrt{\frac{\tilde\Sigma}{3}}\absval{N_+}+\sqrt{\frac{3+\absval {\bparamk}}{3}}\sqrt{\tilde A},
	\end{equation}
	using $\Sigma_+\in\left[{-}1,1\right]$ and $\tilde\Sigma\in\left[0,1\right]$.
	Due to the assumption on the sign of~$\Delta N_+$, this implies
	\begin{equation}
		0<{-}\Delta N_+\le \sqrt{\frac{\tilde\Sigma}{3}}N_+^2+\sqrt{\frac{3+\absval {\bparamk}}{3}}\absval{N_+}\sqrt{\tilde A}
	\end{equation}
	at time~$\tau_1$.
	With this and \equ ~\eqref{eqn_definitionqmatter} for~$q$ one computes, suppressing the time~$\tau_1$ for readibility, that
	\begin{align}
		(\Sigma_++\tilde\Sigma)'(\tau_1)={}&{-}\frac23(\Sigma_++1)N_+^2-\frac43\tilde\Sigma N_+^2-4\Delta N_+
			-\frac32(2-\gamma)(\Sigma_++2\tilde\Sigma)\Omega\\
			&{}-2\tilde A((3-\frac {\bparamk}3)\Sigma_++2(1-\frac {\bparamk}3)\tilde\Sigma-\frac {\bparamk}3)\\
			\le{}&\frac43({-}\frac12(\Sigma_++1)-\tilde\Sigma+\sqrt{3\tilde\Sigma})N_+^2 -\frac32(2-\gamma)(\Sigma_++2\tilde\Sigma)\Omega+f_1\sqrt{\tilde A},
	\end{align}
	for a function~$f_1$ which is
	bounded due to the compactness on the state space, say
	\begin{equation}
	\absval{f_1}\le M
	\end{equation}
	for some constant~$M>0$.
	The term containing~$\Omega$ is non-positive. For the first bracket, one easily sees that $\Sigma_+\in[1/2,1]$ and $\tilde\Sigma\in[0,3/4]$ imply
	\begin{align}
		-\frac12(\Sigma_++1){}&\in\left[{-}1,{-}\frac34\right],\\
		-\tilde\Sigma+\sqrt{3\tilde\Sigma}{}&\in\left[0,\frac34\right],
	\end{align}
	hence the term containing~$N_+^2$ is non-positive as well. In total, this implies that
	\begin{equation}
	(\Sigma_++\tilde\Sigma)(\tau_1)'\le M\sqrt{\tilde A}.
	\end{equation}

	We now integrate inequality~\eqref{eqn_evolutionsigmaplusplustildesigma} from~$\tau<\tau_0$ to~$\tau_0$.
	Using the estimate on~$\tilde A$ yields
	\begin{align}
		(\Sigma_++\tilde\Sigma)(\tau)\ge{}&(\Sigma_++\tilde\Sigma)(\tau_0)-\int_{\tau}^{\tau_0}M\sqrt{\tilde A(s)}ds\\
			\ge{}&(\Sigma_++\tilde\Sigma)(\tau_0)-\int_{\tau}^{\tau_0}M\sqrt\delta e^{(s-\tau_0)/2}ds\\
			\ge{}&(\Sigma_++\tilde\Sigma)(\tau_0)-2M\sqrt\delta,
	\end{align}
	and setting $\delta\coloneqq\eps_1^2(2M)^{{-}2}$ concludes the proof.
\end{proof}
\begin{proof}[Proof of Prop.~\ref{prop_rightoftaub2source}]
	Set~$\tilde U$ to be the neighborhood of the limiting point $(\bar\slimit,1-\bar\slimit^2,0,0,0)$ that satisfies
	\begin{equation}
		\absval{\Sigma_+-\bar\slimit}<\eps.
	\end{equation}
	One now constructs an even smaller neighborhood~$U$ whose closure is contained in~$\tilde U$ and such that all orbits starting in this smaller neighborhood are contained in~$\tilde U$.

	The construction proceeds as follows: Note that due to the restriction on the state space,
	equation~\eqref{eqn_constraintgeneraltwo}, one finds that
	\begin{equation}
		0\le\Sigma_+^2+\tilde\Sigma\le1.
	\end{equation}
	That is, the orbit projected to the $(\Sigma_+,\tilde\Sigma)$-plane lies below the Kasner parabola, which is the graph of a function with slope~$-2\Sigma_+$. For~$\Sigma_+$ in the interval $(1/2,1)$, this slope is strictly less than~$-1$ and decaying.
	One can therefore choose a constant $d<1-\bar\slimit^2$ such that the straight line with slope~${-}1$ through $(\bar\slimit,d)$
	\begin{equation}
	\label{eqn_proofrightoftaub2straightline}
		\tilde\Sigma-d={-}(\Sigma_+-\bar\slimit)
	\end{equation}
	intersects the Kasner parabola at some $\bar\slimit<\Sigma_+<\bar\slimit+\eps$. Let $0<\eps_1<\min((1-\bar\slimit^2-d)/2,\eps)$, choose $\delta=\delta(\eps_1)$ as in Lemma~\ref{lemm_rightoftaub2deltanplusbothsigns} and let~$U$ be the set defined by
	\begin{equation}
		\absval{\Sigma_+-\bar\slimit}<\eps_1,\qquad
		\Sigma_++\tilde\Sigma>\bar\slimit+d+\eps_1,\qquad
		\tilde A<\delta.
	\end{equation}

	By Lemma~\ref{lemm_rightoftaub2deltanplusbothsigns}, any orbit which is contained in~$U$ at time~$\tau_0$ satisfies
	\begin{align}
		\Sigma_+(\tau)\ge{}&\Sigma_+(\tau_0)>\bar\slimit-\eps_1,\\
		\tilde\Sigma(\tau)\ge{}&{-}\Sigma_+(\tau)+(\Sigma_++\tilde\Sigma)(\tau_0)-\eps_1>{-}\Sigma_+(\tau)+\bar\slimit+d,
	\end{align}
	for all~$\tau\le\tau_0$.
	From the first inequality we conclude that~$\Sigma_+(\tau)>\bar\slimit-\eps$ at all times~$\tau\le\tau_0$.
	The second inequality implies that the graph of the solution lies above the straight line from \equ~\eqref{eqn_proofrightoftaub2straightline}. Both inequalities are visualised in Figure~\ref{figure_rightoftaubtwo}. Because the graph also has to lie below the Kasner parabola, but we have chosen the constant~$d$ such that the straight line and the Kasner parabola intersect at some $\bar\slimit<\Sigma_+<\bar\slimit+\eps$, this implies that
	\begin{equation}
		\Sigma_+(\tau)<\bar\slimit+\eps
	\end{equation}
	for all times~$\tau\le\tau_0$ and every $\alpha$-limit point. This concludes the proof.
\end{proof}

\section{Asymptotics towards the plane wave equilibrium solutions}

\label{section_asymptoticsplanewave}

In this section, we use the theory of dynamical systems to determine the qualitative behaviour of solutions converging to the plane wave equilibrium points~$\planewave {\bparamk}$ as $\tau\rightarrow{-}\infty$.
The statements are not qualitatively new, as certain parts of~$\planewave {\bparamk}$ have already been identified as ``saddles'' or ``sinks'' in~\cite{hewittwainwright_dynamicalsystemsapproachbianchiorthogonalB}. Here, we state and prove more detailed properties of solutions converging to~$\planewave {\bparamk}$. The approach we use here is similar to the one in Section~\ref{section_leftoftaubtwo}, where we discussed the behaviour of solutions converging to a Kasner point situated to the left of the point Taub~2.
\begin{prop}
\label{prop_centreunstablelocallyplanewave}
	Consider the evolution \equ s~\eqref{eqns_evolutionbianchib}--\eqref{eqn_definitiontilden}
	in the extended state space, \ie without assuming the constraint \equ s~\eqref{eqn_constraintgeneralone}--\eqref{eqn_constraintgeneraltwo}.
	Let~$K_1$, $K_2$ be compact subsets of the arc $\planewave {\bparamk}\cap\{{-}1<\Sigma_+<{-}(3\gamma-2)/4\}$ and the arc $\planewave {\bparamk}\cap\{{-}(3\gamma-2)/4<\Sigma_+<0\}$, if these are non-empty. Let~$K_3$ and~$K_4$ denote the points on~$\planewave {\bparamk}$ with $\Sigma_+={-}(3\gamma-2)/4$ and $\Sigma_+=0$, \resp Then there are neighborhoods~$U_i$ of~$K_i$, $i=1,\ldots,4$, and $C^1$~sub\mf s~$\centreunstableplanewaveleft$, $\centreunstableplanewaveright$, $\centreunstableplanewavespecial$, $\centreunstableplanewavezero$ of~$\RR^5$ in~$U_1$, $U_2$, $U_3$, $U_4$ \resp with the following properties:
	\begin{itemize}
		\item $\centreunstableplanewaveleft$ contains the set~$K_1$ and in these points is tangent to those eigenvectors of the linearised evolution equations in the extended five-dimensional space, given in Appendix~\ref{subsect_appendixlinearisedevolutionplanewave}, which correspond to eigenvalues with non-negative real parts.
		\item $\centreunstableplanewaveright$ contains the set~$K_2$ and in these points
		is tangent to those eigenvectors which correspond to eigenvalues with non-negative real parts.
		\item $\centreunstableplanewavespecial$ contains the point~$K_3$ and in this point
		is tangent to those eigenvectors
		which correspond to eigenvalues with non-negative real parts.
		\item $\centreunstableplanewavezero$ contains the point~$K_4$ and in this point
		is tangent to those eigenvectors
		which correspond to eigenvalues with non-negative real parts.
	\end{itemize}
	Further:
	\begin{itemize}
		\item Points in~$U_1$ are either contained in~$\centreunstableplanewaveleft$ or
		their evolution under \equ s~\eqref{eqns_evolutionbianchib}--\eqref{eqn_definitiontilden} leaves~$U_1$ as $\tau\rightarrow{{-}\infty}$, and similarly for the remaining~$U_i$.
	\end{itemize}
	The dimensions of the sub\mf s are
	\begin{align}
		\dim \centreunstableplanewaveleft ={}&2,\\
		\dim \centreunstableplanewaveright ={}&1,\\
		\dim \centreunstableplanewavespecial ={}&2,\\
		\dim \centreunstableplanewavezero
			={}& \left\{\begin{array}{ll} 1 & \text{if}\quad3\gamma-2>0 ,\\
			2 & \text{if}\quad3\gamma-2\le0 \end{array}\right..
	\end{align}
\end{prop}
\begin{rema}
	We see in the proof that we could change the regularity of the sub\mf s to~$C^r$, for some~$r<\infty$.
	Further, the eigenvalues on~$\planewave k$ which are not real have real part~${-}2(1+\Sigma_+)$. The only possibility where this is non-negative for a point in one of the compact sets is if~$\gamma=2$ and~$\Sigma_+={-}1$. This implies that the solution converges to the point Taub~1 and hence is constant, see Prop.~\ref{prop_tauboneothermatter}. For the current statement, this situation is not of interest, and excluding the case~$K_4=\planewave k\cap\{\Sigma_+={-}1\}$, the sub\mf s are tangent to those eigenvectors which correspond to non-negative eigenvalues.
\end{rema}
For the proof, we make use of the concepts and notation introduced in Appendix~\ref{section_appendixdynamsystheo}, see also the explanation in the beginning of Section~\ref{section_leftoftaubtwo}. The set of points which remain in~$U_i$ under the evolution in negative time direction is the maximal negatively invariant set~$A^-(U_i)$, and we prove the proposition using properties of the centre-unstable \mf~$\C^u$.
\begin{proof}
	The arc $\planewave {\bparamk}\cap\{{-}1<\Sigma_+<{-}(3\gamma-2)/4\}$ is a \mf\ in~$\RR^5$ consisting of equilibrium points of the evolution \equ, with three eigenvalues of the linearised evolution equations lying in the left half-plane, one vanishing, and one lying in the right half-plane, see Def.~\ref{defi_planewaveexpansionnorm} and the adjacent text.

	According to Thm~\ref{theo_centremftheory}, there is a centre-unstable \mf~$\C^u$ near~$K_1$ and a neighborhood~$U_1$ of~$K_1$ such that the maximal negatively invariant set~$A^-(U_1)$ is contained in~$\C^u$.
	Without loss of generality, we can restrict the manifold~$\C^u$ to the open set~$U_1$.

	The \mf~$\C^u$ is by definition tangential to $E_m{}^c\oplus E_m{}^u$, which are the subspaces of the tangent spaces at points on~$K_1$ associated with eigenvalues on the imaginary axis and in the right half-plane.
	These are the eigenvalues~$0$ and ${-}4\Sigma_+-(3\gamma-2)$ whose eigenvectors span a two-dimensional subspace, which implies that~$\C^u$ is a sub\mf\ of dimension~two.
	The evolution equations~\eqref{eqns_evolutionbianchib} are polynomial and consequently~$C^\infty$, and the plane wave equilibrium points form a smooth curve. We can therefore apply Thm~\ref{theo_centremftheory} with some finite~$r$, for example~$r=1$.

	The proof for the arc $\planewave {\bparamk}\cap\{{-}(3\gamma-2)/4<\Sigma_+<0\}$ proceeds in the same way, with~$0$ being the only eigenvalue with non-negative real part. The two individual points on the arc~$\planewave {\bparamk}$ constitute (zero-dimensional) \mf s of equilibrium points on their own, to which we can also apply Thm~\ref{theo_centremftheory}. In case $\Sigma_+=0$, the value of~$\gamma$ determines where the eigenvalue ${-}4\Sigma_+-(3\gamma-2)$ is situated in the complex plane.
\end{proof}
As in~Section~\ref{section_leftoftaubtwo}, we now have to restrict the sub\mf s to the constraint surface. There are two main differences between the situation in that section and the present one: The dimension of all centre-unstable \mf s is now at most~two, and the constraint surface is singular only in the point $\Sigma_+=0$, see Remark~\ref{rema_singularconstraintequ}.

On the other hand,
intersecting the centre-unstable~\mf s for plane wave equilibrium points~$\planewave {\bparamk}$ with the constraint surface does not necessarily result in a lower dimension:
We know from the proof of the previous proposition that the centre-unstable \mf s are tangent to the eigenvectors to eigenvalues with positive or zero real parts, which are~$0$ and possibly ${-}4\Sigma_+-(3\gamma-2)$, depending on the relation between~$\Sigma_+$ and~$\gamma$, see Appendix~\ref{subsect_appendixlinearisedevolutionplanewave}. It follows by direct computation that these eigenvectors are orthogonal to the gradient of equation~\eqref{eqn_constraintgeneralone}.
Consequently, the centre-unstable \mf s and the constraint surface do not intersect transversally.
We can only conclude that restricting the sub\mf s to the constraint surface yields sub\mf s of at most the same dimension, that is dimension at most~one or~two.

\begin{theo}
\label{theo_centreunstablegloballyplanewave}
	For a given~${\bparamk}\in\RR$ and~$\gamma\in[0,2)$, consider the set of solutions to \equs ~\eqref{eqns_evolutionbianchib}--\eqref{eqn_evolutionomega} converging to a point on~$\planewave {\bparamk}$ as $\tau\rightarrow{-}\infty$. Assume that the limit point satisfies~$\Sigma_+=\slimit$. Then
	\begin{equation}
	\label{eqn_inequalityplanewave}
		{\bparamk}(1+\slimit)\ge3\slimit
	\end{equation}
	has to hold. Furthermore:
	\begin{itemize}
		\item If~${-}(3\gamma-2)/4<\slimit<0$, then the solution is the constant solution. In particular, the solution is a vacuum solution.
		\item There is a countable family of $C^1$ sub\mf s $\{\submfsplanewave_m\}_{m\in\NN}$ of dimension at most~two
				such that if $\slimit\le{-}(3\gamma-2)/4$ and~$\slimit<0$,
				then the solution is contained in $\bigcup_{m\in\NN}\submfsplanewave_m$.
		\item There is a countable family of $C^1$ sub\mf s $\{\submfsplanewavezero_m\}_{m\in\NN}$ of dimension at most~two
				such that if $\slimit=0$,
				then the solution is contained in $\bigcup_{m\in\NN}\submfsplanewavezero_m$.
	\end{itemize}
\end{theo}
\begin{proof}
	The relation between the parameter~${\bparamk}$ and the limiting value~$\slimit$ follows from the definition of the plane wave equilibrium points, Def.~\ref{defi_planewaveexpansionnorm}, as~$N_+^2\ge0$.

	We start with the case ${-}(3\gamma-2)/4<\slimit<0$.
	Fix a compact subarc~$K_2$ of the arc $\planewave {\bparamk}\cap\{{-}(3\gamma-2)/4<\Sigma_+<0\}$ containing the limiting point of the solution. As the solution converges to a point in~$K_2$, it is contained in the neighborhood~$U_2$ from the previous statement for sufficiently negative times, and therefore the solution has to lie in the one-dimensional sub\mf ~$\centreunstableplanewaveright$ for sufficiently negative times. The set of plane wave equilibrium points~$\planewave {\bparamk}$ itself forms a one-dimensional sub\mf\ as well. It consists of constant solutions, and~$K_2$ consequently must be contained in the centre-unstable \mf ~$\centreunstableplanewaveright$.
	Due to the dimension, both sub\mf s coincide in a sufficiently small neighborhood of the limiting point, and the first part of the statement follows. The solution satisfies~$\Omega=0$ due to the definition of~$\planewave {\bparamk}$, see Def.~\ref{defi_planewaveexpansionnorm} and below.

	For the second case,
	choose a countable family of compact subarcs~$K_m$, $m\in\NN$, which exhaust the arc $\planewave {\bparamk}\cap\{{-}1<\Sigma_+<{-}(3\gamma-2)/4\}$. For every~$K_m$ as well as for the
	point on~$\planewave {\bparamk}$ with $\Sigma_+={-}(3\gamma-2)/4$,
	consider the sub\mf s found in Prop.~\ref{prop_centreunstablelocallyplanewave}. Convergence implies that for sufficiently negative times the solution cannot escape the corresponding open neighborhoods, and it consequently lies in one of these sub\mf s for sufficiently negative times.

	Still without restricting to the constraint \equs\, we apply the flow corresponding to the evolution \equ s and integer times to these sub\mf s. The flow resulting from a polynomial evolution equation implies that the regularity of the sub\mf s is preserved. As the graph of any solution is invariant under this flow, solutions to the evolution \equ s in the extended state space are fully contained in the resulting family of sub\mf s~$\submfsplanewave_m$. By construction, the family is countable. In the extended state space, the statement about the dimension follows from Prop.~\ref{prop_centreunstablelocallyplanewave} and the fact that the dimension of a sub\mf\ is invariant under \diffeo s, consequently invariant under the flow. Restricting to the constraint \equ s cannot increase the dimension, which concludes the proof in this case. The remaining case where~$\slimit=0$ is treated in the same way.
\end{proof}
\begin{rema}
\label{rema_planewaveconvergencebianchiseparated}
	Consider the element of the plane wave equilibrium points~$\planewave {\bparamk}$ which satisfies~$\Sigma_+=\slimit\in({-}1,0]$.
	We have found in the previous theorem that the question whether there is a solution to equations~\eqref{eqns_evolutionbianchib}--\eqref{eqn_evolutionomega} converging to this point, and whether it is constant or not, depends on the relation between~$\slimit$ and the group parameter~${\bparamk}$, as well as the relation between~$\slimit$ and the matter parameter~$\gamma$.

	In case of vacuum~$\Omega=0$, the only solutions converging to~$\planewave {\bparamk}$ are the constant solutions, see Prop.~\ref{prop_alphalimitsets_vacuum_inflat}.
	Assume that~$\Omega>0$. All solutions have to be non-constant and in particular cannot converge to limit points with~${-}(3\gamma-2)/4<\slimit<0$. We analyse the possible combinations of~$\slimit$,~${\bparamk}$ and~$\gamma$, depending on the Bianchi type.
	\begin{itemize}
		\item Bianchi class~A solutions cannot converge to the plane wave equilibrium points~$\planewave {\bparamk}$, due to~$\tilde A=0$ for class~A but non-vanishing on~$\planewave {\bparamk}$.
		\item Bianchi type VI$_{\binvparam}$: As~${\bparamk}<0$ in this Bianchi type, solutions can only converge to limit points satisfying
		\begin{equation}
			\slimit\le\frac{{\bparamk}}{3-{\bparamk}}<0, \qquad \slimit\le{-}(3\gamma-2)/4,
		\end{equation}
		due to inequality~\eqref{eqn_inequalityplanewave} and~$\Omega>0$. The solutions are then contained in the union of submanifolds~$\submfsplanewave_m$.
		\item Bianchi type VII$_{\binvparam}$: In this Bianchi type, the parameter~$\bparamk$ is positive.
		Inequality~\eqref{eqn_inequalityplanewave} is satisfied for all possible values of~$\slimit$.
		Then either~$\slimit\le{-}(3\gamma-2)/4$ and~$\slimit<0$, or~$\slimit=0$ has to hold, and the solutions are contained in the union of submanifolds~$\submfsplanewave_m$ or~$\submfsplanewavezero_m$.
		\item Bianchi type IV: As~${\bparamk}=0$, inequality~\eqref{eqn_inequalityplanewave} is satisfied for all possible values of~$\slimit$.
		Then either~$\slimit\le{-}(3\gamma-2)/4$ and~$\slimit<0$, or~$\slimit=0$ has to hold, and the solutions are contained in the union of submanifolds~$\submfsplanewave_m$ or~$\submfsplanewavezero_m$.
		\item Bianchi type V: These models are defined by~${\bparamk}=0$, $\tilde A>0$, $\Sigma_+=\Delta=N_+=0$, which implies that solutions can only converge to the limit point with~$\slimit=0$.
		The evolution equation~\eqref{eqns_evolutionbianchib} for~$\tilde\Sigma$ reads
		\begin{equation}
			\tilde\Sigma'=2(q-2)\tilde\Sigma
		\end{equation}
		on Bianchi type~V solutions. As~$q\le2$ and~$\tilde\Sigma\ge0$, due to~\eqref{eqn_constraintgeneralthree}, this implies that~$\tilde\Sigma$ is monotone decreasing or constant. As the only element of the plane wave equilibrium points~$\planewave {\bparamk}$ with~$\slimit=0$ satisfies~$\tilde\Sigma=0$,
		this shows that~$\tilde\Sigma$ vanishes at all times.
		Further, the element in the plane wave equilibrium point with~$\slimit=0$ satisfies~$\tilde A=1$. With the information on the other variables, the evolution equation of~$\tilde A$ reads
		\begin{equation}
			\tilde A'=(3\gamma-2)(1-\tilde A)\tilde A.
		\end{equation}
		Its range, see~\eqref{eqn_constraintgeneralthree} then implies that for~$\gamma\ge2/3$, only the constant orbit is possible.
		In case~$0\le\gamma<2/3$, the solution is contained in~$\Sigma_+=\tilde\Sigma=\Delta=N_+=0$ and~$\tilde A$ decreases mononotically from~$1$ to~$0$.
	\end{itemize}
\end{rema}

\section{Equivalence of the initial data perspective and the expansion-normalised variables}

\label{section_equivalenceinitialdataexpansionnormalised}

The goal of this section is to justify the use of expansion-normalised variables, and show that under the correct transformation the evolution of these variables is equivalent to solving Einstein's equation for \ogon\ Bianchi~B perfect fluid initial data.
At the same time, we show how to construct, for given initial data $(G,{\idmetric},{\idfundform},\mu_0)$ as in Def.~\ref{defi_initialdatabianchib}, the \mghd\ and prove properties of this spacetime. This is done via the expansion-normalised variables $(\Sigma_+,\tilde\Sigma,\Delta,\tilde A,N_+)$ for Bianchi class~B models.

The expansion-normalised variables in Bianchi class~B models were developed in~\cite{hewittwainwright_dynamicalsystemsapproachbianchiorthogonalB}, motivated by a similar set of coordinates for Bianchi~A models introduced in~\cite{wainwrighthsu_dynamicalsystemsapproachbianchiorthogonalA}.
Their deduction starts out with given structure constants~$\gamma_{\alpha\beta}^\delta$ of a suitably chosen four-dimensional \onorm\ frame of the spacetime in question. To connect this to the initial data perspective, we have to understand how such a frame can be constructed from the knowledge of the metric~$\idmetric$ and the two-tensor~${\idfundform}$ on the \liegr ~$G$. In particular, this means choosing a suitable three-dimensional basis of the \liealg\ associated with the \liegr.
For this reason, we collect the necessary background on three-dimensional \liegr s in Subsection~\ref{subsect_appendixbianchiclassification}.

\smallskip

We recall the deduction of the expansion-normalised variables from given structure constants~$\gamma_{\alpha\beta}^\delta$ in Subsection~\ref{subsect_appendixdeducenormalisedvariablesbianchib}.
Note that this deduction already starts with the full spacetime, in particular with structure constants $\gamma_{\alpha\beta}^\delta$ in four dimensions. For initial data, this information is not available, only three-dimensional spacelike structure constants~$\gamma_{ij}^k$ make sense. The structure constants~$\gamma_{0i}^j$ require the existence of a timelike vector field~$e_0$.
However, to begin with there is no such vector field.
Nonetheless, we can define objects~$\tilde\gamma_{0i}^j$ using the metric~$\idmetric$ and the two-tensor~${\idfundform}$ and need to make sure that they have the form required for the construction by~\cite{hewittwainwright_dynamicalsystemsapproachbianchiorthogonalB}, which means that
the only non-vanishing ones are
\begin{equation}
	\tilde\gamma_{01}^1,\quad \tilde\gamma_{0A}^B,
\end{equation}
$A,B\in\{2,3\}$.
In order to see that in our setting we can indeed choose a suitable basis such that the commutators have the required form, we discuss in more detail initial data sets where the three-dimensional \riem\ \mf\ is a \liegr\ with left-invariant metric. This is done in Subsection~\ref{subsect_liegrinitialdata}, and it is here that we explain the terms 'exceptional' and '\ogon ' as well as the reason for excluding \liegr s of type~VI$_{{-}1/9}$.

The objects~$\tilde\gamma_{01}^1,\tilde\gamma_{0A}^B$ have to be understood as merely numbers, devoid of any geometric meaning. It is only a posteriori that we can interpret these numbers as structure constants of a suitable four-dimensional frame.
With these objects at our disposal, we can follow through with the transformation explained in Subsection~\ref{subsect_appendixdeducenormalisedvariablesbianchib}. We wish to point out however that we make use only of the algebraic relations, not their geometric interpretation.
This yields initial data $(\Sigma_+,\tilde\Sigma,\Delta,\tilde A,N_+)(0)$ for the evolution \equs\ in expansion-normalised variables, equations~\eqref{eqns_evolutionbianchib}--\eqref{eqn_evolutionomega}.

In order to avoid confusion and shorten notation, we are going to denote initial data to Einstein's field equation, \ie initial data as in Def.~\ref{defi_initialdatabianchib}, by \textit{geometric initial data}, and initial data to the evolution equations in expansion-normalised variables, equations~\eqref{eqns_evolutionbianchib}--\eqref{eqn_evolutionomega}, by \textit{dynamical initial data}.
The details of how to translate geometric initial data into dynamical initial data are given in the first part of the construction of the \mghd, Subsection~\ref{constr_initialdatatoexpnorm}.

\smallskip

Once geometric initial data is translated into the expansion-normalised variables setting, \ie into dynamical initial data, and we have obtained a solution in these variables, the main work lies in the construction of a global four-dimensional frame with structure constants behaving correctly over time, and such that our initially defined objects~$\tilde\gamma_{0i}^j$ are consistent with the geometric objects~$\gamma_{0i}^j$ at the starting time~$t=0$, \ie on the initial Cauchy hypersurface.
Finally, we use this four-dimensional frame to construct a spacetime metric. This is done in the second and third part of the construction, Subsections~\ref{constr_expnormtosolutionbasic} and~\ref{constr_solutionbasictomghd}.
Having obtained through this construction a spacetime into which the initial data is embedded in the correct way, we investigate its properties in Subsection~\ref{constr_propertiesdevelopmentinitialdata}.
We can show that this spacetime is in fact the \mghd\ of the geometric initial data.

\smallskip

In Subsection~\ref{constr_highersymmetrysolutions}, we then consider Bianchi spacetimes with additional symmetries, namely local rotational symmetry and plane wave equilibrium solutions. We compare their definitions from the point of view of geometric initial data with their definitions in expansion-normalised variables and show that these definitions coincide, respectively.

\smallskip

The construction of expansion-normalised variables by~\cite{hewittwainwright_dynamicalsystemsapproachbianchiorthogonalB} which we explain in Subsection~\ref{subsect_appendixdeducenormalisedvariablesbianchib} is related to the study of \ogon ly transitively~$G_2$ cosmologies. Even though~\cite{hewittwainwright_dynamicalsystemsapproachbianchiorthogonalB} appeal to results having been obtained in this context, our construction does not make use of these additional statements but is self-contained. We do nonetheless give several remarks explaining the relations, but these are logically independent from
the construction presented here.

\smallskip

In the Bianchi class~A setting, results equivalent to what we achieve in this section have been obtained by~\cite{ringstrom_bianchiixattr}, however with a somewhat different approach.

\subsection{Bianchi classification of three-dimensional \liegr s}
\label{subsect_appendixbianchiclassification}

We give a brief introduction to the classification of three-dimensional \liegr s proposed by Bianchi in 1903. The article~\cite{krasinskibehrschueckingestabrookwahlquistellisjantzenkundt_bianchiclass} gives a historical overview and provides insights into how our modern understanding of this classification came to be.
In~\cite{ellismaccallum_classofhomogcosmmodels}, the details are laid out, and for the first part of this subsection we refer to~\cite[App.~E]{ringstrom_topologyfuturestabilityuniverse} for the details.

For a three-dimen\-sional Lie group~$G$, let $\{e_i\}$, $i=1,2,3$, be a basis of the associated Lie algebra~$\ggg$, and let~$\gamma_{ij}^k$ denote the structure constants, \ie
\begin{equation}
	\liebr{e_i}{e_j}=\gamma_{ij}^ke_k.
\end{equation}
The equivalent information is encoded in the symmetric matrix~$n$ and the vector~$a$ given by
\begin{equation}
\label{eqn_data3dimliegr}
	n^{ij}=\frac12\gamma_{kl}^{\(i}\epsilon^{j\)kl},\quad a_k=\frac12\gamma_{ki}^i,
\end{equation}
with $\epsilon_{ijk}=\epsilon^{ijk}$ the permutation symbol which satisfies $\epsilon_{123}=1$ and is antisymmetric in all indices. The brackets~$^{()}$ denote symmetrisation, in fact $\gamma_{kl}^{(i}\epsilon^{j)kl}=(\gamma_{kl}^{i}\epsilon^{jkl}+\gamma_{kl}^{j}\epsilon^{ikl})/2$.
Equivalence follows from
\begin{equation}
\label{eqn_data3dimliegrback}
	\gamma_{jk}^i=\epsilon_{jkl}n^{li}+a_j\delta^i_k-a_k\delta^i_j.
\end{equation}
The structure constants have to satisfy the Jacobi identity, which is equivalent to the condition
\begin{equation}
\label{eqn_jacobiidentequiv}
	n a=0.
\end{equation}
Applying a suitable change of basis yields~$n$ and~$a$ in a specific simplified form and gives the following classification of three-dimensional \liegr s, of which we sketch a proof further down in Lemma~\ref{lemm_classthreedimliealg}. As there is a one-to-one correspondence between simply connected \liegr s and their \liealg s, see~\cite[Thm~3.28]{warner_founddiffmfsliegrs}, we formulate this classification in terms of (simply connected) \liegr s.
\begin{enumerate}
	\item In the case of simply connected unimodular Lie groups (Bianchi class~A) which are defined by $a=0$, one can choose~$n$ diagonal, \ie $n=\diag(\nu_1,\nu_2,\nu_3)$, such that it falls in exactly one of the categories given in Table~\ref{table_bianchiaclassification}.

	Given a left-invariant metric in~$G$, the basis~$e_1,e_2,e_3$ producing~$a$ and~$n$ of this form can be chosen \onorm.
	\begin{table}[htp]
		\centering
		\caption{Bianchi class~A}
	\label{table_bianchiaclassification}
		\begin{tabular}{l|c|c|c}
			Type & $\nu_1$ & $\nu_2$ & $\nu_3$ \\
			\hline
			I & 0 & 0 & 0 \\
			II & + & 0 & 0 \\
			VI$_0$ & 0 & + & -- \\
			VII$_0$ & 0 & + & + \\
			VIII & -- & + & + \\
			IX & + & + & +
		\end{tabular}
	\end{table}
	\item In the case of simply connected non-unimodular Lie groups (Bianchi class~B) which are defined by $a\not=0$, one can choose $a_1\not=0$, $a_2=a_3=0$, and~$n$ diagonal, \ie $n=\diag(\nu_1,\nu_2,\nu_3)$. The Jacobi identity then implies $\nu_1=0$. All possible types and their \resp names are listed in Table~\ref{table_bianchibclassification}.

	Again, given a left-invariant metric in~$G$, the basis~$e_1,e_2,e_3$ can be chosen \onorm.

	Lie groups of Bianchi type~VI and~VII have an additional degree of freedom which is captured in the quantity~$\binvparam$. In the chosen basis this parameter satisfies
	\begin{equation}
	\label{eqn_binvparamdefi}
		\binvparam \nu_2\nu_3=a_1^2
	\end{equation}
	and is invariant under scaling.
	We give a more geometric definition further down, see Lemma~\ref{lemm_invardefibparaminv}.
	\begin{table}[htp]
		\centering
		\caption{Bianchi class~B}
	\label{table_bianchibclassification}
		\begin{tabular}{l|c|c|c}
			Type & $\nu_1$ & $\nu_2$ & $\nu_3$ \\
			\hline
			V & 0 & 0 & 0 \\
			IV & 0 & 0 & + \\
			VI$_{\binvparam}$ & 0 & + & -- \\
			VII$_{\binvparam}$ & 0 & + & +
		\end{tabular}
	\end{table}
\end{enumerate}

In order to understand the parameter~$\binvparam$, we have to gather more detailed information on the structure of the \liealg s of the different Bianchi types.
\begin{lemm}
\label{lemm_abeliansubalgstructconstants}
	Let~$G$ a three-dimensional \liegr, and~$\ggg$ the associated \liealg. Then~$\ggg$ admits a two-dimensional Abelian subalgebra \iif~$G$ is not of Bianchi type~VIII or~IX.
\end{lemm}
\begin{proof}
	Choose a basis $e_1,e_2,e_3$ of the \liealg\ such that~$n$ and~$a$ are of one of the forms given in i) and ii). If the \liegr\ is not of Bianchi type~II, VIII or~IX, then~$\nu_1=0$. One computes that
	\begin{equation}
		\liebr{e_2}{e_3}=\gamma_{23}^ie_i=(\epsilon_{23l}n^{li}+a_2\delta^i_3-a_3\delta^i_2)e_i
			=\epsilon_{231}n^{1i}e_i=\nu_1e_1=0,
	\end{equation}
	which implies that $e_2,e_3$ span an Abelian subalgebra. For a \liegr\ of Bianchi type~II one finds~$\nu_2=0$, and~$e_1,e_3$ commute.

	Assume now that the \liegr\ is of type~VIII or~IX and suppose that there are two linearly independent vectors~$b,c\in\ggg$ which commute, \ie
	\begin{equation}
		b=\sum_{i=1}^3b_ie_i,\qquad c=\sum_{i=1}^3c_ie_i,
	\end{equation}
	for $b_i,c_i\in\RR$, and
	\begin{equation}
		\liebr bc=0.
	\end{equation}
	This implies
	\begin{equation}
		0=\liebr bc=\liebr{\sum_{i=1}^3b_ie_i}{\sum_{j=1}^3c_je_j}=\sum_{i,j}b_ic_j\liebr{e_i}{e_j}=\sum_{i,j}b_ic_j\gamma_{ij}^ke_k,
	\end{equation}
	hence for all $k=1,2,3$
	\begin{equation}
	\label{eqn_abeliansubgroupcondition}
		0=\sum_{i,j}b_ic_j\gamma_{ij}^k=\sum_{i<j}(b_ic_j-b_jc_i)\gamma_{ij}^k.
	\end{equation}
	As the two Bianchi types in question are of class~A, one finds from equation~\eqref{eqn_data3dimliegrback} that
	\begin{equation}
		\gamma_{12}^k=\delta_{3k}\nu_3,\qquad
		\gamma_{13}^k=-\delta_{2k}\nu_2,\qquad
		\gamma_{23}^k=\delta_{1k}\nu_1.
	\end{equation}
	All~$\nu_i$, $i=1,2,3$, are non-vanishing, and equation~\eqref{eqn_abeliansubgroupcondition} is therefore equivalent to the system of equations
	\begin{equation}
		b_2c_3-b_3c_2=0,\qquad
		b_1c_3-b_3c_1=0,\qquad
		b_1c_2-b_2c_1=0,
	\end{equation}
	or in other words $b\times c=0$, where~$\times$ denotes the cross product. Hence,~$b$ and~$c$ are parallel, a contradiction.
\end{proof}
In case the \liegr\ is of class~B, there is also a geometric way of defining the Abelian subalgebra. For this, we consider the adjoint~$\ad:\ggg\rightarrow\operatorname{End}(\ggg)$, $\ad_xy=\liebr xy$, see~\cite[Prop.~3.47]{warner_founddiffmfsliegrs}.
\begin{lemm}
\label{lemm_abeliansubalginvardefi}
	Let~$G$ a three-dimensional \liegr\ of class~B and~$H$ the one-form 
	\begin{equation}
		H:\ggg\rightarrow\RR, \quad x\mapsto \frac12\tr(\ad_x).
	\end{equation}
	Denote by~$\ggg_2$ the kernel of~$H$.
	Then~$\ggg_2$ is an Abelian subalgebra and coincides with the subalgebra identified in Lemma~\ref{lemm_abeliansubalgstructconstants}.
\end{lemm}
\begin{proof}
	We choose a basis~$e_1,e_2,e_3$ of~$\ggg$ and denote by~$e^k$, $k=1,2,3$, the dual basis. Setting~$a_k$, $k=1,2,3$, as in~\eqref{eqn_data3dimliegr}, we find
	\begin{equation}
		H = a_ke^k.
	\end{equation}
	As the \liegr\ is of class~B, we can choose the basis~$e_1,e_2,e_3$ as above in ii), meaning that~$a_1\not=0$, $a_2=a_3=0$, and $n$ is diagonal, \ie $n=\diag(0,\nu_2,\nu_3)$. With this, we find that the kernel~$\ggg_2$ is spanned by the two basis elements~$e_2$ and~$e_3$.
	Comparison with the proof of Lemma~\ref{lemm_abeliansubalgstructconstants} shows that it is the same subalgebra we identified there.
\end{proof}
\begin{rema}
	Given a class~B \liegr~$G$ with a left-invariant metric, we obtain the following useful interpretation of the kernel~$\ggg_2$: it provides a splitting of the \liealg~$\ggg_2$ into a two-dimensional Abelian subalgebra and the direction \ogon\ to it. For any \onorm\ basis~$e_1,e_2,e_3$ of~$\ggg$ such that~$e_2,e_3$ span~$\ggg_2$, the basis element~$e_1$ is the one with the non-vanishing~$a_i$.
\end{rema}

\begin{rema}
\label{rema_frameforggg2rotationtodiagonalise}
	Given a \liegr\ of class~B with \liealg~$\ggg$ and a left-invariant metric on~$G$, we have at our disposal the uniquely defined Abelian subalgebra~$\ggg_2$. If instead the \liegr\ is of class~A but not of type VIII or IX, we can fix an Abelian subalgebra, which exists due to Lemma~\ref{lemm_abeliansubalgstructconstants}. In both cases, we can then introduce an \onorm\ basis $e_1,e_2,e_3$ of~$\ggg$ such that~$e_2,e_3$ span the Abelian subalgebra and~$e_1$ is \ogon\ to this span. Given~$\ggg_2$, this basis is uniquely defined up to a rotation in the $e_2e_3$-plane and a choice of sign in~$e_1$.

	One easily sees that it is equivalent to say that one uses the choice of \onorm\ basis from~i) and~ii), and then allows for a rotation in the~$e_2e_3$-plane. This holds in all cases apart from the case of Bianchi type~II, where it is necessary to first switch the basis elements~$e_1$ and~$e_2$. In the later subsections, in particular for the deduction of the different variables in Subsection~\ref{subsect_appendixdeducenormalisedvariablesbianchib}, this rotation in the~$e_2e_3$-plane is left as a gauge freedom.
\end{rema}

With this knowledge, we are now in a position to give an invariant definition of the parameter~$\binvparam$ appearing in the \liegr s of type~VI$_\binvparam$ and VII$_\binvparam$, which are of class~B.
\begin{lemm}
\label{lemm_invardefibparaminv}
	Let~$G$ a three-dimensional \liegr\ of type~VI$_\binvparam$ or VII$_\binvparam$, $\binvparam\not=0$ in either case, with associated \liealg~$\ggg$, and~$\ggg_2$ the uniquely defined Abelian subalgebra which is the kernel of the one-form~$H$.
	Let~$v_1\in\ggg\setminus\ggg_2$. Then
	\begin{equation}
		A_2\coloneqq \ad_{v_1} |_{\ggg_2}:\ggg_2\rightarrow\ggg_2,
	\end{equation}
	and a different choice~$\hat v_1\in\ggg\setminus\ggg_2$ results into a map~$\hat A_2$ differing from~$A_2$ by a constant non-zero multiple.
	Further, setting
	\begin{equation}
		\binvparam\coloneqq \frac{(\tr A_2)^2}{4\det A_2-(\tr A_2)^2}
	\end{equation}
	is well-defined independently of the choice of~$v_1$ and consistent with equation~\eqref{eqn_binvparamdefi}, \ie under a choice of basis as in ii) one finds
	\begin{equation}
		\binvparam\nu_2\nu_3=a_1^2.
	\end{equation}
\end{lemm}
\begin{proof}
	Consider a basis~$e_1,e_2,e_3$ as in~ii), \ie satisfying~$a_1\not=0$, $a_2=a_3=0$, and $n=\diag(0,\nu_2,\nu_3)$. By Remark~\ref{rema_frameforggg2rotationtodiagonalise}, the Abelian subalgebra~$\ggg_2$ is spanned by~$e_2,e_3$. Consequently, we can express the given unit vector~$v_1$ as~$v_1=b_1e_1+\sum_Ab_Ae_A$ and find
	\begin{equation}
		A_2(e_B)=\liebr{b_1e_1+\sum_Ab_Ae_A}{e_B}=b_1\liebr{e_1}{e_B}=b_1\gamma_{1B}^Ce_C\quad\in\ggg_2,
	\end{equation}
	as~$\liebr{e_2}{e_3}=0$ and~$\gamma_{1A}^1=0$. This proves the first statement.

	For the second statement, we realise that~$v_1,e_2,e_3$ also forms a basis of~$\ggg$, hence~$\hat v_1=c_1v_1+\sum_Ac_Ae_A$ with~$c_1\not=0$, and
	\begin{equation}
		(A_2-\hat A_2)(e_B)=\liebr{e_1-c_1e_1}{e_B}=(1-c_1)A_2(e_B).
	\end{equation}

	From the last computation, it even follows immediately that~$\binvparam$ is independent of the choice of~$v_1$, and we can for all purposes assume that~$v_1=e_1$.
	Doing so, the linear mapping~$A_2$ in the chosen basis~$e_2,e_3$ is described by the matrix
	\begin{equation}
		\begin{pmatrix}
			a_1 & \nu_3\\
			-\nu_2 & a_1
		\end{pmatrix},
	\end{equation}
	which implies
	\begin{equation}
		\tr A_2=2a_1,\qquad \det A_2=a_1^2+\nu_2\nu_3,
	\end{equation}
	and yields
	\begin{equation}
		\binvparam=\frac{a_1^2}{\nu_2\nu_3}.
	\end{equation}
	This is equivalent to the requested relation and in addition shows that~$\binvparam$ is well-defined, as both~$\nu_2$ and~$\nu_3$ are non-vanishing for the Bianchi types in question.
\end{proof}
We conclude this subsection by proving that the types we listed in Tables~\ref{table_bianchiaclassification} and~\ref{table_bianchibclassification}, together with the parameter~$\binvparam$, indeed provide a classification of all three-dimensional \liealg s.
\begin{lemm}
\label{lemm_classthreedimliealg}
	Two three-dimensional \liealg s are isomorpic \iif\ they have the same Bianchi type, and in case of Bianchi type~VI or~VII additionally the same quantity~$\binvparam$.
\end{lemm}
\begin{proof}
	Given a \liealg\ of either class, the matrix~$n$ and the vector~$a$ defined in equations~\eqref{eqn_data3dimliegr} admit the form~$n=\diag(\nu_1,\nu_2,\nu_3)$ and~$a=(a_1,0,0)$ after applying a suitable change of basis.
	As a consequence, the \liealg\ falls in one of the types defined in Tables~\ref{table_bianchiaclassification} and~\ref{table_bianchibclassification}.
	From the transformation behaviour of~$n$ and~$a$ under a change of basis, see~\cite[eq.~(E.4)]{ringstrom_topologyfuturestabilityuniverse}, we see that the number of non-vanishing diagonal elements of~$\nu$ as well as the number of diagonal elements of~$\nu$ having the same sign is fixed. As a conclusion, the given \liealg\ cannot fall in two different of the types from Tables~\ref{table_bianchiaclassification} and~\ref{table_bianchibclassification}.
	For \liealg s of class~A, scaling the basis such that~$\nu_i\in\{{-}1,0,1\}$, $i=1,2,3$, shows that each of the types in Table~\ref{table_bianchiaclassification} has one unique representative, which concludes the proof for this class. For a \liealg\ of type V or IV, one can apply scaling to achieve~$a_1=1$, $\nu_3\in\{0,1\}$, and thereby uniqueness. In case of a \liealg\ of type~VI or~VII, the quantities~$a_1$ and~$\nu_2,\nu_3$ cannot be scaled independently of one another but have to satisfy equation~\eqref{eqn_binvparamdefi}. Requesting~$\nu_2=1$ and~$\nu_3\in\{\pm1\}$ fixes~$a_1$ up to sign, which by a change of direction in~$e_1$ can be set to be positive. This concludes the proof.
\end{proof}

\subsection{The expansion-normalised variables from a four-dimensional spacetime point of view}
\label{subsect_appendixdeducenormalisedvariablesbianchib}

For this subsection and this subsection alone, we assume the existence of a four-dimensional spacetime from the start. That is, we assume that we already have a four-dimensional solution to Einstein's equations at our disposal, not only geometric initial data consisting of information on a three-dimensional \mf. With the four-dimensional information given, we recall the deduction of the expansion-normalised variables $(\Sigma_+,\tilde\Sigma,\Delta,\tilde A,N_+)$ from given structure constants~$\gamma_{\alpha\beta}^\delta$, as it was developed in~\cite{hewittwainwright_dynamicalsystemsapproachbianchiorthogonalB}.
In the following subsections, we then connect this to the initial data perspective we started with in the beginning of this paper.

We restrict ourselves to spacetimes with a Bianchi symmetry, \ie with a three-dimen\-sional \liegr~$G$ acting on the spacetime.
We assume a stress-energy tensor which is either that of vacuum or that of a perfect fluid, and additionally assume that the fluid velocity~$u$ is \ogon\ to the group orbits of~$G$.

We further assume that we are given an \onorm\ frame $(e_0,e_1,e_2,e_3)$ such that the only non-zero structure constants are
\begin{equation}
\label{eqn_structureconstantsnonzero}
	\gamma_{01}^1,\quad \gamma_{10}^1,\quad \gamma_{0A}^B,\quad\gamma_{A0}^B,\quad\gamma_{1A}^B,\quad
	\gamma_{A1}^B,
\end{equation}
with $A,B=2,3$, which is the setting considered in~\cite{hewittwainwright_dynamicalsystemsapproachbianchiorthogonalB}.
\begin{rema}
	In~\cite{hewittwainwright_dynamicalsystemsapproachbianchiorthogonalB}, the property that the only non-vanishing structure constants are those in~\eqref{eqn_structureconstantsnonzero} is justified by building upon the study of orthogonally transitive~$G_2$ cosmologies, which are a generalisation of \ogon\ Bianchi cosmologies. In~\cite{wainwright_classschemenonrotinhomogcosm}, a suitably adapted \onorm\ frame $(e_0,e_1,e_2,e_3)$ of a~$G_2$ spacetime is introduced, such that the vector fields $e_2,e_3$ are tangential to the group orbits of the two-dimensional symmetry group~$G_2$, the vector field~$e_1$ is spacelike, and the vector field~$e_0$ is aligned with the fluid velocity.
	The study of orthogonally transitive~$G_2$ cosmologies is continued in~\cite{hewittwainwright_orthogtransg2cosm} where the authors, following a construction proposed by~\cite{maccallum_cosmmodelsgeometricpointofview}, decompose the structure constants and deduce the corresponding evolution \equs\ equivalent to the Einstein field \equs\ of an \ogon\ perfect fluid.
	This decomposition is explicitly adapted to the setting of Bianchi~B orthogonal perfect fluid models by~\cite{hewittwainwright_dynamicalsystemsapproachbianchiorthogonalB}.

	We do not make use of these arguments but instead treat the non-vanishing of all structure constants other than those in~\eqref{eqn_structureconstantsnonzero} as an assumption. In the following subsection, we discuss in more detail initial data sets where the three-dimensional \mf\ is a \liegr. We notice that by excluding one specific Bianchi type, we can ensure that we end up in a setting where the structure constants have the requested form, see Lemma~\ref{lemm_exceptionalbianchiasinitialdata} and the beginning of Subsection~\ref{constr_initialdatatoexpnorm}.
\end{rema}

With structure constants as in~\eqref{eqn_structureconstantsnonzero} at hand, we set
\begin{equation}
\label{eqns_prebasic_definitionpartone}
\begin{subaligned}
	\sigma_{22}={}&{-}\gamma_{02}^2+\frac13(\gamma_{01}^1+\gamma_{02}^2+\gamma_{03}^3), &
	n_{22}={}&{-}\gamma_{13}^2 ,\\
	\sigma_{23}={}&\frac12(-\gamma_{02}^3-\gamma_{03}^2)=\sigma_{32}, &
	n_{23}={}&\frac12\gamma_{12}^2-\frac12\gamma_{13}^3=n_{32} ,\\
	\sigma_{33}={}&{-}\gamma_{03}^3+\frac13(\gamma_{01}^1+\gamma_{02}^2+\gamma_{03}^3), \qquad&
	n_{33}={}&\gamma_{12}^3,\\
	\theta={}&{-}\gamma_{01}^1-\gamma_{02}^2-\gamma_{03}^3, &
	a_1={}&\frac12(\gamma_{12}^2+\gamma_{13}^3) .\\
	\Omega_1={}&\frac12(\gamma_{02}^3-\gamma_{03}^2), &&
\end{subaligned}
\end{equation}
This gives a decomposition of the structure constants as proposed for general cosmological models by~\cite{maccallum_cosmmodelsgeometricpointofview} and deduced in the present setting by~\cite{hewittwainwright_dynamicalsystemsapproachbianchiorthogonalB}:
\begin{equation}
\label{eqns_commutators}
\begin{subaligned}
	\liebr{e_0}{e_1}={}&(\sigma_A{}^A-\frac13\theta)e_1,\\
	\liebr{e_0}{e_A}={}&({-}\sigma_A{}^B-\frac13\theta\delta_A^B+\Omega_1\epsilon_A{}^B)e_B,\\
	\liebr{e_1}{e_A}={}&(\epsilon_{AB}n^{BC}+a_1\delta_A^C)e_C,\\
	\liebr{e_A}{e_B}={}&0,
\end{subaligned}
\end{equation}
with $A,B\in\{2,3\}$, $\epsilon_{AB}$ the two-dimensional permutation symbol, and~$\delta_{AB}$ the Kronecker delta which is also used for lifting and lowering indices.
The non-vanishing structure constants~$\gamma_{\alpha\beta}^\delta$, see~\eqref{eqn_structureconstantsnonzero}, of the frame elements on the one hand and the set of variables $(\theta,\sigma_{AB},\Omega_1,a_1,n_{AB})$ as given in~\eqref{eqns_commutators} on the other hand encode the same information.

The choice of \onorm\ frame $(e_0,e_1,e_2,e_3)$ is not unique but allows for a rotation in the $e_2e_3$-plane. In the classification of \liegr s in Subsection~\ref{subsect_appendixbianchiclassification}
we have chosen a specific frame which additionally diagonalises~$n_{AB}$, but for the remainder of this section, we leave this rotation as a gauge freedom. Having this freedom of rotation in choosing the frame however becomes an issue when constructing the \mghd\ in Subsections~\ref{constr_initialdatatoexpnorm}--\ref{constr_solutionbasictomghd}.

\begin{rema}
\label{rema_geometricmeaningthetasigmaAB}
	The quantities appearing in the left column of~\eqref{eqns_prebasic_definitionpartone} can be interpreted geometrically. In order to do so, assume that the spacetime is~$I\times G$, with~$I$ an open interval described by a time parameter~$t$, and the frame is such that~$e_0=\partial_t$, and~$e_i$, $i=1,2,3$, are tangential to every~$\{t\}\times G$ and invariant under the action of~$G$.

	The decomposition of structure constants in a spacetime was originally proposed by~\cite{maccallum_cosmmodelsgeometricpointofview}.
	Upon comparison, one notices that the quantity~$\theta_{ij}$ given there (but in our index convention) equals the second fundamental form of $\{t\}\times G$ in the spacetime.
	Consequently, the trace~$\theta$ of~$\theta_{ij}$ equals the mean curvature of~$\{t\}\times G$, and
	\begin{equation}
		\sigma_{ij}=\theta_{ij}-\frac{\theta}{3}h_{ij}
	\end{equation}
	is the trace-free part of the second \fundform, where~$h$ is the spacetime metric restricted to~$\{t\}\times G$.

	We compute
	\begin{equation}
		\gamma_{0i}^i=\scalprod{\liebr{e_0}{e_i}}{e_i}=\scalprod{\nabla_{e_0}e_i-\nabla_{e_i}e_0}{e_i}={-}\theta_{ii},
	\end{equation}
	where~$\nabla$ denotes the four-dimensional \levi\ connection corresponding to the metric on~$I\times G$, which we here denote by~$\scalprod{\cdot}{\cdot}$.
	As~$\sigma_{ij}$ is trace free, we find
	\begin{equation}
		0=\sigma_{11}+\sigma_{22}+\sigma_{33}=\theta_{11}-\frac\theta3+\sigma_{22}+\sigma_{33},
	\end{equation}
	or equivalently
	\begin{equation}
		{-}\theta_{11}=\sigma_A{}^A-\frac13\theta,
	\end{equation}
	which should be compared to the first equation in~\eqref{eqns_commutators}.

	The quantity~$\Omega_1$ can be expressed in terms of the four-dimensional frame as follows:
	\begin{equation}
		\Omega_1=\scalprod{e_3}{\nabla_{e_0}e_2},
	\end{equation}
	where~$\nabla$ denotes the four-dimensional \levi\ connection. The quantity~$\Omega_1$ thus describes a certain timelike derivative, but not one which is encoded in the second fundamental form~$\theta_{ij}$.
\end{rema}

Following the deduction of the expansion-normalised variables $(\Sigma_+,\tilde\Sigma,\Delta,\tilde A,N_+)$ as given in~\cite{hewittwainwright_dynamicalsystemsapproachbianchiorthogonalB}, one in a first step replaces~$n_{AB}$ and~$\sigma_{AB}$, $A,B\in\{2,3\}$, by their trace and trace-free part (\wrt\ the two-dimensional trace in the $23$-components), \ie one defines
\begin{equation}
\label{eqns_prebasic_definitionparttwo}
\begin{subaligned}
	\sigma_+={}&\frac32\delta^{AB}\sigma_{AB}, &
	n_+={}&\frac32\delta^{AB}n_{AB} ,\\
	\tilde\sigma_{AB}={}&\sigma_{AB}-\frac13\sigma_+\delta_{AB}, \qquad &
	\tilde n_{AB}={}&n_{AB}-\frac13n_+\delta_{AB}.
\end{subaligned}
\end{equation}
The information encoded in these new variables is equivalent to that encoded in~$n_{AB}$, $\sigma_{AB}$. Additionally, one sets
\begin{equation}
\label{eqns_basic_definitionpartone}
	{}^*\tilde\sigma_{AB}={}\tilde\sigma_A{}^C\epsilon_{BC}, \qquad
	{}^*\tilde n_{AB}={}\tilde n_A{}^C\epsilon_{BC}.
\end{equation}

Assuming that the only non-vanishing structure constants are those given in~\eqref{eqn_structureconstantsnonzero} which are decomposed as above, and further assuming a stress-energy tensor of an \ogon\ perfect fluid~\eqref{eqn_stressenergyperfectfluid} with linear equation of state~\eqref{eqn_lineareqnofstate}, the evolution \equs\ for the variables $(\theta,\sigma_+,\tilde\sigma_{AB},a_1,n_+,\tilde n_{AB})$ have been given explicitly in~\cite[App.~A]{hewittwainwright_dynamicalsystemsapproachbianchiorthogonalB}:
\begin{equation}
\label{eqns_prebasic_evolutionbianchib}
\begin{subaligned}	\dot\theta={}&{-}\frac13\theta^2-\frac23\sigma_+^2-\tilde\sigma^{AB}\tilde\sigma_{AB}-\frac12(3\gamma-2)\mu,\\
	\dot\sigma_+={}&{-}\theta\sigma_+-\tilde n^{AB}\tilde n_{AB},\\
	\dot{\tilde\sigma}_{AB}={}&{-}\theta\tilde\sigma_{AB}+2\Omega_1{}^*\tilde\sigma_{AB}-\frac23n_+\tilde n_{AB}+2a_1{}^*\tilde n_{AB},\\
	\dot a_1={}&\frac13(2\sigma_+-\theta)a_1,\\
	\dot n_+={}&\frac13(2\sigma_+-\theta)n_++3\tilde\sigma^{AB}\tilde n_{AB},\\
	\dot{\tilde n}_{AB}={}&\frac13(2\sigma_+-\theta)\tilde n_{AB}+2\Omega_1{}^*\tilde n_{AB}+\frac23n_+\tilde\sigma_{AB},
\end{subaligned}
\end{equation}
with constraint \equ
\begin{equation}
\label{eqn_prebasic_constraint}
	3a_1\sigma_+-\frac32{}^*\tilde\sigma^{AB}\tilde n_{AB}=0,
\end{equation}
\noeqref{eqn_prebasic_constraint}
defining equation for~$\mu$
\begin{equation}
\label{eqn_prebasic_mu}
	\mu=\frac13\theta^2-\frac13\sigma_+^2-\frac12\tilde\sigma^{AB}\tilde\sigma_{AB}-\frac12\tilde n^{AB}\tilde n_{AB}-3a_1^2,
\end{equation}
and auxiliary \equ
\begin{equation}
\label{eqn_prebasic_evolutionmu}
	\dot\mu={-}\gamma\mu\theta.
\end{equation}
\begin{rema}
\label{rema_equivalenceeinsteinperfectfluidprebasic}
	Consider a four-dimensional spacetime with an \onorm\ frame~$(e_0,e_1,e_2,e_3)$.
	Under the assumption that the only non-vanishing structure constants~$\gamma_{\alpha\beta}^\delta$ are those given in~\eqref{eqn_structureconstantsnonzero} and they additionally satisfy
	\begin{equation}
		e_i(\gamma_{\alpha\beta}^\delta)=0,
	\end{equation}
	the system of evolution equations with constraints~\eqref{eqns_prebasic_evolutionbianchib}--\eqref{eqn_prebasic_evolutionmu}
	holds \iif\
	the Ricci \curv\ of the spacetime is invariant under a three-dimensional \liegr\ action in the~$e_1e_2e_3$-space such that in the chosen basis the only non-vanishing terms are
	\begin{equation}
		Ric(e_0,e_0)=\frac12(3\gamma-2)\mu,\qquad Ric(e_i,e_i)=\frac12(2-\gamma)\mu,
	\end{equation}
	$i=1,2,3$, and the Jacobi identities hold in the $e_1e_2e_3$-space. This follows from direct computation,
	expressing the Ricci \curv\ in terms of the structure constants~$\gamma_{\alpha\beta}^\delta$ and then replacing them by the variables $(\theta,\sigma_+,\tilde\sigma_{AB},a_1,n_+,\tilde n_{AB})$.
	A Ricci \curv\ of this form is equivalent to Einstein's equations of a perfect fluid~\eqref{eqn_stressenergyperfectfluid} with linear equation of state~\eqref{eqn_lineareqnofstate}.

	In detail, one finds that the evolution equations for the variables~$\theta$, $\sigma_+$ and~$\tilde\sigma_{AB}$ correspond to the Ricci \curv\ terms~$Ric(e_i,e_j)$ with~$i\not=0\not=j$, while the evolution of~$a_1$, $n_+$ and~$\tilde n_{AB}$ follows from the Jacobi identities on the \liegr. The constraint equation~\eqref{eqn_prebasic_constraint} is equivalent to the momentum constraint, and equation~\eqref{eqn_prebasic_mu} defining~$\mu$ is equivalent to the Hamiltonian constraint. The evolution of~$\mu$ is equivalent to the matter equation of a perfect fluid with linear equation of state.
\end{rema}

In the next step of the construction, one introduces the variables
\begin{equation}
\label{eqns_basic_definitionparttwo}
\begin{subaligned}
	\tilde\sigma=\frac32\tilde\sigma^{AB}\tilde\sigma_{AB}, \qquad&\tilde n={}\frac32\tilde n^{AB}\tilde n_{AB},\\
	\delta=\frac32\tilde\sigma^{AB}\tilde n_{AB}, \qquad& ^*\delta={}\frac32\,^*\tilde\sigma^{AB}\tilde n_{AB},
\end{subaligned}
\end{equation}
which are invariant under the freedom of rotation and satisfy $\tilde\sigma\tilde n=\delta^2+\,^*\delta^2$.
The constraint equation~\eqref{eqn_prebasic_constraint}
implies $^*\delta=3a_1\sigma_+$, which can be used to eliminate~$^*\delta$ in the following.
One further finds that there is a constant~${\bparamk}$ such that
\begin{equation}
\label{eqn_basic_tildenrelation}
	\tilde n=\frac13(n_+^2-9{\bparamk}a_1^2).
\end{equation}
In the case of Bianchi type~VI and~VII, this relation immediately follows from the definition of the Bianchi group parameter~$\binvparam$ in equation~\eqref{eqn_binvparamdefi}, and the constant satisfies~${\bparamk}=1/{\binvparam}$. In all other Bianchi cases, one computes directly from the form of~$n_{AB}$ and~$a_1$ that equation~\eqref{eqn_basic_tildenrelation} holds for~${\bparamk}=0$.
One further sets
\begin{equation}
\label{eqn_basic_definitionpartthree}
	\tilde a=9a_1^2
\end{equation}
for simplification. The resulting variables
$(\theta ,\sigma_+ ,\tilde{\sigma},\delta ,\tilde{a},n_+)$
are called the \textit{basic variables} and evolve according to
\begin{equation}
\label{eqns_basic_evolutionbianchib}
\begin{subaligned}
	\dot\theta={}&{-}\frac13\theta^2-\frac23(\sigma_+^2+\tilde\sigma)-\frac12(3\gamma-2)\mu,\\
	\dot\sigma_+={}&{-}\theta\sigma_+-\frac23\tilde n,\\
	\dot{\tilde\sigma}={}&{-}2\theta\tilde\sigma-\frac43n_+\delta-\frac43\tilde a\sigma_+,\\
	\dot\delta={}&\frac23(\sigma_+-2\theta)\delta+\frac23n_+(\tilde\sigma-\tilde n),\\
	\dot{\tilde a}={}&\frac23(2\sigma_+-\theta)\tilde a,\\
	\dot n_+={}&\frac13(2\sigma_+-\theta)n_++2\delta.
\end{subaligned}
\end{equation}
They satisfy
\begin{equation}
\label{eqn_basic_constraintnonnegative}
	\tilde\sigma\ge0, \qquad \tilde a\ge0, \qquad\tilde n\ge0, \qquad \sigma_+^2+\tilde\sigma+\tilde a+\tilde n\le \theta^2,
\end{equation}
and the constraint
\begin{equation}
\label{eqn_basic_constraint}
	\tilde\sigma \tilde n -\delta^2-\tilde a\sigma_+^2=0,
\end{equation}
where
\begin{equation}
\label{eqn_basic_definitiontilden}
	\tilde n=\frac13(n_+^2-{\bparamk}\tilde a).
\end{equation}
The matter $\mu$ is given by
\begin{equation}
\label{eqn_basic_mu}
	\mu=\frac13(\theta^2-\sigma_+^2-\tilde\sigma-\tilde a-\tilde n),
\end{equation}
and the auxiliary equation is satisfied
\begin{equation}
\label{eqn_basic_evolutionmu}
	\dot\mu={-}\gamma\theta\mu.
\end{equation}
\begin{rema}
	One sees that the evolution equations~\eqref{eqns_prebasic_evolutionbianchib}--\eqref{eqn_prebasic_evolutionmu} imply the evolution equations~\eqref{eqns_basic_evolutionbianchib}--\eqref{eqn_basic_evolutionmu}, provided relation~\eqref{eqn_basic_definitiontilden} holds and the matter~$\mu$ in the first set of equations is assumed to be non-negative.

	Note also that the relations in~\eqref{eqn_basic_constraintnonnegative} and the constraint equation~\eqref{eqn_basic_constraint} are preserved under the evolution equations~\eqref{eqns_basic_evolutionbianchib}. This follows from similar arguments as for the expansion-normalised variables, see Remark~\ref{rema_expnormvarconstraintspreserved}.
\end{rema}

\begin{lemm}
	\label{lemm_excludethetazerobasic}
	Consider a solution to equations~\eqref{eqns_basic_evolutionbianchib}--\eqref{eqn_basic_evolutionmu}. Then either the solution is the trivial solution
	\begin{equation}
		\theta=\sigma_+=\tilde{\sigma}=\delta=\tilde{a}=n_+=\mu=0,
	\end{equation}
	or $\theta\not=0$ at all times.
\end{lemm}
\begin{proof}
	Using non-negativity of~$\mu$ in equation~\eqref{eqn_basic_mu} shows that if~$\theta=0$ at one time then
	\begin{equation}
		\sigma_+=0,\quad \tilde\sigma=0,\quad \tilde a=0,\quad \tilde n=0,\quad \mu=0
	\end{equation}
	at this time, and from equations~\eqref{eqn_basic_definitiontilden} and~\eqref{eqn_basic_constraint}, even~$n_+$ and~$\delta$ vanish at this time.
	The evolution equations~\eqref{eqns_basic_evolutionbianchib} imply that all variables vanish at all times.
\end{proof}

In the final step, the basic variables are normalised with appropriate powers of the rate of expansion scalar~$\theta$. By excluding the trivial solution in basic variables, we ensure that $\theta(t)\not=0$ at all times~$t$, see Lemma~\ref{lemm_excludethetazerobasic}. Setting then
\begin{equation}
\label{eqn_expansionnormalisation}
	\Sigma_+=\frac{\sigma_+}{\theta},\quad \tilde\Sigma=\frac{\tilde\sigma}{\theta^2},\quad
	\Delta=\frac{\delta}{\theta^2},\quad \tilde A=\frac{\tilde a}{\theta^2},\quad
	N_+=\frac{n_+}{\theta},\quad
	\tilde N=\frac{\tilde n}{\theta^2},
\end{equation}
and replacing the energy density~$\mu$ by the density parameter~$\Omega$ defined by
\begin{equation}
\label{eqn_definitionomega}
	\Omega=\frac{3\mu}{\theta^2}
\end{equation}
yields the \textit{expansion-normalised and dimensionless variables}
$(\Sigma_+,\tilde\Sigma,\Delta,\tilde A,N_+)$. After changing to the dimensionless time~$\tau$ which satisfies
\begin{equation}
\label{eqn_evolutiontau}
	\frac{dt}{d\tau}=\frac{3}{\theta}
\end{equation}
and some arbitrarily chosen initial condition,
they evolve according to \equs ~\eqref{eqns_evolutionbianchib}--\eqref{eqn_evolutionomega}. The deceleration parameter~$q$ appearing in these evolution equations is related to the basic variable~$\theta$ via
\begin{equation}
\label{eqn_evolutiontheta}
	\theta'={-}(1+q)\theta.
\end{equation}
\begin{rema}
\label{rema_equivalencebasicvarexpnormvar}
	It follows by direct computation that the evolution in basic variables, \equs~\eqref{eqns_basic_evolutionbianchib}--\eqref{eqn_basic_evolutionmu}, on the one hand and the evolution in expansion-normalised variables, equations \eqref{eqns_evolutionbianchib}--\eqref{eqn_evolutionomega}, together with equation~\eqref{eqn_evolutiontheta} on the other hand are equivalent to one another, provided~$\theta\not=0$, the variables are related as in~\eqref{eqn_expansionnormalisation}--\eqref{eqn_definitionomega},
	and the time~$t$ is related to the dimensionless time~$\tau$ as in equation~\eqref{eqn_evolutiontau}.
\end{rema}

\subsection{\liegr s as initial data subsets}
\label{subsect_liegrinitialdata}

	In this subsection, we consider initial data sets where the three-dimensional \mf\ is a \liegr, and the initial metric and \fundform\ are invariant under the group action.
	In addition to the purely three-dimensional \liegr\ properties which we have discussed in Subsection~\ref{subsect_appendixbianchiclassification}, the Hamilton and momentum constraint equations pose restrictions on the two-tensor~$\idfundform$ which we now investigate in detail.
	In the following subsections, these restrictions are used to connect geometric initial data to expansion-normalised variables and the construction of the \mghd.

	We start with initial data $(G,{\idmetric},{\idfundform},\mu_0)$ as in Def.~\ref{defi_initialdatabianchib}:
	A \liegr~$G$ of class~B or
	of type~I or~II,
	a left-invariant \riem\ metric~${\idmetric}$ on~$G$, a left-invariant symmetric covariant two-tensor~${\idfundform}$ on~$G$, and a constant $\mu_0\ge0$, satisfying the constraint equations
	\begin{align}
		\overline R-{\idfundform}_{ij}{\idfundform}^{ij}+(\tr_{\idmetric}{\idfundform})^2={}&2\mu_0,\\
		\overline\nabla_i\tr_{\idmetric}{\idfundform}-\overline\nabla^j{\idfundform}_{ij}={}&0.
	\end{align}
	In a first step, we fix a two-dimensional Abelian subalgebra~$\ggg_2$ of the \liealg~$\ggg$ corresponding to the \liegr~$G$. In case of a group of Bianchi class~B, this is the uniquely defined subalgebra which is the kernel of the one-form~$H$, see Lemma~\ref{lemm_abeliansubalginvardefi}. In case of a group of Bianchi type~I or~II, the existence of such a subalgebra is ensured by Lemma~\ref{lemm_abeliansubalgstructconstants}. We then choose an orthonormal basis $\tilde e_2,\tilde e_3$ of this subalgebra.
	Having fixed~$\ggg_2$, these two vectors can be chosen uniquely up to rotation and reflection.
	Using the given initial metric~${\idmetric}$ to fix a third unit vector~$\tilde e_1$ in the \liealg\ orthogonal to the span of those two yields a basis of~$\ggg$.
	This choice of basis gives
	structure constants~$\gamma_{ij}^k$, $i,j,k\in\{1,2,3\}$, or equivalently $n^{ij}$, $a_k$, see~\eqref{eqn_data3dimliegr} and~\eqref{eqn_data3dimliegrback} in Subsection~\ref{subsect_appendixbianchiclassification}.

	For a fixed~$\ggg_2$, the chosen frame of the \liealg~$\ggg$ is unique up to rotation and reflection in~$\ggg_2$ in case of~$\tilde e_2$, $\tilde e_3$, and a choice of sign in case of~$\tilde e_1$.
	We can therefore use the frame from the classification of \liegr s in Subsection~\ref{subsect_appendixbianchiclassification}, see also Remark~\ref{rema_frameforggg2rotationtodiagonalise}. Note that in the case of a Bianchi type~II \liegr, we rename the basis elements such that~$\tilde e_2$ and~$\tilde e_3$ commute. For this special frame, the structure constants are such that~$a_2=a_3=0$, the matrix~$n$ is diagonal with~$n_{11}=0$, and~$a_1\not=0$ for~$G$ of Bianchi class~B. The reflection in~$\tilde e_1$ and as well as the rotation and reflection in~$\ggg_2$ do not affect~$a_i$ or~$n_{1i}=n_{i1}$, $i=2,3$.

	Let us now have a closer look at the constraint equations, which the metric~$\idmetric$ and two-tensor~$\idfundform$ have to satisfy by Def.~\ref{defi_initialdatabianchib}. As both tensors are left-invariant, the momentum constraint reduces to
	\begin{equation}
		\overline\nabla^j\idfundform_{ij}=0.
	\end{equation}
	In terms of the chosen \onorm\ frame, we find
	\begin{equation}
		\overline\nabla_l\idfundform_{ij}
		=\tilde e_l(\idfundform(\tilde e_i,\tilde e_j)) -\idfundform(\overline\nabla_{\tilde e_l}\tilde e_i, \tilde e_j) -\idfundform(\tilde e_i,\overline\nabla_{\tilde e_l}\tilde e_j)
		={-}\overline\Gamma_{li}^m\idfundform_{mj}-\overline\Gamma_{lj}^m\idfundform_{im},
	\end{equation}
	where~$\overline\Gamma_{ij}^k$ denote the Christoffel symbols~$\overline\Gamma_{ij}^k\tilde e_k=\overline\nabla_{\tilde e_i}\tilde e_j$
	corresponding to the \levi\ connection~$\overline\nabla$ on the \riem\ \mf~$(G,h)$.
	We conclude that the momentum constraint is equivalent to
	\begin{equation}
	\label{eqn_momentumconstraintchristoffel}
		0={-}\sum_{j=1}^3(\overline\Gamma_{ji}^m\idfundform_{mj}+\overline\Gamma_{jj}^m\idfundform_{im}).
	\end{equation}
	As a consequence of Koszul's formula, the Christoffel symbols can be expressed in terms of the structure constants via
	\begin{equation}
	\label{eqn_christoffelinstructureconstants}
		\overline\Gamma_{ij}^k=\frac12({-}\gamma_{jk}^i+\gamma_{ki}^j+\gamma_{ij}^k).
	\end{equation}
	Using the definitions for~$n^{ij}$ and $a_k$, equations~\eqref{eqn_data3dimliegr}, as well as the symmetry of~$\idfundform$
	we then find
	\begin{equation}
	\label{eqn_threericcizeroione}
		0={-}2a_1\idfundform_{11}+a_1(\idfundform_{22}+\idfundform_{33})+(n_{33}-n_{22})\idfundform_{23}
	\end{equation}
	for the~$i=1$-component, while
	the~$i=2$ and~$i=3$-component of equation~\eqref{eqn_momentumconstraintchristoffel} read
	\begin{equation}
	\label{eqns_threericcizeroitwothree}
	\begin{subaligned}
		0={}&{-}3a_1\idfundform_{12}-n_{33}\idfundform_{13}\\
		0={}&{+}n_{22}\idfundform_{12}-3a_1\idfundform_{13}.
	\end{subaligned}
	\end{equation}
	We conclude that there is a non-vanishing solution~$(\idfundform_{12},\idfundform_{13})$ of the system of equations~\eqref{eqns_threericcizeroitwothree} \iif
	\begin{equation}
		9 a_1^2+n_{22}n_{33}=0.
	\end{equation}
	In Bianchi class~B, this is only possible for Bianchi type~VI due to~$a_1^2>0$ and the signs of~$n_{ii}$,  compare Table~\ref{table_bianchibclassification}.  Furthermore, the invariance of the parameter~$\binvparam$, equation~\eqref{eqn_binvparamdefi}, implies that only the case~$\binvparam={-}1/9$ allows for a non-vanishing solution. In total, we have shown the following:
	\begin{lemm}
	\label{lemm_exceptionalbianchiasinitialdata}
		Let~$G$ be a \liegr\ of class~B, ${\idmetric}$ a left-invariant \riem\ metric on~$G$, and ${\idfundform}$ a left-invariant symmetric covariant two-tensor on~$G$
		satisfying the momentum constraint
		\begin{equation}
			\overline\nabla_i\tr_{\idmetric}{\idfundform}-\overline\nabla^j{\idfundform}_{ij}=0.
		\end{equation}
		Let~$\ggg_2$ be the kernel of the one-form~$H$ from Lemma~\ref{lemm_abeliansubalginvardefi} which is an Abelian subalgebra of the \liealg~$\ggg$ corresponding to~$G$. Let~$(\tilde e_1,\tilde e_2,\tilde e_3)$ an \onorm\ basis of~$\ggg$ such that~$\tilde e_2$ and~$\tilde e_3$ span~$\ggg_2$.
		If~$\idfundform_{12}$ or $\idfundform_{13}$ is non-vanishing, then the \liegr~$G$ is of Bianchi type~VI$_{{-}1/9}$.
	\end{lemm}
	Note that in the previous Lemma, the subalgebra~$\ggg_2$ is defined geometrically, as the kernel of a well-defined one-form. Consequently, whether
	\begin{equation}
	\label{eqn_nonexceptionalsigmaij}
		\idfundform_{12}=\idfundform_{21}=0=\idfundform_{13}=\idfundform_{31}
	\end{equation}
	or not is also a well-defined geometric property, and independent of the exact choice of \onorm\ basis, as long as~$\tilde e_2$ and~$\tilde e_3$ span~$\ggg_2$.

	For a \liegr\ of Bianchi class~A, a basis can be chosen which satisfies the properties of i) in the classification and diagonalises both the metric and the \fundform, see~\cite[Cor.~9.14]{ringstrom_cauchyproblem}. Therefore in this basis the relations~\eqref{eqn_nonexceptionalsigmaij}
	hold as well.

	Lemma~\ref{lemm_exceptionalbianchiasinitialdata} is the reason why we have explicitly excluded Bianchi type~VI$_{{-}1/9}$ \liegr s in the definition of geometric initial data, Def.~\ref{defi_initialdatabianchib}: By doing so, we ensure that equations~\eqref{eqn_nonexceptionalsigmaij} hold.
	In fact, we do not have to exclude the special Bianchi type~VI$_{{-}1/9}$ altogether. It can be included in the discussion as long as we adopt equation~\eqref{eqn_nonexceptionalsigmaij} as an additional assumption.
	\begin{rema}
	\label{rema_exceptionalityexplained}
		It has been shown in~\cite{ellismaccallum_classofhomogcosmmodels} that condition~\eqref{eqn_nonexceptionalsigmaij} is equivalent to `non-exceptio\-nality' of a Bianchi spacetime.
		More precisely, this reference considers the trace-free part
		\begin{equation}
			\sigma_{ij}= {\idfundform}_{ij}-\frac{\theta}3\idmetric_{ij}
		\end{equation}
		of~$\idfundform$, where~$\theta=\tr_\idmetric\idfundform$ is the mean \curv, see also Remark~\ref{rema_geometricmeaningthetasigmaAB}, and makes use of the equivalent formulation
		\begin{equation}
			\sigma_{12}=\sigma_{21}=0=\sigma_{13}=\sigma_{31}.
		\end{equation}
		This notion of 'non-exceptionality' is a property of the four-dimensional spacetime, namely not being one of the 'exceptional' spacetimes which~\cite{ellismaccallum_classofhomogcosmmodels} denote by Bbii. These 'exceptional' spacetimes are those which admit an \onorm\ frame~$e_i$, $i=0,1,2,3$, such that~$e_1,e_2,e_3$ are tangential to the spacelike hypersurfaces and
		\begin{equation}
		\label{eqn_definitionBbii}
			\sigma_{12}\sigma_{13}\not=0,\qquad {\binvparam}={-}1/9,\qquad n_{22}=\frac{3a_1\sigma_{13}}{\sigma_{12}},\qquad n_{33}={-}\frac{3a_1\sigma_{12}}{\sigma_{13}}
		\end{equation}
		holds.

		The property of being `non-exceptional' is further related to the notion of \ogon ly transitive~$G_2$ cosmologies, see~\cite[Thm~3.1(i)]{wainwright_classschemenonrotinhomogcosm} for a characterisation of \ogon\ transitivity. In combination with several results from~\cite{ellismaccallum_classofhomogcosmmodels}, in particular Lemma~4.1 and Thm~5.1 in that reference, we conclude the following: A Bianchi class~B spacetime admits an Abelian subgroup acting \ogon ly transitively \iif\ it is `non-exceptional'.
		For this reason, we use the terms '\ogon ' and 'non-exceptional' initial data interchangeably.
		 
		In terms of the Bianchi classification, only certain initial data sets for one specific parameter in one Bianchi type of class~B satisfy the 'exceptional' properties. However, due to the additional degree of freedom, the phase space effectively has a higher dimension than that of all remaining Bianchi class~B types.
		For this reason, the notion 'exceptional' should be interpreted as 'having exceptional behaviour', not 'being special enough to discard'. In fact, towards the initial singularity such spacetimes are expected to show chaotic oscillatory behaviour, see~\cite{hewitthorwoodwainwright_asymptdynamexceptbianchicosm}.
	\end{rema}

\subsection{Construction part I: From geometric initial data to the expansion-normalised evolution equa\-tions}
\label{constr_initialdatatoexpnorm}

	In this subsection and the two following ones, we carry out the construction of the \mghd\ for given geometric initial data, by which we mean initial data to Einstein's field equations.

	We start with initial data $(G,{\idmetric},{\idfundform},\mu_0)$ as in Def.~\ref{defi_initialdatabianchib}.
	We fix a two-dimensional Abelian subalgebra~$\ggg_2$ of the \liealg~$\ggg$ corresponding to the \liegr~$G$. As we detailed in the previous subsection, this is the kernel of the one-form~$H$ from Lemma~\ref{lemm_abeliansubalginvardefi} in case of a \liegr\ of class~B, otherwise such a subalgebra exists due to Lemma~\ref{lemm_abeliansubalgstructconstants}. We then choose an orthonormal basis $\tilde e_1,\tilde e_2,\tilde e_3$ of~$\ggg$ such that~$\tilde e_2,\tilde e_3$ span~$\ggg_2$.
	Such a choice of basis gives initial structure constants~$\gamma_{ij}^k(0)$, $i,j,k\in\{1,2,3\}$, and this information is equivalent to\begin{equation}
		n^{ij}(0),\qquad a_k(0)
	\end{equation}
	due to~\eqref{eqn_data3dimliegr} and~\eqref{eqn_data3dimliegrback} in Subsection~\ref{subsect_appendixbianchiclassification}.
	We have explained in the previous subsection that in this basis one has~$a_2=a_3=0$, the matrix~$n$ is diagonal with~$n_{11}=0$, and~$a_1\not=0$, in case of a \liegr~$G$ of Bianchi class~B. For groups of type~I or~II, we can assume that~$a_1=a_2=a_3=0$ and the matrix~$n$ is diagonal with~$n=\diag(0,n_{22},n_{33})$. In particular, the initial structure constants are such that the third and fourth equation in~\eqref{eqns_commutators} hold.

	We further set initial values for~$\theta$ and~$\sigma_{ij}$
	\begin{equation}
	\label{eqn_initialvaluesthetasigmaAB}
		\theta(0)\coloneqq\tr_{\idmetric}{\idfundform},\qquad \sigma_{ij}(0)\coloneqq {\idfundform}_{ij}-\frac{\theta(0)}3\idmetric_{ij},
	\end{equation}
	in accordance with Remark~\ref{rema_geometricmeaningthetasigmaAB}.
	For the initial value~$\Omega_1(0)$, we choose an arbitrary value. This is not determined by the geometric initial data. We will see further down that one can \woutlog\ assume $\Omega_1\equiv0$, but at this stage this is a non-determined variable.
	From Lemma~\ref{lemm_exceptionalbianchiasinitialdata}, we conclude that
	the~$12$- and~$13$-component of~$\sigma_{ij}$ (or equivalently of~$\idfundform_{ij}$) vanish. In particular, we are now in a position to connect geometric initial data to the construction carried out in Subsection~\ref{subsect_appendixdeducenormalisedvariablesbianchib}, as the only non-vanishing initial quantities are those appearing on the \lhs\ of equations~\eqref{eqns_prebasic_definitionpartone}.

	With the definitions made so far, we can introduce quantities
	\begin{equation}
		\tilde\gamma_{01}^1(0)\coloneqq\sigma_A{}^A(0)-\frac13\theta(0),\qquad
		\tilde\gamma_{0A}^B(0)\coloneqq{-}\sigma_A{}^B(0)-\frac13\theta(0)\delta_A^B+\Omega_1(0)\epsilon_A{}^B.
	\end{equation}
	We stress that these are merely numbers at this point. Only after we have constructed the spacetime can they be interpreted as the structure constants on the hypersurface~$\{0\}\times G$ of a suitable four-dimensional frame.
	In combination with the three-dimensional structure constants~$\gamma_{ij}^k(0)$ we found above, we have now constructed, at time~$t=0$, a set of numbers as appear on the~\rhs\ of \equ s~\eqref{eqns_commutators}.

	In the next step, we apply the algebraic operations~\eqref{eqns_prebasic_definitionparttwo} and~\eqref{eqns_basic_definitionpartone} to the set of numbers
	\begin{equation}
		(\theta,\sigma_{AB},\Omega_1,a_1,n_{AB})(0)
	\end{equation}
	to obtain a set of numbers
	\begin{equation}
		(\theta,\sigma_+,\tilde\sigma_{AB},a_1,n_+,\tilde n_{AB})(0).
	\end{equation}
	Direct computation shows that the Hamiltonian constraint
	\begin{equation}
		\overline R-{\idfundform}_{ij}{\idfundform}^{ij}+(\tr_{\idmetric}{\idfundform})^2=2\mu_0
	\end{equation}
	which the geometric initial data has to satisfy
	is equivalent to equation~\eqref{eqn_prebasic_mu}, using for example~\cite[eq.~{[}E.12{]}]{ringstrom_topologyfuturestabilityuniverse} to compute the three-dimensional scalar \curv~$\overline{R}$.
	The~$i=1$-component of the momentum constraint implies equation~\eqref{eqn_prebasic_constraint}, as can be concluded from equation~\eqref{eqn_threericcizeroione}.
	We therefore interpret~$(\theta,\sigma_+,\tilde\sigma_{AB},a_1,n_+,\tilde n_{AB})(0)$ as initial data for the evolution equations~\eqref{eqns_prebasic_evolutionbianchib}--\eqref{eqn_prebasic_evolutionmu}. Note that from this point on, the arbitrary value~$\Omega_1(0)$ does no longer appear.
	Applying now also the transformations~\eqref{eqns_basic_definitionparttwo} and~\eqref{eqn_basic_definitionpartthree}, one obtains a set of numbers
	\begin{equation}
		(\theta,\sigma_+,\tilde\sigma,\delta,\tilde a,n_+)(0)
	\end{equation}
	which can be interpreted as initial data in basic variables. Note that equation~\eqref{eqn_basic_definitiontilden} holds due to the underlying \liegr\ structure of the geometric initial data we started with.

	We assume that we have started out with initial data~$(G,{\idmetric},{\idfundform},\mu_0)$
	such that after the steps carried out so far, the initial data in basic variables is different from the trivial data, \ie
	\begin{equation}
		(\theta,\sigma_+,\tilde\sigma,\delta,\tilde a,n_+)(0)\not=(0,0,0,0,0,0).
	\end{equation}
	Further down in Lemma~\ref{lemm_excludeminkowski}, we show that this restriction is equivalent to geometric initial data whose \mghd\ is not isometric to a quotient of Minkowski space.
	Going through the construction so far, looking for a solution to~\eqref{eqns_prebasic_evolutionbianchib}--\eqref{eqn_prebasic_evolutionmu} for the initial values we constructed, we can assume that $\theta\not=0$ at all times due to Lemma~\ref{lemm_excludethetazerobasic}. We therefore assume that
	\begin{equation}
		\theta>0.
	\end{equation}
	Once we have obtained a spacetime which admits the geometric initial data induced on a hypersurface, this choice of sign corresponds to the choice of normal vector~$e_0$, or in other words fixes the orientation of time~$t$.

	Essentially, the transformation of geometric initial data into initial values in basic variables is a one-to-one correspondence up to sign. We can recover the signs from the geometric initial data, see the third part of the construction.

	In the next step, we normalise the initial values via equations~\eqref{eqn_expansionnormalisation} to obtain initial values
	\begin{equation}
		(\Sigma_+,\tilde\Sigma,\Delta,\tilde A,N_+)(0).
	\end{equation}
	These are interpreted as initial data to the evolution equations in expansion-normalised variables~\eqref{eqns_evolutionbianchib}--\eqref{eqn_evolutionomega}, where we also change to an expansion-normalised time~$\tau$ which we choose to satisfy~$\tau(t=0)=0$.

	In total, we have achieved the following: Geometric initial data $(G,{\idmetric},{\idfundform},\mu_0)$ as in Def.~\ref{defi_initialdatabianchib} (or even of Bianchi type~VI$_{{-}1/9}$ satisfying equation~\eqref{eqn_nonexceptionalsigmaij}) is transformed into dynamical initial data $(\Sigma_+,\tilde\Sigma,\Delta,\tilde A,N_+)(0)$ for the evolution equations~\eqref{eqns_evolutionbianchib}--\eqref{eqn_evolutionomega} of expansion-normalised variables. In the process, an arbitrary number~$\Omega_1(0)$ was chosen, but the resulting dynamical initial data is independent of the choice of~$\Omega_1(0)$.

\subsection{Construction part II: From the expansion-normalised evolution equations to a solution in basic variables}
\label{constr_expnormtosolutionbasic}

	At the end of the first part of the construction, Subsection~\ref{constr_initialdatatoexpnorm}, we have obtained dynamical initial data $(\Sigma_+,\tilde\Sigma,\Delta,\tilde A,N_+)(0)$. One can apply the evolution equations~\eqref{eqns_evolutionbianchib}--\eqref{eqn_evolutionomega} for the expansion-normalised variables, which form a system of polynomial differential equations on a compact subset of~$\RR^5$, see Remark~\ref{rema_statespacecompact}. Consequently, for given dynamical initial data, there exists a unique solution which is defined at all times $\tau\in({-}\infty,\infty)$. In this part of the construction, this solution is now translated back to a solution in basic variables, and then in the following subsection even further, to a spacetime solving the correct Einstein's equations and with correct initial data.

	In order to retrieve basic variables at all times~$t$, one considers the evolution equation for the basic variable~$\theta$, equation~\eqref{eqn_evolutiontheta}. The initial value~$\theta(0)$ is given from the geometric initial data, see equation~\eqref{eqn_initialvaluesthetasigmaAB}. Consequently this evolution equation has a unique maximal solution~$\theta(\tau)$. As~$q\in[{-}1,2]$ due to the relations in~\eqref{eqn_constraintgeneralthree}, there are two cases to consider:
	\begin{enumerate}
		\item $q={-1}$ at some time. This corresponds to $\gamma=0$ and $\Omega=1$, which implies
		\begin{equation}
			\Sigma_+=\tilde\Sigma=\Delta=\tilde A=N_+=0.
		\end{equation}
		This point is an equilibrium point of the evolution \equs ~\eqref{eqns_evolutionbianchib}, from which we conclude that $q={-1}$ at all times and $\theta=\theta(0)$ constant.
		\item $q>{-}1$ at all times.
		In case of inflationary matter~$\Omega>0$, $\gamma\in[0,2/3)$, we have shown in Prop.~\ref{prop_alphalimitsets_vacuum_inflat} that the only solution whose $\alpha$-limit set contains the point $\Sigma_+=\tilde\Sigma=\Delta=\tilde A=N_+=0$ is the constant solution, which is excluded in this case.
		Consequently, $1+q$ is bounded from below by a positive number for~$\tau$ sufficiently negative. In vacuum and in the remaining matter cases~$\Omega>0$, $\gamma\in[2/3,2]$, equation~\eqref{eqn_definitionqvacuum} and equation~\eqref{eqn_definitionqmatter} even yield~$1+q>1$. With this, equation~\eqref{eqn_evolutiontheta} implies that~$\theta$ is monotone and, for sufficiently negative times, $\absval\theta$ can be estimated from above and below by exponential functions of~$\tau$ with a negative exponent.
		Due to choice of sign we have made in the first part of the construction, $\theta>0$ at all times and therefore~$\theta$ monotone decaying.
		Thus, $\theta$ converges
		to some non-negative number~$\theta_\infty$
		as $\tau\rightarrow\infty$, and diverges to~$\infty$ as $\tau\rightarrow{-}\infty$.
	\end{enumerate}
	In both cases, $\theta$ is positive and defined for all~$\tau\in({-}\infty,\infty)$. One can therefore multiply the individual components of the solution $(\Sigma_+,\tilde\Sigma,\Delta,\tilde A,N_+)(\tau)$ with the corresponding power of~$\theta$ such that \equ ~\eqref{eqn_expansionnormalisation} holds, and retrieves the remaining basic variables $(\sigma_+ ,\tilde{\sigma},\delta ,\tilde{a},n_+)(\tau)$. As was the case for~$\theta$, these variables are defined at all times~$\tau\in({-}\infty,\infty)$.

	To replace the expansion-normalised time~$\tau$, we define a different time scale~$t$ via equation~\eqref{eqn_evolutiontau} and the initial condition~$\tau(t=0)=0$.
	Due to $\theta>0$, this definition immediately yields that~$t$ and~$\tau$ have the same time orientation. In case $q={-}1$, the time~$t$ is simply a rescaling of~$\tau$, and the basic variables are defined at all times $t\in({-}\infty,\infty)$. In case~$q>{-}1$, our discussion in~ii) implies that time~$\tau\rightarrow{-}\infty$ corresponds to diverging~$\theta$ and a finite time~$t$, while~$\tau\rightarrow\infty$ corresponds to~$\theta\rightarrow\theta_\infty\ge0$
	and $t\rightarrow\infty$.

	In total, we obtain a curve in basic variables $(\theta ,\sigma_+ ,\tilde{\sigma},\delta ,\tilde{a},n_+)(t)$. By our construction, they are related to the solution in expansion-normalised variables $(\Sigma_+,\tilde\Sigma,\Delta,\tilde A,N_+)(\tau)$ exactly as in Subsection~\ref{subsect_appendixdeducenormalisedvariablesbianchib}. We have noted there,
	in Remark~\ref{rema_equivalencebasicvarexpnormvar}, that the evolution \equs\ for the expansion-normalised variables and the evolution \equ s for the basic variables are equivalent. Therefore, the construction carried out here yields the maximal solution to equations~\eqref{eqns_basic_evolutionbianchib}--\eqref{eqn_basic_evolutionmu}. The resulting maximal interval of existence is $(t_-,\infty)$, with $t_->{-}\infty$ apart from when $\gamma=0$ and $\Omega=1$, in which case $t_-={-}\infty$.

\subsection{Construction part III: From a solution in basic variables to a spacetime}
\label{constr_solutionbasictomghd}

	Having found a solution in basic variables $(\theta ,\sigma_+ ,\tilde{\sigma},\delta ,\tilde{a},n_+)(t)$ via the previous two parts of the construction, Subsections~\ref{constr_initialdatatoexpnorm} and~\ref{constr_expnormtosolutionbasic}, we now translate this into a spacetime with suitable Lorentzian metric which is consistent with the geometric initial data~$(G,{\idmetric},{\idfundform},\mu_0)$ we started with.

	Recall how the construction started: In order to obtain initial data in basic variables from the geometric initial data, one introduced a frame~$\tilde e_1,\tilde e_2,\tilde e_3$ and then had to choose an arbitrary initial value~$\Omega_1(0)$, whose information was lost in the first algebraic manipulation. This corresponded to introducing variables which were explicitly invariant under the choice of frame, in the sense that they did no longer depend on the freedom of rotation in the $\tilde e_2\tilde e_3$-plane. In order to now recover from a maximal solution $(\theta ,\sigma_+ ,\tilde{\sigma},\delta ,\tilde{a},n_+)(t)$, $t\in(t_-,\infty)$, a spacetime frame with four-dimensional structure constants and then a spacetime metric, one needs to break this gauge invariance and reverse this process by choosing a suitable frame. This is done in several steps:
	\begin{enumerate}
		\item Given an arbitrary, sufficiently smooth~$\Omega_1$ defined on the same~$t$-time interval~$(t_-,\infty)$ as the solution in basic variables, one retrieves the variables $(\sigma_{AB},\theta,n_{AB},a_1)(t)$.
		\item With this, one defines a \mf\ and constructs a frame~$e_0,e_1,e_2,e_3$.
		\item Then, one checks that this construction holds through time, independently of the choice of~$\Omega_1$.
		\item Once the frame is obtained, one defines this frame as \onorm, which uniquely defines the corresponding spacetime metric on the manifold~$(t_-,\infty)\times G$.
		\item Finally, one checks that in the resulting spacetime Einstein's equations are satisfied and the correct initial data induced.
	\end{enumerate}

	We begin with the first step, which is achieved via constructing a solution in the variables~$(\theta,\sigma_+,\tilde\sigma_{AB},a_1,n_+,\tilde n_{AB})$ satisfying equations~\eqref{eqns_prebasic_evolutionbianchib}--\eqref{eqn_prebasic_evolutionmu}. With this solution at our disposal, we can immediately deduce the variables $(\sigma_{AB},\theta,n_{AB},a_1)(t)$. It is the first part which is more intricate, as we have to be careful about notation. We start with a solution in basic variables, which for the sake of precision, we denote by
	\begin{equation}
		(\theta_\basic \,,\,\sigma_{+,\basic} \,,\,\tilde{\sigma}_\basic\,,\,\delta_\basic \,,\,\tilde{a}_\basic\,,\,n_{+,\basic})(t)\,,
	\end{equation}
	as the solution satisfies the evolution equations for basic variables, equations~\eqref{eqns_basic_evolutionbianchib}--\eqref{eqn_basic_evolutionmu}.

	We want to find a solution $(\theta,\sigma_+,\tilde\sigma_{AB},a_1,n_+,\tilde n_{AB})$ to equations~\eqref{eqns_prebasic_evolutionbianchib}--\eqref{eqn_prebasic_evolutionmu}. Even though we do not know that the functions~$\theta_\basic$, $\sigma_{+,\basic}$ and~$n_{+,\basic}$ evolve correctly, we carry them over, but keep the subscript as a reminder.
	The variable~$a_1$ is supposed to satisfy $\tilde a=9a_1^2$, see equation~\eqref{eqn_basic_definitionpartthree}. As $\tilde a=0$ is an invariant set due to the evolution equations~\eqref{eqns_basic_evolutionbianchib}, the sign of~$a_1$ cannot change and is determined by the initial data. From the knowledge of~$\tilde a_\basic$ and the initial data, we can uniquely define a function~$a_{1,\basic}$. Upon comparison with the evolution equation for~$\tilde a$ in~\eqref{eqns_basic_evolutionbianchib}, we find that~$a_{1,\basic}$ satisfies the evolution equation for~$a_1$ in~\eqref{eqns_prebasic_evolutionbianchib}, if one replaces~$\sigma_+$ and~$\theta$ by~$\sigma_{+,\basic}$ and~$\theta_\basic$.

	Now consider the evolution \equs\ of~$\tilde \sigma_{AB}$ and~$\tilde n_{AB}$ in~\eqref{eqns_prebasic_evolutionbianchib}. Together, they form a system of linear differential equations whose coefficients are the function~$\Omega_1$ and---after adding the subscript~${}_\basic$ at all necessary places---quantities whose time development is known from the solution in basic variables. For every sufficiently smooth function~$\Omega_1$, there is a maximal solution~$(\tilde n_{AB,\basic},\tilde\sigma_{AB,\basic})$ to this system.
	Due to linearity of the differential equations, the maximal solution is defined on the whole time interval on which the coefficients are defined. By construction, this is the same time interval as for the solution in basic variables.

	Next, we consider the expressions
	\begin{align}
		\difftildesigma\coloneqq{}&\tilde\sigma_\basic-\frac32\tilde\sigma^{AB}_\basic\tilde\sigma_{AB,\basic},\\
		\difftilden\coloneqq{}&\tilde n_\basic-\frac32\tilde n^{AB}_\basic\tilde n_{AB,\basic},\\
		\diffdelta\coloneqq{}&\delta_\basic-\frac32\tilde\sigma^{AB}_\basic\tilde n_{AB,\basic},\\
		\diffconstraint\coloneqq{}&3a_{1,\basic}\sigma_{+,\basic}-\frac32\,^*\tilde\sigma^{AB}_\basic\tilde n_{AB,\basic},
	\end{align}
	with~$\tilde n_{\basic}$ defined by equation~\eqref{eqn_basic_definitiontilden}.
	These expressions should be compared to the constraint equation~\eqref{eqn_prebasic_constraint} and equations~\eqref{eqns_basic_definitionparttwo}. The derivative of ~$(\difftildesigma,\difftilden,\diffdelta,\diffconstraint)$ is a homogeneous system of equations, and as all four functions vanish initially due to their components having been constructed from the same geometric initial data, we see that all four functions vanish identically. In particular, the constructed functions~$\tilde n_{\basic,AB}$ and~$\tilde\sigma_{\basic,AB}$ are related to the basic variables~$\tilde\sigma_\basic$, $\tilde n_\basic$ and~$\delta_\basic$ as in equations~\eqref{eqns_basic_definitionparttwo}, and we can replace these expressions at every occurence. This in particular implies that the function~$\mu_\basic$, defined via equation~\eqref{eqn_basic_mu} in basic variables, also satisfies equation~\eqref{eqn_prebasic_mu}. Further, we conclude from the argumentation above that the constraint equation~\eqref{eqn_prebasic_constraint} holds for the functions indexed~${}_\basic$.

	It remains to check that the variables~$\theta_\basic$, $\sigma_{+,\basic}$ and~$n_{+,\basic}$ satisfy the evolution equations~\eqref{eqns_prebasic_evolutionbianchib}. To this end, consider the initial data to the evolution equations~\eqref{eqns_prebasic_evolutionbianchib}--\eqref{eqn_prebasic_evolutionmu} which we, in the construction in Subsection~\ref{constr_initialdatatoexpnorm}, obtained from geometric initial data and then subsequently used to obtain initial data to the evolution in basic variables. There is a unique solution with this initial data, and in order to distinguish its individual variables from the ones constructed above, we denote this solution by
	\begin{equation}
		(\theta_\initialdata \,,\,\sigma_{+,\initialdata} \,,\,\tilde{\sigma}_\initialdata\,,\,\delta_\initialdata \,,\,\tilde{a}_\initialdata\,,\,n_{+,\initialdata})(t)\,.
	\end{equation}
	The initial values of these variables coincide with the initial values of the ones constructed, and uniqueness of systems of ordinary differential equations yields that~$\theta_\initialdata=\theta_\basic$ at all times, and equivalently for all other variables.
	We conclude that via this construction, we have obtained the unique solution to equations~\eqref{eqns_prebasic_evolutionbianchib}--\eqref{eqn_prebasic_evolutionmu} to the initial data coming from the geometric initial data.

	In this transformation, the maximal interval of the solution~$(\theta,\sigma_+,\tilde\sigma_{AB},a_1,n_+,\tilde n_{AB})$ cannot exceed that of the solution in basic variables, as otherwise this would yield an extension of the solution in basic variables, a contradiction. Consequently, the maximal interval of existence of the solution after transformation coincides with the interval~$(t_-,\infty)$.
	With this solution at hand, we use equations~\eqref{eqns_prebasic_definitionparttwo} and uniquely retrieve the variables~$\sigma_{AB}$ and~$n_{AB}$, defined on the same interval.

	\smallskip

	In the next step, we construct a four-dimensional \mf\ with a frame whose structure constants have the correct form.
	To this end, we define scalar functions~$\gamma_{\alpha\beta}^\delta(t)$ on the time interval~$(t_-,\infty)$
	via setting
	\begin{equation}
	\label{eqns_structureconstantlikeobjectsalltimes}
	\begin{subaligned}
		\gamma_{01}^1(t)\coloneqq{}&\sigma_A{}^A(t)-\frac13\theta(t),\\
		\gamma_{0A}^B(t)\coloneqq{}&{-}\sigma_A{}^B(t)-\frac13\theta(t)\delta_A^B+\Omega_1(t)\epsilon_A{}^B,\\
		\gamma_{1A}^C(t)\coloneqq{}&\epsilon_{AB}n^{BC}(t)+a_1(t)\delta_A^C,
	\end{subaligned}
	\end{equation}
	and setting all other~$\gamma_{\alpha\beta}^\delta$ to vanish identically. The form of these scalar functions coincides with the form of the structure constants in equation~\eqref{eqns_commutators}.
	By construction, these objects are consistent with the initial data:
	At time~$t=0$, the spacelike ones~$\gamma_{ij}^k(0)$ coincide with the \liegr\ structure constants we chose in Subsection~\ref{constr_initialdatatoexpnorm}, and the remaining ones~$\gamma_{0i}^j(0)$ coincide with the structure constant-like object~$\tilde\gamma_{0i}^j(0)$ defined there.
	We construct now a four-dimensional frame whose structure constants coincide with these~$\gamma_{\alpha\beta}^\delta(t)$ at all times.
	Note that it is at this point that the structure constant-like object~$\tilde\gamma_{0i}^j(0)$ can finally be interpreted geometrically.

	From the initial data, we have obtained an initial frame $\tilde e_1,\tilde e_2,\tilde e_3$ on the \liegr~$G$. We consider the \mf~$(t_-,\infty)\times G$ and in this \mf\ extend this frame globally, \ie time-independently to every $\{t\}\times G$.
	Further, we choose $e_0=\partial_t$ globally. This gives a four-dimensional frame~$e_0,\tilde e_1,\tilde e_2,\tilde e_3$, but the relation between its structure constants and the solution~$(\theta ,\sigma_{AB},\Omega_1,a_1,n_{AB})(t)$ does not necessarily fulfill equations~\eqref{eqns_commutators} at times other than~$t=0$. The final frame $e_0,e_1,e_2,e_3$ which does satisfy these relations can be constructed as follows.

	We wish the frame vectors~$e_2$,~$e_3$ to be tangent to the subalgebra~$\ggg_2$ which we have chosen in the beginning of Subsection~\ref{constr_initialdatatoexpnorm}, which implies that we need to construct time-dependent functions~$f_A^B(t)$ such that
	\begin{equation}
		e_A=f_A^B\tilde e_B, \qquad f_A^B(0)=\delta_A^B.
	\end{equation}
	Assuming the existence of such functions gives the following commutator relation:
	\begin{equation}
		\liebr{e_0}{e_A}=\liebr{e_0}{f_A^B\tilde e_B}=e_0(f_A^B)\tilde e_B+f_A^B\liebr{e_0}{\tilde e_B}=e_0(f_A^B)\tilde e_B,
	\end{equation}
	where we seek to find
	\begin{equation}
		\liebr{e_0}{e_A}=(-\sigma_A{}^B-\frac13\theta\delta_A^B+\Omega_1\epsilon_A{}^B)e_B=(-\sigma_A{}^B-\frac13\theta\delta_A^B+\Omega_1\epsilon_A{}^B)f_B^C\tilde e_C.
	\end{equation}
	Consequently, solving the system
	\begin{equation}
		e_0(f_A^C)=(-\sigma_A{}^B-\frac13\theta\delta_A^B+\Omega_1\epsilon_A{}^B)f_B^C,\qquad f_A^C(0)=\delta_A^C,
	\end{equation}
	yields the final frame vectors~$e_2$ and~$e_3$. Note that in order to do so, we have to fix a function~$\Omega_1(t)$.

	For the missing spatial frame component~$e_1$, we make the ansatz
	\begin{equation}
		e_1=f_1^1\tilde e_1+f_1^A\tilde e_A,\qquad f_1^1(0)=1,\quad f_1^A(0)=0.
	\end{equation}
	After a computation similar to the above, the resulting commutator reads
	\begin{equation}
		\liebr{e_0}{e_1}=\liebr{e_0}{f_1^1\tilde e_1+f_1^A\tilde e_A}=e_0(f_1^1)\tilde e_1+e_0(f_1^A)\tilde e_A.
	\end{equation}
	Comparison with the desired commutator
	\begin{equation}
		\liebr{e_0}{e_1}=(\sigma_A{}^A-\frac13\theta)e_1=(\sigma_A{}^A-\frac13\theta)(f_1^1\tilde e_1+f_1^A\tilde e_A)
	\end{equation}
	gives the following systems:
	\begin{align}
		e_0(f_1^1)={}&(\sigma_A{}^A-\frac13\theta)f_1^1,& \qquad f_1^1(0)={}&1,\\
		e_0(f_1^A)={}&(\sigma_A{}^A-\frac13\theta)f_1^A,& \qquad f_1^A(0)={}&0.
	\end{align}
	In particular, $f_1^A(t)=0$ at all times.
	The solution to this system of equations is independent of~$\Omega_1$, and yields the final frame vector~$e_1$.

	Having found functions~$f_A^B,f_1^A,f_1^1$, we have constructed a frame which satisfies all but the third commutator relation in~\eqref{eqns_commutators}. These remaining equations can be considered as constraint equations. From the construction of the final frame vectors, we obtain
	\begin{align}
		\liebr{e_1}{e_A}={}&\liebr{f_1^1\tilde e_1+f_1^C\tilde e_C}{f_A^B\tilde e_B}=f_1^1f_A^B\liebr{\tilde e_1}{\tilde e_B}+f_1^Cf_A^B\liebr{\tilde e_C}{\tilde e_B}=f_1^1f_A^B\liebr{\tilde e_1}{\tilde e_B}\\
		={}&f_1^1f_A^B\gamma_{1B}^C(0)\tilde e_C,
	\end{align}
	which has to coincide with
	\begin{equation}
		\gamma_{1A}^Be_B=\gamma_{1A}^Bf_B^C\tilde e_C,
	\end{equation}
	$\gamma_{1A}^B$ as defined in~\eqref{eqns_structureconstantlikeobjectsalltimes},
	in order for our construction to by consistent.
	Thus, we have to check whether
	\begin{equation}
	\label{eqn_constraintincommutatorform}
		\LHS_A^C\coloneqq f_1^1f_A^B\gamma_{1B}^C(0)=\gamma_{1A}^Bf_B^C \eqqcolon \RHS_A^C
	\end{equation}
	holds. Initially, the \lhs\ and \rhs\ coincide, and in order to prove that the relation holds at all times, it is therefore enough to prove an evolution equation of the form
	\begin{equation}
		\partial_t((\LHS-\RHS)_A^C)=(\LHS-\RHS)_A^B \cdot M_B^C,
	\end{equation}
	for some time-dependent functions~$M_B^C$.

Using the derivatives of the functions~$f_1^1$, $f_1^A$ and~$f_A^B$ which we determined above, we find that the derivative of the \lhs\ of~\eqref{eqn_constraintincommutatorform} reads
\begin{equation}
\label{eqn_constructionmghdconstraint}
\begin{split}
	\partial_t(\LHS_A^C)={}&
		\dot f_1^1f_A^B\gamma_{1B}^C(0)+f_1^1\dot f_A^B\gamma_{1B}^C(0)
		=\gamma_{01}^1f_1^1f_A^B\gamma_{1B}^C(0)+f_1^1\gamma_{0A}^Df_D^B\gamma_{1B}^C(0)\\
		={}&\gamma_{01}^1\LHS_A^C+\gamma_{0A}^D\LHS_D^C.
\end{split}
\end{equation}
For the \rhs, we find
\begin{equation}
	\partial_t(\RHS_A^C)={}
		\dot\gamma_{1A}^Bf_B^C+\gamma_{1A}^B\dot f_B^C
		=\dot\gamma_{1A}^Bf_B^C+\gamma_{1A}^B\gamma_{0B}^Df_D^C
		=(\dot\gamma_{1A}^B+\gamma_{1A}^D\gamma_{0D}^B)f_B^C.
\end{equation}
The time-derivative of the structure constants can be obtained by expressing~$\gamma_{1A}^B$ in terms of the variables $(\theta,\sigma_+,\tilde\sigma_{AB},a_1,n_+,\tilde n_{AB})$, applying their evolution \equs ~\eqref{eqns_prebasic_evolutionbianchib}--\eqref{eqn_prebasic_evolutionmu}, and converting the results back to structure constants. This yields
\begin{equation}
	\dot\gamma_{1A}^B+\gamma_{1A}^D\gamma_{0D}^B=\gamma_{01}^1\gamma_{1A}^B+\gamma_{0A}^D\gamma_{1D}^B,
\end{equation}
and therefore
\begin{equation}
	\partial_t(\RHS_A^C)
	=(\gamma_{01}^1\gamma_{1A}^B+\gamma_{0A}^D\gamma_{1D}^B)f_B^C
	=\gamma_{01}^1\RHS_A^C+\gamma_{0A}^D\RHS_D^C,
\end{equation}
which is the same structure we also found for the evolution of the \lhs, \equ ~\eqref{eqn_constructionmghdconstraint}. We conclude that the constraint is preserved independently of the choice of~$\Omega_1$. For all our purposes, we can assume $\Omega_1\equiv0$ at all times.
The above construction results in a time-dependent left-invariant frame $e_1(t),e_2(t),e_3(t)$ on the \liegr, which in combination with $e_0=\partial_t$ yields a spacetime frame admitting the correct structure constants~\eqref{eqns_commutators}.
In checking this consistency, we have completed the third step of our construction of the spacetime.

\smallskip

There is for every~$t$ a unique metric ${}^t\overline g$ on $\{t\}\times G$ such that the frame~$e_1(t),e_2(t),e_3(t)$ is orthonormal. With this, we define the spacetime
\begin{equation}
\label{eqn_constructedspacetime}
	I\times G,\qquad g=-dt^2+{}^t\overline g,
\end{equation}
where~$I$ is the existence interval of the basic variables and contains~$0$, \ie $I=(t_-,\infty)$, with ${-}\infty<t_-<0$ apart from the case $\gamma=0$, $\Omega=1$, where $t_-={-}\infty$.

\smallskip

By our construction above, the only non-vanishing structure constants of the four-dimensional frame~$(e_0,e_1,e_2,e_3)$ are those in~\eqref{eqn_structureconstantsnonzero}, and as the variables composing the structure constants satisfy the evolution \equs~\eqref{eqns_prebasic_evolutionbianchib}--\eqref{eqn_prebasic_evolutionmu}, we can deduce from Remark~\ref{rema_equivalenceeinsteinperfectfluidprebasic} that the spacetime constructed in this way solves Einstein's field \equs\ for an \ogon\ perfect fluid with linear equation of state.

On the~$t=0$ timeslice, the frame vectors~$e_i$ and~$\tilde e_i$ coincide, $i=1,2,3$. The metric induced on this timeslice by the spacetime metric~\eqref{eqn_constructedspacetime} therefore coincides with the initial metric~$\idmetric$.
Due to our choice of metric, the vector field~$e_0$ is the timelike unit normal to every~$\{t\}\times G$ timeslice.
It follows from the construction that the \fundform\ of~the timeslice~$\{0\}\times G$ in~$(I\times G,g)$ coincides with the initial two-tensor~$\idfundform$. We conclude that the constructed spacetime induces the correct initial data.
This completes the fifth step, and as a consequence the following definition is reasonable.
\begin{defi}
\label{defi_mghd}
	Given initial data as in Def.~\ref{defi_initialdatabianchib}, we call a spacetime~$(I\times G,g,\mu)$ constructed as in the three parts of the construction, Subsections~\ref{constr_initialdatatoexpnorm}--\ref{constr_solutionbasictomghd}, a Bianchi~B development of the data.
\end{defi}
Due to the choice of frame~$\tilde e_2,\tilde e_3$ in the beginning of the construction, the spacetime constructed this way is not necessarily unique.

In order to transform to expansion-normalised variables in Subsection~\ref{constr_initialdatatoexpnorm}, we had to exclude initial data with zero mean curvature~$\theta$. We now prove that this corresponds to excluding Bianchi~B developments which are part of Minkowski spacetime.
\begin{lemm}
\label{lemm_excludeminkowski}
	Consider initial data as in Def.~\ref{defi_initialdatabianchib}. Then either
	its universal covering space is initial data for the four-dimensional Minkowski spacetime, or the mean \curv~$\theta=\tr_{\idmetric}\idfundform$ is non-vanishing at all times.
\end{lemm}
The proof proceeds similar to that of~\cite[Lemma~20.6]{ringstrom_cauchyproblem} for Bianchi class~A developments.
\begin{proof}
	As in the first part of the construction, Subsection~\ref{constr_initialdatatoexpnorm}, we choose a suitable frame and carry out the transformations into initial data in basic variables, \ie initial data for the evolution equations~\eqref{eqns_basic_evolutionbianchib}--\eqref{eqn_basic_evolutionmu}. From the construction, we see that the mean \curv\ coincides with the initial value for~$\theta$. Using Lemma~\ref{lemm_excludethetazerobasic}, we find that if the initial value~$\theta(0)$ vanishes, then the solution in basic variables is the trivial solution.

	With this trivial solution in basic variables at hand, we retrace the steps of the third part of the construction, Subsection~\ref{constr_solutionbasictomghd}, and recover a spacetime with an \onorm\ frame~$(e_0,e_1,e_2,e_3)$ such that the initial data is correctly induced. We have been able to choose~$\Omega_1\equiv0$ in this construction, and from this we conclude that all structure constants~$\gamma_{\alpha\beta}^\delta$ vanish identically.

	Consider now the \liegr~$G$ with structure constants~$\gamma_{ij}^k=0$. The three-dimensional Christoffel symbols~$\overline\Gamma_{ij}^k$ vanish due to equation~\eqref{eqn_christoffelinstructureconstants}, and we conclude that the three-dimensional Ricci \curv
	\begin{equation}
		\overline{R}_{ij}=\sum_{l=1}^3\scalprod{\overline R(e_i,e_l)e_l}{e_j}
		=\sum_{k,l=1}^3(\Gamma_{ll}^k\Gamma_{ik}^j-\Gamma_{il}^k\Gamma_{lk}^j-\gamma_{il}^k\Gamma_{kl}^j)
	\end{equation}
	is zero identically.
	As~$G$ is three-dimensional, this implies that the \riem\ \curv\ tensor vanishes as well. As the metric is left-invariant, it is complete, from which it follows that~$G$ with this metric is isometric to a quotient of Euclidean space, by virtue of having the same dimension, index and curvature, see~\cite[Prop.~8.23]{oneill_semiriemgeomappl}.

	We have explained in Remark~\ref{rema_geometricmeaningthetasigmaAB} that the \fundform\ of the timeslice~$\{t\}\times G$ can be computed from~$\theta$ and the trace-free variable~$\sigma_{ij}$, which satisfies~$\sigma_{12}=\sigma_{13}=0$. We can therefore conclude that the \fundform\ vanishes. This concludes the proof.
\end{proof}

\subsection{Properties of Bianchi~B developments of initial data}

\label{constr_propertiesdevelopmentinitialdata}

We now show that a development as in Def.~\ref{defi_mghd} of given geometric initial data, \ie one which we obtain through our construction in the previous three subsections, is in fact the \mghd\ of the geometric initial data. The proof of global hyberbolicity works identically as in the case of Bianchi class~A and has been carried out in~\cite[p.~217]{ringstrom_cauchyproblem}.
We only state the proposition:
\begin{prop}
\label{prop_developmentglobhyperbolic}
	Let $(G,{\idmetric},{\idfundform},\mu_0)$ be initial data as in Def.~\ref{defi_initialdatabianchib}, and~$(I\times G,g,\mu)$ a Bianchi~B development of the data.
	Then $\{t\}\times G$ is a Cauchy hypersurface for every $t\in I$.
\end{prop}
For the proof of maximality, we wish to apply the following proposition, which is part of~\cite[Prop.~18.16]{ringstrom_cauchyproblem}:
\begin{prop}
\label{prop_mghdnecessaryproperties}
	Let $(M,g)$ be a connected and time oriented Lorentzian \mf\ and assume that the following holds:
	\begin{itemize}
		\item $(M,g)$ is future geodesically complete and there are real numbers~$\kappa_j$ \st $\kappa_j\rightarrow\infty$ as $j\rightarrow\infty$ and smooth spacelike Cauchy hypersurfaces~$\Sigma_j$ in $(M,g)$ with constant mean curvature~$\kappa_j$.
	\end{itemize}
	Assume $(\tilde M,\tilde g)$ to be a time oriented and connected Lorentzian \mf\ satisfying the timelike convergence condition: $Ric(v,v)\ge0$ for all timelike vectors $v\in TM$. Assume further $i:M\rightarrow\tilde M$ to be a smooth embedding \st $i(M)$ is an open set, and that there is a Cauchy hypersurface~$S$ in $(M,g)$ \st $i(S)$ is a Cauchy hypersurface in $(\tilde M,\tilde g)$. Then $i(M)=\tilde M$.
\end{prop}
We further adapt~\cite[Prop.~20.10]{ringstrom_cauchyproblem} to our setting:
\begin{lemm}
	Consider a Bianchi~B development~$(M=I\times G,g,\mu)$ of initial data as in Def.~\ref{defi_initialdatabianchib}. Let $c:(s_-,s_+)\rightarrow M$ be a future directed inextendible causal geodesic, and
	\begin{equation}
		f_0(s)=\scalprod{c'(s)}{e_0|_{c(s)}}.
	\end{equation}
	Then either $\theta\equiv0$ and~$f_0$ is a constant, or
	\begin{equation}
		\frac{d}{ds}(f_0\theta)\ge C\theta^2f_0^2,
	\end{equation}
	with $C\coloneqq\min(1/2-\sqrt2/3,\gamma/2)$ if~$\mu>0$, and~$C\coloneqq(2-\sqrt2)/3 $ if~$\mu=0$.
\end{lemm}
\begin{proof}
	By construction of the Bianchi~B development in Subsections~\ref{constr_initialdatatoexpnorm}--\ref{constr_solutionbasictomghd}, we know that the trace-free part~$\sigma_{ij}$ of the \fundform~$\theta_{ij}$ of the sub\mf~$\{t\}\times M$ satisfies
	\begin{equation}
		\sigma_{12}=\sigma_{21}=0=\sigma_{13}=\sigma_{31}.
	\end{equation}

	For a fixed time~$s$, we apply a rotation in the~$e_2e_3$-plane to diagonalise~$\sigma_{ij}$ (and simultaneously~$\theta_{ij}$), then define
	\begin{equation}
		f_i(s)=\scalprod{c'(s)}{e_i|_{c(s)}},
	\end{equation}
	for~$i=1,2,3$.

	As $f_0(s)=\scalprod{c'(s)}{e_0|_{c(s)}}$ and~$c$ is a geodesic, we find
	\begin{equation}
		\frac{df_0}{ds}=\scalprod{c'(s)}{\nabla_{c'(s)}e_0}
			=\sum_{i,j}f_i(s)f_j(s)\scalprod{e_i}{\nabla_{e_j}e_0}
			=\sum_{i}f_i^2\theta_{ii}
			=\frac13\theta\sum_{i}f_i^2+\sum_{i}f_i^2\sigma_{ii},
	\end{equation}
	where we applied
	$\nabla_{e_0}e_0=0$.
	The diagonalisation we applied does not affect the basic variables~$(\theta,\sigma_+,\tilde\sigma,\delta,\tilde a,n_+)$, as they are invariant under rotation in the~$e_2e_3$-plane. We can therefore apply Lemma~\ref{lemm_excludethetazerobasic} to see that~$\theta=0$ at some time implies that
	\begin{align}
		0=\tilde\sigma=\frac34(\sigma_{22}-\sigma_{33})^2, \qquad
		0=\sigma_+=\frac32(\sigma_{22}+\sigma_{33}),
	\end{align}
	at all times. In particular $\sigma_{ij}\equiv0$, and from this we see that~$f_0$ is constant.

	The Raychaudhuri \equ\ for an \ogon\ perfect fluid with linear \equ\ of state~\eqref{eqn_lineareqnofstate} is
	\begin{equation}
	\label{eqn_raychaudhuri}
		\dot\theta+\theta^{ij}\theta_{ij}={-}\frac\mu2(3\gamma-2),
	\end{equation}
	and coincides with the evolution equation for~$\theta$ in~\eqref{eqns_prebasic_evolutionbianchib}. Here, the dot~$\dot{ }$ denotes differentiation \wrt\ the time parameter~$t$ of the foliation~$M=I\times G$, in contrast to~$'$ denoting differentiation \wrt\ the parameter~$s$ of the geodesic~$c$.
	The Raychaudhuri \equ\ follows from the~$00$-component of the Ricci \curv\ for a perfect fluid~\eqref{eqn_stressenergyperfectfluid} with linear equation of state~\eqref{eqn_lineareqnofstate}
	and holds independently of the Bianchi class, see~\cite[Eq.~(25.10)]{ringstrom_topologyfuturestabilityuniverse}. The general expression in terms of an \onorm\ tetrad for a perfect fluid can be found in~\cite[Eq.~(82)]{maccallum_cosmmodelsgeometricpointofview}.

	Combining the Raychaudhuri \equ ~\eqref{eqn_raychaudhuri} with the fact that the surfaces have constant mean \curv, it follows that
	\begin{equation}
	\label{eqn_mghdprooffunctionfzero}
		\frac{d}{ds}(f_0\theta)=\frac13\theta^2\sum_{i}f_i^2+\sum_{i}\theta
		f_i^2\sigma_{ii}+\frac13f_0^2\theta^2+f_0^2\sum_i\sigma_{ii}^2+\frac12f_0^2(3\gamma-2)\mu.
	\end{equation}
	We observe that apart from the last term including the energy density~$\mu$ this coincides with the expression for vacuum Bianchi~A developments used in the proof of~\cite[Lemma~20.10]{ringstrom_cauchyproblem}. In particular, in case this last term is non-negative, the identical argument works, yielding that
	\begin{equation}
		\frac{d}{ds}(f_0\theta)\ge\frac{2-\sqrt2}{3}\theta^2f_0^2.
	\end{equation}
	Note that $(2-\sqrt2)/3$ is in particular larger than the constant~$C$ we aim to achieve in case~$\mu>0$.

	To conclude the proof, we have to deal with the case that the last term in equation~\eqref{eqn_mghdprooffunctionfzero} is negative, \ie the case $\mu>0$ and $0\le\gamma<\frac23$.
	Without loss of generality one can assume that $\absval{\sigma_{11}}\le\absval{\sigma_{22}}\le\absval{\sigma_{33}}$, and as~$\sigma_{ij}$ is trace-free, one can estimate
	\begin{equation}
		3\sigma_{33}^2=2\sigma_{33}^2+(\sigma_{11}+\sigma_{22})^2\le 2\sum_i\sigma_{ii}^2.
	\end{equation}
	Consequently,
	\begin{equation}
		\absval{\sum_{i}f_i^2\sigma_{ii}}\le(\frac23)^{1/2}(\sum_i\sigma_{ii}^2)^{1/2}\sum_if_i^2.
	\end{equation}

	We compute that
	\begin{equation}
	\label{eqns_auxiliaryexpressionsumsigmaij}
		\begin{subaligned}
			\sigma_+^2+\tilde\sigma
			={}&\frac94(\sigma_{11})^2+\frac34(\sigma_{22}-\sigma_{33})^2+\frac32(\sigma_{23})^2+\frac32(\sigma_{32})^2\\
			={}&\frac32(\sigma_{11})^2+\frac34(\sigma_{22}+\sigma_{33})^2+\frac34(\sigma_{22}-\sigma_{33})^2+\frac32(\sigma_{23})^2+\frac32(\sigma_{32})^2\\
			={}&\frac32\sigma^{ij}\sigma_{ij},
		\end{subaligned}
	\end{equation}
	and as~$\sigma_{ij}$ is diagonal by assumption, we can use the definition of~$\mu$ in terms of the basic variables, equation~\eqref{eqn_basic_mu} to find
	\begin{equation}
		\frac13\theta^2-\mu
			=\frac13(\sigma_+^2+\tilde\sigma+\tilde a+\tilde n)
			\ge\frac13(\sigma_+^2+\tilde\sigma)
			=\frac12\sum_i\sigma_{ii}^2.
	\end{equation}
	We conclude that the derivative of $f_0\theta$ can be estimated by
	\begin{equation}
	\label{eqn_initialestimatef0thetawithoutgamma}
		\frac{d}{ds}(f_0\theta)\ge\frac13\theta^2\sum_{i}f_i^2+\sum_{i}\theta
		f_i^2\sigma_{ii}+\frac32f_0^2\sum_i\sigma_{ii}^2+\frac32f_0^2\gamma\mu.
	\end{equation}
	On the other hand, we can use the fact that $3\gamma-2<0$ to find
		\begin{align}
		\frac13\theta^2+\frac12(3\gamma-2)\mu ={}&\frac13\theta^2+\frac12(3\gamma-2)\frac13(\theta^2-\sigma_+^2-\tilde\sigma-\tilde a-\tilde n)\\
		\ge{}&\frac12\gamma\theta^2+(\frac13-\frac12\gamma)(\sigma_+^2+\tilde\sigma)\\
		={}&\frac12\gamma\theta^2+(\frac12-\frac34\gamma)\sum_i\sigma_{ii}^2
	\end{align}
	and start out with the estimate
	\begin{equation}
	\label{eqn_initialestimatef0thetawithgamma}
		\frac{d}{ds}(f_0\theta)\ge\frac13\theta^2\sum_{i}f_i^2+\sum_{i}\theta
		f_i^2\sigma_{ii}+\frac12\gamma\theta^2f_0^2+(\frac32-\frac34\gamma)f_0^2\sum_i\sigma_{ii}^2.
	\end{equation}
	We now divide into the three cases $\frac13\theta^2\le\sum_i\sigma_{ii}^2$, $\frac16\theta^2\le\sum_i\sigma_{ii}^2\le\frac13\theta^2$ and $\sum_i\sigma_{ii}^2\le\frac16\theta^2$. In the first two cases, inequality~\eqref{eqn_initialestimatef0thetawithoutgamma} and causality of the curve yields the desired estimate, while in the third case we apply inequality~\eqref{eqn_initialestimatef0thetawithgamma}.
\end{proof}
\begin{lemm}
\label{lemm_bianchibdevelopmentproperties}
	Consider a Bianchi~B development
	of initial data as in Def.~\ref{defi_initialdatabianchib}.
	If $\mu=0$ or~$\mu>0$, $\gamma>0$, then all inextendible causal geodesics are future complete and past incomplete.
	If instead $\mu>0$, $\gamma=0$, then all inextendible causal geodesics are future complete.
\end{lemm}
\begin{proof}
	In case $\mu=0$ or~$\mu>0$, $\gamma>0$, the previous lemma implies that
	\begin{equation}
		\frac{d}{ds}(f_0\theta)\ge C\theta^2f_0^2
	\end{equation}
	for a positive constant~$C$. Additionally, from equation~\eqref{eqn_basic_mu} one finds
	\begin{equation}
		\theta^2=3\mu+\sigma_+^2+\tilde\sigma+\tilde a+\tilde n\ge\sigma_+^2+\tilde\sigma,
	\end{equation}
	and in combination with the computation carried out in~\eqref{eqns_auxiliaryexpressionsumsigmaij} concludes
	\begin{equation}
		\sigma^{ij}\sigma_{ij}\le\frac23\theta^2,
	\end{equation}
	which is identical to~\cite[Eq.~(20.19)]{ringstrom_cauchyproblem}. These two inequalities are the only ingredients necessary for the proof of~\cite[Lemma~20.12]{ringstrom_cauchyproblem}, which is the Bianchi class~A variant of the present lemma.

	In case $\mu>0$, $\gamma=0$, the constant in the previous lemma is zero. Consequently, $f_0\theta$ is non-decreasing, which is enough to conclude future completeness of inextendible causal geodesics, retracing the arguments in the proof of~\cite[Lemma~20.12]{ringstrom_cauchyproblem}.
\end{proof}
\begin{coro}
\label{coro_mghdforgammapositive}
	Consider a Bianchi~B development
	of initial data as in Def.~\ref{defi_initialdatabianchib}.
	If $\mu=0$ or~$\mu>0$, $\gamma>0$, then this development is isometric to the \mghd\ of the initial data.
\end{coro}
\begin{proof}
	The case of vacuum~$\mu=0$ is equivalent to~$Ric=0$, thus the timelike convergence condition is satisfied. The convergence properties of the mean \curv ~$\theta$ have been discussed in Subsection~\ref{constr_expnormtosolutionbasic}. With this, the statement is a direct consequence of Lemma~\ref{lemm_bianchibdevelopmentproperties}, Prop.~\ref{prop_developmentglobhyperbolic}, Prop.~\ref{prop_mghdnecessaryproperties} and the definition of the \mghd.

	In the matter case $\mu>0$, $\gamma>0$, Lemma~\ref{lemm_bianchibdevelopmentproperties} yields future completeness of inextendible cau\-sal geodesics. If the development was extendible towards the past as a globally hyperbolic \mf, then it would in particular be extendible as a semi-\riem\ \mf.
	In Lemma~\ref{lemm_curvblowupricci}, we show that the Ricci \curv\ contraction $R_{\alpha\beta}R^{\alpha\beta}$ is unbounded as $\tau\rightarrow{-}\infty$, using the expansion-normalised variables. The statement translates into the setting of the Bianchi~B development~$I\times G$, $I=(t_-,\infty)$, when we translate the condition on the time~$\tau$ into the condition that $t\rightarrow t_-$. By an argument similar to~\cite[Lemma~18.18]{ringstrom_cauchyproblem}, this contradicts extendibility.
\end{proof}
The only case not covered by this last statement is when~$\mu>0$, $\gamma=0$, in which case one can interpret the stress-energy tensor~$T_{\alpha\beta}$ as that of vacuum with a positive cosmological constant $\Lambda>0$.

We have to distinguish between the two situations which we encountered in Subsection~\ref{constr_expnormtosolutionbasic}: $q\equiv{-}1$ and $q>{-}1$. In the first case, $\theta=\theta_0$ is constant and all expansion-normalised variables vanish at all times. One concludes from the construction of Bianchi~B development that the commutators have the form
\begin{equation}
	\liebr{e_0}{e_i}={-}\frac13\theta_0 e_i,\qquad
	\liebr{e_i}{e_j}=0.
\end{equation}
From the second relation, we can conclude that a left-invariant metric on the universal covering group of the \liegr ~$G$ is isometric to Euclidean three-space $(\RR^3,\delta)$, while the first relation implies
\begin{equation}
	{}^tk=\frac{\theta_0}{3}\cdot{}^tg,
\end{equation}
with~${}^tg$ and~${}^tk$ the metric and second fundamental form of $\{t\}\times G$ \wrt\ the development.
Combined, the second fundamental form of $\{t\}\times G$ equals
\begin{equation}
	{}^tk=\frac{\theta_0}{3}\delta.
\end{equation}
One compares this with the spacetime
\begin{equation}
	(\RR\times\RR^3\,,\,g={-}dt^2+e^{2Ht}\delta),
\end{equation}
with $H=\sqrt{\Lambda/3}$, which can be isometrically embedded in the well-known \desitter\ spacetime, see~\cite[eq.~(52)]{schroedinger_expandinguniverses}. Upon comparison, one realises that the~$t=const.$ slices in this spacetime have the same induced metric and \fundform\ as we computed for our case~$\gamma=0$, $q\equiv{-}1$. We conclude that
a Bianchi~B development, which by Prop.~\ref{prop_developmentglobhyperbolic} is a globally hyperbolic development of the given initial data, corresponds to a development of a hypersurface in a quotient of the \desitter\ spacetime.
This concludes the case~$\gamma=0$, $q={-}1$.

In case~$q>{-}1$, we can show that the development we constructed cannot be isometrically embedded into a globally hyperbolic spacetime which extends the development to the past.
\begin{prop}
\label{prop_cosmconstpastinextend}
	Consider a Bianchi~B development~$(M=I\times G,g,\mu)$
	of initial data as in Def.~\ref{defi_initialdatabianchib}.
	If $\mu>0$, $\gamma=0$, and the mean \curv\ of~$\{t\}\times G$ covers the interval~$(\theta_\infty,\infty)$, then the development is past inextendible as a globally hyperbolic spacetime.
\end{prop}
\begin{proof}
	Suppose that the opposite holds, \ie there exists a proper globally hyperbolic extension~$(\tilde M,\tilde g,\tilde\mu)$ of the Bianchi~B development~$(M=I\times G,g,\mu)$, $I=(t_-,\infty)$, to the past. Then there exists a point~$p$ in the interior of~$\tilde M\setminus M$.
	As the timeslice~$\{0\}\times G$ is a Cauchy hypersurface by Prop.~\ref{prop_developmentglobhyperbolic}, the set
	\begin{equation}
		J^+(p)\cap J^-(\{0\}\times G)
	\end{equation}
	is compact by~\cite[Lemma~40, p.~423]{oneill_semiriemgeomappl}, where~$J^\pm$ denote the causal future and past. Consequently, there exists a unit speed causal curve
	\begin{equation}
		c:[s_-,s_+]\rightarrow\tilde M
	\end{equation}
	maximising the length between the point~$p$ and the timeslice~$\{0\}\times G$. Denote by~$q$ the endpoint of this curve on~$\{0\}\times G$.

	By construction, the curve also maximises the length between~$q$ and every earlier timeslice~$\{t\}\times G$, $t\in(t_-,0)$ and therefore the causal vector~$c'$ is \ogon\ to every one of these hypersurfaces. We conclude that in~$M$, the curve~$c$ has the form
	\begin{equation}
		c:s\mapsto (t(s),q).
	\end{equation}

	Assume that~$s_0\in(s_-,s_+)$ is the curve parameter such that~$c(s_0)\in
	\partial M\subset\tilde M$ and choose a sequence of times~$s_n\in(s_0,s_+)$ such that~$s_n\searrow s_0$. Set~$t_n$ to be the time satisfying~$c(s_n)\in\{t_n\}\times M$, for every~$n\in\NN$.

	Let now~$E_i$, $i=1,2,3$, be vector fields which are parallelly propagated along~$c$ and such that they form an \onorm\ basis of~$T_{c(0)}(\{0\}\times G)$. Consequently, they are \onorm\ along the whole curve and span the \ogon\ complement to~$c'(s)$, for every $s\in[s_-,s_+]$. Construct a local spacelike hypersurface~$\Sigma_{c(s_-)}$ through the point~$c(s_-)$ which is tangent to this frame, for example by using the exponential map. Consider then a smooth function~$\chi:[s_-,s_+]\rightarrow[0,1]$ with $\chi(s_-)=0$ and $\chi(s)=1$ for all~$s\ge s_0$. According to~\cite[Lemma~18.4]{ringstrom_cauchyproblem}, for every~$i=1,2,3$ and every~$n\in\NN$ there is a piecewise smooth variation~$\boldsymbol x$ of~$c|_{[s_-,s_n]}$, such that the variation vector field of~$\boldsymbol x$ is~$\chi E_i$, \ie there exists a continuous and piecewise smooth curve~$\boldsymbol x:[s_-,s_n]\times({-}\delta,\delta)\rightarrow\tilde M$ such that
	\begin{itemize}
		\item $\boldsymbol x(s,0)=c(s)$ for all $s\in[s_-,s_n]$,
		\item $\boldsymbol x_v(s,0)=\chi(s) E_i$ for all $s\in[s_-,s_n]$,
		\item $\boldsymbol x(s_-,\cdot)$ is contained in~$\Sigma_{c(s_-)}$ and $\boldsymbol x(s_n,\cdot)$ is contained in~$\{t_n\}\times G$.
	\end{itemize}
	Note that the subscript~$v$ denotes derivation with respect to the second component.

	We see that~$\boldsymbol x$ is a variation of~$c$ with fixed endpoint~$\boldsymbol x(s_-,\cdot)=c(s_-)=q$. By assumption, the curve~$c$ maximises the length between~$q$ and the hypersurface~$\{t_n\}\times G$. We can therefore apply~\cite[Lemma~18.7]{ringstrom_cauchyproblem} to find
	\begin{equation}
		0\ge\int_{s_-}^{s_n}\scalprod{R_{c'\,\chi\cdot E_i}\chi\cdot E_i}{c'}ds+\idfundform_{\{t_n\}\times G}[E_i(c(s_n)),E_i(c(s_n))]-\idfundform_\Sigma[0,0],
	\end{equation}
	where we have used that~$E_i'=0$ and the properties of the function~$\chi$. Here~$\idfundform_{\{t_n\}\times G}$ and~$\idfundform_\Sigma$ denote the \fundform\ of~$\{t_n\}\times G$ and~$\Sigma$, respectively.
	Summing over~$i$, we obtain the inequality
	\begin{equation}
	\label{eqn_estimatepastinextendibility}
		\int_{s_-}^{s_n}\chi^2(s)Ric(c'(s),c'(s))ds+\theta(t_n)\le0,
	\end{equation}
	where, as in the construction of a Bianchi~B development, $\theta(t)$ denotes the mean \curv\ of the timeslice~$\{t\}\times M$.

	We are interested in estimate~\eqref{eqn_estimatepastinextendibility} in the limit~$n\rightarrow\infty$.
	Due to the assumption on extendibility, the Ricci \curv\ is bounded along~$c$, consequently
	\begin{equation}
		\int_{s_-}^{s_+}\absval{Ric(c'(s),c'(s))ds}<\infty.
	\end{equation}
	However, the mean \curv~$\theta(t=t_n)$ diverges to~$\infty$ as~$n\rightarrow\infty$ by construction of the development, which is a contradiction.
\end{proof}
We can conclude from the previous proof that in a Bianchi~B development with~$\mu>0$, $\gamma=0$ and~$\theta$ covering~$(\theta_\infty,\infty)$, a causal curve connecting the hypersurfaces~$\{t_1\}\times G$ and~$\{t_2\}\times G$ has length at most~$\absval{t_1-t_2}$. In particular, causal curves have finite length towards the past. This stands in contrast to the \desitter\ spacetime, where causal curves have infinite length towards the past.

With Prop.~\ref{prop_cosmconstpastinextend} at hand, we can now prove that
in the remaining case~$\mu>0$ with $\gamma=0$ but~$q>{-}1$,
the Bianchi~B development is isometric to the \mghd.
\begin{coro}
\label{coro_mghdforgammazero}
	Consider a Bianchi~B development
	of initial data as in Def.~\ref{defi_initialdatabianchib}.
	If~$\mu>0$, $\gamma=0$, and the development does not correspond to a development of a hypersurface in a quotient of the \desitter\ spacetime, it
	is isometric to the \mghd\ of the initial data.
\end{coro}
\begin{proof}
	Recall the construction we carried out in Subsections~\ref{constr_initialdatatoexpnorm}--\ref{constr_solutionbasictomghd}. We argued in the second part of the construction that there are two different cases which can occur: $q\equiv{-}1$ and~$q>{-}1$. We showed above that the first case results into a development which is part of the \desitter\ spacetime, and is hence excluded by assumption.
	As $q>{-}1$, the mean \curv~$\theta(t)$ covers the interval~$(\theta_\infty,\infty)$, see the end of Subsection~\ref{constr_expnormtosolutionbasic}. As a result, we can apply Prop.~\ref{prop_cosmconstpastinextend} to see that the Bianchi~B development cannot be extended to the past as a globally hyperbolic development.
	We further know from Lemma~\ref{lemm_bianchibdevelopmentproperties} that every timelike geodesic is future complete.
	This concludes the proof.
\end{proof}
We collect the statements we made about Bianchi~B developments in the different cases in the following proposition.
\begin{prop}
\label{prop_mghdincompletedirections}
	Given \ogon\ perfect fluid Bianchi class~B initial data as in Def.~\ref{defi_initialdatabianchib}, consider the \mghd~$(\tilde M,\tilde g)$ to this data. In case the universal covering of this spacetime is not a part of the Minkowski or the \desitter\ spacetime, it is isometric to every Bianchi~B development~$(M=I\times G,g,\mu)$, $I=(t_-,t_+)$, as in Def.~\ref{defi_mghd}. The mean \curv~$\theta$ of $\{t\}\times G$ in~$M$ is positive, monotone, and tends to~$\infty$ at~$t_-$,
	where~${-}\infty<t_-<0$ and~$t_+=\infty$.

	Consider a causal geodesic in~$\tilde M$. The behaviour of any geometric quantity in the incomplete direction of this geodesic is the same as the behaviour of this geometric quantity on~$\{t\}\times G$ as~$t\rightarrow t_-$, while the complete direction corresponds to~$t\rightarrow t_+$.
\end{prop}
\begin{proof}
	If we exclude spacetimes isometric to (quotients of)
	the Minkowski spacetime, we can carry out the construction from Subsections~\ref{constr_initialdatatoexpnorm}--\ref{constr_solutionbasictomghd}, see also Lemma~\ref{lemm_excludeminkowski}. We have seen in Corollaries~\ref{coro_mghdforgammapositive} and~\ref{coro_mghdforgammazero} that the resulting Bianchi~B development is isometric to the \mghd\ of the data, if in case~$\gamma=0$ we exclude \desitter\ spacetime.

	The properties of $t_-,t_+$ and the convergence behaviour of~$\theta$ are a consequence of the construction, see also at the end of Subsection~\ref{constr_expnormtosolutionbasic}. The one case with~$t_-={-}\infty$ mentioned there corresponds to~$\gamma=0$ and~$\Omega=1$, or in other words~$q\equiv{-}1$, which implies \desitter\ spacetime and is therefore excluded by assumption.

	Consider an inextendible causal geodesic in~$\tilde M$. Using the isometry between~$\tilde M$ and the Bianchi~B development~$M$, we obtain an inextendible causal geodesic in~$M$. All geometric quantities are invariant under isometries, and in addition constant on every timeslice~$\{t\}\times G$ due to the metric being invariant under the action of the \liegr~$G$.

	As in the proof of Corollary~\ref{coro_mghdforgammazero}, the case~$q={-}1$ is excluded by assumption, and~$\theta$ covers the interval~$(\theta_\infty,\infty)$. Consequently, Prop.~\ref{prop_cosmconstpastinextend} applies and in combination with
	Lemma~\ref{lemm_bianchibdevelopmentproperties} implies that the geodesic in~$M$ is future complete and past incomplete. The incomplete direction of an inextendible geodesic in both~$\tilde M$ and~$M$ therefore corresponds to~$t\rightarrow t_-$, and the complete direction to~$t\rightarrow t_+$. This concludes the proof.
\end{proof}

\subsection{Solutions with additional symmetry}

\label{constr_highersymmetrysolutions}

In the results of this paper, several types of geometric initial data sets as well as invariant subsets are of interest which have additional properties. In particular, those with local rotational symmetry and those describing plane wave equilibrium solutions appear frequently.
We have defined both notions twice, first in the setting of geometric initial data and then in the setting of expansion-normalised variables. In this subsection, we show that these definitions are consistent provided we transform between geometric initial data sets and expansion-normalised variables the way we described in Subsections~\ref{constr_initialdatatoexpnorm}--\ref{constr_solutionbasictomghd}.

\smallskip

The notion of locally rotationally symmetric geometric initial data was given in Def.~\ref{defi_lrsinitialdata}, where such initial data is defined via the properties of a specific basis of the associated \liealg.
Given more generally some \onorm\ basis~$e_1,e_2,e_3$ of the \liealg~$\ggg$ such that
\begin{equation}
	\liebr{e_2}{e_1}=0=\liebr{e_2}{e_3},
	\qquad
	\liebr{e_1}{e_3}\parallel e_2,
	\qquad
	\idfundform=\diag(\idfundform_{11},\idfundform_{22},\idfundform_{33}=\idfundform_{11}),
\end{equation}
but not necessarily with~$e_2,e_3$ spanning a particular subalgebra, we find that these properties are preserved under a rotation in the~$e_1e_3$-plane. The term local rotational symmetry is thereful meaningful and should be compared to the definition of local rotational symmetry in Bianchi class~A spacetimes in~\cite[Def.~19.16]{ringstrom_cauchyproblem} and in more general spacetimes in~\cite{stewartellis_solseinsteinequsfluidlrs}.

Let us now discuss Definition~\ref{defi_lrsinitialdata} in connection with the different Bianchi types.
For a \liegr\ of type~V or~IV, the form of the structure constants from Table~\ref{table_bianchibclassification} reveals
\begin{equation}
	\liebr{e_1}{e_2}=a_1e_2+\nu_3e_3,\qquad
	\liebr{e_1}{e_3}=a_1e_3,\qquad
	\liebr{e_2}{e_3}=0,
\end{equation}
for a basis chosen as in the classification of \liegr s in Subsection~\ref{subsect_appendixbianchiclassification}. Note that the basis elements~$e_2,e_3$ span~$\ggg_2$, see Remark~\ref{rema_frameforggg2rotationtodiagonalise}. There is no rotation in the~$e_2e_3$-plane such that the commutators of the rotated basis have the form requested for local rotational symmetry. Consequently, no initial data set of Bianchi type~V or~IV can have local rotational symmetry.

Let us therefore assume that the \liegr\ is of type~VI or~VII. When comparing the properties of local rotational symmetry from Def.~\ref{defi_lrsinitialdata} with the decomposition of the structure constants in~\eqref{eqns_commutators}, we see that the vanishing of~$\liebr{e_2}{e_1}$ and~$\liebr{e_2}{e_3}$ is equivalent to~$a_1={-}n_{32}$ and~$n_{33}=0$, while~$\liebr{e_1}{e_3}\parallel e_2$ holds \iif~$a_1=n_{23}$.
In particular, we can consider the linear map~$A_2=\ad_{v_1}|_{\ggg_2}$, $v_1\in\ggg\setminus\ggg_2$, from Lemma~\ref{lemm_invardefibparaminv} and conclude that local rotational symmetry implies~$\det A_2=0$. The image of either~$\ad_{v_1}$ or~$A_2$, which coincides with~$\liebr{\ggg}{\ggg}$, is one-dimensional and spanned by~$e_2$. After a rotation in the~$e_2e_3$ plane, we can assume that~$n$ is diagonal, but lose the properties~$a_1=n_{23}={-}n_{23}$ and~$n_{33}=0$, see Remark~\ref{rema_frameforggg2rotationtodiagonalise}. A vanishing determinant of~$A_2$ implies that~$\binvparam={-}1$, hence the \liegr\ can only be of type~VI$_{{-}1}$.

Alternatively, let~$e_1,\tilde e_2,\tilde e_3$ a basis of the \liealg~$\ggg$ with~$\tilde e_2,\tilde e_3$ spanning~$\ggg_2$ and such that~$a=(a_1,0,0)$, $n=\diag(0,\nu_2,\nu_3)$, \ie a basis as in the classification of \liegr s in Subsection~\ref{subsect_appendixbianchiclassification}. Let us further assume that~$\binvparam={-}1$, \ie the \liegr\ is of type~VI$_{{-}1}$. In this case, we find from the matrix representation of~$A_2$ that
\begin{equation}
	\liebr{\ggg}{\ggg}=\operatorname{span}(\tilde e_2\pm\sqrt{\absval{\frac{\nu_2}{\nu_3}}}\tilde e_3),
\end{equation}
where the sign depends on the sign of~$a_1$. We further find that if we set~$e_2$ and~$e_3$ to be unit vectors in~$\ggg_2$, one spanning~$\liebr{\ggg}{\ggg}$ and one \ogon\ to it, and by a change of sign in~$e_1$ ensure that~$a_1$ is negative, then we recover that the commutator of this new basis satisfy the properties
\begin{equation}
	\liebr{e_2}{e_1}=0=\liebr{e_2}{e_3},
	\qquad
	\liebr{e_1}{e_3}\parallel e_2
\end{equation}
from the definition of local rotational symmetry, Def.~\ref{defi_lrsinitialdata}. Whether geometric initial with an associated \liealg\ of type~VI$_{{-}1}$ is locally rotationally symmetric then depends solely on the two-tensor~$\idfundform$.

Given geometric initial data as in Def.~\ref{defi_lrsinitialdata}, we can deduce from the commutators the corresponding basic variables via expressions~\eqref{eqns_prebasic_definitionpartone},~\eqref{eqns_prebasic_definitionparttwo},~\eqref{eqns_basic_definitionpartone} and~\eqref{eqns_basic_definitionparttwo} and find that local rotational symmetry implies that~$3\sigma_+^2=\tilde\sigma$ and~$\sigma_+n_+=\delta$. Further, the \liegr\ has to be of class~A or of type~VI$_{{-}1}$ due to the above discussion. The expansion-normalised variables result from applying~\eqref{eqn_expansionnormalisation} and~\eqref{eqn_definitionomega}, and one recovers the definition of locally rotationally solutions in expansion-normalised variables, Def.~\ref{defi_lrsexpansionnorm}, as Bianchi class~A corresponds to~$\tilde A=0$.

\smallskip

For the notion of plane wave equilibrium initial data, we turn to Def.~\ref{defi_planewaveinitialdata}.
In terms of structure constants for the four-dimensional spacetime, the property~$\gamma_{1A}^B+\gamma_{1B}^A=-2\idfundform_{AB}$ is equivalent to
\begin{equation}
\label{eqn_defiplanewavestructconst}
  \gamma_{12}^2=\gamma_{02}^2,\qquad \gamma_{13}^3=\gamma_{03}^3,\qquad\gamma_{12}^3+\gamma_{13}^2=\gamma_{02}^3+\gamma_{03}^2.
\end{equation}
This should be compared with the definition given in~\cite{hervikkundurilucietti_homogplanewavestability}, where the special spacelike direction is set to be~$e_1$ instead of~$e_n$. Applying an appropriate permutation of the basis elements to fit with our choice of frame, the definition given there reads
\begin{equation}
  \gamma_{1A}^B=\gamma_{0A}^B.
\end{equation}
In terms of geometric initial data, the \rhs\ is not defined, but its symmetrisation can be replaced by the expression~${-}\idfundform_{AB}$, see equations~\eqref{eqns_commutators} and Remark~\ref{rema_geometricmeaningthetasigmaAB}. This tensor is symmetric by definition, and also replacing the \lhs\ by its symmetrisation one obtains Def.~\ref{defi_planewaveinitialdata}.

Translating structure constants into basic variables via expressions~\eqref{eqns_prebasic_definitionpartone},~\eqref{eqns_prebasic_definitionparttwo},~\eqref{eqns_basic_definitionpartone}, and~\eqref{eqns_basic_definitionparttwo}, and using the constraint equation~\eqref{eqn_prebasic_constraint}, one finds that the relations~\eqref{eqn_defiplanewavestructconst} carry information equivalent to
\begin{equation}
  {-}\sigma_+(\theta+\sigma_+)=\tilde\sigma=\tilde n,\qquad \delta=0,\qquad \mu=0.
\end{equation}
Via~\eqref{eqn_expansionnormalisation} and~\eqref{eqn_definitionomega} one obtains that this is equivalent to expansion-normalised variables satisfying~$Z=0$, see~\eqref{eqn_functionzzeros}. This characterises the set~$\planewave {\bparamk}$, Def.~\ref{defi_planewaveexpansionnorm}, together with the point Taub~1, and
is equivalent to dropping the condition $\Sigma_+>{-}1$ in \equ ~\eqref{eqn_defiplanewave}.

\begin{rema}
\label{rema_taubpointsasinitialdata}
	In terms of initial data, the Kasner parabola~$\kasnerparabola$ corresponds to vacuum Bianchi type~I data. This can easily be deduced from Table~\ref{table_bianchiasubsets} and using the condition for vacuum, $\Omega=0$, in equation~\eqref{eqn_omegageneral}.

	The two special Taub points have the following characterisation:
	\begin{itemize}
		\item The point Taub~1 corresponds to initial data of Bianchi type~I which is of plane wave equilibrium type.
		\item The point Taub~2 corresponds to initial data of Bianchi type~I which is locally rotationally symmetric  and additionally
		the unique eigenvalue of~$\idfundform$ in the rotation plane is smaller than the unique eigenvalue along the rotation axis. In a basis satisfying the conditions for local rotation symmetry from Definition~\ref{defi_lrsinitialdata}, this means that
		${\idfundform}_{11}<{\idfundform}_{22}$ holds.
	\end{itemize}
	This follows from the consistency check we carried out before. Geometric initial data being locally rotationally symmetric is equivalent to the corresponding expansion-normalised variables satisfying Def.~\ref{defi_lrsexpansionnorm}, and geometric initial data of plane wave equilibrium type corresponds to Def.~\ref{defi_planewaveexpansionnorm}, including the possibility that $\Sigma_+={-}1$.
	Intersection with the Kasner parabola~$\kasnerparabola$, which consist of all Bianchi type~I vacuum spacetimes, yields the characterisation: Taub~1 is the unique point on the Kasner parabola which satisfies the plane wave equilibrium point relations. Taub~2 is one of two intersection points between the Kasner parabola~$\kasnerparabola$ and the set of LRS solutions, the one with~$\Sigma_+>0$. This last property translates into $2{\idfundform}_{11}-{\idfundform}_{22}-{\idfundform}_{33}<0$, and as ${\idfundform}_{11}={\idfundform}_{33}$ for LRS geometric initial data, the statement follows.
\end{rema}

\section{Proof of the main theorems}
\label{section_proofmainthms}

In this final section, we give the proofs of the main theorems stated in the introduction. Four of them, Thm~\ref{theo_fullmeasurelimitset}, Thm~\ref{theo_curvblowupexpansionnormalisedmatter}, Thm~\ref{theo_curvblowupexpansionnormalisedvacuum}, and Thm~\ref{theo_curvblowupexpansionnormalisedcosmconst}, are stated in the setting of expansion-normalised variables, and these are proven first. The remaining statements, Thm~\ref{theo_curvblowupinitialdatamatter} and Thm~\ref{theo_curvblowupinitialdatavacuum}, are translated versions of Thm~\ref{theo_curvblowupexpansionnormalisedmatter} and Thm~\ref{theo_curvblowupexpansionnormalisedvacuum}, giving the results in terms of geometric initial data and the corresponding \mghd.

In the previous sections, we have determined the asymptotic behaviour of solutions to the evolution \equ s~\eqref{eqns_evolutionbianchib}--\eqref{eqn_evolutionomega}. We have found possible~$\alpha$-limit sets and determined in detail the solutions which converge to the Kasner parabola~$\kasnerparabola$ or the plane wave equilibrium points~$\planewave {\bparamk}$ as~$\tau\rightarrow{-}\infty$.
Thm~\ref{theo_curvblowupexpansionnormalisedmatter}, Thm~\ref{theo_curvblowupexpansionnormalisedvacuum}, and Thm~\ref{theo_curvblowupexpansionnormalisedcosmconst} state that apart from a short list of exceptional solutions, either the Kretschmann scalar $R_{\alpha\beta\gamma\delta}R^{\alpha\beta\gamma\delta}$ or the contraction of the Ricci tensor with itself $R_{\alpha\beta}R^{\alpha\beta}$ or both become unbounded along solutions as~$\tau\rightarrow{-}\infty$. This implies Strong Cosmic Censorship in the~$C^2$-sense.
We start with a discussion of these two geometric invariants and their form and asymptotic behaviour in terms of expansion-normalised variables, before we are then in a position to give the proof of these three theorems.

We then continue this section with a proof of Thm~\ref{theo_fullmeasurelimitset}, where we show that apart from a `small` subset, all solutions to the evolution \equ s~\eqref{eqns_evolutionbianchib}--\eqref{eqn_evolutionomega} converge to a Kasner point to the right of Taub~2, as $\tau\rightarrow{-}\infty$. For this proof as well, we heavily rely on the results on asymptotic behaviour of solutions which we obtained in the previous sections.

At the end of this section, we conclude with the proofs of the main statements in the initial data perspective, Thm~\ref{theo_curvblowupinitialdatamatter} and Thm~\ref{theo_curvblowupinitialdatavacuum}. The equivalent statements in expansion-normalised variables are given in Thm~\ref{theo_curvblowupexpansionnormalisedmatter} and Thm~\ref{theo_curvblowupexpansionnormalisedvacuum}. Using the transformation between this set of variables and the \mghd\ to given geometric initial data which we constructed in Subsections~\ref{constr_initialdatatoexpnorm}--\ref{constr_solutionbasictomghd}, we re-translate the statements back to the setting of geometric initial data.
Applying the proof to~$\mu>0$, $\gamma=0$, \ie the stress-energy tensor of a positive cosmological constant in vacuum, even justifies the statement given in Remark~\ref{rema_curvblowupcosmconst} as a re-translated version of Thm~\ref{theo_curvblowupexpansionnormalisedcosmconst}.

\smallskip

We start with computing several curvature expressions in terms of expansion-normalised variables.
The Weyl tensor of a four-dimensional \mf\ with metric~$g$ is given by
\begin{equation}
	C_{\alpha\beta\gamma\delta}=R_{\alpha\beta\gamma\delta}-(g_{\alpha[\gamma}R_{\delta]\beta}-g_{\beta[\gamma}R_{\delta]\alpha})+\frac13Sg_{\alpha[\gamma}g_{\delta]\beta},
\end{equation}
see~\cite[Eq.~(3.2.28)]{wald_generalrelativity}.
For a spacetime satisfying Einstein's \equ ~\eqref{eqn_einsteineqn} with stress-energy tensor of a perfect fluid~\eqref{eqn_stressenergyperfectfluid} with linear \equ\ of state~\eqref{eqn_lineareqnofstate}, one computes that the contraction of the Ricci tensor with itself is
\begin{equation}
\label{eqn_contractionricciricci}
	R_{\alpha\beta}R^{\alpha\beta}=\mu^2+3p^2=(1+3(\gamma-1)^2)\mu^2,
\end{equation}
and finds the Kretschmann scalar
\begin{align}
R_{\alpha\beta\gamma\delta}R^{\alpha\beta\gamma\delta}
	={}&C_{\alpha\beta\gamma\delta}C^{\alpha\beta\gamma\delta}+2R_{\alpha\beta}R^{\alpha\beta}-\frac13S^2\\
	={}&C_{\alpha\beta\gamma\delta}C^{\alpha\beta\gamma\delta}+\frac13(4+(3\gamma-2)^2)\mu^2.
\end{align}
The relation between the Weyl tensor and the expansion normalised variables $(\Sigma_+,\tilde\Sigma,\Delta,\tilde A,N_+)$ has been computed in~\cite[App.~B]{hewittwainwright_dynamicalsystemsapproachbianchiorthogonalB}, referring to~\cite{ellismaccallum_classofhomogcosmmodels}:
\begin{equation}
\label{eqn_weyltensorexpansionnormalised}
	C_{\alpha\beta\gamma\delta}C^{\alpha\beta\gamma\delta}=\theta^4(\tilde E_{AB}\tilde E^{AB}+\frac23E_+^2)-\theta^4(\tilde H_{AB}\tilde H^{AB}+\frac23H_+^2),
\end{equation}
where the electrical components are
\begin{align}
	\tilde E_{AB}\tilde E^{AB}={}&\frac2{27}[\tilde\Sigma(2\Sigma_+-1)^2+4(\tilde A\Sigma_++\tilde N(\tilde A+3\tilde N-2\tilde\Sigma))\\
		& +4\Delta^2-4N_+\Delta(2\Sigma_+-1)+4\bparamk\tilde A\tilde N],\\
	E_+={}&\frac13(\tilde \Sigma -\Sigma_+(1+\Sigma_+)-2\tilde N),
\end{align}
and the magnetic ones are
\begin{align}
	\tilde H_{AB}\tilde H^{AB}={}&\frac2{27}[9\Sigma_+^2\tilde N+6\tilde\Sigma\tilde N+\tilde\Sigma\tilde A+6\Delta^2+12\Sigma_+N_+\Delta+4\bparamk\tilde A\tilde\Sigma],\\
	H_+={}&{-}\Delta.
\end{align}
\begin{lemm}
\label{lemm_curvblowupricci}
	Consider a solution to \equs ~\eqref{eqns_evolutionbianchib}--\eqref{eqn_evolutionomega}. If $\Omega>0$ and $\gamma>0$, then
	\begin{equation}
		\lim_{\tau\rightarrow{-}\infty}R_{\alpha\beta}R^{\alpha\beta}=\infty.
	\end{equation}
	If $\Omega>0$, $\gamma=0$, or if $\Omega=0$, then $R_{\alpha\beta}R^{\alpha\beta}$ remains bounded as $\tau\rightarrow{-}\infty$.
\end{lemm}
The central idea of proof has already been given in~\cite{ringstrom_asymptbianchiaspacetimes}, where the Bianchi~A case was discussed. It only relies upon \equ ~\eqref{eqn_contractionricciricci}, which follows immediately from the assumption on the stress-energy tensor and is independent of the Bianchi class of the \liegr.
Note that only the evolution equations for~$\theta$ and~$\Omega$ come into play. In particular, this statement can be obtained without knowledge on the detailed asymptotic properties of the individual variables. 
\begin{proof}
	The contraction of the Ricci tensor with itself is given by~\eqref{eqn_contractionricciricci}.
	The density~$\mu$ satisfies \equ ~\eqref{eqn_definitionomega}, and the evolution equations~\eqref{eqn_evolutiontheta} and~\eqref{eqn_evolutionomega} for~$\theta$ and~$\Omega$ yield
	\begin{align}
		9\mu^2(\tau)={}&\Omega^2\theta^4(\tau)\\
			={}&\Omega(\tau_0)^2\theta(\tau_0)^4\exp(\int_{\tau_0}^\tau -4-4q +4q -2(3\gamma-2)ds)\\
			={}&\Omega(\tau_0)^2\theta(\tau_0)^4\exp({-}6\gamma(\tau-\tau_0)).
	\end{align}
	As $\Omega>0$ is an invariant set by \equ ~\eqref{eqn_evolutionomega}, and $\theta\not=0$ by construction of the expansion-normalised coordinates in Subsections~\ref{constr_initialdatatoexpnorm}--\ref{constr_solutionbasictomghd}, the statement follows.
\end{proof}
\begin{lemm}
\label{lemm_curvblowupvacuum}
	Assume either vacuum or inflationary matter with $\gamma=0$, \ie $\Omega=0$ or $\Omega>0$, $\gamma=0$, and consider a solution to \equs ~\eqref{eqns_evolutionbianchib}--\eqref{eqn_evolutionomega}. Assume that the solution has an $\alpha$-limit point $(\Sigma_+,1-\Sigma_+^2,0,0,0)$ with $\Sigma_+\notin\{{-}1,1/2\}$. Then
	\begin{equation}
		\limsup_{\tau\rightarrow{-}\infty}\absval{\kretschmann}=\infty.
	\end{equation}
\end{lemm}
\begin{proof}
	Vacuum is equivalent to vanishing Ricci curvature, consequently
	\begin{equation}
		\kretschmann=C_{\alpha\beta\gamma\delta}C^{\alpha\beta\gamma\delta}
	\end{equation}
	in this case,
	while for $\Omega>0$, $\gamma=0$, the argument from the previous proof implies that
	\begin{equation}
		\kretschmann-C_{\alpha\beta\gamma\delta}C^{\alpha\beta\gamma\delta}
		=2R_{\alpha\beta}R^{\alpha\beta}-\frac13S^2=\frac13(4+(3\gamma-2)^2)\mu^2
	\end{equation}
	is a constant.

	Direct computation shows that on the Kasner parabola~$\kasnerparabola$ one finds
	\begin{equation}
		\tilde E_{AB}\tilde E^{AB}+\frac23E_+^2-\tilde H_{AB}\tilde H^{AB}-\frac23H_+^2=\frac4{27}(2\Sigma_+-1)^2(\Sigma_++1),
	\end{equation}
	which is non-zero \iif\ $\Sigma_+\notin\{{-}1,1/2\}$. The statement then follows from \equ ~\eqref{eqn_weyltensorexpansionnormalised} and the fact that $\theta\rightarrow\infty$ as $\tau\rightarrow{-}\infty$, due to \equ ~\eqref{eqn_evolutiontheta}.
\end{proof}
\begin{lemm}
\label{lemm_expressionforweyltensorvanishes}
	The expression
	\begin{equation}
		\tilde E_{AB}\tilde E^{AB}+\frac23E_+^2-\tilde H_{AB}\tilde H^{AB}-\frac23H_+^2
	\end{equation}
	vanishes for the Kasner points Taub~1 and~2, for the plane wave equilibrium points~$\planewave {\bparamk}$ and for the point $\Sigma_+=\tilde\Sigma=\Delta=\tilde A=N_+=0$.
\end{lemm}
\begin{proof}
	This follows from direct computation, which we in the case of the Taub points already carried out in the previous proof, and where the set~$\planewave {\bparamk}$ is defined in Def.~\ref{defi_planewaveexpansionnorm}.
\end{proof}
\begin{rema}
	On the Taub points~1 and~2, on the plane wave equilibrium points~$\planewave {\bparamk}$, and on the point $\Sigma_+=\tilde\Sigma=\Delta=\tilde A=N_+=0$, the contraction of the Weyl tensor with itself, expression~\eqref{eqn_weyltensorexpansionnormalised}, vanishes due to Lemma~\ref{lemm_expressionforweyltensorvanishes}. Together with the result from Lemma~\ref{lemm_curvblowupricci} and equation~\eqref{eqn_contractionricciricci}, this implies that both the Kretschmann scalar $R_{\alpha\beta\gamma\delta}R^{\alpha\beta\gamma\delta}$ and the contraction of the Ricci tensor with itself $R_{\alpha\beta}R^{\alpha\beta}$ remain bounded for the constant solutions
	in the points Taub~1 and Taub~2, the plane wave equilibrium points and the point $\Sigma_+=\tilde\Sigma=\Delta=\tilde A=N_+=0$.
\end{rema}
We now give the proofs of the theorems stated in expansion-normalised variables.
\begin{proof}[Proof of Thm~\ref{theo_curvblowupexpansionnormalisedmatter}]
	This statement is an immediate consequence of Lemma~\ref{lemm_curvblowupricci}.
\end{proof}
\begin{proof}[Proof of Thm~\ref{theo_curvblowupexpansionnormalisedvacuum}]
	Prop.~\ref{prop_alphalimitsets_vacuum_inflat} restricts the possible $\alpha$-limit sets of solutions in vacuum~$\Omega=0$ to the Kasner parabola~$\kasnerparabola$ and the plane wave equilibrium points~$\planewave {\bparamk}$.
	Non-constant solutions with $\alpha$-limit set in~$\planewave {\bparamk}$ are immediately excluded by the same statement.
	One knows further from Prop.~\ref{prop_convergencetokasner} that solutions with $\alpha$-limit set in~$\kasnerparabola$ have a unique $\alpha$-limit point, \ie they converge as $\tau\rightarrow{-}\infty$.
	We now discuss the different possible locations of limit points for solutions to equations~\eqref{eqns_evolutionbianchib}--\eqref{eqn_evolutionomega} and check in which cases the Kretschmann scalar possibly remains bounded.

	Lemma~\ref{lemm_curvblowupvacuum} states that the Kretschmann scalar becomes unbounded upon convergence to every Kasner points but the points Taub~1 and~2.
	The only solution converging to the point Taub~1 is the constant one, see Prop.~\ref{prop_taubone}.
	All solutions converging to Taub~2 are characterised in Thm~\ref{theo_taubtwocharacterisationoforbits}.
	This concludes the proof.
\end{proof}
\begin{proof}[Proof of Thm~\ref{theo_curvblowupexpansionnormalisedcosmconst}]
	The proof proceeds similarly to the previous one.
	Prop.~\ref{prop_alphalimitsets_vacuum_inflat} restricts the possible $\alpha$-limit sets of solutions with~$\Omega>0$, $\gamma=0$, to the Kasner parabola~$\kasnerparabola$, the plane wave equilibrium points~$\planewave {\bparamk}$, and the point $\Sigma_+=\tilde\Sigma=\Delta=\tilde A=N_+=0$.
	Non-constant solutions converging to the point $\Sigma_+=\tilde\Sigma=\Delta=\tilde A=N_+=0$ are immediately excluded by the same proposition, which also states that solutions whose $\alpha$-limit set intersects both~$\kasnerparabola\setminus\taubone$ and~$\planewave {\bparamk}$ do not exist.
	One knows further from Prop.~\ref{prop_convergencetokasner} and Prop.~\ref{prop_convergencetoplanewave} that solutions with $\alpha$-limit set in either~$\kasnerparabola$ or~$\planewave {\bparamk}$ have a unique $\alpha$-limit point, \ie they converge as $\tau\rightarrow{-}\infty$.
	We now discuss the different possible locations of limit points and check in which cases the Kretschmann scalar possibly remains bounded.

	The constant solution in the point $\Sigma_+=\tilde\Sigma=\Delta=\tilde A=N_+=0$ satisfies~$\Omega=1$. Due to the previous remark, the Kretschmann scalar
	remains bounded.

	Let us now turn to solutions converging to a plane wave equilibrium point, \ie solutions contained in~$\setconvplanewave$.
	We apply Thm~\ref{theo_centreunstablegloballyplanewave} and find that the case~${-}(3\gamma-2)/4<\slimit<0$ is excluded by~$\gamma=0$. Every solution in~$\setconvplanewave$ is contained in the union of sub\mf s~$\submfsplanewave_m$ or~$\submfsplanewavezero_m$, and the limiting value~$\slimit$ and the parameter~${\bparamk}$ have to satisfy relation~\eqref{eqn_inequalityplanewave}. In Remark~\ref{rema_planewaveconvergencebianchiseparated}, we list the possible values depending on the Bianchi type. Setting the union of sub\mf s~$\{\submfsplanewavetheorem_m\}$ to include all~$\submfsplanewave_m$ and~$\submfsplanewavezero_m$ gives the statement.

	The arguments regarding solutions converging to a point on the Kasner parabola~$\kasnerparabola$ are identical to the ones given in the previous proof. Note that only the non-constant solutions, \ie locally rotationally symmetric solutions of Bianchi type~I, II or~VI$_{{-}1}$ can satisfy~$\Omega>0$. This concludes the proof.
\end{proof}
For the case of vacuum, our statement Thm~\ref{theo_curvblowupexpansionnormalisedvacuum} makes precise and proves a claim by~\cite[p.~165]{wainwrightellis_dynamsystemsincosm} saying that all solutions except those contained in the unstable manifold of the point Taub~2 have an initial curvature singularity.

\begin{proof}[Proof of Thm~\ref{theo_fullmeasurelimitset}]
	The three sets describing solutions of Bianchi type~VI$_{\binvparam}$, VII$_{\binvparam}$ and~IV are open sets of~$\RR^5$. Restricted to the constraint \equ s~\eqref{eqn_constraintgeneralone} and~\eqref{eqn_constraintgeneraltwo}, they form open subsets of the set defined by these equations. Consequently, these three Bianchi sets are of dimension four. Restricting to vacuum~$\Omega=0$ yields sets of dimension three. The set describing solutions of Bianchi type~V is a two-dimensional closed subset of~$\RR^5$ and contained in the set defined by the constraint \equ s~\eqref{eqn_constraintgeneralone} and~\eqref{eqn_constraintgeneraltwo}. Restricted to vacuum~$\Omega=0$, this Bianchi type yields a set of dimension~one. We prove the theorem by showing that all solutions with convergence behaviour different from the one in the statement are contained in countable unions of smooth sub\mf s of positive codimension.

	Independently of the Bianchi type, the $\alpha$-limit set of a solution to \equ s~\eqref{eqns_evolutionbianchib}--\eqref{eqn_evolutionomega} with either $\Omega=0$ or $\Omega>0$, $\gamma\in\left[0,2/3\right)$ is given in Prop.~\ref{prop_alphalimitsets_vacuum_inflat}: We exclude constant solutions as they are contained in two smooth curve arcs together with one additional point, namely the Kasner parabola~$\kasnerparabola$, the plane wave equilibrium points~$\planewave \bparamk$ and the point Taub~1.
	In Bianchi type~V, we additionally notice that the intersection between these two arcs and the subset defining this Bianchi type consists of two points.

	Prop.~\ref{prop_taubone} states that only the constant solution has Taub~1 as an $\alpha$-limit point, so in particular we exclude solutions with Taub~1 as an $\alpha$-limit point. We conclude from Prop.~\ref{prop_alphalimitsets_vacuum_inflat} that the $\alpha$-limit set is contained in either the Kasner parabola~$\kasnerparabola$ without the point Taub~1 or in the plane wave equilibrium points~$\planewave {\bparamk}$, and also see that latter set can only occur for inflationary matter, \ie for $\Omega>0$, $\gamma\in\left[0,2/3\right)$.

	In Prop.~\ref{prop_convergencetokasner} and Prop.~\ref{prop_convergencetoplanewave}, we further showed that the $\alpha$-limit sets consist of exactly one point meaning that solutions converge as $\tau\rightarrow{-}\infty$. All that remains to show now is that the set of solutions converging to a plane wave equilibrium point~$\planewave {\bparamk}$ or a point on the Kasner parabola~$\kasnerparabola$ with ${-}1<\Sigma_+\le1/2$, \resp the complement of this set of solutions, has the necessary properties.

	\smallskip

	We start with solutions converging to a point on the Kasner parabola~$\kasnerparabola$.
	According to Thm~\ref{theo_taubtwocharacterisationoforbits}, the only Bianchi class~B solutions converging to the point Taub~2 are those of Bianchi type~VI$_{{-}1}$ which are locally rotationally symmetric. They are characterised by
	\begin{equation}
		{\bparamk}={-}1,\qquad \tilde A>0,\qquad 3\Sigma_+^2=\tilde\Sigma,\qquad \Sigma_+N_+=\Delta,
	\end{equation}
	and are solutions of Bianchi type~VI$_{{-}1}$. In particular, no Bianchi type~VII$_{\binvparam}$, type~IV  or type~V solution converges to the point Taub~2. The set is a~$C^1$ sub\mf\ of dimension three, and its restriction to vacuum of dimension~two. In particular, it forms a subset of positive codimension in the set of all Bianchi type~VI$_{{-}1}$ solutions, both in case of vacuum and inflationary matter.

	The solutions converging to a Kasner point to the left of Taub~2 are given in Thm~\ref{theo_leftoftaubtworesult}.
	In case of Bianchi type~VI$_{\binvparam}$,
	solutions either satisfy $\tilde A>0$, $\Delta=0=N_+$, $3\Sigma_+^2+{\bparamk}\tilde\Sigma=0$ (Bianchi type VI$_{\binvparam}$ with $n^\alpha{}_\alpha=0$, see Table~\ref{table_bianchihighersymmetry}), or have to be contained in a countable union of $C^1$ sub\mf s satisfying $\tilde A>0$ and~$\Delta$, $N_+$ not both vanishing identically.
	These sub\mf s are contained either in the set of all non-vacuum solutions or the set of all vacuum solutions, and in the respective set have codimension at least~one.
	Solutions of Bianchi type~VII$_{\binvparam}$ cannot converge to a Kasner point to the left of Taub~2 due to~$\bparamk>0$. In the remaining two cases, Bianchi type~IV and~V, the restriction~${\bparamk}=0$ implies that solutions can only converge to the Kasner point with~$\slimit=0$.

	For non-constant solutions converging to a plane wave equilibrium point, Thm~\ref{theo_centreunstablegloballyplanewave} states that they converge to the arc $\planewave {\bparamk}\cap\{\Sigma_+\le{-}(3\gamma-2)/4\}$, which includes the point~$\planewave {\bparamk}\cap\{\Sigma_+=0\}$ due to the assumption on~$\gamma$. All solutions are contained in the countable union of $C^1$ sub\mf s~$\submfsplanewave_m$ or~$\submfsplanewavezero$ whose dimension is at most two. This is enough to identify them as subsets of positive codimension in the Bianchi type~VI$_{\binvparam}$,~VII$_{\binvparam}$ and~IV. For Bianchi type~V solutions, we have argued in Remark~\ref{rema_planewaveconvergencebianchiseparated} that the solution is contained in the subset~$\Sigma_+=\tilde\Sigma=\Delta=N_+=0$, which is a subset of positive codimension in non-vacuum.

	\smallskip

	Collecting the different partial results, we conclude the following: In the case of Bianchi type~VI$_{\binvparam}$ or VII$_{\binvparam}$, if a solution does not converge to a Kasner point to the right of Taub~2, then it is contained in a countable union of $C^1$ sub\mf s of positive codimension, both in case of vacuum and inflationary matter.
	Therefore, this set of exceptions has the right properties: its complement is of full measure and a countable intersection of open and dense sets.
	We argue similarly in case of Bianchi type~IV, with the additional possibility of convergence to the Kasner point with~$\slimit=0$. For the case of Bianchi type~V, convergence to Kasner point to the right of Taub~2 is not possible, as all such Bianchi solutions satisfy~$\Sigma_+=0$ by definition. We have shown that all non-constant solutions converge to the point on the Kasner parabola~$\kasnerparabola$ satisfying~$\slimit=0$.
\end{proof}

We conclude this section with the proof of the two remaining theorems, which are stated in the setting of geometric initial data. As large portions of the proofs are identical, we combine them into one.
\begin{proof}
[Proofs of Thm~\ref{theo_curvblowupinitialdatamatter} and Thm~\ref{theo_curvblowupinitialdatavacuum}]
	We consider initial data which is of Bianchi class~B. This implies in particular that
	the Minkowski spacetime or quotients of this space cannot occur as development of the data.
	This excludes the case~$\theta=\tr_\idmetric\idfundform=0$, see Lemma~\ref{lemm_excludeminkowski}, and we can therefore consider a development of the data~$(I\times G,g,\mu)$ as in Def.~\ref{defi_mghd}, \ie one constructed as in Subsections~\ref{constr_initialdatatoexpnorm}--\ref{constr_solutionbasictomghd}. According to Corollary~\ref{coro_mghdforgammapositive}, this development is isometric to the \mghd\ of the given initial data. We further know from Prop.~\ref{prop_mghdincompletedirections} that the behaviour of any geometric quantity in the incomplete direction of causal geodesics in~$M$ corresponds to the behaviour of this geometric quantity on~$\{t\}\times G$ as~$t\rightarrow t_-$.

	We now switch to the point of view of expansion-normalised variables, as explained in Subsection~\ref{constr_initialdatatoexpnorm}. The behaviour of any geometric quantity in the incomplete direction of causal geodesics in~$M$ consequently corresponds to the behaviour of this geometric quantity as~$\tau\rightarrow{-}\infty$, see also the construction of the development, specifically the end of Subsection~\ref{constr_expnormtosolutionbasic}.

	The result in the matter case~$\mu_0>0$, $\gamma>0$, which is the setting of Thm~\ref{theo_curvblowupinitialdatamatter} is an immediate consequence of Thm~\ref{theo_curvblowupexpansionnormalisedmatter}, as the definition of the density parameter in equation~\eqref{eqn_definitionomega} yields~$\Omega>0$ at~$\tau=0$, and this property is conserved by the evolution equation~\eqref{eqn_evolutionomega}.

	Consider now the vacuum case~$\mu_0=0$, \ie the setting of Thm~\ref{theo_curvblowupinitialdatavacuum}.
	In the setting of expansion-normalised variables this corresponds to~$\Omega=0$, and we have listed all exceptions to \curv\ blow-up as $\tau\rightarrow{-}\infty$ in Thm~\ref{theo_curvblowupexpansionnormalisedvacuum}. To prove the theorem for geometric initial data, we have to carry over the individual exceptions to the geometric initial data perspective. As we only consider initial data with a \liegr\ of Bianchi class~B, a number of these exceptions cannot occur by assumption, namely those assuming Bianchi type~I and~II. Note also that the point Taub~1 is a constant Bianchi type~I solution.

	We have argued in Subsection~\ref{constr_highersymmetrysolutions} that the definitions of local rotational symmetry in terms of geometric initial data, Def.~\ref{defi_lrsinitialdata}, and in terms of expansion-normalised variables, Def.~\ref{defi_lrsexpansionnorm}, carry equivalent information. This proves the statement.
\end{proof}
\begin{rema}
	We can adapt the proof to justify the statement made in Remark~\ref{rema_curvblowupcosmconst}. Assuming~$\mu>0$ and~$\gamma=0$, we switch to expansion-normalised variables. The \desitter\
	spacetime or quotients of this space cannot occur as development of the data due to the assumption on the Bianchi class. It therefore follows from Prop.~\ref{prop_mghdincompletedirections} that the spacetime constructed as in Subsections~\ref{constr_initialdatatoexpnorm}--\ref{constr_solutionbasictomghd} is in fact the \mghd\ of the given geometric initial data, otherwise the arguments are identical to those given in the previous proof. The exceptions to \curv\ blow-up are LRS Bianchi~VI$_{{-}1}$ solutions and solutions converging to plane wave equilibrium points, see Thm~\ref{theo_curvblowupexpansionnormalisedcosmconst}.
\end{rema}

\appendix

\section{Additional properties of expansion-normalised variables}
\label{section_appendixbianchivariables}

\subsection{The linearised evolution equations on the Kasner parabola}
\label{subsect_appendixlinearisedevolutionkasner}

By the linearised evolution equations in the extended five-dimensional state space we mean the linear approximation of the evolution \equ s~\eqref{eqns_evolutionbianchib}. For points on the Kasner parabola~$\kasnerparabola$, this is the vector field $\RR^5\rightarrow\RR^5$ given by the matrix
\begin{equation}
	\begin{pmatrix}
	3(2-\gamma)\Sigma_+^2 & \frac32(2-\gamma)\Sigma_+ & 0 & \frac12(3\gamma-2)(\frac {\bparamk}3-1)\Sigma_++\frac23{\bparamk} & 0\\
	6(2-\gamma)\Sigma_+\tilde\Sigma & 3(2-\gamma)\tilde\Sigma & 0 & (3\gamma-2)(\frac {\bparamk}3-1)\tilde \Sigma-4\Sigma_+ & 0 \\
	0 & 0 & 2\Sigma_++2 & 0 & 2\tilde\Sigma \\
	0 & 0 & 0 & 2(2+2\Sigma_+) & 0 \\
	0 & 0 & 6 & 0 & 2+2\Sigma_+
	\end{pmatrix}.
\end{equation}
Note that~$\gamma$ and~${\bparamk}$ are constants, and we did not replace~$\tilde\Sigma$ by $1-\Sigma_+^2$ to make it more readable.

As there appear to be typos in the eigenvalues in both~\cite[Sect.~4.4]{hewittwainwright_dynamicalsystemsapproachbianchiorthogonalB} and~\cite[Sect.~7.2.3]{wainwrightellis_dynamsystemsincosm},
we give here the corrected eigenvalues and state the corresponding eigenvectors.
\begin{itemize}
	\item The eigenspace to eigenvalue~$0$ is tangential to the Kasner parabola, it is spanned by
	\begin{equation}
		(1\,,\,{-}2\Sigma_+\,,\,0\,,\,0\,,\,0).
	\end{equation}
	\item The eigenspace to eigenvalue $2(1+\Sigma_++\sqrt{3(1-\Sigma_+^2)})$ lies in the $\Delta N_+$-plane and is spanned by
	\begin{equation}
		(0\,,\,0\,,\,+\frac13\sqrt{3(1-\Sigma_+^2)}\,,\,0\,,\,1).
	\end{equation}
	\item The eigenspace to eigenvalue $2(1+\Sigma_+-\sqrt{3(1-\Sigma_+^2)})$ lies in the $\Delta N_+$-plane and is spanned by
	\begin{equation}
		(0\,,\,0\,,\,{-}\frac13\sqrt{3(1-\Sigma_+^2)}\,,\,0\,,\,1).
	\end{equation}
	\item The eigenspace to eigenvalue $4(1+\Sigma_+)$ is spanned by
	\begin{equation}
		(\frac16({\bparamk}\Sigma_+-3\Sigma_++{\bparamk})\,,\,\frac13(3\Sigma_+^2-{\bparamk}\Sigma_+^2+{\bparamk}-3-3\Sigma_+)\,,\,0\,,\,1+\Sigma_+\,,\,0).
	\end{equation}
	\item The eigenspace to eigenvalue $3(2-\gamma)$ is spanned by
	\begin{equation}
		(\Sigma_+\,,\,2(1-\Sigma_+^2)\,,\,0\,,\,0\,,\,0)
	\end{equation}
\end{itemize}
Whenever two eigenvalues coincide, the eigenspace consists of all linear combinations of the respective eigenvectors.

\subsection{The linearised evolution equations on the plane wave equilibrium points}
\label{subsect_appendixlinearisedevolutionplanewave}

For points on the plane wave equilibrium points~$\planewave {\bparamk}$, linearising the evolution equations~\eqref{eqns_evolutionbianchib} gives raise to the vector field $\RR^5\rightarrow\RR^5$ given by the matrix
\begin{align}
&\left(
\begin{matrix}
	{-}2\Sigma_+-2+3(2-\gamma)\Sigma_+^2 & \frac32(2-\gamma)\Sigma_+ & \dots \\
	{-}6(2-\gamma)\Sigma_+^2(\Sigma_++1)-4(\Sigma_++1)^2 & -4\Sigma_+-4-3(2-\gamma)\Sigma_+(\Sigma_++1) & \dots \\
	0 & {-}2\absval{N_+} & \dots \\
	(6(2-\gamma)\Sigma_++4)(\Sigma_++1)^2 & 3(2-\gamma)(\Sigma_++1)^2 & \dots \\
	{-}(3(2-\gamma)\Sigma_++2)\absval{N_+} & {-}\frac32(2-\gamma)\absval{N_+} & \dots
\end{matrix}\right.
\\
&\left.
\begin{matrix}
	\dots & 0 & \frac12(3\gamma-2)(\frac {\bparamk}3-1)\Sigma_++\frac23{\bparamk} & \frac13(3\gamma-2)\Sigma_+\absval{N_+}+\frac43\absval{N_+}\\
	\dots & 4\absval{N_+} & {-}(3\gamma-2)(\frac {\bparamk}3-1)\Sigma_+(\Sigma_++1)-4\Sigma_+ & {-}\frac23(3\gamma-2)\Sigma_+(\Sigma_++1)\absval{N_+} \\
	\dots & {-}2\Sigma_+-2 & {-}\frac23{\bparamk}\absval{N_+} & {-}\frac43N_+^2 \\
	\dots & 0 & (3\gamma-2)(\frac {\bparamk}3-1)(\Sigma_++1)^2 & \frac23(3\gamma-2)(\Sigma_++1)^2\absval{N_+} \\
	\dots & 6 & {-}\frac12(3\gamma-2)(\frac {\bparamk}3-1)\absval{N_+} & {-}\frac13(3\gamma-2)N_+^2
\end{matrix}
\right),
\end{align}
where $N_+^2=(1+\Sigma_+)({\bparamk}(1+\Sigma_+)-3\Sigma_+)$ as in \equ ~\eqref{eqn_defiplanewave}, by definition of the plane wave equilibrium points.
The eigenvalues to this vector field are
\begin{equation}
	0,\qquad {-}4(1+\Sigma_+),\qquad {-}4\Sigma_+-(3\gamma-2),\qquad {-}2(1+\Sigma_+)\pm2iN_+.
\end{equation}
We do not make use of the explicit form of corresponding eigenvectors in any of the proofs, and therefore only give those with a relatively short form.
\begin{itemize}
	\item The eigenspace to eigenvalue~$0$ is tangential to the plane wave equilibrium points, it is spanned by
	\begin{equation}
		({-}2\absval{N_+}\,,\,(4\Sigma_++2)\absval{N_+}\,,\,0\,,\,-4\absval{N_+}(\Sigma_++1)\,,\,\frac{2N_+^2}{\Sigma_++1}-3).
	\end{equation}
	\item The eigenspace to eigenvalue ${-}4(1+\Sigma_+)$ is spanned by
	\begin{align}
		&\left((1-\Sigma_+)N_+^2\,,\,(\Sigma_++1)(2\Sigma_+N_+^2-N_+^2-2{\bparamk}\Sigma_+-5\Sigma_++2{\bparamk}+1),\right.\\
		&\quad\left.(1-\Sigma_+)(1+\Sigma_+)\absval{N_+}\,,\,{-}(\Sigma_++1)^2(\Sigma_++2N_+^2+1)\,,\,\absval{N_+}(2\Sigma_++N_+^2-1)\right).
	\end{align}
\end{itemize}
The eigenspace to eigenvalue ${-}4\Sigma_+-(3\gamma-2)$ has an explicit expression which is lengthy but can be computed using mathematical software.
One can further check that this vector is \ogon\ to the gradient of equation~\eqref{eqn_constraintgeneralone}, as we claim in Section~\ref{section_asymptoticsplanewave}.

Direct computation shows that the eigenvector to eigenvalue ${-}4(1+\Sigma_+)$ is not \ogon\ to the gradient~\eqref{eqn_gradientconstraintequ} of the constraint \equ ~\eqref{eqn_constraintgeneralone}. Wherever the hypersurface defined by this constraint is non-singular, this means that this eigenvector is transverse to the constraint hypersurface, see Remark~\ref{rema_singularconstraintequ}.

\section{Results from dynamical systems theory}
\label{section_appendixdynamsystheo}

In this section, we state a theorem about the qualitative behaviour of solutions to a differential equation close to a sub\mf\ of equilibrium points. It is part of a more general discussion of singular perturbation theory given by~\cite{fenichel_geomsingperturbtheoryode}. We describe the setting and state the result in the form necessary for our work, then explain the relation to the original statement.

Let~$M$ be a~$C^{r+1}$ \mf, $1\le r<\infty$, and let~$Y$ be a~$C^r$ vector field on~$M$. The solution to
\begin{equation}
	\frac{d}{dt}\phi_t(m)=Y(\phi_t(m)),\qquad
	\phi_0(m)=m,
\end{equation}
with~$t$ in an interval chosen to be maximal, is the flow corresponding to the vector field. The family~$\{\phi_t\}$ defines a~$1$-parameter family of \diffeo s of~$M$. In the setting we are interested in, the interval
is~$\RR$, and we assume this from now on.
\begin{defi}
\label{defi_maxinvariantsets}
	For an open set $U\subset M$, the maximal positively invariant set~$A^+(U)$, the maximal negatively invariant set~$A^-(U)$, and the maximal invariant set~$I(U)$ are defined by
	\begin{align}
		m\in A^+(U) &\qquad \Leftrightarrow \qquad \overline{\{\phi_t(m):t\in[0,\infty\)\}}\subset U,\\
		m\in A^-(U) &\qquad \Leftrightarrow \qquad \overline{\{\phi_t(m):t\in\({-}\infty,0]\}}\subset U,\\
		m\in I(U) &\qquad \Leftrightarrow \qquad \overline{\{\phi_t(m):t\in\RR\}}\subset U,
	\end{align}
	where the bar denotes closure.
\end{defi}
Consider a point~$m\in M$ which is an equilibrium point of~$Y$, \ie a fixed point of the flow. The vector field~$Y$
induces a linear mapping
\begin{equation}
	TY(m):T_mM\rightarrow T_mM.
\end{equation}
Let~$\E$ be a~$C^r$ sub\mf\ consisting entirely of equilibrium points of~$Y$, and~$K\subset\E$ a compact subset such that the numbers of non-vanishing eigenvalues of~$TY(m)$
situated in the left half-plane, on the imaginary axis, and in the right half-plane, respectively, are constant for all points~$m\in K$. For every~$m\in K$, denote by~$E_m{}^s$, $E_m{}^c$ and~$E_m{}^u$ the invariant subspaces of~$T_mM$ associated with the eigenvalues of~$TY(m)$ in the left half-plane, on the imaginary axis, and in the right half-plane, respectively.
\begin{defi}
\label{defi_centremf}
	A $C^1$ \mf~$\C^s$ is called a centre-stable \mf\ for~$Y$ near~$K$ if
	\begin{itemize}
		\item $K\subset\C^s$;
		\item $\C^s$ is locally invariant under the flow of~$Y$, \ie there is an open neighborhood~$V$ of~$\C^s$ such that $m\in\C^s$, $t_0,t_1\ge 0$ and $\{\phi_t(m):t\in[{-}t_0,t_1]\}\subset V$ implies $\{\phi_t(m):t\in[{-}t_0,t_1]\}\subset\C^s$;
		\item for all $m\in K$, $\C^s$ is tangent to $E_m{}^s\oplus E_m{}^c$ at~$m$.
	\end{itemize}
	Centre-unstable and centre \mf s~$\C^u$ and~$\C$ are defined the same way, with $E_m{}^s\oplus E_m{}^c$ replaced by $E_m{}^c\oplus E_m{}^u$ and $E_m{}^c$, respectively.
\end{defi}
The following theorem states that these \mf s exist, at least locally, around the compact set~$K$, and that they contain the different maximal invariant sets.
\begin{theo}
\label{theo_centremftheory}
	Let~$M$ be a~$C^{r+1}$ \mf, $1\le r<\infty$, let~$Y$ be a~$C^r$ vector field on~$M$. Let~$\E$ be a~$C^r$ sub\mf\ consisting entirely of equilibrium points of~$Y$, and $K\subset\E$ a compact subset such that the number of non-vanishing eigenvalues of~$TY|_K$ situated in the left half-plane, on the imaginary axis, and in the right half-plane, respectively, is constant. Then there is a~$C^r$ centre-stable \mf~$\C^s$, a~$C^r$ centre-unstable \mf~$\C^u$, and a~$C^r$ centre \mf~$\C$ for~$Y$ near~$K$. Furthermore, there is a neighborhood~$U$ of~$K$ such that
	\begin{equation}
		A^+(U)\subset \C^s,\qquad A^-(U)\subset\C^u,\qquad I(U)\subset\C.
	\end{equation}
	In addition, the following uniqueness properties hold:
	\begin{itemize}
		\item Let~$\D^s$ be any~$C^r$ \mf\ which is locally invariant relative to~$U$ and tangent to~$\C^s$ at a point $m\in A^+(U)$. Then~$\C^s$ and~$\D^s$ have contact of order~$r$ at~$m$.
		\item Let~$\D^u$ be any~$C^r$ \mf\ which is locally invariant relative to~$U$ and tangent to~$\C^u$ at a point $m\in A^-(U)$. Then~$\C^u$ and~$\D^u$ have contact of order~$r$ at~$m$.
		\item Let~$\D$ be any~$C^r$ \mf\ which is locally invariant relative to~$U$ and tangent to~$\C$ at a point $m\in I(U)$. Then~$\C$ and~$\D$ have contact of order~$r$ at~$m$.
	\end{itemize}
\end{theo}
The statement appears as Thm~9.1~(i) and~(iv) in~\cite{fenichel_geomsingperturbtheoryode}, where one considers not only one vector field~$Y$, but a family $X^\epsilon$, $\epsilon\in({-}\epsilon_0,\epsilon_0)$. The theorem stated here is a direct consequence of restricting to a vector field which remains unchanged in~$\epsilon$, \ie $X^\epsilon=X^0=Y$, and ignoring the~$\epsilon$-direction in the definition of centre-stable, centre-unstable, and centre \mf s.

\bibliographystyle{amsalpha}
\bibliography{researchbib}
\vfill

\end{document}